\documentclass[11pt]{book}

\title{\large{\bf The Functional Analysis of Quantum
  Information Theory} \\
  \vspace{4mm}
\small{\it --- a collection of notes based on lectures by \\
Gilles Pisier, K.~R.~Parthasarathy, Vern Paulsen and Andreas Winter}}

\author{Ved Prakash Gupta, Prabha Mandayam and V.~S.~Sunder}
\date{}

\usepackage{makeidx}
\makeindex

\usepackage{amsmath}
\usepackage{amsfonts}
\usepackage{latexsym}
\usepackage{amssymb}
\usepackage{mathrsfs}
\usepackage{amscd}
\usepackage{hyperref}
\usepackage{fullpage}
\usepackage[mathscr]{eucal}

\newtheorem{thm}{Theorem}[section]
\newtheorem{lem}[thm]{Lemma}
\newtheorem{cor}[thm]{Corollary}

\newtheorem{prop}[thm]{Proposition}
\newtheorem{defn}[thm]{Definition}
\newtheorem{rem}[thm]{Remark}
\newtheorem{exer}[thm]{Exercise}
\newtheorem{example}[thm]{Example}
\newtheorem{exam}[thm]{Example}
\newtheorem{obs}[thm]{Observation}
\newtheorem{ques}[thm]{Question}

 \newenvironment{pf}{\noindent{\em Proof:}}{\hfill
   $\Box$\vspace*{1mm}}
 \newenvironment{proof}{\noindent{\em Proof:}}{\hfill
   $\Box$\vspace*{1mm}}
\newcommand{\comments}[1]{}
\newcommand{\ra}{\rightarrow}

\newcommand{\mbb}{\mathbb}
\newcommand{\mcal}{\mathcal}

\newcommand{\N}{\mathbb N}

\newcommand{\C}{\mathbb{C}}

\newcommand{\F}{\mathbb{F}}
\newcommand{\Id}{\mathbb{I}}
\newcommand{\cA}{\mathcal{A}}
\newcommand{\cB}{\mathcal{B}}
\newcommand{\cC}{\mathcal{C}}
\newcommand{\cD}{\mathcal{D}}
\newcommand{\cF}{\mathcal{F}}
\newcommand{\cH}{\mathcal{H}}
\newcommand{\cI}{\mathcal{I}}
\newcommand{\cJ}{\mathcal{J}}
\newcommand{\cK}{\mathcal{K}}
\newcommand{\cL}{\mathcal{L}}
\newcommand{\cM}{\mathcal{M}}
\newcommand{\cN}{\mathcal{N}}
\newcommand{\cO}{\mathcal{O}}
\newcommand{\cP}{\mathcal{P}}
\newcommand{\cS}{\mathcal{S}}
\newcommand{\cT}{\mathcal{T}}
\newcommand{\cU}{\mathcal{U}}
\newcommand{\tr}{\text{Tr}}

\newcommand{\ot}{\otimes}
\newcommand{\ptimes}[1]{\hat{\otimes}{#1}}
\renewcommand{\tilde}[1]{\widetilde{#1}}

\renewenvironment{frontmatter}{\pagenumbering{roman}}{\clearpage\pagenumbering{arabic}}

\begin{document}

\maketitle

\begin{frontmatter}

\chapter*{Preface}

\bigskip \noindent These notes grew out of a two-week workshop on {\em The Functional Analysis of
Quantum Information Theory} that was held at the Institute of
Mathematical Sciences, Chennai during 26/12/2011-06/01/2012. This was
initially the brain-child of Ed Effros, who was to have been one of
the four principal speakers; but as things unfolded, he had to pull
out at the eleventh hour due to unforeseen and unavoidable
circumstances. But he mentored us through the teething stages
with suggestions on possible and desirable substitutes. After a few
hiccups, we arrived at a perfectly respectable cast of characters,
largely owing to Prof. K.R. Parthasarathy agreeing to fill the breach to the extent
that his somewhat frail health would permit. While everybody else had
been asked to give five $90$ minute lectures each, he agreed gamely
to give three $60$ minute lectures, which were each as delightful as
one has expected his lectures to be.

\bigskip \noindent The three other speakers were Gilles Pisier, Vern Paulsen and Andreas
Winter. Given the impeccable clarity in their lectures, it was not
surprising that the workshop had a substantial audience (a mixed bag of
mathematicians, physicists and theoretical computer scientists) for the entire
two weeks, and several people wanted to know if the proceedings of the
workshop could be published. Given the wide scope of the problems
discussed here, we hope these notes will prove useful to students and research scholars across all three
disciplines.

\bigskip \noindent The quite non-trivial and ambitious task of trying to put
together a readable account of the lectures was taken upon by the
three of us, with periodic assistance from Madhushree Basu and Issan
Patri. When we finally finished a first draft some 28 months later and
solicited permission from the speakers for us
to publicise this account, they responded quite promptly and
positively, with Vern even offering to slightly edit and expand `his'
part and Gilles uncomplainingly agreeing to edit the last few pages of
`his' part.

\bigskip \noindent As even a casual reading of these notes will show, different parts
have been written by different people, and some material has been
repeated (the Kraus decomposition for instance). Further the schism in
the notational conventions of physicists and mathematicians (eg., in which of its two arguments the inner product is (conjugate) linear), reflects which convention which speaker followed! This schism was partly because we did not want to introduce new errors by trying to adopt a convention other than in the notes taken from the lectures.
Thus, this must be viewed not as a textbook in a polished and complete form, but rather as a transcription of notes derived from lectures, reflecting the spontaneity of live classroom interactions and conveying the enthusiasm and intuition for the subject. It goes without saying that any shortcomings in these notes are to be attributed to the scribes, while anything positive should be credited to the speakers.




\newpage
\begin{center}
{\bf Acknowledgements}
\end{center}

\bigskip \noindent Firstly, we should  express our gratitude to so many people in the institute (IMSc) who helped in so many ways to make the workshop as successful as it was,
not least our genial Director Balu for generously funding the whole enterprise. We should also thank our systems people who cheerfully videographed all the lectures, these videos being made available as a You-tube playlist at http://www.imsc.res.in/~sunder/faqit.html; and the reader is strongly urged to view these videos in parallel with reading these notes, to really get the full benefit of them.

\bigskip \noindent
PM was a post-doctoral fellow with the quantum information group at IMSc during the writing of this manuscript. She would like to thank IMSc and the quantum information group in particular for the opportunity to participate in the workshop and be a part of this effort.

\bigskip \noindent
VSS also wishes to explicitly thank Dr. S. Ramasami, the former Secretary of the Department of Science and Technology, for having gladly, promptly and unconditionally permitted the use of the generous grant (the J.C. Bose Fellowship) he receives from them in a wonderful example of the `reasonable accommodation' that the UN Convention on the Rights of Persons with Disability speaks about, which has permitted him to travel when needed and be a functional mathematician in spite of health constraints.

\tableofcontents

\end{frontmatter}


\chapter{Operator Spaces}

\bigskip \noindent
{\bf Chapter Abstract:}

 \bigskip \noindent
{\it Starting from the definitions of operator spaces and operator systems (the natural ambience to discuss the notions of complete boundedness and complete positivity, respectively) this chapter, which is based on lectures by Gilles Pisier, quickly gets to the heart of the matter by discussing several closely related theorems -  Stinespring's theorem on dilation of completely positive maps, Arveson's extension theorems (for completely bounded as well as completely positive maps), related results by Haagerup and Paulsen, Nachbin's Hahn-Banach theorem  - as more or less different perspectives of the same result; we also see the power/use of the clever matrix trick relating positivity of  self-adjoint $2 \times 2$ block matrices and operator norm relations between its entries. After discussing the classical row (R) and column (C) Hilbert spaces, Schur multipliers etc., we proceed with Ruan's axioms for an `abstract' operator space, applications of the formalism to highlight analogy with classical Banach space theory (e.g.: $R^* \cong C$), {\rm min } and {\rm max } operator space structures, tensor products of $C^*$-algebras, nuclear pairs (whose tensor products have a unique `$C^*$-cross-norm'), and the chapter concludes with Kirchberg's theorem on ($C^*(\F_n), B(H)$) being a nuclear pair.}

\section{Operator spaces}
\begin{defn}
A closed subspace $E \subset B(H)$ for some Hilbert space $H$ is
called an operator space\index{operator space}.
\end{defn}
The requirement of `closed'-ness is imposed because we want to think
of operator spaces as `quantised (or non-commutative) Banach
spaces'. This assumption ensures that operator spaces are indeed
Banach spaces with respect to the CB-norm (see Definition \ref{cbdef}
and the subsequent remarks).  Conversely every Banach space can be
seen to admit such an embedding.

\noindent ({\em Reason:} If $E$ is a Banach space, then the unit ball
of $E^*$, equipped with the weak$^*$-topology is a compact Hausdorff
space (by Alaoglu's theorem), call it $X$; and the Hahn-Banach theorem
shows that $E$ embeds isometrically into $C(X)$. Finally, we may use
an isometric representation of the $C^*$-algebra $C(X)$ on some
Hilbert space $H$; and we would finally have our isometric embedding
of $E$ into $B(H)$.)

Note that, for each Hilbert space $H$, there is a natural
identification between $M_n(B(H))$ and $B(H^{(n)})$, where $H^{(n)}:=
\underbrace{H \oplus \cdots \oplus H}_{n-\text{copies}}$. For each
operator space $E \subset B(H)$, and $n \geq 1$, let $$M_n(E)=\{ [
  a_{ij} ] : a_{ij} \in E, 1\leq i, j \leq n \} \subset M_n(B(H)) .$$
Then, each $M_n(E)$ inherits a norm, given by
\[
||a||_n := ||a||_{B(H^{(n)})} = \sup \left\{ \Big(\sum_{i=1}^n
\Big\| \sum_{j=1}^n a_{ij}(h_j)\Big\|^2\Big)^{1/2} : h_j \in H,
\sum_{j=1}^n||h_j||^2 \leq 1 \right\}
\]
for all $a=[a_{ij}] \in M_n(E)$. In particular, each operator space $E
\subset B(H)$ comes equipped with a sequence of norms $(||\cdot||_n ~(\mbox{on } M_n(E)),
n\geq 1)$.

\subsection{Completely bounded and completely positive maps}
For a linear map $u: E \rightarrow F$ between two operator spaces, for
each $n \geq 1$, let $u_n : M_n(E) \rightarrow M_n(F)$ be defined by
$u_n([a_{i,j}])=[u(a_{i,j})]$ for all $[a_{ij}] \in M_n(E)$.  .

\begin{defn}\label{cbdef}
A linear map $u: E \rightarrow F$ is called {\bf completely bounded (CB)} \index{completely! bounded (CB)}
if $||u||_{cb} := \sup_{n \ge 1} ||u_n|| < \infty$. Write $CB(E,F) =
\{u:E \xrightarrow{CB} F\}$ and equip it with the CB norm $\| \cdot
\|_{cb}$.
\end{defn}

Note that $||u|| = ||u||_1 \le ||u||_{cb}$; so, any Cauchy sequence in
$CB(E,F)$ is also a Cauchy sequence in $B(E,F)$ and hence one deduces
that $CB(E,F)$ is a Banach space (as $B(E,F)$ is a Banach space).
Complete boundedness has many properties similar to boundedness, for
example: $||vu||_{cb} \le ||v||_{cb}||u||_{cb}$. We shall see latter
that $CB(E,F)$ also inherits an appropriate operator space structure.

\subsection{Operator systems}

\begin{defn}
A (closed) subspace $S \subset B(H)$ for some Hilbert space $H$ is
called an {\bf operator system} \index{operator system}, if
\begin{itemize}
 \item $1 \in S$
 \item $x \in S \Rightarrow x^* \in S$.
\end{itemize}
\end{defn}

Note that if $S \subset B(H)$ is an operator system, then $S_+ := S
\cap B(H)_+$ linearly spans $S$ since if $a \in S_h = \{x \in S:
x=x^*\}$, then $a = ||a||1 - (||a||1 - a)$. Also, if $S$ is an
operator system, so also is $M_n(S) ~(\subset B(H^{(n)}))$.

\begin{defn}\label{def:k-positive}
If $S$ is an operator system, a linear mapping $u:S \rightarrow B(K)$
is said to be:
\begin{itemize}
 \item {\bf  positive} if $x \in S_+  \Rightarrow u(x) \in B(K)_+$.
 \item {\bf  $n$-positive} if $u_n:M_n(S) \rightarrow M_n(B(K))$ is positive.
 \item {\bf completely positive} (CP)  \index{completely! positive (CP)} if it is $n$-positive for all $n
   \in \mathbb{N}$.
\item {\bf completely contractive} \index{completely! contractive} if $\|u_n\| \le 1 ~\forall n\in \N$
\end{itemize}
\end{defn}

\subsection{Fundamantal Factorisation of CB maps}

Our goal, in this section, is to prove the following fundamental
factorisation theorem:

\begin{thm}[Fundamental Factorisation]\label{ff} \index{factorisation! fundamental} \index{Theorem! Stinespring's} Given an operator space $E$
 contained in a unital $C^*$-algebra $A$, the following conditions on
 a linear map $u:E \rightarrow B(K)$ are equivalent:
 \begin{enumerate}
 \item $u \in CB(E,B(K))$ and $||u||_{cb} \le 1$;
 \item there is a Hilbert space $\hat{H}$, a unital
   $\ast$-representation $\pi:A \rightarrow B(\hat{H})$ and linear
   maps $V, W :K \rightarrow \hat{H}$ with $||V||, ||W|| \le 1$ such
   that $ u(a)=V^*\pi(a)W$ for all $a \in E$.
\end{enumerate}
Moreover if $E$ is an operator system then $u$ is CP if and only if it
admits such a factorisation with $V=W$.
\end{thm}

There are many names and a long history associated with this
fact. Initially, Stinespring established his version of this in 1955
as a factorisation of CP maps defined on $C^*$-algebras, making it
appear as an operator analogue of the classic GNS construction
associated to a state on a $C^*$-algebra. Then Arveson proved his
extension theorem  in 1969 - to the effect that a CP map on an operator
system can be extended to a CP map on a $C^*$-algebra containing it
(thereby generalizing Stinespring's theorem so as to be valid for CP
maps on operator sytems). The final full generalization to CB maps on
operator spaces, which may be attributed to several authors -
Wittstock, Haagerup, Paulsen - came around $1981$. We shall give a
proof of the theorem which exhibits it as an extension theorem. But
first, some examples:

\begin{example}
  The `transpose' mapping $M_n \xrightarrow{T} M_n$, given by
  $T(a)=a^t$, is CB with $||T||=1$ but $||T||_{cb}=n$.
\end{example}

\begin{proof}
 Let $e_{i,j}$ be the matrix units of $M_n$. Consider the matrix $a
 \in M_n(M_n)$ with $(i,j)$-th entry given by $e_{j,i}$. It is not
 hard to see that $a$ is a permutation matrix, hence unitary and
 $||a||=1$.

Also, $\frac{T_n(a)}{n}$ is seen to be the projection in $M_n(M_n)$
onto the one-dimensional subspace spanned by the vector $v=\ e_1 +
e_{n+2} + e_{2n+3} + \cdots + e_{n^2}$ in $\ell_2^{n^2}$, where
$\{e_i\}$ is the standard orthonormal basis of $\ell_2^{n^2}$. Hence
$||T_n(a)||=n$. In particular, $$||T||_{cb} \ge ||T_n|| \ge
\frac{||T_n(a)||}{||a||} = n.$$

Conversely, consider the subalgebra $\Delta_n \subset M_n$ of $n
\times n$ diagonal matrices. Let $E:M_n \rightarrow \Delta_n$ be the
       {\em conditional expectation} \index{conditional expectation} defined by
       $E(a)=E([a_{i,j}])=\mathrm{diag} (a_{1,1}, \ldots, a_{n,n})$ -
       which is easily verified to satisfy $E(dad') = d E(a) d'$ for
       all $d,d' \in \Delta_n$ and $a \in M_n$. Consider the
       permutation matrix
\[ u = \begin{pmatrix}
0 & 0 & \cdots  & 0 & 1\\
1 & 0 & \cdots & 0 & 0 \\
0 & 1 & \cdots & 0 & 0 \\
\vdots & \vdots & \ddots & \vdots & \vdots \\
0 & 0 & \cdots & 1 & 0
\end{pmatrix}~.\]
A little calculation shows that $E(au^k) =\mathrm{diag} (a_{1,k+1},
a_{2,k+2},\ldots, a_{n,n})$, and hence $E(au^k)u^{-k}$ has $a_{i,k+i}$
at $(i,k+i)^{th}$ position for each $1 \le k \le n, 1 \le i \le n$,
(indices are modulo $n$) as its only non-zero entries. Thus
$a=\sum_{k=1}^{n}E(au^k)u^{-k}$ and $a^t=\sum_{k=1}^nu^kE(au^k)$.

Define $u_l,u_r:M_n \rightarrow M_n$ by $u_l(a) = ua$ and $u_r(a) =
au$, respectively. Then $u_l^k$ and $u_r^k$ are isometries for all $k$.
As $(u_l)_n$ (resp., $(u_r)_n$) is seen to be the map of $M_n(M_n)$
given by left- (resp., right-) multiplication by the block-diagonal
unitary matrix with all diagonal blocks being given by $u$, we
conclude that $u_l,u_r$ are CB maps. (They are in fact {\em complete
  isometries}.) \index{complete! isometry}

Moreover, a conditional expectation is a CB map (since $E:M_n
\rightarrow \Delta_n$ being a conditional expectation implies that,
for all $k \in \mathbb{N}$, $E_k:M_k(M_n) \rightarrow M_k(\Delta_n)$
is also a conditional expectation whence a projection and $ ||E_k||
\le 1$ for all $k$).

Hence the equation $T = \sum_{k=1}^nu_r^k \circ E \circ u_l^k$
expresses $T$ as the sum of $n$ {\em complete contractions} thereby showing that also $\|T\|_{cb} \leq n$, whereby $\|T\|_{cb} = n$.
\end{proof}

\begin{example}
 It follows from the previous example that the transpose mapping, when
 thought of as a self-map $T:K(\ell^2) \rightarrow K(\ell^2)$ of the
 space of compact operators on $\ell^2$, is isometric but not CB
 (since all finite rank matrices are in $K(\ell_2)$).
\end{example}

\begin{thm}[Stinespring] \label{st} \index{Theorem! Stinespring's}
 Given a unital C*-algebra $A$, a Hilbert space $K$, and a unital CP
 map $u:A \rightarrow B(K)$, there exist a Hilbert space $H$, a
 $\ast$-representation $\pi:A \rightarrow B(H)$ and an isometric embedding
 $V:K \hookrightarrow H$ such that $u(a)=V^*\pi(a)V$ for all $a$ in
 $A$. Equivalently, if we identify $K$ via $V$ as a subspace of $H$,
 then $u(a)=P_K\pi(a)|_K$.
\end{thm}

\begin{proof}
 Let $\alpha =\sum_{i \in I} a_i \otimes k_i, \beta= \sum_{j \in J}
 b_j \otimes k_j \in A \otimes K$. The fact that $u$ is CP implies
 that the equation $\langle \alpha, \beta \rangle = \sum_{i \in I,j
   \in J} \langle u(b_i^{\ast}a_j)k_j,k_i \rangle_K$ defines a
 positive semi-definite sesquilinear form on $A \otimes K$.( {\em
   Reason:} $u_n$ is positive and if $I=\{1,\cdots,n\}$, then
\begin{align*}
\begin{pmatrix}
a_1^* & 0 & \cdots & 0 \\
a_2^* & 0 & \cdots & 0 \\
\vdots & \vdots & \ddots & \vdots \\
a_n^* & 0 & \cdots & 0
\end{pmatrix}
\begin{pmatrix}
a_1 & a_2 & \cdots & a_n \\
0 & 0 & \cdots & 0 \\
\vdots & \vdots & \ddots & \vdots \\
0 & 0 & \cdots & 0
\end{pmatrix} \ge 0 \Rightarrow  \langle \alpha, \alpha \rangle \ge 0.)
\end{align*}
 Define $H$ to be the result of `separation and completion of this
 semi-inner-product space'; i.e., $H$ is the completion of the
 quotient of $A \otimes K$ by the radical $N$ of the form $\langle
 \cdot, \cdot \rangle$. And it is fairly painless to verify that the
 equations $$\pi(a)(b \otimes k + N)=ab \otimes k + N \ \text{ and
 }\ V(k) = (1 \otimes k) + N$$ define a $\ast$-representation $\pi : A
 \rightarrow B(H)$ and an isometry $V: K \rightarrow H$ which achieve
 the desired results.
\end{proof}

Before we get to a CB version of Arveson's extension theorem, which
may be viewed as a non-commutative Hahn-Banach theorem, we shall pause
to discuss a parallel precursor , i.e., a (commutative) Banach space
version due to Nachbin, namely:

\begin{prop}[Nachbin's Hahn-Banach Theorem]\label{nachbext} \index{Theorem! Nachbin's Hahn-Banach}
 Given an inclusion $E \subset B$ of Banach spaces, any $u \in
 B(E,L^\infty(\mu)$)\footnote{When we write symbols such as $L^1(\mu)$
   or $L^\infty(\mu)$, it will be tacitly assumed that we are talking
   about (equivalence classes of almost everywhere agreeing)
   complex-valued functions on some underlying measure space
   $(\Omega,\mu)$; when we wish to allow vector-valued functions, we
   will write $L^1(\mu;X)$, etc.} admits an extension $\tilde{u} \in
 B(B,L^\infty(\mu))$ such that $||\tilde{u}|| = ||u||$.
\end{prop}
The proof of Nachbin's theorem relies on the following observation.

\begin{obs}\label{HBN}
{\rm For Banach spaces $E$ and $F$, we have isometric identifications:
\begin{eqnarray*}
  B(E,F^*) & := & \{\text{all bounded linear maps from }E \text{ to }F^*\}
  \\ & \simeq & Bil(E \times F) := \{\text{all bounded
    bilinear forms on }E \times F \}\\ & \simeq & (E \ptimes F)^*,
\end{eqnarray*}
where $\hat{\otimes}$ is the projective tensor product, \index{tensor product! projective}
i.e., the completion of the linear span of the elementary tensor
products with respect to the norm $||t||_{\wedge} := \inf \{\sum_i ||a_i|| \cdot
||b_i||: t=\sum_i a_i \otimes b_i\}$.

The first isomorphism follows from the identification
$B(E,F^*) \ni T \mapsto \psi_T \in Bil(E \times F)$ given by
$\psi_T(a,b)=T(a)(b)$ for $a \in E, b \in F$, since
\begin{align*}
 ||T||&= \sup \{||T(a)||:||a|| \le 1 \} \\
&= \sup\{|T(a)(b)|:||a||, ||b|| \le 1 \} \\
&=||\psi_T||,
\end{align*}
while the second isomorphism follows from the identification $Bil(E
\times F) \ni S \mapsto \phi_S \in (E \ptimes F)^*$ given by $\phi_S(a
\otimes b)= S(a,b)$. On one hand, $||\phi_S||= \sup \{|\phi_S(t)|:
||t||_{E \ptimes F} \le 1\}$. On the other hand, for each
representation $t=\sum_i a_i \otimes b_i \in E\ot F$,
\begin{eqnarray*}
 |\phi_S(t)|&=&|\sum_iS(a_i,b_i)| \\ &\leq & \sum_i |S(a_i, b_i)| \\ &
 \leq & ||S|| \sum_i||a_i||\cdot||b_i||.
\end{eqnarray*}
Taking infimum over all such representations of $t$, we obtain
$|\phi_S(t)| \le ||S||\, ||t||_{\wedge}$ implying $ ||\phi_S||
\le ||S||$. Conversely, since $\|a \otimes b\|_{E \hat{\otimes} F} = \|a\| \cdot \|b\|$ (as is easily verified),
\begin{align*}
 ||\phi_S|| & \ge \sup \{ |\phi_S(a \otimes b)| :
 ||a||\cdot||b||\le 1\} \\ &= \sup \{|S(a, b)| : ||a||\cdot||b||\le
 1\}\\ &= ||S||.
\end{align*}}
\end{obs}

\noindent{\em Proof of Theorem \ref{nachbext}:} Deduce from the above
observation that
\begin{align*}
B(E,L^\infty(\mu)) & \cong (E \ptimes L^1(\mu))^*\\
& \cong L^1(\mu, E)^*\\
& \subseteq L^1(\mu, B)^*\\
& \cong B(B, L^\infty(\mu))~,
\end{align*}
where the first and last isomorphisms follow from the Observation
\ref{HBN}, the inclusion is an isometric imbedding as a consequence of
the classical Hahn-Banach theorem, and the second isomorphism is
justified in two steps as follows:
\vspace*{2mm}

 \noindent {\em Step I}: If $X \ptimes Y$ denotes the
projective tensor product of Banach spaces $X$ and $Y$, if $D_X$
(resp. $D_Y$) is a dense subspace of $X$ (resp., $Y$), and if $\phi$
is any map from the algebraic tensor product $D_X \otimes D_Y$ to any
Banach space $Z$ such that $\|\phi(x \otimes y)\|_Z = \|x\| \|y\|~\forall x \in D_X, y \in D_Y$, then
$\phi$ extends uniquely to an element $\tilde{\phi} \in B((X \ptimes
Y),Z)$.

\noindent {\em Reason:} If $\sum_i x_i \otimes y_i = \sum_j u_j \otimes v_j$,
then, observe that
\begin{align*}
\| \phi(\sum_i x_i \otimes y_i)\| & = \|\sum_j \phi(u_j \otimes v_j)\|\\
& \le \sum_j \|\phi(u_j \otimes v_j)\|\\
& = \sum_j \|u_j\| \|v_j\|~.
\end{align*}
Taking the infimum over all such choices of $u_j, v_j$, we find that
\[
\| \phi(\sum_i x_i \otimes y_i)\|_Z \le \|\sum_i x_i \otimes
y_i\|_{\wedge},
\] and the desired conclusion is seen to easily follow..

\noindent {\em Step II}: The assignment
\[
L^1(\mu) \otimes Y \ni f \otimes y \stackrel{\phi}{\mapsto} f(.)y \in
L^1(\mu;Y)
\]
extends to an isometric isomorphism $\tilde{\phi}$ of $L^1(\mu)
\ptimes Y$ onto $L^1(\mu;Y)$ for any Banach space $Y$.

\noindent {\em Reason:} Let $X = L^1(\Omega,\mu)$, and let $D_X$
denote the subspace of measurable functions with finite range, and let
$D_Y =Y$. Note that any element of $D_X$ admits an essentially unique
expression of the form $\sum_i c_i 1_{E_i}$, for some Borel partition
$\Omega = \coprod_i E_i$. It follows that any element of $D_X \otimes
Y$ has an essentially unique expression of the form $\sum_i 1_{E_i}
\otimes y_i$, for some Borel partition $\Omega = \coprod_i E_i$ and
some $y_i \in Y$. It follows from Step I that there exists a
contractive operator $\tilde{\phi} \in B(L^1(\mu) \ptimes Y,
L^1(\mu;Y)$ such that $\tilde{\phi}(1_\Delta \otimes y) =
1_\Delta(.)y$. Observe that if $\sum_i 1_{E_i} \otimes y_i \in D_X
\otimes Y$, then
\begin{align*}
\|\tilde{\phi}(\sum_i 1_{E_i} \otimes y_i)\| & = \int \|\sum_i 1_{E_i}
(.) y_i\| \\ & = \sum_i \mu(E_i) \|y_i\| \\ & = \sum_i \|1_{E_i}\|
\|y_i\| \\ & \ge \|\sum_i 1_{E_i} \otimes y_i)\|_{L^1(\mu) \ptimes Y}~,
\end{align*}
thereby showing that $\tilde{\phi}$ is isometric on the dense subspace
$D_X \otimes Y$ and hence completing the proof.

\begin{thm}[Arveson's Extension Theorem]\label{arv-ext} \index{Theorem! Arveson's Extension}
Given an inclusion $E \subset A$ of an operator space into a unital
C*-algebra, and any CB map $u:E \rightarrow B(K)$ for some Hilbert
space $K$, there exists $\tilde{u}:A \xrightarrow{\text{CB}} B(K)$
extending $u$ such that $||\tilde{u}||_{cb}=||u||_{cb}$.
\end{thm}

\begin{proof} Our strategy will involve proving that:
\begin{enumerate}
  \item If, for $t \in \overline{K} \otimes E \otimes H$, we define
    $$\gamma(t):= \inf \Big\{\Big(\sum||k_i||^2\Big)^{1/2}
    ||[a_{i,j}]||_{M_n(E)}\Big(\sum||h_j||^2\Big)^{1/2}: t = \sum_{i,j=1}^n
    \overline{k_i} \otimes a_{i,j} \otimes h_j\Big\},$$ then $\gamma$ is a
    norm;
 \item there is a natural isometric identification $CB(E,B(H,K)) \cong
   (\overline{K} \otimes E \otimes H)^*$, with $(\overline{K} \otimes E \otimes
   H)^*$ being given the norm dual to $\gamma$; and
   \item if $B$ is an `intermediate' operator space, meaning $E
     \subset B \subset A$, then $(\overline{K} \otimes E \otimes H,
     \gamma) \subset (\overline{K} \otimes B \otimes H, \gamma)$ is an
     isometric embedding; i.e.,
\[
\| \cdot \|_{\overline{K} \otimes E \otimes H} = (\| \cdot \|_{\overline{K}
  \otimes B \otimes H})_{|_{{\overline{K} \otimes E \otimes H}}}.
\]
\end{enumerate}
It is obvious that saying $\gamma(t) \ge 0$ amounts to showing that
$t$ admits an expression of the form $t = \sum_{i,j=1}^n \overline{k_i}
\otimes a_{i,j} \otimes h_j$. In fact, there is such an expression
with the matrix $[a_{i,j}]$ being diagonal: indeed, if $t = \sum_{j=1}^n
\overline{k_j} \otimes e_j \otimes h_j$, we can set $a_{ij} = \delta_{ij}
e_i$. For $t = \sum_{i,j=1}^n \overline{k_i} \otimes a_{i,j} \otimes h_j$, let us  use
the suggestive notation $t = \begin{pmatrix}\overline{k} \end{pmatrix}
\otimes a \otimes \begin{pmatrix}h \end{pmatrix}$.

We may assume that the infimum defining $\gamma$ may be taken only over
collections $(\bar{k}) \otimes(a) \otimes (h)$ which satisfy
\[
\|(a)\|_{M_n(E)} = 1 \mbox{ and } \|(\bar{k})\| = \sqrt{\sum_i
\|\bar{k}_i \|^2} = \|(h)\|.
\]

({\em Reason:} This may be achieved by `spreading scalars over the tensor factors'. More precisely, if $C^2 = \|(\bar{k}\| \cdot \|a\| \cdot \|(h)\|$ and if we define

\[ (\bar{k'}) = \frac{C}{\|(\bar{k})\|} (\bar{k}), (a') = \frac{1}{\|(a)\|} (a), (h') = \frac{C}{\|(h)\|} (h),\]

then $(\bar{k}') \otimes (a') \otimes (h')  = (\bar{k}) \otimes(a) \otimes (h)$ and
$\|(a')\|_{M_n(E)} = 1 \mbox{ and } \|(\bar{k}')\| = \|(h')\|$.)

For any other $t' = \begin{pmatrix}\overline{k}' \end{pmatrix}
\otimes a' \otimes \begin{pmatrix}h' \end{pmatrix}$, we find that
$t+t'=\begin{pmatrix} \overline{k} & \overline{k'} \end{pmatrix} \otimes \begin{pmatrix}
a & 0 \\
0 & a'
\end{pmatrix} \otimes
\begin{pmatrix}
h \\
h'
\end{pmatrix}$. Thus, $$ \gamma(t+t') \le
\Big(\sum||k_i||^2 +
\sum||k'_p||^2\Big)^{1/2}\Big(\sum||h_j||^2+\sum||h'_q||^2\Big)^{1/2}.$$

(Note the crucial use of the fact that the operator norm of a direct
sum of operators is the maximum of the norms of the summands - this is
one of the identifying features of these matrix norms, as will be
shown when we get to {\em Ruan's theorem}.)

For now, let us make the useful observation that:
\vspace*{1mm}

\noindent $\gamma(t) < 1$ iff there exists a decomposition $t =
(\overline{k}) \otimes (a) \otimes (h)$ with $|| a||_{M_n(E)} =
1$ and $\|(\overline{k})\| = \|(h)\|^2 < 1$. \hfill{($\dagger$)}
\vspace*{1mm}

 Now conclude from remark ($\dagger$) that the above decompositions
 may be chosen to satisfy $\|a\|_{M_n(E)} = 1 $ and
 $\|(\bar{k})\|= \|(h)\| \sim \sqrt{\gamma(t)}$
 and $\|(a')\|_{M_n(E)}=1$ and $\|(\bar{k'})\|= \|(h')\| \sim \sqrt{\gamma(t')}$; here $\sim$ denotes
 approximate equality. (We are being slightly sloppy here in the interest of sparing ourselves the agony of performing calisthenics with epsilons.)

Thus, $\gamma(t+t')\le(\gamma(t)+\gamma(t'))^{1/2}(\gamma(t) +
\gamma(t'))^{1/2} = \gamma(t)+\gamma(t')$. Hence the triangle
inequality holds for $\gamma$. It remains only to prove that
$\gamma(t) = 0 \Rightarrow t = 0$.

So, suppose $\gamma(t)=0$. Fix an arbitrary $(\phi,\chi,\psi) \in
(\overline{K})^* \times E^* \times H^*$ and $\epsilon > 0$. By
assumption, there exists a decomposition $t = \sum_{i,j=1}^n
\overline{k_i} \otimes a_{i,j} \otimes h_j$, with
$\|[a_{ij}]\|_{B(H^{(n)}} = 1$ and
$\|(\overline{k})\|_{\overline{K}^{(n)}} = \|(h)\|_{H^{(n)}} <
\epsilon$. Hence,
\begin{align*}
|(\phi \otimes \chi \otimes \psi)(t)| &= |\sum
\phi(\overline{k}_j)\chi(a_{ij}) \psi(h_j)|\\ & = \langle \chi_n(a) \begin{pmatrix} \psi(h_1) \\ \vdots \\ \psi(h_n) \end{pmatrix}  \begin{pmatrix} \phi(\bar{k}_1) \\ \vdots \\ \phi(\bar{k}_n) \end{pmatrix} \\ & \le \|\chi_n(a)\| (\sum_{j=1}^n |\psi(h_j)|^2)^\frac{1}{2} (\sum_{i=1}^n |\phi(\bar{k}_i)|^2)^\frac{1}{2} \\ & \le \|\chi_n)\| \|a\| \|\phi\| \|(\bar{k})\| \|\psi\| \|(h)\|\\ & \le \|\chi\|_n \|\phi\| \|\psi\| \epsilon^2 .
\end{align*}

Hence,
\[
(\phi \otimes \chi \otimes \psi)(t) = 0\ \forall\ \phi \in
(\overline{K})^*, \chi \in E^*, \psi \in H^*.
\]
But this is seen to imply that
$t=0$, thereby proving (1).
\vspace*{2 mm}

As for (2), first note that $B(H,K)$ has a natural operator space structure, by viewing it as the (2,1)-corner of $B(H \oplus K)$ or equivalently by identifying $M_n(B(H,K))$ with $B(H^{(n)},K^{(n)})$. Consider the mapping
\[
CB(E,B(H,K)) \ni u \mapsto \Phi_u \in (\overline{K} \otimes E \otimes H)^*
\]
defined by $\Phi_u(t)= \sum_{i,j} \langle u(a_{i,j})h_j,k_i \rangle$,
if $t=\sum_{i,j} \overline{k_i} \otimes a_{i,j} \otimes h_j$. Note
that
\begin{align*}
 ||\Phi_u||_{\gamma^*} \leq 1 &\Leftrightarrow \sup_{\gamma(t) < 1}
 |\Phi_u(t)| \le 1\\ & \Leftrightarrow \sup \left\{ |\sum_{i,j} \langle
 u(a_{i,j}) h_j, k_i \rangle| : \sum_j ||h_j||^2 = \sum_i ||k_i||^2 < 1,
 ||[a_{i,j}]||_{M_n(E)} = 1 \right\}\le 1 \\
& \Leftrightarrow    \sup \left\{ ||[u(a_{i,j})]||_{M_n(B(H,K))}  : [a_{i,j}] \in M_n(E), ||[a_{i,j}]||_{M_n(E)} = 1, n \ge 1 \right\} \le 1
   \\ & \Leftrightarrow ||u||_{cb} \le 1
\end{align*}
Hence $||\Phi_u||_{\gamma}=||u||_{cb}$ and the map $u \stackrel{\Phi}{\mapsto} \Phi_u$ is isometric.

We only need to verify now that $\Phi$ is surjective. Suppose $\Psi \in (\bar{K} \otimes E \otimes H)^\ast$. Fix $a \in E$. If $h \in H, k \in K$ and $k \mapsto \bar{k}$ is an antiunitary map, then consider the clearly sesquilinear form defined on $H \times K$ by $[h,k] = \Psi(\bar{k} \otimes a \otimes h)$. By definition of $\gamma$, we have $|[h,k]| \le \|\Psi\| \cdot \gamma(\bar{k} \otimes a \otimes h) \le \|\Psi\| \cdot \|a\|_E \|h\| \|k\|$, and hence $[\cdot , \cdot]$ is a bounded sesquilinear form and there exists $u \in B(H,K)$ such that $\Psi(\bar{k} \otimes a \otimes h) = \langle uh,k \rangle$. It is a routine application of the definition of $\gamma$ to verify that $u \in CB(H,K)$, thus completing the proof of (2).

As for $(3)$, it is clear that if $t \in \overline{K} \otimes E \otimes
H,$ it follows that $||t||_B \le ||t||_E$ - where we write $\| \cdot
\|_C = \| \cdot\|_{(\hat{K} \otimes C \otimes H, \gamma)}$ - since the
infimum over a larger collection is smaller.  In order to prove the
reverse inequality, we shall assume that $||t||_B < 1$, and show that
$||t||_E < 1$. The assumption $\|t\|_B < 1$ implies that there exists
a decomposition
\[t = \sum_{i=1}^m \sum_{j=1}^n \overline{k}_i \otimes a_{i,j} \otimes h_j =
(\overline{k}_1, \ldots,  \overline{k}_m) \otimes
[a_{i,j}]_{m \times n} \otimes
\begin{pmatrix}
h_1 \\
\vdots \\
h_n
\end{pmatrix}
\] with $a_{i,j} \in B$ such that
\[\Big(\sum_{i=1}^m \|\overline{k_i}\|^2\Big)^\frac{1}{2} \cdot \|[a_{i,j}]\|_{M_{m
    \times n}(B)} \cdot \Big(\sum_{j=1}^n \|h_j\|^2\Big)^\frac{1}{2} < 1~.
\]

It follows from the linear algebraic Observation \ref{pdhs} (discussed
at the end of this proof) that there exist $r \le m$, a linearly
independent set $\{\overline{k_p'}: 1 \le p \le r\}$ of vectors, and a
matrix $C \in M_{m \times r}$ (resp., $s \le n$, a linearly
independent set $\{h_q: 1 \le q \le s\} $ of vectors, and a matrix $D
\in M_{n \times s}$) such that $\|C\| \le 1$, $\overline{k_i} =
\sum_{p=1}^r \overline{c_{ip}}\,  \overline{k'_p}$, and $\sum_{i=1}^m
\|\overline{k_i}\|^2 = \sum_{p=1}^r \|\overline{k_p'}\|^2$ (resp., $\|D\| \le
1$, $h_j = \sum_{q=1}^s \overline{d_{jq}}h'_q$, and $\sum_{j=1}^n \|h_j\|^2
=\sum_{q=1}^s \|h_q'\|^2$).

Then, note that
\begin{align*}
t & = \sum_{i=1}^m \sum_{j=1}^n \bar{k_i} \otimes a_{ij} \otimes h_j\\
& = \sum_{i=1}^m \sum_{j=1}^n (\sum_{p=1}^r \bar{c_{ip}} \bar{k_p}) \otimes
a_{ij} \otimes (\sum_{q=1}^s \bar{d_{jq}} h'_q)\\
& = \sum_{p=1}^r \sum_{q =1}^s \bar{k'_p} \otimes a'_{pq} \otimes h'_q~,
\end{align*}
where $a'_{pq} = \sum_{i=1}^m \sum_{j =1}^n \bar{c_{ip}} a_{ij} \bar{d_{jq}}$,
i.e.,
$A' =: [a'_{pq}] = C^*AD$ and hence $\|A'\| \le \|A\| = 1$.

Observe that since $\{k'_p:1 \leq p \leq r\}$ and $\{(h_q': 1 \leq q
\leq s\}$ are linearly independent sets, then $a'_{pq}$
for all $p,q$ are forced to be in $E$, as $t \in \overline{K} \otimes E \otimes H$!
({\em Reason:} For each $p,q$ take $f_p \in (\overline{K})^*$ such that
$f_i(\overline{k_{i'}'})=\delta_{i,i'}$ and $g_q \in H^*$ such that
$g_q(h_{q'}')=\delta_{q,q'}$. Then $f_p \otimes Id_E \otimes g_q :\overline{K}
\otimes E \otimes H \rightarrow E$ and maps $t$ to $a'_{pq}$.)

Thus we find - since $\sum_{i=1}^n \|\overline{k_i}\|^2 =
\sum_{j=1}^n \|\overline{k_j'}\|^2$, $\sum_{j=1}^n \|h_j\|^2 =
\sum_{j=1}^n \|h_j'\|^2$, $A'\in M_{r \times s}(E)$ and
$\|A'\| = \|C^*AD\| \le \|A\|$ - that
\begin{align*}
\|t\|_E & \le (\sum_{p=1}^r \|\bar{k'_p}\|^2)^{\frac{1}{2}} \cdot \|A'\| \cdot
(\sum_{q=1}^s \|h'_q\|^2)^{\frac{1}{2}}\\
& \le (\sum_{i=1}^m \|\bar{k_i}
\|^2)^{\frac{1}{2}} \cdot \|A\| \cdot
(\sum_{j=1}^n \|h_j\|^2)^{\frac{1}{2}}\\
& < 1~,
\end{align*}
thus establishing that $\|t\|_B < 1 \Rightarrow \|t\|_E \le 1$; hence,  indeed $\|t\|_B = \|t\|_E$.

To complete the proof of the theorem, note that if $u \in CB(E,B(K))$,
and if $\Phi_u \in (\bar{K} \otimes E \otimes K)^*$ corresponds to $u$ as in the
proof of (2), then it is a consequence of (3) and the classical
Hahn-Banach theorem that there exists a $\tilde{\Phi} \in (\bar{K}
\otimes A \otimes K)^*$ which extends, and has the same norm, as
$\Phi_u$. Again, by (2), there exists a unique $\tilde{u} \in
CB(A,B(K))$ such that $\tilde{\Phi} = \Phi_{\tilde{u}}$. It follows
from the definitions that $\tilde{u}$ extends, and has the same
CB-norm as, $u$.
\end{proof}

Now for the `linear algebraic' observation used in the above proof:
\begin{obs}\label{pdhs}
If $h_1, \cdots , h_n$ are elements of a Hillbert space $H$, then
there exist $r \le n$, vectors $\{h_j': 1 \le j \le r\}$ in $H$ and a
rectangular matrix $C = [c_{ij}] \in M_{n \times r}$ such that
\begin{enumerate}
\item $CC^*$ is an orthogonal projection, and in particular,
  $\|[C]\|_{M_{n \times r}} \le 1$;
\item $h_i = \sum_{j=1}^r \bar{c}_{ij} h_j'$ for $1 \le i \leq n$;
\item $\sum_{i=1}^n \| h_i\|^2 = \sum_{j=1}^r \|h'_j\|^2$; and
\item $\{h'_j:1 \le j \le r\}$ is a linearly independent basis for the
  linear span of $\{h_i: 1 \le i \le n\}$.
\end{enumerate}
\end{obs}

\begin{proof}
Consider the linear operator $T:\ell^2_n \rightarrow H$ defined by
$Te_i = h_i$ for each $1 \le i \le n$, where of course $\{e_i:1 \le i
\le n\}$ is the standard orthonormal basis for $\ell^2_n$.  Let
$\{f_j: 1 \le j \le r\}$ denote any orthonormal basis for $M =
\mathrm{ker}^\perp(T)$ and let $P$ denote the orthogonal projection of
$\ell^2_n$ onto $M$. Set $h_j' = Tf_j$ for $1 \le j \le r$. Define
$c_{ij} = \langle f_j, e_i \rangle$ for $1 \le i \le n, 1 \le j \le r$
and note that the matrix $C = [c_{ij}] \in M_{n \times r}$ satisfies
\begin{align*}
(CC^*)_{ii'} & = \sum_{j=1}^r c_{ij}\bar{c}_{i'j}\\ & = \sum_{j=1}^r
  \langle f_j,e_i \rangle \langle e_{i'},f_j \rangle\\ & = \langle
  \sum_{j=1}^r \langle e_{i'},f_j \rangle f_j, \sum_{j=1}^r \langle
  e_{i},f_j \rangle f_j \rangle\\ & = \langle Pe_{i'}, Pe_i \rangle
  \\ & = \langle Pe_{i'}, e_i \rangle
\end{align*}
and so, $CC^*$ denotes the projection $P$.  Now observe that
\[
h_i = Te_i = TP e_i = \sum_j \langle e_i, f_j \rangle Tf_j= \sum_j
\bar{c}_{ij} h'_j.
\]

Finally observe that if $\{f_j: r < j \le n\}$ is an orthonormal basis for
$\mathrm{ker}(T)$, then,
\[ \sum_{j=1}^r \|h'_j\|^2  = \sum_{j=1}^r \|Tf_j\|^2
 = \|TP\|_{HS}^2
 = \|T\|_{HS}^2
 = \sum_{i=1}^n \|Te_i\|^2
 = \sum_{i=1}^n \| h_i\|^2 ~,
\]
thereby completing the proof of the observation.
\end{proof}

Now we discuss completely positive (CP) maps. We begin with a key
lemma relating positivity to norm bounds:

\begin{lem}\label{posbd}
For Hilbert spaces $H$ and $K$, let $a\in B(H)^+$, $b \in B(K)^+$ and
$x\in B(K, H)$. Then,
\begin{enumerate}
  \item $ \begin{pmatrix} 1 & x \\ x^* & 1 \end{pmatrix} \in B(H
    \oplus K)_ + \Leftrightarrow ||x|| \le 1$; and,
 \item more generally $\begin{pmatrix} a & x \\ x^* & b \end{pmatrix}
   \in B(H \oplus K)_ + \Leftrightarrow |\langle xk,h \rangle| \le
   \sqrt{\langle ah,h \rangle \langle bk,k \rangle} \ \forall\  h \in H,
   k \in K$.
 \end{enumerate}
\end{lem}
\begin{proof}
 First note that $(2) \Rightarrow ||x|| \le \sqrt{||a|| \cdot ||b||}
 \Rightarrow (1)$

As for (2),
\begin{align*}
 \begin{pmatrix}
a & x \\ x^* & b \end{pmatrix} \ge 0 & \Leftrightarrow
 \langle \begin{pmatrix} a & x \\ x^* & b \end{pmatrix} \begin{pmatrix}
   h \\ k \end{pmatrix} , \begin{pmatrix}
   h \\ k \end{pmatrix} \rangle \ge 0 ~\forall h \in H, k \in
 K\\ &\Leftrightarrow \langle ah,h \rangle+ \langle bk, k \rangle + 2
 Re \langle xk,h \rangle \ge 0 ~\forall h \in H, k \in
 K\\ &\Leftrightarrow |\langle xk,h \rangle| \le \frac{\langle ah,h
   \rangle + \langle bk,k \rangle}{2} ~\forall h \in H, k \in
 K\\ &\Leftrightarrow \forall t > 0, |\langle xk, h \rangle| \le
 \frac{1}{2} (t \langle ah,h \rangle +\frac{1}{t}\langle bk,k \rangle)
 ~\forall h \in H, k \in K. \\
& \hspace*{50mm} (\text{on replacing } h \text{ by }
 \sqrt{t}h \text{ and } k \text{ by } \frac{k}{\sqrt{t}})
\end{align*}
Hence, $$|\langle xk, h \rangle| \le \inf_{t > 0}\  \frac{1}{2} \left(t
\langle ah,h \rangle +\frac{1}{t}\langle bk,k \rangle \right) = \sqrt{\langle
  ah,h \rangle \langle bk,k \rangle}$$ (using the fact that
$\frac{a+b}{2} \ge \sqrt{ab}$ and that $\frac{a+b}{2}=\sqrt{ab}
\Leftrightarrow a=b$).
\end{proof}

\begin{prop}\label{cbcp1}
 Let $A$ be a unital $C^*$-algebra, $S \subset A$ be an operator
 system and $K$ be any Hilbert space. Then
\begin{enumerate}
 \item $\forall u \in CP(S, B(K)), ||u||=||u||_{cb}=||u(1)||$.
 \item If $ u : S \ra B(K)$ is linear  with $u(1)=1$, then $||u||_{cb} \le
   1 \Leftrightarrow u$ is CP.
\end{enumerate}
\end{prop}
\begin{proof}
(1) $u$ is CP $\Rightarrow u(S_+) \subset B(K)_+ \Rightarrow u(S_h)
  \subset B(K)_h \Rightarrow u(x^*)=u(x)^* \forall x \in S$ (by
  Carteasian decomposition). Hence by Lemma \ref{posbd}(1), we see
  that
\[
||x|| \le 1 \Rightarrow \begin{pmatrix} 1 & x \\ x^* &
    1 \end{pmatrix} \ge 0 \Rightarrow \begin{pmatrix} u(1) & u(x)
    \\ u(x^*) & u(1) \end{pmatrix} \ge 0.
\]
 Then, by Lemma \ref{posbd}(2),we find that $||u(x)|| \le ||u(1)||$
 and hence $||u||=||u(1)||$. Similarly, $x \in M_n(S)$ with $||x|| \le
 1$ implies $\begin{pmatrix} 1 & x \\ x^* & 1 \end{pmatrix} \ge 0 $
 and hence \[||u_n(x)|| \le ||u_n(1)|| = ||
 \mathrm{diag}_n(u(1),\ldots, u(1))|| = ||u(1)||.
\]
Hence, $||u||_{cb} = \sup_n \|u_n\| \le ||u(1)||$,
i.e., $||u||_{cb}=||u(1)||$, as desired.

\bigskip
(2) The proof of $(\Leftarrow)$ is an immediate consequence of part
(1) of this proposition and the assumed unitality of $u$.

We shall prove $(\Rightarrow)$ using the fact that if $\phi \in C^*$
for a commutative unital $C^*$-algebra $C$ (i.e., we may assume
$C=C(\Omega)$ for some compact $\Omega$), and if $\phi(1)=1$, then
$||\phi|| \le 1 \Leftrightarrow \phi \ge 0$. (It is a fact, which we
shall not go into here, that that fact is a special case of (2)!)

We first prove positivity of $u$, i.e., we need to verify that
$\langle u(x)h, h \rangle \ge 0, \ \forall\ x \in S_+, h \in K$.  We
may assume, without loss of generality that $||h|| =1$. Consider the
commutative unital $C^*$-subalgebra $A_0 = C^*(\{1,x\})$ of $A$. The
linear functional $\phi_0$ defined on $S \cap A_0$ by
$\phi_0(a):=\langle u(a)h,h \rangle$ is seen to be bounded with
$\|\phi_0\| = \phi_0(1) = 1$. Let $\phi \in A_0^*$ be a Hahn-Banach
extension of $\phi_0$. Since $\|\phi\| = 1 = \phi(1)$, it follows from
the fact cited in the previous paragraph that $\phi \ge 0$. Hence
$\langle u(x)h, h \rangle = \phi(x) \ge 0$, and the arbitrariness of
$h$ yields the positivity of $u(x)$. Thus, indeed $u$ is a positive
map.

To prove positivity of $u_n $, $n > 1$, we need to verify that
$\langle u_n(x)h, h \rangle \ge 0, ~\forall x \in M_n(S)_+, h \in
K^{(n)}$.  First deduce from the positivity of $u$ that $u(a^*) =
u(a)^* ~\forall a \in S$, from which it follows that $u_n(S_h) \subset
M_n(B(K))_h$. Also, note that $u(1) = 1 \Rightarrow u_n(1) = 1.$ The
assumption that $\|u_n\| = 1$ now permits us, exactly as in the case
$n=1$, (by now considering the commutative $C^*$-subalgebra $A_n :=
C^*(\{1, x\})$ of $M_n(A)$) that $u_n$ is also
positive. Thus, indeed $u$ is CP.
\end{proof}

\begin{cor}[Arveson's Extension Theorem - CP version]\label{arvcpext} \index{Theorem! Arveson's Extension (CP)}
If $S$ is an operator system contained in a $C^*$-algebra $A$, then
any CP map $u:S \rightarrow B(K)$ extends to a CP map $\tilde{u}:A
\rightarrow B(K)$.
\end{cor}

\begin{proof} Normalise and assume $\|u(1)\| \le 1$, whence it follows
  from the positivity of $u$ that $0 \le u(1) \le 1$. Pick any state
  on $A$, i.e., a positive (=positivity-preserving element) of $A^*$
  such that $\phi(1) = 1$. (The existence of an abundance of states
  is one of the very useful consequences of the classical Hahn-Banach
  Theorem.)

Now, consider the map $U: S \oplus S ~(\subset A\oplus A)~ \rightarrow
B(K)$ defined by
\[
U\Big(\begin{pmatrix} x & 0\\ 0 & y \end{pmatrix}\Big) = u(x) + \phi(y) (1 -
u(1))
\]
and note that $U$ is unital and completely positive, in view of Lemma
\ref{cbcp1}(2); and hence (by Proposition \ref{cbcp1}(1))
$\|U\|_{cb} = 1$. Appeal to Theorem \ref{arv-ext} to find an
extension $\tilde{U} \in CB(A \oplus A,B(K))$ with $\|\tilde{U}\|_{cb}
= 1$. As $\tilde{U}$ inherits the property of being unital from $U$,
it follows by an application of Proposition \ref{cbcp1}(2) that
$\tilde{U}$ is CP. Finally, if we define $\tilde{u}(x) = \tilde{U}
\Big(\begin{pmatrix} x & 0\\ 0 & 0 \end{pmatrix}\Big)$ for $x \in A$, it is
clear that $\tilde{u}$ is a CP extension to $A$ of $u$.
\end{proof}

\begin{lem}\label{posfctcp} Let $A$ be a unital $C^*$-algebra. Then,
 $\phi \in A_+^* \Rightarrow \phi \in CP(A, \mathbb{C})$.
\end{lem}

\begin{proof}
We need to show that if $[a_{i,j}] \in M_n(A)_+$, then we have
$\sum_{i,j} \phi(a_ {i,j})h_j \bar{h_i} \ge 0 ~\forall h \in
\mathbb{C}^n$. But that follows since $\phi$ is positive and
$\sum_{i,j}h_j\bar{h_i}a_{i,j} = {\begin{pmatrix} h_1 \\ \vdots
  \\ h_n \end{pmatrix}}^* [a_{i,j}] \begin{pmatrix} h_1 \\ \vdots
  \\ h_n \end{pmatrix} \in A_+$.
\end{proof}

We now discuss yet another useful $2 \times 2$ matrix trick; this one
also serves as a conduit from operator spaces to operator systems.

\begin{thm}\label{cbcp2}
 Suppose $E \subset B(H)$ is a subspace and $w:E \rightarrow B(K)$ is a linear map. Define $S =\left\{ \begin{pmatrix} \lambda1 & x \\ y^* & \mu
   1 \end{pmatrix}: x, y \in E \right.$, $\left. \lambda, \mu \in
 \mathbb{C}\right\} \subset M_2(B(H))$ and
 $W:S \rightarrow M_2(B(K))$ by $W \Big(\begin{pmatrix}
   \lambda1 & x \\ y^* & \mu 1
\end{pmatrix}\Big) =
\begin{pmatrix}
\lambda1 & w(x)
\\ w(y)^* & \mu 1
\end{pmatrix}$.
Then $E$ is an operator space and $\|w\|_{cb} \le 1 \Leftrightarrow S$
is an operator system and $W$ is CP.

(Here the operator space/system structure on $E/S$ is the natural one
induced from the other.)
\end{thm}

\begin{proof}
 $(\Leftarrow)$ First note that if $S$ is an operator system and $W$
  is linear and necessarily unital, then $E$ (identified as the
  (1,2)-corner of $S$) is an operator space, while it follows from two
  applications of Lemma \ref{posbd}(1) that if $x \in E$, then
  $||x||_E \le 1 \Rightarrow X = \begin{pmatrix} 1 & x \\ x^* &
    1 \end{pmatrix} \in S_+$ and so  $W$ positive $\Rightarrow  W(X)
  = \begin{pmatrix} 1 & w(x) \\ w(x)^* & 1 \end{pmatrix} \in M_2(B(K))_+ \Rightarrow
  ||w(x)||_{B(K)} \le 1$.

Now suppose $x(n) \in M_n(E)$ and $\|x(n)\|_{M_n(E)} \le 1.$
Define $X(n) \in M_n(S)$ by $X(n)_{ij} = \begin{pmatrix} \delta_{ij} & x(n)_{ij} \\ (x(n)^*)_{ij} & \delta_{ij} \end{pmatrix}$, where the Kronecker symbol $\delta_{ij}$ denotes the $ij$-th entry of the identity matrix.
 It is not hard to see that there exists a permutation matrix $P \in M_{2n}$ - independent of $x(n)$ - such that $PX(n)P^* = \begin{pmatrix} 1 & x(n) \\ (x(n))^* & 1 \end{pmatrix}$.

Similarly, define $y(n) = w_n(x(n)) \in M_n(B(K))$ and $Y(n) \in M_n(S)$ by $Y(n)_{ij} = \begin{pmatrix} \delta_{ij} & y(n)_{ij} \\ (y(n)^*)_{ij} & \delta_{ij} \end{pmatrix}$ and deduce that with $P$ as above, we have $PY(n)P^* = \begin{pmatrix} 1 & y(n)) \\ (y(n))^* & 1 \end{pmatrix}$.

It follows from the definitions that
\begin{align*}
(y(n)^*)_{ij} &= (w_n(x(n))^*)_{ij}\\
&= (w(x(n)_{ji})^*\\
&= (w_n(x(n)^*))_{ij}
\end{align*}
and hence, $W_n(X(n)) = Y(n)$. Now,
\begin{align*}
\|x(n)\| \le 1 &\Rightarrow  PX(n)P^* \ge 0\\ & \Rightarrow  X(n) \ge 0 \\ & \Rightarrow  Y(n) = W_n(X(n)) \ge 0\\ & \Rightarrow PY(n)P^* \ge 0\\ & \Rightarrow  \|y(n)\| \le 1,\end{align*}
and hence $ \|w_n\| \le 1$.

$(\Rightarrow)$ It is clear that $S$ is naturally an operator system
if $E$ is an operator space. We first prove positivity of $W$. (That
of the $W_n$'s is proved similarly.)

We need to show that $A = \begin{pmatrix} \lambda & b \\ c^* &
  \mu \end{pmatrix} \in S_+ \Rightarrow W(A) \ge 0$. Since $A \ge 0$,
we have $b = c \in E$ and $\lambda, \mu \ge 0$. Also, $ A + \epsilon 1
\ge 0$ and since, by the definition of $W$, it is seen that $W(A) =
\lim_{\epsilon \downarrow 0} W(A + \epsilon 1)$, we may assume without
loss of generality that $\lambda , \mu > 0$.

Then, $A = \begin{pmatrix} \lambda1 & b \\ b^* & \mu 1 \end{pmatrix}
= \begin{pmatrix} \sqrt{\lambda} & 0 \\ 0 & \sqrt{\mu} \end{pmatrix}
\begin{pmatrix} 1 & x \\ x^* & 1 \end{pmatrix} \begin{pmatrix}
\sqrt{\lambda} & 0\\0 & \sqrt{\mu} \end{pmatrix}$, where $x =
b/\sqrt{\lambda \mu}$. Hence
\begin{align*}
 \begin{pmatrix}
\lambda1 & b \\ b^* & \mu 1 \end{pmatrix} \ge 0 &
 \Leftrightarrow \begin{pmatrix} 1 & x \\ x^* & 1 \end{pmatrix} \ge
 0\\& \Leftrightarrow ||x||\le 1\\ & \Rightarrow \|w(x)\| \le 1\\ &
 \Rightarrow \begin{pmatrix} 1 & w(x) \\ w(x^*) & 1 \end{pmatrix} \ge
 0\\ & \Rightarrow W (A) = \begin{pmatrix} \sqrt{\lambda} & 0 \\ 0 & \sqrt{\mu} \end{pmatrix} \begin{pmatrix} 1 & w(x) \\ w(x^*) & 1 \end{pmatrix} \begin{pmatrix} \sqrt{\lambda} & 0 \\ 0 & \sqrt{\mu} \end{pmatrix}
 \ge 0.
\end{align*}
Arguing with the permutation $P$ as in the proof of the reverse implication, and proceeding as above (in the $n=1$ case), we find that $\|w_n\| \le 1 \Rightarrow W_n$ is
positive.
\end{proof}

Now we are ready {\bf to prove the Fundamental Factorisation Theorem \ref{ff}.} \index{Theorem! Fundamental Factorisation}

\begin{proof} Normalise and assume $\|u\|_{cb} \le 1$.
Let the operator system $S$ and $U:S \rightarrow M_2(B(K)) =
B(K^{(2)})$ be the (unital) CP map associated to the operator space
$E$ and the completely contractive map $u:E \rightarrow B(K)$ (like
the $w \leftrightarrow W$) as in Theorem \ref{cbcp2}. As $E$ is a
subset of the $C^*$-algebra $A$ by assumption, we find that $S \subset
M_2(A)$

Then by the CP version (Corollary \ref{arvcpext}) of Arveson's
extension theorem, we can extend $U$ to a CP map $\tilde{U}: M_2(A)
\rightarrow B(K^{(2)})$ .

As $M_2(A)$ is a unital $C^*$-algebra and $\tilde{U}$ is a unital CP
map, we may conclude from Stinespring's theorem that there exists a
representation $\sigma:M_2(A) \rightarrow B(H)$ and an isometry $ T :K
\oplus K \rightarrow H$ such that
$\tilde{U}(\tilde{x})=T^*\sigma(\tilde{x}) T$ for all $\tilde{x} \in
M_2(A)$. Hence,
\[\begin{pmatrix} 0 & u(x) \\ 0 & 0 \end{pmatrix}
= U \begin{pmatrix} 0 & x \\ 0 & 0 \end{pmatrix}
=\tilde{U} \begin{pmatrix} 0 & x \\ 0 & 0 \end{pmatrix}
= T^*\sigma\Big(\begin{pmatrix} 0 & x \\ 0 & 0 \end{pmatrix}\Big)T
=T^* \sigma\Big(
\begin{pmatrix}
 x & 0 \\
 0 & 0
\end{pmatrix}
\begin{pmatrix}
 0 & 1 \\
 0 & 0
\end{pmatrix}\Big)T = T^* \pi(x)T'
,\]
where $T'=\sigma \Big(
\begin{pmatrix} 0 & 1 \\
0 & 0
\end{pmatrix}\Big) T$ and
$\pi(x)=\sigma\Big(\begin{pmatrix} x & 0 \\ 0 &
  0 \end{pmatrix}\Big)$. Now, consider the projection $P :=
\sigma\Big(\begin{pmatrix} 1 & 0 \\ 0 & 0\end{pmatrix}\Big)$ and define
  $\hat{H} = P(H), Vk = PT(k,0), Wk =
  \sigma\Big(\begin{pmatrix} 0 & 1 \\ 0 & 0 \end{pmatrix}\Big)T(k,0)$, and note that $P$ commutes with
  $\sigma\Big(\begin{pmatrix}x & 0\\0 & 0\end{pmatrix}\Big)$ for each
    $x \in A$, so that the equation $\pi(x) =
    \sigma\Big(\begin{pmatrix}x & 0\\0 &
      0\end{pmatrix}\Big)_{_{\hat{H}}}$ does define a (unital)
      representation of $A$ on $\hat{H}$; and we finally see that $V,W: K \rightarrow \hat{H}$ and indeed
\begin{equation}\label{cbfact}
u(x) = V^* \pi(x) W ~\mathrm{for\ all}~ x \in E~,
\end{equation}
with $\|V\|, \|W\| \le 1$, as desired.

In the converse direction, if $u$ admits the factorisation
(\ref{cbfact}), then it is seen that also
\[u_n([a_{ij}]) = \begin{pmatrix} V^*\\ \vdots \\ V^* \end{pmatrix}
 \pi_n([a_{ij} ])(W , \ldots, W) \] and as $\pi_n$ is as much of a
representation (and hence contractive) as $\pi$, it is seen that
$\|u\|_{cb} \le 1$.

The non-trivial implication in the final assertion of the theorem is a
consequence of Arveson's extension theorem \ref{arvcpext} and
Stinespring's theorem.
\end{proof}

\begin{rem}\label{factcp}
For $V=W, u(x)=V^* \pi(x) V \Rightarrow u$ is CP (since a
representation is obviously CP).
\end{rem}

\begin{cor}
  Any $u \in CB(S,B(K))$ can be decomposed as $u=u_1-u_2 + i(u_3 - u_4)$,
 with $u_j \in CP(S, B(K)), j=1,2,3,4$.
\end{cor}

\begin{proof}
 The proof is a consequence of the fundamental factorisation, the
 polarisation identity and Remark \ref{factcp}.
\end{proof}


\section{More on CB and CP maps}

\begin{thm}\textbf{(Kraus Decomposition)} \label{kraus} \index{Kraus decomposition}

A linear map $u : M_n \ra M_m$ is CP if and only if there exists a

family $\{V_p : 1 \le p \le N \} \subset M_{n \times m}$ with $N \leq nm$

such that $u(a) = \sum_{p=1}^N V_p^* a V_p$ for all $a \in M_n$.

\end{thm}

\begin{pf}

Suppose $u$ is CP. Then, by the fundamental factorization Theorem
\ref{ff}, there is a Hilbert space $\widehat{H}$, a
$\ast$-representation $\pi : M_n \ra B(\widehat{H})$ and a map $V :
\ell_2^m \ra \widehat{H}$ such that $u(a) = V^* \pi(a)V$. It is a
basic fact that for any representation of $M_n$, as above, there
exists a Hilbert space $H$ such that $\widehat{H} =(\cong) \ell_2^n
\otimes H$ and $\pi(a) = a \otimes 1$.

Further, it is also true that there exists a subspace $H_1 \subset H$
with $\text{dim}\, H_1 \leq mn$ such that $V(\ell_2^m) \subset
\ell_2^n \otimes H_1$.

({\em Reason:} If $\{e_j: 1 \le j \le m\}$ and $\{e_i: 1 \le i \le
n\}$ are orthonormal bases for $\ell_2^m$ and $\ell_2^n$,
respectively, we see that there must exist operators $T_i :\ell_2^m
\rightarrow H$ such that
\[ Ve_j = \sum_{i=1}^n e_i \otimes T_ie_j ~\forall 1 \leq j \leq n~.\]
Clearly, then, if $H_1 = span\{T_i e_j: 1 \le j \le m, 1 \le i \le n\}$, then dim$~H_1 \le nm$ and $V(\ell_2^m) \subset \ell_2^n \otimes H_1$.)

Therefore, it follows that if $\{e_p, 1 \le p \le N\}$ is an orthonormal basis for $H_1$, then $N \le nm$ and there exist $V_p:\ell_2^m \ra \ell_2^n, 1 \le p \le N$ such that $V(\xi) = \sum_p V_p(\xi) \otimes e_p$ for all $\xi \in \ell_2^m$.

Finally, it is seen that for all $a\in M_n$, $\xi, \eta \in \ell^m_2$, we have
\begin{eqnarray*}
\langle u(a)\xi, \eta\rangle & = & \langle (a \otimes 1) \Big( \sum_p V_p (\xi) \otimes e_p\Big)  , \sum_q V_{q} (\eta) \otimes e_q \rangle \\
& = & \sum_p \langle a V_{p} (\xi), V_{p} (\eta)\rangle \\
& = & \langle \sum_p V_{p}^{*} a V_{p} (\xi), \eta \rangle
\end{eqnarray*}

Conversely, it is not hard to see that any map admitting a
decomposition of the given form is necessarily CP.

\end{pf}

The above result may be regarded as one of the first links between
Operator Space Theory and Quantum Information Theory.

\begin{defn}
If a linear map $u : M_n \ra M_m$ is CP and preserves the trace, then
it is called a {\bf quantum channel}. \index{quantum channel}
\end{defn}

\begin{rem}
In the set up of Theorem \ref{kraus}, $u$ preserves the trace if and
only if $\sum_p V_p V_p^* = I$; while it is identity-preserving if and
only if $\sum_p V_p^*V_p = I$.
\end{rem}

\begin{thm}[Choi]\label{choi}
The following conditions on a linear map $u : M_n \ra B(K)$ are equivalent:
\begin{enumerate}
\item $u$ is CP.

\item $u$ is $n$-positive.

\item $[u(e_{ij})] \in M_n(B(K))_+$, where $\{e_{ij} : 1 \leq i, j
  \leq n\}$ is the canonical system of matrix units for $M_n$.

\end{enumerate}
\end{thm}

\begin{pf}

$(3) \implies (1).$ $[u(e_{ij})] \geq 0 \Rightarrow [u(e_{ij})] = X^*
  X$ for some $X = [x_{ij}] \in M_n(B(K))$. So, $u(e_{ij}) = \sum_k
  x_{ki}^* x_{kj}$. In particular,

\[
 u(a) = \sum_{ij} a_{ij} u(e_{ij}) = \sum_{i, j, k}
 x^*_{ki}a_{ij}x_{kj} = \sum_k ( \sum_{ij}x^*_{ki}a_{ij}x_{kj} )
\]
for all $a = [a_{ij}] \in M_n$.

Define $u_k(a)$ to be the element of $M_n(B(K))$ with $(i,j)$-th entry
equal to $\sum_{ij}x^*_{ki}a_{ij}x_{kj}$. Then, clearly,
\[
u_k(a) = \begin{pmatrix} x_{k1}^* & x_{k2}^* & \cdots &
  x_{kn}^* \end{pmatrix} \begin{pmatrix} a_{11} & a_{12} & \cdots &
  a_{1n}\\ a_{21} & a_{22} & \cdots & a_{2n} \\ \cdots \\ a_{n1} &
  a_{n2} & \cdots & a_{nn} \ \end{pmatrix} \begin{pmatrix}
    x_{k1}\\ x_{k2}\\ \cdots \\ x_{kn} \end{pmatrix} \]

Hence $u$ has the form $u(a) = \sum_{k=1}^n V_k^* a V_k$, where $V_k:K
\rightarrow \ell_2^n \otimes K$ is given by $V_k(\xi) = \sum_j e_j
\otimes x_{kj}(\xi)$ with $\{e_i\}$ denoting the the standard
o.n.b. of $\ell_2^n$. The desired implication follows now from Theorem
\ref{kraus}.

The implication $(1) \implies (2)$ is trivial while $(2) \implies (3)$ is a consequence of the fact that $[e_{ij}] \in (M_n)_+$.
\end{pf}

\begin{lem}[Roger Smith]\label{rrsmlem}
Fix $N \geq 1$ (resp., a compact Hausdorff space $\Omega$) and an
operator space $E$. Then every bounded linear map $u: E \ra M_N$
(resp., $u: E \ra C(\Omega, M_N)$) is CB with $||u||_{cb} = ||u_N:
M_n(E) \ra M_N(M_N)||$ (resp., $||u||_{cb} = ||u_N: M_n(E) \ra
M_N(C(\Omega, M_N))||$).  In particular, every bounded linear map $ u
: E \ra C(\Omega)$ is CB with $||u||_{cb} = ||u||$.
\end{lem}

\begin{pf}
Suppose first that $u \in B(E, M_n(\C))$. To prove that $\|u\|_{cb} =
\|u_N\|$, it clearly suffices to verify that $\|u_n\| \le \|u_N\|
~\forall n \ge N$. We need to verify that if $a = [a_{ij}] \in M_n(E)$
and $x_1, \cdots, x_n, y_1, \cdots, y_n \in \C^N$ satisfy
$\|a\|_{M_n(E)} \le 1$ and $\sum_{i=1}^n \|x_i\|^2 = \sum_{i=1}^n
\|y_i\|^2 = 1$, then
\[
|\sum_{i,j=1}^n \langle u(a_{ij})y_j, x_i \rangle | \le \|u_N\|~.
\]
For this, appeal first to Observation \ref{pdhs} to find $\alpha,
\beta \in M_{n \times N}(\C)$ and $x'_1, \cdots , x'_N, y'_1, \cdots ,
y'_N \in \C^N$ such that

\begin{enumerate}
\item $\sum_{l=1}^N \|x'_l\|^2 = \sum_{j=1}^N \|y'_j\|^2 = 1$;
\item $x_i = \sum_{l=1}^N \alpha_{il} x'_l$ and $y_j = \sum_{k=1}^N
  \beta_{jk} y'_k$ for $ 1 \le i,j \le n$, and
\item $\|\alpha\|_{M_{n \times N}(\C)}, \|\beta\|_{M_{n \times N}(\C)} \le 1$
\end{enumerate}

Deduce, then, that
\begin{eqnarray*}
|\sum_{i,j=1}^n \langle u(a_{ij})y_j, x_i \rangle | & = &
|\sum_{i,j=1}^n \sum_{k,l=1}^N \langle u(a_{ij}) \beta_{jk} y'_k,
\alpha_{il} x'_l \rangle |\\ &=& |\sum_{k,l=1}^N \langle
u(\alpha^*a\beta)_{lk} y'_k, x'_l \rangle|\\ &\le&
\|u_N(\alpha^*a\beta)\|\\ &\le& \|u_N \|\cdot
\|\alpha^*a\beta\|\\ &\le& \|u_N \| ~,
\end{eqnarray*}
as desired.

Next, suppose $u \in B(E, C(\Omega;M_N(\C))$. Let us introduce the
notation $\Omega \ni \omega \mapsto \phi_\omega \in E^*$ where
$\phi_\omega(x) = u(x)(\omega)$, and note that

\begin{eqnarray*}
\|\phi_\omega\| & = & \sup \{|\phi_\omega(x)| : \|x\| \le 1\}\\
&=& \sup \{|u(x)(\omega)| : \|x\| \le 1\}\\
&\le& \sup \{\|u(x)\| : \|x\| \le 1\}\\
&=& \|u\| ~.
\end{eqnarray*}

Now conclude that if  $[a_{ij}] \in M_n(E)$ and $u(a_{ij}) = f_{ij}$, then
\begin{eqnarray*}
\|u_n([a_{ij}])\| &=& \sup_{\omega \in \Omega} \|[f_{ij}(\omega)]\|\\
&=& \sup_{\omega \in \Omega} \|[\phi_\omega(a_{ij})]\|\\
&=& \sup_{\omega \in \Omega} \|(\phi_\omega)_n([a_{ij}])\|\\
&\le& \sup_{\omega \in \Omega} \|(\phi_\omega)_n\| \cdot \|[a_{ij}]\|_{M_n(E)}\\
&=& \sup_{\omega \in \Omega} \|\phi_\omega\| \cdot \|[a_{ij}]\|_{M_n(E)}\\
&\le&  \|u\| \cdot \|[a_{ij}]\|_{M_n(E)}~,
\end{eqnarray*}

where we have used the fact that $\|\phi_\omega\|_{cb} =
\|\phi_\omega\|$ as $\phi_\omega \in B(E,\C)$, and we hence find that
$\|u\|_{cb} \le \|u\|$.
\end{pf}

\begin{prop}
Let $E \subset B(H)$ and $F \subset B(K)$ be operator spaces and $u
\in B(E, F)$. Then $||u||_{cb} = ||u||$ in the following cases:
\begin{enumerate}

\item $Rank\,(u) \leq 1$.

\item $C^*(F)$ is commutative.

\item $E = F = R$ or $E = F= C$.

\end{enumerate}

\end{prop}

\begin{pf}

\begin{enumerate}
\item If $0 \neq f \in u(E)$, then, clearly there exists $\phi \in
  E^*$ such that $u(e) = \phi(e)f ~\forall e \in E$. Deduce from the
  case $N=1$ of Roger Smith's Lemma \ref{rrsmlem} that $\|\phi\|_{cb}
  = \|\phi\|$, and hence if $a =[a_{ij}] \in M_n(E)$, then $\|u_n(a)\|
  = \|[\phi(a_{ij})f]\| \le \|\phi_n(a)\| \cdot \|f\| \le
  \|\phi\|\cdot \|a\| \cdot \|f\| = \|u\| \cdot \|a\|$ and hence
  indeed $\|u\|_{cb} = \|u\|$.

\item This is an immediate consequence of Roger Smith's Lemma \ref{rrsmlem}.

\item Let us discuss the case of $C:= \overline{span}\{e_{i1} : i \geq
  1\} \subset B(\ell_2)$ which clearly admits the identification $C
  \cong \ell^2 \cong B(\C, \ell^2)$ thus:
\[
C(\xi) = \sum_{i=1}^\infty \lambda_i e_{i1} \leftrightarrow \xi =
\sum_{i=1}^\infty \lambda_i e_i \leftrightarrow \rho_\xi(\alpha) =
\alpha \xi.
\]

Observe now that if $u \in B(C)$ corresponds as above to $u' \in
B(\ell^2)$ and to $u''\in B(B(\C,\ell^2))$, then
\[
u(C(\xi)) = C(u'\xi) \mbox{ and } u"(\rho_\xi) =\rho_{u'\xi}~,
\]
 and hence it follows that $\|u\|_{cb} = \|u"\|_{cb} = \|u'\| =
 \|u\|$, where we have used the obvious fact that
\[
\|u"_n([\rho_{\xi_{ij}}]\| = \|[\rho_{u'\xi_{ij}}]\| = \|\left[\begin{array}{cccc} u" & 0 & \cdots & 0\\ 0 & u" & \cdots & 0\\ \vdots & \vdots & \ddots & \vdots\\
  0 & 0 & \cdots & u"\end{array} \right]
[\rho_{\xi_{ij}}]\| \le \|u"\| \cdot \|[\rho_{\xi_{ij}}]\|~.
\]

The case of $R$ is proved similarly.

Note that $R$ and $C$ also have Hilbert space structures with o.n.b.'s
$\{ e_{1j}\}$ and $\{e_{i1}\}$ respectively.
\end{enumerate}
\end{pf}

\begin{prop}
For all $u \in CB(R, C)$ (resp., $u \in CB(C, R)$, $||u||_{cb} = ||u||_{HS}$.
\end{prop}

\begin{pf}
First consider $u \in CB(R,C)$. (The case of $u \in CB(C,R)$ is
treated analogously.) Notice that if
\[x_n = \left( \begin{array}{cccc}e_{11} & 0 & \cdots & 0\\e_{12} & 0 & \cdots &
  0\\\vdots & \vdots & \vdots & \vdots\\e_{1n} & 0 & \cdots & 0 \end{array}
\right) \in M_n(R)~, \]
then $\|x_n\| = 1$, and hence $\|u_n(x_n)\|_{M_n} \le \|u\|_{cb}$.

Hence,
\[\|u\|_{cb} \ge \|u_n(x_n)\| = \|\left( \begin{array}{cccc}u(e_{11}) & 0 & \cdots & 0\\u(e_{12}) & 0
  & \cdots &  0\\\vdots & \vdots & \vdots & \vdots\\u(e_{1n}) & 0 & \cdots & 0 \end{array}
\right)\| = \|\sum_{j=1}^n u(e_{1j}^*)u(e_{1j})\|^\frac{1}{2}~.\]

So, if $u(e_{1j}) = \sum_{i=1}^\infty u_{ij}e_{i1}$, then
\begin{eqnarray*}
\|\sum_{j=1}^n u(e_{1j}^*)u(e_{1j})\|^\frac{1}{2} &=& \|\sum_{j=1}^n \sum_{i,k = 1}^\infty (u_{ij} e_{i1})^*(u_{kj} e_{k1})\|^\frac{1}{2}\\
&=& \|\sum_{j=1}^n\sum_{i,k = 1}^\infty \overline{u_{ij}}u_{kj}e_{1i}e_{k1}\|^\frac{1}{2}\\
&=& \|\sum_{j=1}^n\sum_{i,k = 1}^\infty \overline{u_{ij}}u_{kj} \delta_{ik}e_{11}\|^\frac{1}{2}\\
&=&(\sum_{j=1}^n \sum_{k=1}^\infty |u_{kj}|^2)^\frac{1}{2}~.
\end{eqnarray*}

Letting $n \rightarrow \infty,$ we see that $\|u\|_{HS} \le \|u\|_{cb}$.

For the reverse inequality, since the linear span of the $e_{1j}$'s is
dense in $R$, deduce that $\{\sum_{j=1}^mx_j \otimes e_{1j}: m \in \N,
x_1, \cdots, x_m \in M_n(\C)\}$ is dense in $M_n(\C) \otimes R =
M_n(R)$, so it suffices to verify that
\[\| \sum_{j=1}^m x_j \otimes u(e_{1j})\|_{M_n(C)} \le \|u\|_{HS}~\cdot~\|\sum_{j=1}^m x_j \otimes e_{1j}\|_{M_n(R)}\]

As $R$ and $C$ are both isometric to Hilbert space, $u$ is in
particular a bounded operator between Hilbert spaces, and its matrix
with respect to the natural orthonormal bases for domain and range is
$((u_{ij}))$ as above. And we see that
\begin{eqnarray*}
\| \sum_{j=1}^m x_j \otimes u(e_{1j})\|_{M_n(C)}&=& \|\sum_{j=1}^m \left( x_j \otimes \sum_{i=1}^\infty   u_{ij}e_{i1}\right) \|_{M_n(C)}\\
&=& \| \sum_{i=1}^\infty \left( \sum_{j=1}^m u_{ij}x_j \right) \otimes e_{i1}\|_{M_n(C)}\\
&=& \|\left( \sum_{i=1}^\infty \left( \sum_{j=1}^m u_{ij}x_j \otimes e_{i1}\right)\right)^*\left( \sum_{k=1}^\infty \left( \sum_{j=1}^m u_{kj}x_j\right) \otimes e_{k1}\right)\|^\frac{1}{2}_{M_n(C)}\\
&=& \|\left(\sum_{i=1}^\infty \left( \sum_{j=1}^m \overline{u_{ij}}x_j^*\right) \otimes e_{1i}\right)
\left(\sum_{k=1}^\infty \left(\sum_{j=1}^m u_{kj}x_j \right) \otimes e_{k1} \right)\|^\frac{1}{2}_{M_n(C)}\\
&=& \|\sum_{i=1}^\infty \left( \sum_{j=1}^m \overline{u_{ij}}x_j^*\right) \left( \sum_{j=1}^m u_{ij}x_j \right)\|^\frac{1}{2}_{M_n(\C)}\\
&\le& \left( \sum_{i=1}^\infty \| \sum_{j=1}^m u_{ij}x_j\|^2 \right)^\frac{1}{2}~.
\end{eqnarray*}

On the other hand,
\begin{eqnarray*}
\| \sum_{j=1}^m u_{ij} x_j\| &=& \|\left[ \begin{array}{ccc}x_1 & \cdots & x_m\\0 & \cdots & 0\\ \vdots & \vdots & \vdots\\0 & \cdots & 0 \end{array} \right] \left[ \begin{array}{cccc}u_{i1} & 0 & \cdots & 0\\u_{i2} & 0 & \cdots & 0\\ \vdots & \vdots & \vdots & \vdots\\ u_{i2} & 0 & \cdots & 0 \end{array} \right] \|\\
& \le & \|\left[ \begin{array}{ccc}x_1 & \cdots & x_m\\0 & \cdots & 0\\ \vdots & \vdots & \vdots\\0 & \cdots
    & 0 \end{array} \right] \|\cdot \| \left[ \begin{array}{cccc}u_{i1} & 0 & \cdots & 0\\u_{i2} & 0 & \cdots & 0\\ \vdots & \vdots & \vdots & \vdots\\ u_{i2} & 0 & \cdots & 0 \end{array} \right] \|\\
&=& \| \sum_{j=1}^mx_jx_j^*\|^\frac{1}{2} \left( \sum_{j=1}^m  |u_{ij}|^2 \right)^\frac{1}{2}~,
\end{eqnarray*}

and hence we see that
\begin{eqnarray*}
\| \sum_{j=1}^m x_j \otimes u(e_{1j})\|_{M_n(C)}&\le& \left( \sum_{i=1}^\infty \|\sum_{j=1}^m x_jx_j^*\| \sum_{j=1}^m |u_{ij}|^2 \right)^\frac{1}{2}\\
& \le & \left( \|\sum_{j=1}^m x_jx_j^*\| \sum_{i,j=1}^\infty |u_{ij}|^2 \right)^\frac{1}{2}~\\
&=& \| \sum_{j=1}^m x_j \otimes e_{1j}\|_{M_n(R)} \|u\|_{HS}~.
\end{eqnarray*}
\end{pf}

\begin{defn}
Let $E$ and $F$ be operator spaces. A linear map $u : E \ra F$ is said to be a

\begin{enumerate}

\item complete isomorphism \index{complete! isomorphism} if $u$ is a linear isomorphism such that $u$ and $u^{-1}$ are both CB.

\item complete contraction \index{complete! contraction} if $u$ is CB with $||u||_{cb}\leq 1$.

\item complete isometry \index{complete! isometry} if $u_n$ is an isometry for all $n \geq 1$.

\end{enumerate}

\end{defn}

\begin{cor}
$R$ and $C$ are not isomorphic to each other as operator spaces.
\end{cor}

\begin{pf}
Suppose $u: R \ra C$ and $u^{-1}:C \ra R$ are both CB maps. Then $u$
and $u^{-1}$ are Hilbert Schmidt and hence compact on an infinite
dimensional Hilbert space, a contradiction.
\end{pf}

It is a fact that even for  $n$-dimensional row and column spaces $R_n$ and $C_n$, we have
$$\inf\{ ||u||_{cb}||u^{-1}||_{cb} : u \in CB(R_n, C_n)\, \text{invertible}\} = n.$$
This is ``worst possible'' in view of the fact that for any two $n$-dimensional operator spaces $E$ and $F$, $$\inf\{ ||u||_{cb}||u^{-1}||_{cb} : u \in CB(F, F)\, \text{invertible}\} \leq n.$$

\begin{prop}
Let $A$ be a commutative $C^*$-algebra and $u \in B(A, B(K))$. If $u$ is positive, then $u$ is completely positive.
\end{prop}

\begin{pf}
We first indicate the proof for finite-dimensional $A$.

Thus, suppose $\text{dim}\, A =n$. Then
$$A \simeq \ell^{\infty}_n = \left\{ \left(\begin{array}{ccc}
  \lambda_1 & \cdots & 0\\ \vdots & \ddots & \vdots \\ 0 & \cdots &
  \lambda_n \end{array} \right) : \lambda_i \in \mbb{C} \right\}.$$

Now $u : \ell^{\infty}_n \ra B(K)$ is positive if and only if $a_j:=u(e_{j}) \geq 0, \forall\, j$. Finally, we get $$u(x) = u(\sum_j x_j e_j) \sim V^* x V ~\forall x \in \ell^{\infty}_n ~,$$
where $V= \left[ \begin{array}{c}a_1^{\frac{1}{2}}\\ \vdots \\ a_n^{\frac{1}{2}} \end{array} \right]$ and hence $u$ is CP.

We now give a sketch of the proof in the general case. It suffices to
show that if $X$ is a compact Hausdorff space, and $\phi:C(X)
\rightarrow B$ is a positive map into a unital $C^*$-algebra $B$, then
$\phi$ is $n$-positive for an arbitrarily fixed $n$. The first thing
to observe is that there exist natural identifications of $M_n(C(X))
\cong C(X;M_n(\mbb{C}))$ with $C^*$-completions of the algebraic tensor
product $C(X) \otimes M_n(\mbb{C})$ whereby the elementary tensor $f
\otimes T$ corresponding to $f(\cdot)T$; and further $\phi_n = \phi
\otimes \Id_{M_n(\mbb{C})}$ which maps elementary tensors $f \otimes
T$, with $f,T \ge 0$ to the element $\phi(f) \otimes T$ of $B \otimes
M_n(\C)$ which is positive, being the product of two commuting
positive elements. Finally, an easy partition of identity argument
shows that if $P(\cdot) \in C(X;M_n(\C)$ is positive, then $P$ is
approximable by elements of the form $\sum_i f_i \otimes T_i$, with
$f_i \in C(X)_+, T_i \in M_n(\C)_{+}$, so we may conclude that indeed
$\phi_n(P) \ge 0$.

\end{pf}

The corresponding fact for CB maps is false: i.e., if $u
:\ell^{\infty}_n \ra B(K)$ and $u(e_j) = a_j$ (with $\{e_j:1 \le j \le
n\}$ denoting the standard basis of $\ell^{\infty}_n$), then $||u|| (=
\sup \{ ||\sum_k z_j a_j || : |z_j| \leq 1, z_j \in \mbb{C} \} \neq
||u||_{cb}$ - it turns out that
\[
||u||_{cb} = \inf \{ ||\sum_j x_j x_j^*||^{1/2} ||\sum_j 
y_j^*y_j||^{1/2} : a_j = x_j y_j, x_j, y_j \in B(K)\}.
\]

\begin{prop}[Schur multipliers] \index{Schur multipliers}
\begin{enumerate}
\item Let $\varphi_{ij} \in \mbb{C}, 1 \leq i, j \leq n$ and consider
  $u_{\varphi} : M_n \ra M_n$ given by $[a_{ij}] \mapsto
  [a_{ij}\varphi_{ij}]$. Then,
\[
||u_{\varphi}||_{cb} = \inf \{ \sup_i ||x_i||_H\, \sup_j||y_j||_H :
\varphi_{ij} = \langle x_i, y_j \rangle_H, x_i, y_j \in H \}
\]
where the infimum runs over all such possible Hilbert spaces $H$ and
vectors $x_i, y_j$.Also, $u_\varphi$ is CP if and only if
$[\varphi_{ij}] \geq 0$.
\vspace*{1mm}

\item Let $G$ be a discrete group and $C^*_\lambda (G)$ \index{group
  $C^*$-algebra! reduced} denote the norm-closure in $B(\ell_2(G))$ of
  $\{\lambda_g : g \in G \}$, its `reduced group $C^*$-algebra' where
  $\lambda : G \ra B(\ell^2(G))$ is the left regular representation of
  $G$ given by $\lambda_g(\xi_h) = \xi_{gh}$, where $\{\xi_g : g\in G
  \}$ is the canonical orthonormal basis of $\ell_2(G)$. Let $\varphi
  : G \ra \mbb{C}$ be any function. Consider $T_\varphi : span\{
  \lambda_g : g\in G\} \ra B(\ell^2(G))$ given by $T_\varphi (\sum_g
  c_g \lambda_g) = \sum_g c_g \varphi(g) \lambda_g$. Then,
 \[
||T_\varphi||_{cb} = \inf \{ ||x||_{\infty} ||y||_{\infty} : x, y \in
\ell^{\infty}(G , H) \text{ such that } \langle x(t), y(s) \rangle_H =
\varphi (st^{-1} ), \forall\, s, t \in G \}
\]

\end{enumerate}
\end{prop}

\begin{pf}
(1) Suppose $\phi_{ij} = \langle x_i, y_j \rangle_H, x_i, y_j \in
  H$. If $(a_{ij}) \in (M_n)$, we can - by polar decomposition and
  diagonalisation of positive operators - find vectors $u_i,v_j$ in
  some Hilbert space $K$ such that $a_{ij} = \langle u_i, v_j
  \rangle_K$ and $\|(a_{ij})\| = (\max_{i}\|u_i\|)
  (\max_{j}\|v_j\|)$. Then it follows that

\[\phi_{ij}a_{ij} = \langle x_i \otimes u_i, y_j \otimes v_j  \rangle_{H\otimes K}\]

and hence,
\[\|(\phi_{ij}a_{ij}))\| = (\max_{i}\|x_i\|)  (\max_{j}\|y_j\|)(\max_{i}\|u_i\|)  (\max_{j}\|v_j\|)\]

and we see that

\[ ||u_{\varphi}||_{cb} \leq \inf \{ \sup_i ||x_i||_H\, \sup_j||y_j||_H : \varphi_{ij} = \langle x_i, y_j \rangle_H, x_i, y_j \in H \}
\]

Conversely, by the fundamantal factorisation Theorem \ref{ff}, we may
find a representation $\pi:M_n \rightarrow B(K)$ and operators $V,W
\in B(\ell^2_n,K)$ for some Hilbert space $K$, with $\|V\| \cdot \|W\|
= \|u_{\varphi}\|_{cb}$, such that
\[u_{\varphi}(x) = V^*\pi(x)W~.\]

Then,

\begin{eqnarray*}
|\phi_{ij}| &=& |\langle u_\phi(e_{ij}) e_j, e_i \rangle |\\
&=& |\langle \pi(e_{ij})We_j, Ve_i \rangle |\\
&=& |\langle \pi(e_{i1})\pi(e_{1j})We_j, Ve_i \rangle |\\
&=& |\langle \pi(e_{1j})We_j, \pi(e_{1i})Ve_i \rangle |\\
& \leq & \|u_{\varphi}\|_{cb}~.
\end{eqnarray*}

\bigskip (ii) See Theorem 8.3 in \cite{Pis03}.
\end{pf}

It is true, although a bit non-trivial, that $||u_\varphi||_{cb} =
||u_\varphi||$, while, however, $||T_{\varphi}||_{cb} \neq
||T_{\varphi}||$.

\section{Ruan's Theorem and its applications}

\subsection{Ruan's Theorem}

Let us call a complex vector space $E$ an {\bf abstract operator
  space} \index{operator space! abstract} if it comes equipped with a
family of norms $||\cdot ||_n$ on $M_n(E), n\geq 1$,which satisfy the
properties:

\begin{enumerate}

\item $ \left\| \left[ \begin{array}{cc} x & 0\\ 0 & y \end{array} \right] \right\|_{m+n} = \text{max}\ \left( ||x||_n, ||y||_m \right)$ for all $x   \in M_n(E)$, $y \in M_m(E)$;

\item $||a x b ||_n \leq ||a||_{M_{n\times m}} ||x||_m ||b||_{M_{m \times n}}$ for all $a, b^* \in M_{n\times m}$, $x \in M_m(E)$.
\end{enumerate}

In contrast, any vector subspace of a $B(H)$ is called a {\bf concrete operator space} \index{operator space! concrete} (with respect to the norms $\| \cdots \|_n = \|\cdot \|_{M_n(B(H))}$).

The equivalence of these two notions is the content of {\bf Ruan's theorem} established in his PhD Thesis \cite{Rua87}.

\begin{thm}[Ruan]\label{ruan} \index{Theorem! Ruan's}
Any abstract operator space is completely isometrically isomorphic to a concrete operator space; i.e., for any abstract operator space $(E, \{\| \cdot \|_n : n \ge 1\})$, there is a Hilbert space $H$ and a linear map
$J : E \ra B(H)$ such that, for each $n \geq 1$ and $x = [x_{ij}] \in M_n(E)$, $||x||_n = \|[J(x_{ij})]\|_{M_n(B(H))}$.
\end{thm}

The proof of Ruan's theorem bears a strong similarity to the proof of Gelfand Naimark's Theorem that an abstract $C^*$-algebra is isomorphic to a norm-closed (concrete) $*$-subalgebra of $B(H)$ for some Hilbert space $H$.

\begin{pf} [(Ruan's) Theorem \ref{ruan}]
The proof of this theorem mainly depends on proving the following:

\noindent {\bf Claim:} For each $n \geq 1$ and $x \in M_n(E)$ with $||x||_n = 1$, there is a Hilbert space $H_{n, x}$ and an operator $J_{n, x}: E \ra B(H_{n, x})$ satisfying
\begin{enumerate}
\item $||[J_{n,x}(x_{ij})]||_n  = 1$

\item $\forall m$, $\forall y=[y_{ij}] \in M_m(E)$, $||[J_{n,x}(y_{ij})]\||_m \leq ||y||_m$.

\end{enumerate}

Indeed, once we have the above, taking $H := \oplus_{n, x} H_{n, x}$ and $J:=\oplus_{n, x} J_{n, x} : E \ra B(H)$ we observe that $J$ satisfies the properties stated in the theorem.

To prove the claim, fix $n \geq 1$ and an $x \in M_n(E)$ with $\|x\|_n = 1$. Observe first that $(M_k(E), ||\cdot||_k)$ embeds (as the north-east corner, with 0's on the $(n+1)$-th row and column) isometrically in
$(M_{k+1}(E), ||\cdot||_{k+1})$ for all $k \geq 1$ (thanks to condition (2) in the requirements for an abstract operator space). Thus, there exists a $\varphi \in \left( \cup_{k \geq 1} M_k(E) \right)^*$ satisfying $||\varphi || = 1$ and $\varphi (x) = 1$. Then, $\forall\, y \in M_m(E)$ with $||y||_m \leq 1$ and
$\forall\, a , b \in M_{m \times n}, n \ge 1$, we have $|\varphi (a^* y b) | \leq ||a||_{M_{m \times n}} ||b||_{M_{m \times n}}$.

We see, in particular that for each $m, N, n \geq 1$, we have $\forall\, a_j, b_j \in M_{m \times n}$,
$\forall\, y_j \in M_m(E)$ with $||y_j||_m\leq 1$, $1 \leq j \leq N$, we have

\begin{equation}\label{star}
| \sum_j \varphi(a_j^* y_j b_j) | \leq \|\sum_j a_j^*a_j \|^{1/2}\, \|\sum_jb_j^* b_j \|^{1/2}.
\end{equation}

This is true because taking $a = \left[ \begin{array}{c} a_1\\ \vdots \\ a_N \end{array} \right]$, $y$ to be the diagonal matrix $y = y_1 \oplus \cdots \oplus y_N \in M_N(M_m(E))$ and $b = \left[ \begin{array}{c} b_1\\ \vdots \\ b_N \end{array} \right]$, we clearly have $|\varphi (a^*yb)| \leq ||a||_{M_{mN\times n}} \, ||y||_N\, ||b||_{M_{mN \times n}}$, which is precisely the above inequality since $|| [b_1, \ldots, b_N]^T ||_{M_{mN \times n}} = || \sum_j b^*_j b_j||^{1/2}$.

In fact, by some suitable convexity and Hahn Banach separation arguments, the inequality in (\ref{star}) can be improved (see \cite{Rua87}) to:
There exist states $f_1, f_2$ on $B(\ell^2)$ such that
\begin{equation}
| \sum_j \varphi(a_j^* y_j b_j) | \leq (\sum_j f_1(a_j^*a_j))^{1/2} (\sum_j f_2(b_j^*b_j))^{1/2}
\end{equation}
for all $N$, $a_j, b_j $, $ y_j $ as above.

In particular, $\forall\, a, b \in M_{1 \times n}, e \in E \mbox{ with }~\|e\| \le 1$, we have
\[ |\varphi (a^*[e]b)| \leq f_1(a^*a)^{1/2}f_2(b^*b)^{1/2} = ||\hat{a}||_{H_1} ||\hat{b}||_{H_2},
\]

where $H_i := L^2(B(\ell^2), f_i), i = 1, 2$ and $\hat{a}$ represents $a$ as an element of $H_1$ and likewise $\hat{b}$.

Consider $v: E \ra B(H_2, H_1)$ given by $\langle v(e)(\hat{b}), \hat{a}\rangle_{H_1} =\varphi(a^*[e]b)$. Clearly $||v(e)|| \leq ||e||$ for all $e \in E$ and hence $||v|| \leq 1$. Also, for $y =[y_{ij}] \in M_m(E)$,

\begin{eqnarray*}
||v_m(y)|| &  = & \sup \{| \sum_j \langle v(y_{ij})\hat{b}_j, \hat{a}_j | \rangle : \sum_j||\hat{b}_j||^2 \leq 1, \sum_j||\hat{a}_j||^2 \leq | 1\}\\
&=& \sup \{   |\varphi \left([a_1^*, \ldots, a_m]^*\, y\,  [b_1, \ldots, b_m]^T  \right)|  : \sum_j||\hat{b}_j||^2 \leq 1, \sum_j||\hat{a}_j||^2 | \leq 1 \}\\
& \leq & ||y||_m.
\end{eqnarray*}

Let $H_{n, x}= H_1 \oplus H_2$ and $J_{n,x}:E \ra B(H_{n,x})$ be given by $J_{n,x}(e) = \left[\begin{array}{cc} 0 & v(e) \\ 0 & 0 \end{array}\right]$. We find that\[ 1 = \phi(x) = \phi([x_{ij}]) = \sum_{ij} \langle v(x_{ij}) \widehat{e_{1j}}, \widehat{e_{1i}} \rangle = \langle v_n(x) \left[ \begin{array}{c}\widehat{e_{11}}\\ \ddots \\ \widehat{e_{1n}} \end{array} \right], \left[ \begin{array}{c}\widehat{e_{11}}\\ \ddots \\ \widehat{e_{1n}} \end{array} \right] \rangle ,\]
where $\{e_{ij}\}$ is the system of matrix units in $M_n$. Since $\|\widehat{e_{1j}}\|_{H_i}^2 = \sum_j f_i(e_{jj}) \le f_i(1) = 1$, we find that
\[1 = !\phi(x)| \leq \|v_n(x)\| \cdot \left( \sum_j \| \widehat{e_{1j}} \|_{H_1}^2 \right)^{1/2} \cdot \left( \sum_j \| \widehat{e_{1j}} \|_{H_2}^2 \right)^{1/2} \leq \|v_n(x)\|.\]
Thus, $\|[J_{n,x}(x_{ij}]\| = 1$ and the proof is complete.

\end{pf}

\subsection{Some applications and some basic facts.}

$\bullet$ $CB(E, F).$ Let $E$ and $F$ be operator spaces and $G:=CB(E, F)$. For $[x_{ij}] = x \in M_n(G)$, let $||x||_n :=||\tilde{x} : E \ra M_n(F)||_{CB(E,M_n(F))}$, where $\tilde{x}(e):= [x_{ij}(e)]$, $e \in E$. Routine checking shows that the above sequence of norms satisfies  Ruan's axioms and, hence, by Theorem \ref{ruan}, $CB(E,F)$ admits an operator space structure.

$\bullet$ {\em Dual operator space. } \index{operator space! dual} In particular, the dual space $E^*=CB(E, \mbb{C}) = B(E, \mbb{C})$ also inherits an operator space structure. We now see that this operator space structure on $E^*$ is
the appropriate one, in the sense that many properties of Banach dual space structure carry over to this theory.

$\bullet$ {\em Adjoint operator.} \index{ adjoint! CB map} For $u \in CB(E, F)$, its usual adjoint $u^*$ is also a CB map from $F^* \ra E^*$ with $||u^*||_{cb} = ||u||_{cb}$.

$\bullet$ {\em Quotient spaces.} Let $E_2 \subset E_1 \subset B(H)$ be operator spaces. For $x = [x_{ij}] \in M_n(E_1/E_2)$, set $||x||_n = ||q_n(\hat{x})||_{M_n(E_1)/M_n(E_2)}$, where $ x_{ij} = q\hat{x}_{ij}$ for $\hat{x}_{ij} \in E_1$ and $\hat{x} =[\hat{x}_{ij}]$ and $q: E_1 \ra E_1/E_2$ is the canonical quotient map.

The above sequence of norms are seen to be well-defined and to satisfy Ruan's axioms and thus equip the quotient space $E_1/E_2$ with the structure of an operator space.

Note there is also an obvious operator space structure on $E_1/E_2$ via the embedding $E_1/E_2 \subset B(H)/E_2$; but this is not consistent with other properties.

\bigskip \noindent {\bf Analogy with Banach space properties.}

$\bullet$ With $E_1$ and $E_2$ as above,

\[ \left(E_1/E_2 \right)^* \simeq E_2^{\perp} \ \text{and}\ E_2^* \simeq E_1^*/E_2^{\perp} \]
as operator spaces, where of course $E_2^{\perp}\subset E_1^*$ is defined as $E_2^\perp = \{f \in E_1^* : E_2 \subset ~ker(f)\}$.

$\bullet$ The row and column operator spaces are dual to each other, i.e., $R^*\simeq C$ and $C^* \simeq R$ as operator spaces.

$\bullet$ The canonical embedding $E \subset E^{**}$ is a complete isometry - see \cite{BP91}.

\subsection{$\min$ and $\max$ operator space structures on a Banach space}

Given a Banach space $X$, two operator space structures attract
special attention, which are described by the adjectives $max$ and
$min$. These structures are characterised as follows.

\begin{prop}
There exist operator space structures $\max(X)$ \index{$\max(X)$} and $\min(X)$ \index{$\min(X)$} on a
Banach space $X$ such that $\max(X)$ and $\min(X)$ are isometric to
$X$ and for every operator space $Z$,
\begin{enumerate}

\item $CB(\max(X),\, Z) = B(X, Z)$.

\item  $CB(Z,\, \min(X)) = B(Z, X)$.

\end{enumerate}

Moreover, $(1)$ and $(2)$ characterise $\max(X)$ and $\min(X)$, respectively.
\end{prop}

\begin{pf}

We discuss only the existence of $\max(X)$ and its property $(1)$. For $x\in M_n(X)$, set

\[ ||x||_{M_n(\max(X))} = \inf \{ ||a||_{M_{n \times N}} \sup_{j \leq N} ||D_j||_X ||b||_{M_{N \times n}}: a^*, b \in M_{N \times n}, D_j \in X, N \geq 1\}
\]

where the infimum runs over all possible decompositions $x = a \left[ \begin{array}{cccc}
D_1& 0 & \cdots& 0 \\ 0 & D_2 & \cdots &
    0\\ \vdots & \vdots & \ddots   & \vdots \\ 0 & 0 & \cdots & D_N  \end{array} \right]b$.

This sequence of norms satisfies Ruan's axioms and hence we have an operator space structure $\max(X)$ on $X$. Now, let $u \in B(X, Z)$ for an operator space $Z$ and $x \in M_n(\max(X))$. Assume that $||x||_{M_n(\max(X))} < 1$. Then there exists $N \geq 1$, $a^*, b \in M_{N \times n}, D_j \in X$, $1 \leq j \leq N$ such that $x = a \left[ \begin{array}{cccc}D_1& 0 & \cdots& 0 \\ 0 & D_2 & \cdots &
    0\\ \vdots & \vdots & \ddots & \vdots  \\ 0 & 0 & \cdots & D_N
 \end{array} \right]b$ -- and $||a||_{M_{n \times N}} \sup_{j \leq N} ||D_j||_X ||b||_{M_{N \times n}} < 1$.  Then note that $u_n(x) = a \left[ \begin{array}{cccc} u(D_1)& 0 & \cdots& 0 \\ 0 & (D_2) & \cdots
    & 0\\ \vdots & \vdots & \ddots & \vdots \\ 0 & 0 & \cdots & u(D_N)
 \end{array} \right]b$
-- and that
\[ ||u_n(x)||_{M_n(\max(X))} \leq ||u||\, ||a||_{M_{n \times N}} \sup_{j \leq N} ||D_j||_X ||b||_{M_{N \times n}} \le ||u||.\]

\end{pf}

$\bullet$ $(\min(X))^* \simeq \max(X^*)$ and $(\max(X))^* \simeq \min(X^*)$ as operator spaces - see \cite{BP91,Ble92}.

\section{Tensor products of Operator spaces}

We will be mainly interested in the injective and projective tensor products of operator spaces. \index{tensor product! operator spaces}

\subsection{Injective tensor product}

Let $E_i \subset B(H_i)$, $i = 1, 2$ be operator spaces. Then we have a natural embedding $E_1 \otimes E_2 \subset B(H_1 \otimes H_2)$ and the {\em minimal} tensor product \index{tensor product! minimal} of $E_1$ and $E_2$ is the space $E_1 \otimes_{\min} E_2 := \overline{E_1 \otimes E_2} \subset B(H_1 \otimes H_2)$.

The $\min$ tensor product is independent of the embeddings $E_i \subset B(H_i)$, $i = 1, 2$; it depends only upon the sequence of norms on $E_i$, $i = 1,2$. We see this as follows:
For simplicity, assume $H_1 = H_2 = \ell^2$ and let $t = \sum_{k=1}^r a_k \otimes b_k \in E_1 \otimes E_2$. Consider the natural embeddings $\ell_n^2 \subset \ell^2$, $n \geq 1$ - with respect to some choice
$\{\xi_n: n \ge 1\}$ of orthonormal basis for $\ell^2$ and the corresponding orthogonal projections $P_n: H_1 \otimes H_2 \ra \ell^2_n \otimes H_2$. Then $\overline{\cup_n \ell_n^2 \otimes H_2} = H_1 \otimes H_2$ and so $\|x\| = \sup_n \|P_nxP_n\| ~\forall x \in B(H_1 \otimes H_2)$; and hence,

\[ ||t||_{\min} = \sup_n ||P_n\,  t_{|_{\ell^2_n \otimes H_2}} ||_{B(\ell^2_n \otimes H_2)}.
\]

Suppose $\langle a_k \xi_j, \xi_i \rangle = a_k(i,j)$. Then it is not hard to see that $P_n\,  t_{|_{\ell^2_n \otimes H_2}}$ may be identified with the matrix $t_n \in M_n(E_2)$ given by  $t_n(i, j) = \sum_k (a_k)_{ij} b_k$. This shows that $||t||_{\min} = sup_n||t_n||_{M_n(E_2)}$, and hence that the operator space structure of $E_1 \otimes_{\min} E_2$ depends only on the operator space structure of $E_2$ and not on the embedding $E_2 \subset B(H_2)$. In an entirely similar manner, it can be seen that the operator space structure of $E_1 \otimes_{\min} E_2$ depends only on the operator space structure of $E_1$ and not on the embedding $E_1 \subset B(H_1)$.

More generally, the same holds for $||[t_{ij}]||_{M_n(B(H_1 \otimes H_2)}$ and hence $||\cdot||_{M_n(E_1 \otimes E_2)}$ depends only on the operator space structures of the $E_j$'s and not on the embeddings
$E_j \subset B(H_j)$.

\begin{prop}\label{uG}
Let $u:E\ra F$ be a CB map between two operator spaces and $G$ be any operator space. Let $u_G : G \otimes_{\min} E \ra G \otimes_{\min} F$ be the extension of the map $id_G\otimes u$. Then $||u||_{cb} = \sup_G ||u_G||.$
\end{prop}

\begin{pf}

For this, first observe that if $G = M_n$, then $M_n = B(\ell^2_n)\Rightarrow M_n \otimes_{min} E = M_n(E)$ and that $\|u\|_{M_n} = \|u_n\|$. (In fact the identification $M_n \otimes_{min} E = M_n(E)$ is not only isometric, it is even completely isometric.) It follows - by allowing $G$ to range over $\{M_n: n \in \N\}$ - that \[\sup\{\|u_G\|: G \mbox{ an operator space }\} \ge \|u\|_{cb}~.\]

The same argument that yielded the conclusion $M_n \otimes_{min} E = M_n(E)$ is also seen to show that if $E \subset B(K)$, then $B(H)  \otimes_{min} E = \overline{(B(H) \otimes E)} \subset B(H \otimes K)$ (as a Banach space); assuming, as before, that $H = \ell^2$, and choosing the projections $P_n = P_{\ell^2_n}\otimes K$, we may easily deduce from the strong convergence of $P_n$ to $\Id_{H \otimes K}$ that $\|u_{B(H)}\| = \|u\|_{cb}$. Finally, we may conclude that if $G \subset B(H)$, then $G \otimes_{min} E$ sits isometrically as a subspace of $B(H) \otimes_{min} E$, and that
\begin{eqnarray*}
\|u_G\| &=& \|u_{B(H)}|_{G \otimes_{min} E}\|\\
&\le& \|u_{B(H)}\|\\
&=& \|u\|_{cb}~,
\end{eqnarray*}
thereby establishing the desired equality, as $G$ was arbitrary.
\end{pf}

\begin{rem}\label{mintp}

\begin{enumerate}

\item  For two operator spaces $E$ and $F$ with $F$ being finite dimensional, we have $CB(E, F) \simeq E^* \otimes_{\min} F$ as operator spaces. This is analogous to the fact that $B(E, F) \simeq E^*
\otimes^{\vee} F$ isometrically for two Banach spaces $E$ and $F$.

\item $\min$ tensor product is an associative and commutative tensor product.



\item For operator spaces $E_i$, $F_i$, $i = 1,2$ and $u_j \in CB(E_j, F_j)$, the map $u_1 \otimes u_2$ extends to a CB map $u_1 \otimes u_2 : E_1 \otimes_{\min} E_2 \ra F_1 \otimes_{\min} F_2$ such
that $||u_1 \otimes u_2||_{cb} = ||u_1||_{cb} ||u_2||_{cb}$. Indeed, one easily shows that the composition of the extensions $id_{E_1} \otimes u_2 \in CB(E_1 \otimes_{min} E_2, E_1 \otimes_{min} F_2)$ and $u_1 \otimes id_{F_2} \in CB( E_1 \otimes_{min} F_2, E_2 \otimes_{min} F_2)$ give us a CB map $u_1 \otimes u_2 = (u_1 \otimes id_{F_2}) (id_{E_1} \otimes u_2)$ with required properties.

\item Another nice consequence of Propostion \ref{uG} and above observation about tensor products of CB maps is that the $\min$ tensor product is {\em injective}, \index{tensor product! injective} i.e., if $E_2\subset E_1$ is an inclusion of operator spaces and $G$ is any operator space, then the inclusion $E_2 \otimes_{min} G \subset E_1 \otimes_{min}  G$ completely isometrically.
\end{enumerate}
\end{rem}

\subsection{Projective tensor product}

Let $E_i \subset B(H_i)$ be operator spaces. For $t = [t_{ij}] \in M_n(E_1 \otimes E_2)$, set
\[ ||t||_n = \inf \{||a||_{M_{n \times N^2}} ||x^1||_{M_N(E_1)} ||x^2||_{M_N(E_2)} ||b||_{M_{N^2 \times n}} : a^*,b \in M_{N^2 \times n} , x^i \in M_n(E_i), i = 1,2\},\]
where the infimum runs over all possible decompositions of the form $t = a (x^1 \otimes x^2) b$. This sequence of norms satisfies Ruan's axioms and we obatain an operator space structure on $E_1 \otimes E_2$, which after completion is denoted $E_1 \otimes^{\widehat{}} E_2$ and is called the {\em projective} tensor product \index{tensor product! projective} of the operator spaces $E_1$ and $E_2$.

$\bullet$ For operator spaces $E_1$ and $E_2$, we have operator space isomorphisms $$(E_1\otimes^{\widehat{}} E_2)^* \simeq CB(E_1, E_2^*) \simeq CB(E_2, E_1^*).$$

\begin{rem}

\begin{enumerate}

\item Projective tensor product of operator spaces is not injective.

\item Injective tensor product of operator spaces is not projective, in the sense defined below.

\end{enumerate}
\end{rem}

\begin{defn}
Let $E$ and $F$ be operator spaces and $u \in CB(E, F)$. Then $u$ induces a canonical map $\tilde{u} : E{/\rm ker}\,u \ra F$ and $u$ is said to be a {\bf complete metric surjection} if $u$ is surjective and $\tilde{u}$ is completely isometric.
\end{defn}

\begin{exer}
$u \in CB(E, F)$ is a complete metric surjection if and only if $u^*:F^* \ra E^*$ is a completely isometric embedding.
\end{exer}

$\bullet$ $\otimes^{\widehat{}}$ is ``projective'' in the following sense.  For any two complete metric surjections $u_j : E_j \ra F_j$, $j = 1, 2$, the tensor map $u_1 \otimes u_2: E_1\otimes^{\widehat{}} E_2 \ra F_1\otimes^{\widehat{}} F_2 $ is also a complete metric surjection.

\vspace*{1mm}

Note that for any two Hilbert spaces $H$ and $K$, $B(H, K)$ has a canonical operator space structure via the embedding $B(H, K) \ni x \mapsto \left(\begin{array}{cc} 0 & 0 \\ x & 0 \end{array} \right) \in B(H \oplus K)$.

$\bullet$ Thus, given any Hilbert space $H$, we have two new operator spaces, namely,
\[ H_c:= B(\mbb{C}, H)\ \text{ and } H_r:= B(\overline{H}, \mbb{C}).
\]

$\bullet$ $(H_c)^* \simeq \overline{H}_r$ as operator spaces.

$\bullet$ If $H =\ell_2$, then $H_c = C$ and $H_r = R$.

$\bullet$ For any two Hilbert spaces $H$ and $K$, we have the following operator space isomorphisms:

\[ (\overline{K} \widehat{\otimes} H)^* \simeq B(H, K) \simeq CB(H_c, K_c) \simeq (K_c^* \otimes^{\widehat{}} H_c)^*
\]
where $\widehat{\otimes}$ is the Banach space projective tensor product.  Note that $\overline{K} \widehat{\otimes} H$ as well as $K_c^* \otimes^{\widehat{}} H_c$ are isometrically isomorphic to the space $S_1(H, K)$ of trace class operators.

\subsection{General remarks}\hspace*{1mm}

$\bullet$ Recall the $\gamma$-norm in the proof of Fundamental Factorization Theorem.  In fact, we have an isometric isomorphism
\[ (\overline{K} \otimes E \otimes H)_\gamma \simeq S_1(K, H) \otimes^{\widehat{}} E.
\]

$\bullet$ For any measure space $(X, \mu)$ (unlike in the next remark) the Banach space projective tensor product is injective, at least, if one of the tensor factors is $L^1(X, \mu)$ and the other is an operator space. This is true because for any inclusion of operator spaces $E \subset E_1$, we have
$$L^1(X, \mu) \widehat{\otimes} E = L^1(\mu, E) \subset L^1(\mu, E_1 ) = L^1(X, \mu) \widehat{\otimes} E_1.$$

$\bullet$ For a von Neumann algebra $M$, its predual $M_\ast$, usually called a {\em non-commutative $L^1$-space}, has a natural operator space structure via the embedding $M_\ast \subset M^* (=(M_*)^{\ast \ast}) $.  For an inclusion of operator spaces $E \subset E_1$, only when $M$ is injective, there is an isometric embedding $M_\ast \otimes^{\widehat{}} E \subset M_\ast \otimes^{\widehat{}} E_1$.

\subsection*{A passing remark on the Haagerup tensor product}

Apart from the above two tensor products of operator spaces, there is another extremely important tensor product of operator spaces, namely, the {\em Haagerup} tensor \index{tensor product! Haagerup} product, usually detoned $E\otimes_h F$. Unlike the above two tensor products, the Haagerup tensor product has no analogy in Banach space theory. It is known that the Haagerup tensor product of operator spaces is associative, injective and projective but it is not commutative. Also, $ \overline{K}_r \otimes_h E \otimes_h H_c \simeq (\overline{K} \otimes E \otimes H)_\gamma$ as Banach spaces. We will not get into the details in these notes.

\section{Tensor products of $C^*$-algebras}

\subsection{$\min$ and $\max$ tensor products of $C^*$-algebras}\hspace*{1mm}

\index{tensor product! $C^*$-algebras}
Let $A$ and $B$ be unital $C^*$-algebras. Then $A\otimes_{\min}B$ is canonically a $C^*$-algebra. For each $t = \sum_k a_k \otimes b_k \in A \otimes B$, define \index{tensor product! minimal}
\[
||t||_{\max}= \sup \{ ||\sum_k \pi(a_k) \sigma(b_k)|| : \pi:A \ra B(H), \sigma : B \ra B(H)\}
\]
where the infimum runs over all representations $\pi$ and $\sigma$
with commuting ranges.  Then the completion $A \otimes_{\max} B$ of $A
\otimes B$ with \index{tensor product! maximal} respect to this norm
makes it into a $C^*$-algebra. In general, we have $||t||_{\min} \leq
||t||_{\max}$ for all $t\in A \otimes B$.

\begin{defn}
For $C^*$-algebras $A$ and $B$, the pair $(A, B)$ is called a {\bf nuclear pair} \index{nuclear pair}  if $||t||_{\min} = ||t||_{\max}$ for all $t \in A \otimes B$.
\end{defn}

\begin{defn}
A $C^*$-algebra $A$ is is said to be nuclear if $(A, B)$ is a nuclear pair for all $C^*$-algebras $B$.
\end{defn}

\subsection{Kirchberg's Theorem}

\begin{thm}\cite{Kir93}\label{kirchberg}
Let $ 1 \leq n \leq \infty$. If $A = A_n=C^*(\mbb{F}_n)$, the full $C^*$-algebra on the free group with $n$-generators and $B=B(H)$, then $(A, B)$ is a nuclear pair. \index{Theorem! Kirchberg's}
\end{thm}

The above result has great significance as both $C^*(\mbb{F}_{n})$ and $B(H)$ are universal objects in the sense that every $C^*$-algebra is a quotient of $C^*(\mbb{F}_n)$ and a $*$-subalgebra of $B(H)$ for suitable choices of $n$ and $H$.

A proof of this remarkable theorem involves a fair bit of analysis.

\begin{lem}\label{lemma1}
Let $A\subset B(H)$ be a $C^*$-algebra and $\{ U_i : i\in I\} \subset \mcal{U}(H)$ be a family of unitary operators on $H$ such that $1 \in E :=\overline{\text{span}}\{U_i : i \in I\}$ and $A = C^*(E)$. Let $u : E \ra B(K)$ be a unital c.b map satisfying $u(U_i) \in \mcal{U}(K)$ for all $i\in I$. If $||u||_{cb} \leq 1$, then there exists a $\ast$-representation $\sigma : A \ra B(K)$ such that $\sigma_{|_E} = u$.
\end{lem}

\begin{pf}

Since $u$ is unital and CB with $||u||_{cb} \leq 1$, there is a Hilbert space $\hat{H}$, a $\ast$-representation $\pi: A \ra B(\hat{H})$, a subspace $K \subset \hat{H}$ such that $u(a) = P_K \pi(a)_{|_K}$ for all $a \in
E$, where $P_K: \hat{H} \ra K$ is the orthogonal projection.

Suppose $\pi(U_i) = \left( \begin{array}{cc} a_i & b_i \\ c_i &   d_i \end{array} \right)$ with respect to $\hat{H} = K \oplus K^\perp$, so that $u(U_i)=a_i$. Unitarity of $U_i$ implies that $a_i^*a_i + b_i^*b_i = a_ia_i^* + c_ic_i^* = 1$; while the assumed unitarity of $a_i$ is seen to imply that $b_i = c_i = 0$. Therefore, $\pi(U_i) = \left( \begin{array}{cc} a_i & 0 \\ 0 &   d_i \end{array} \right)$, and we see that $K$ is an invariant subspace for $\pi(U_i)$ as well as for $ \pi(U_i^*)$ for all $i \in I$. Thus, $K$ is invariant under
$C^*(\{\pi(U_i) : i \in I\})$. Finally $\sigma : A \ra B(K)$ defined by $\sigma(a) = \pi(a)_{|_K}$, $a\in A$ is seen to be a $\ast$-representation which extends $u$.

\end{pf}

\begin{lem}\label{Uamin}
Let $U_0 = 1, U_1, \ldots, U_n$ be unitary generators of $A_n := C^*(\mbb{F}_n)$. If $\{a_j : 0 \leq j \leq n\} \subset B(H)$ is such that $\alpha:= ||\sum_0^n U_j \otimes a_j||_{A_n \otimes_{\min} B(H)} \leq 1$, then there exist $b_j, c_j \in B(H)$ such that $a_j = b_j^* c_j$ and $ ||\sum_j b_j^* b_j ||\leq 1$  , $||\sum_j c_j^* c_j ||\leq 1$.
\end{lem}

\begin{pf}
Consider $u : \ell^\infty_{n+1} \ra B(H)$ given by $u(e_j) = a_j$, $0 \leq j \leq n$, where $\{e_j : 0 \leq j \leq n\}$ is the standard basis of $\ell^\infty_{n+1}$.

\bigskip

{\bf Assertion:} $u$ is CB with $||u||_{cb} \leq \alpha$.
\bigskip

Under the identification $M_m(A) \ni [a_{ij}] \mapsto \sum_{ij}e_{ij}\otimes a_{ij} \in M_m \otimes A$ for any algebra $A$, we see that $u_m : M_m(\ell^\infty_{n+1}) \ra M_m(B(H)) $ is given by $u_m(\sum_{j=0}^n x_j \otimes e_j) = \sum_{j=0}^m x_j \otimes a_j$ for $x_j \in M_m$.

It follows that
\begin{eqnarray*}
||u_m|| &=& \sup\{ ||\sum_{j = 0}^n x_j \otimes a_j||_{M_m \otimes_{\min} B(H)} : \max_j ||x_j||_{M_m} \leq 1\}\\
&\leq & \sup\{ ||\sum_{j=0}^n V_j \otimes a_j|| : V_j \in M_m \text{ is unitary}\}~,
\end{eqnarray*}
where the first equality is because $\|\sum_j x_j \otimes e_j\|_{ M_m(\ell_\infty^{n+1})} = \max_j ||x_j||_{M_m}$, and the second equality is the consequence of the fact that the extreme points of the unit ball of $M_m(\C)$ are precisely the unitary matrices.

By definition of the \index{group $C^*$-algebra! full} full free group
$C^*$-algebra, there exists a $\ast$-homomorphism $\sigma:A_n
\rightarrow M_m$ such that $\sigma(U_j) = V_j, \forall 0 \leq j \leq
n$. Since any $\ast$-homomorphism is a complete contraction, it
follows from Remark \ref{mintp} (3) that $\sigma \otimes id_{B(H)}:
A_n \otimes_{min} B(H) \rightarrow M_n \otimes_{min} B(H)$ is also a
complete contraction, and consequently, we find that indeed
$||u||_{cb} =\sup ||u_m|| \leq \alpha \leq 1$. Thus $u$ is a complete
contraction, and we may deduce from the fundamental factorisation
theorem that there exists a representation $\pi : \ell^\infty_{n+1}
\rightarrow B(K)$ and contractions $V, W: H \rightarrow K$ such that
$(a_j =) u(e_j) = V^* \pi(e_j) W ~\forall j$. Pick some isometry $S:K
\rightarrow H$ and define $b_j = S\pi(e_j)V, c_j = S\pi(e_j)W$.

Then, notice that indeed $b_j^*c_j = V^* \pi(e_j) S^*S \pi(e_j) W = a_j$, and as $V$ and $W$ are contractions, we find that
\[\sum_j b_j^*b_j = V^*\left(\sum_j\pi(e_j)\right) V = V^*V \leq 1\]
and
\[\sum_j c_jc_j^* = \sum_j S \pi(e_j) W^*W \pi(e_j) S^* \leq \sum_j S\pi(e_j)S^* = SS^* \leq 1~,\]

thereby completing the proof of the Lemma.
\end{pf}

\bigskip

We shall need the following version of the Cauchy-Schwarz inequality. We say nothing about its  proof but it is easy and   may be found, for instance, in any treatment of Hilbert $C^*$-algebras.

\begin{lem}\label{opcs}
For elements $b_j, c_j$ of any $C^*$-algebra, we have
\[
\|\sum_j b_j^*c_j\|^2 \leq \|\sum_j b_j^*b_j\|  \cdot \|\sum_j c_j^*c_j\|
\]
\end{lem}

\vspace*{2mm}

\noindent {\em Proof of Theorem \ref{kirchberg}:} We have $A = C^*(\mbb{F}_n)$ and $B = B(H)$. We need to show $|| A \otimes_{\min} B \ra A \otimes_{\max} B|| \leq 1$.  Let $E = span\{ U_j \otimes B(H) : 0 \leq j \leq n \} = span\{ U_j \otimes V : V \in \mcal{U}(H)\} \subset A \otimes_{\min} B$. Note that $1 \otimes 1 \in E$. Consider $u = id_E : E \ra E \subset A \otimes_{\max} B$.

We shall be done if we show that $||u|| \leq 1$. Suppose $t = \sum_j U_j \otimes a_j \in E$ with $||t||_{\min} \leq 1$. Then, pick $b_j, c_j \in B(H) $  as in Lemma \ref{Uamin} such that $a_j = b_j^*c_j$. For any two representations $\pi:A \ra B(K)$ and $\sigma : B \ra B(K)$ with commuting ranges, consider $\pi \cdot \sigma : E \ra B(K)$ given by $(\pi \cdot \sigma)(t) = \sum_j \pi(U_j)\sigma(a_j) = \sum_j \sigma(b_j^*) \pi(U_j) \sigma(c_j)$.

Then, by Lemma \ref{opcs}, we find that $|| (\pi \cdot \sigma)(t) || \leq ||\sum_j \sigma(b_j^*) \sigma(b_j) ||^{1/2} ||\sum_j \sigma(c_j^*) \sigma( c_j) ||^{1/2} \leq 1$. As $\pi$ and $\sigma$ were arbitrary, we find thus that $||t||_{\max} \leq 1$.  So, we indeed have $||u|| \leq 1$, thus completing the proof of Kirchberg's theorem.

\hfill$\Box$


\chapter{Entanglement in Bipartite Quantum States}

\bigskip \noindent
{\bf Chapter Abstract:}

 \bigskip \noindent
{\it Entanglement is a central feature of quantum mechanical systems that lies at the heart of most quantum information processing tasks. In this chapter, based on lectures by K R Parthasarathy, we examine the important question of quantifying entanglement in bipartite quantum systems. We begin with a brief review of the mathematical framework of quantum states and probabilities and formally define the notion of entanglement in composite Hilbert spaces via the Schmidt number and Schmidt rank. For pure bipartite states in a finite-dimensional Hilbert space, we study the question of how many subspaces exist in which all unit vectors are of a certain minimal Schmidt number k. For mixed states, we describe the role of k-positive maps in identifying states whose Schmidt rank exceeds k. Finally, we compute the Schmidt rank of an interesting class of states called generalized Werner states.}

\section{Quantum States, Observables and Probabilities}\label{sec:qstates}
We begin with a few introductory remarks on notation. We consider finite-level quantum systems labeled by $A, B, C$ etcetera.  Let $\cH_{A}, \cH_{B}, \cH_{C}$, denote the finite-dimensional Hilbert spaces associated with them. The elements of $\cH$  are called {\it ket vectors} denoted as $|u\rangle, |v\rangle$, while the elements of the dual $\cH^{*}$ are called {\it bra vectors} denoted by $\langle u |, \langle v |$. The {\it bra-ket} $\langle u | v\rangle$ denotes the sesquilinear form, which is linear in $|v\rangle$ and anti-linear in $|u\rangle$. If $\cH$ is a $n$-dimensional Hilbert space, $|u\rangle$ can be represented as a column vector and $\langle u|$ as a row vector, with complex entries.
\[|u\rangle  = \left[\begin{array}{c} z_{1} \\ z_{2} \\ \vdots \\ z_{n} \end{array}\right]; \;  \langle u | = \left[\begin{array}{cccc} \bar{z_{1}} & \bar{z_{2}} & \ldots & \bar{z_{n}} \end{array} \right]\]
The bra vector $\langle u| \equiv |u\rangle^{\dagger}$ is thus the {\it adjoint} of the ket vector. The adjoint of any operator $X$ is denoted as $X^{\dagger}$. Also note that for vectors $|u\rangle, |v\rangle$ and operator $X$, $\langle v | X u \rangle = \langle v| X|u\rangle = \langle X^{\dagger} v | u\rangle$. If $|u\rangle \in \cH_{1}$ and $|v\rangle \in \cH_{2}$ are vectors in two different Hilbert spaces $\cH_{1}$ and $\cH_{2}$ respectively, $|u\rangle\langle v|$ is an operator from $\cH_{2}\rightarrow\cH_{1}$. That is, for any $|w\rangle \in \cH_{2}$, $(|u\rangle\langle v|)|w\rangle = \langle v|w\rangle |u\rangle$.

Given a Hilbert space $\cH$ associated with some physical system, let $\cP(\cH)$ denote the set of all orthogonal projectors in $\cH$,  $\cB(\cH)$ denote the set of all bounded linear operators in $\cH$, $\cO(\cH)$ be the set of self-adjoint operators in $\cH$ and $\cS(\cH)$ denote the set of positive operators of unit trace in $\cH$. The elements of $\cS(\cH)$ are the allowed {\it states} of the system, the elements of $\cO(\cH)$ are the {\it observables}, and, the elements of $\cP(\cH)$ are {\it events}.

In classical probability theory, recall that a probability distribution $p$ over a finite set $\Omega$ is defined by assigning a probability $p(w) \geq 0$ to each $w\in \Omega$ such that $\sum_{w}p(w) = 1$.
\begin{defn}[Expectation Value of a Classical Random Variable]
The expectation value of any real-valued random variable $f(w)$ under the distribution $p$ is evaluated as $\mathbb{E}[f] = \sum_{w}f(w)p(w)$.
\end{defn}
The sample space $\Omega$ is simply the finite set of all possible outcomes, and an \emph{event} is any subset $E \subset \Omega$. The probability of event $E$ is then given by $P(E) = \sum_{w \in E} p(w)$.

In quantum probability theory, an event $E$ is a projection. For any $\rho \in \cS(\cH)$, $\tr[E \rho]$ gives the probability of $E$ in the state $\rho$. The states $\rho$ thus plays the role of the probability distribution in the quantum setting. Just as the set of all probability distributions forms a convex set, the set of all states $\cS(\cH)$ is a convex set. The extreme points of $\cS(\cH)$ are one-dimensional projections, which are called {\it pure states}. For any unit vector $|v\rangle$, the operator $|v\rangle\langle v|$ is a one-dimensional projection. It is, in fact, the projection on the subspace (one-dimensional ray) generated by $|v\rangle$. Up to multiplication by a scalar of unit modulus, the unit vector $|v\rangle$ is uniquely determined by the pure state $|v\rangle\langle v|$ and hence we often refer to $|v\rangle$ itself as the pure state. Note that $|v\rangle\langle v|$ is an operator of unit trace\footnote{In the bra-ket notation, $\tr[|u\rangle\langle v|] \equiv \langle v | u\rangle$}, usually
called the {\it density operator} corresponding to the state $|v\rangle$.

Just as quantum states are the analogues of classical probability distributions, observables correspond to random variables. We next define the concept of {\bf expectation value} of an observable. Since observables are self-adjoint operators, by the spectral theorem, every observable $X$ has a spectrum $\sigma(X)$ and a spectral decomposition:
\[X = \sum_{\lambda \in \sigma(X)} \lambda E_{\lambda},\]
where, $E_{\lambda}$ is the spectral projection\index{spectral projection} for $\lambda$. $E_{\lambda}$ is the event that the observable $X$ takes the value $\lambda$. The {\it values} of any observable are simply the points of its spectrum. The probability that the observable $X$ assumes the value $\lambda$ in state $\rho$ is $\tr[\rho E_{\lambda}]$. This naturally leads to the following definition.
\begin{defn}[Expectation Value of an Observable]\index{expectation value}
The expectation value of observable $X$ in state $\rho$ is given by
\begin{equation}
 \mathbb{E}_{\rho}(X) = \sum_{\lambda\in\sigma(X)}\lambda \tr[\rho E_{\lambda}] = \tr[\rho X].
\end{equation}
\end{defn}
The set of projections $\{E_{\lambda}\}$ constitutes a resolution of identity, that is, $\sum_{\lambda}E_{\lambda} = I$.

In a similar fashion, we can compute expectation of any real-valued function of observable $X$. If $f$ is a real-valued function, $f(X)$ is also a self-adjoint operator and it follows from the spectral theorem that, $f(X) = \sum_{\lambda\in \sigma(X)}f(\lambda) E_{\lambda}$. Therefore, $\mathbb{E}_{\rho}(f(X)) = \tr[\rho f(X)]$. In particular, the variance of an observable $X$ in state $\rho$ is given by
\[{\rm Var}_{\rho}X = \tr[\rho X^{2}] - (\tr[\rho X])^{2}.  \]
If $m = \tr[\rho X]$ denotes the mean value of observable $X$, the variance can also be written as ${\rm Var}_{\rho}X  = \tr[\rho (X - m I)^{2}]$. It is then easy to see that the variance of $X$ vanishes in state $\rho$ iff $\rho$ is supported on the eigenspace of $X$ corresponding to eigenvalue $m$.

Now, consider the variance of $X$ in a pure state $|\psi\rangle$.
\begin{eqnarray}
 {\rm Var}_{|\psi\rangle}X  &=& \tr[|\psi\rangle\langle\psi| (X-m I)^{2}]  \nonumber \\
&=& \langle\psi|(X-m I)^{2}|\psi\rangle = \parallel (X-m I)|\psi\rangle\parallel^{2}.
\end{eqnarray}
Thus, in any pure state $|\psi\rangle$, there will always exist an observable $X$ such that the variance of $X$ in $|\psi\rangle$ is non-zero! Contrast this with the situation in classical probability theory. The classical analogues of the pure states are extreme points of the convex set of probability measures on $X$, namely, the point measures. And the variance of any classical random variable vanishes at these extreme points. We thus have an important point of departure between classical and quantum probabilities. In the quantum case, corresponding to the extreme points, namely the pure states, we can always find an observable with non-vanishing variance. The indeterminacy is thus much stronger in the case of quantum probabilities. In fact, examining the variances of pairs of observables leads to a formulation of Heisenberg's uncertainty principle. For a more detailed study of finite-level quantum probability theory, see~\cite{KRPbook1}.

We conclude this section with a brief discussion on the concept of {\bf measurement} in quantum theory\index{quantum! measurement}. Let the eigenvalues of the observable $X$ be labeled as $\{\lambda_{1}, \lambda_{2}, \ldots, \lambda_{k}\}$, with associated projections $\{E_{\lambda_{i}}\}$. Then, the measurement process associated with an observable $X = \sum_{\lambda_{i} \in \sigma(X)} \lambda_{i} E_{\lambda_{i}}$ essentially specifies which value $\lambda_{i}$ is realized for a given state $\rho$. The probability that label ${i}$ is observed is $\tr[\rho E_{\lambda_{i}}]$. Furthermore, the measurement process transforms the state in a non-trivial fashion. The von Neumann collapse postulate states that the measurement process {\it collapses} the state $\rho$ to a state $\rho' = \sum_{i}E_{\lambda_{i}} \rho E_{\lambda_{i}}$. Note that $\rho'$ is also positive and has unit trace.

Associated with every observable $X$ is thus a {\it projection-valued measurement} characterized by the set of spectral projections $\{E_{\lambda_{i}}\}$. More generally, quantum measurement theory also considers {generalized measurements}.
\begin{defn}[Generalized Measurements]
 A generalized quantum measurement $\cL$ for a finite-level system is characterized by a finite family of operators $\{L_{j}\}$ with the property that $\sum_{j=1}^{k}L_{j}^\dagger L_{j} = I$. In state $\rho$, the label $j$ is observed with probability $\tr[\rho L^\dagger_{j}L_{j}]$ and the post measurement state is given by $\rho' = \sum_{j=1}^{k}L_{j}\rho L_{j}^\dagger$.
\end{defn}
The transformation $\rho \rightarrow \rho'$ effected by a generalized measurement is indeed a completely positive trace-preserving (CPTP) map.

\section{Entanglement}\label{sec:entanglement}\index{entanglement}

Given two quantum systems $\cH_{A}$ and $\cH_{B}$, the Hilbert space associated with the {\it composite} system $AB$ is denoted as $\cH_{AB}$. In classical probability theory, if $\Omega_{1}$ and $\Omega_{2}$ are the sample spaces for two experiments, the sample space for the joint experiment is given by the Cartesian product $\Omega_{1}\times\Omega_{2}$. The analogue of that in the quantum case is the tensor product of the two Hilbert spaces, that is, $\cH_{AB} = \cH_{A} \otimes \cH_{B}$. In other words, if $|u\rangle \in \cH_{A}$ and $|v\rangle \in \cH_{B}$, then, $|u\rangle\otimes|v\rangle \in \cH_{AB}$. The tensor symbol is often omitted and the elements of $\cH_{AB}$ are simply denoted as linear combinations of $|u\rangle|v\rangle$.

If ${\rm dim}(\cH_{A}) = d_{A}$ and ${\rm dim}(\cH_{B}) = d_{B}$, then, ${\rm dim}(\cH_{AB}) = d_{A}d_{B}$. Furthermore, if the vectors $\{|e_{1}^{A}\rangle, |e_{2}^{A}\rangle, \ldots, |e_{m}^{A}\rangle \}$ constitute an orthonormal basis for $\cH_{A}$ and $\{|f_{1}^{B}\rangle, |f_{2}^{B}\rangle, \ldots, |f_{n}^{B}\rangle\}$ form an orthonormal basis for $\cH_{B}$, then the vectors $\{|e_{i}^{A}\rangle \otimes |f_{j}^{B}\rangle\}$ constitute an orthonormal basis for $\cH_{AB}$. Once such a basis is fixed, the vectors are often simply denoted as $|e_{i}^{A}\rangle|f_{j}^{B}\rangle \equiv |ij\rangle$. Similarly, we can denote the composite system corresponding to several such Hilbert spaces $\cH_{A_{1}}, \cH_{A_{2}}, \ldots, \cH_{A_{N}}$ as $\cH_{A_{1}A_{2}\ldots A_{N}} = \cH_{A_{1}}\otimes\cH_{A_{2}}\otimes\ldots\otimes\cH_{A_{N}}$. The basis vectors of such a composite system are simply denoted as $|ijk...\rangle \equiv |e_{i}^{A_{1}}\rangle\otimes|f_{j}^{A_{2}}\rangle\otimes|g_{k}^{A_{3}}\rangle\otimes \ldots$.

For example, consider the two-dimensional complex Hilbert space $\mathbb{C}^{2}$. This corresponds to a {\bf one-qubit system} in quantum information theory. The two basis vectors of $\mathbb{C}^{2}$ are commonly denoted as
\[|0\rangle  = \left[\begin{array}{c} 0 \\ 1 \end{array}\right], \qquad  |1\rangle  = \left[\begin{array}{c} 1 \\ 0 \end{array}\right].\]
The $n$-fold tensor product space $(\mathbb{C}^{2})^{\otimes n}$ is the {\bf $n$-qubit} Hilbert space. The basis vectors of $(\mathbb{C}^{2})^{\otimes n}$ can thus be written in terms of binary strings, as $|x_{1}x_{2}\ldots x_{n}\rangle$, where $x_{i} \in \{0,1\}$.

Classically, the extreme points of the joint probability distributions over $\Omega_{1}\times\Omega_{2}$ are Dirac measures on the joint sample space, of the form, $\delta_{(\omega_{1}, \omega_{2})} = \delta_{\omega_{1}}\otimes\delta_{\omega_{2}}$, for points $\omega_{1} \in \Omega_{1}$ and $\omega_{2} \in \Omega_{2}$. The extreme points of the set of joint distributions are in fact products of Dirac measures on the individual sample spaces. The situation is however drastically different for composite Hilbert spaces: there exist pure states of a bipartite Hilbert space which are not product states, but are sums of product states. Such states are said to be {\bf entangled}. We will formalize this notion in the following section.

\subsection{Schmidt Decomposition}\label{sec:schmidt}\index{Schmidt! decomposition}
We first state and prove an important property of pure states in a bipartite Hilbert space $\cH_{AB}$.
\begin{thm}[Schmidt Decomposition]\label{thm:schmidt}\index{Theorem! Schmidt decomposition}
Every pure state $|\psi\rangle \in \cH_{AB}$ can be written in terms of non-negative real numbers $\{\lambda_{k} \, (k \leq r)\} $, and two orthonormal sets $\{|\phi_{k}^{A}\rangle\} \subset \cH_{A}$ and $\{|\psi_{k}^{B}\rangle\} \subset \cH_{B}$, as,
\begin{equation}
|\psi\rangle = \sum_{k=1}^{r}\lambda_{k}|\phi_{k}^{A}\rangle|\psi_{k}^{B}\rangle, \label{eq:schmidt}
 \end{equation}
where, $\lambda_{k}$ satisfy $\sum_{k=1}^{r}\lambda_{k}^{2}=1$.
\end{thm}
\begin{proof}
Consider a pure state $|\psi\rangle \in \cH_{AB}$. Let $\cH_{A}$ be of dimension $d_{A}$ and $\cH_{B}$ of dimension $d_{B}$. In terms of the orthonormal bases for $\cH_{A}$ and $\cH_{B}$, the state $|\psi\rangle$ can be written as
\[|\psi\rangle = \sum_{i=1}^{d_{A}}\sum_{j=1}^{d_{B}} a_{ij}|i_{A}\rangle|j_{B}\rangle. \]
The matrix of coefficients $[A]_{ij} = a_{ij}$ completely characterizes the pure state $|\psi\rangle$. Let $r = {\rm rank}(A)$. $A$ can be represented via singular value decomposition, as $A = UDV^{\dagger}$, where $U, V$ are unitaries of order $d_{A}$, $d_{B}$ respectively, and $D$ is a diagonal matrix of rank $r$, whose entries are the singular values ($\lambda_{k}$) of $A$. Thus,
\begin{eqnarray}
 |\psi\rangle &=& \sum_{i,j,k}u_{ik}\lambda_{k}v_{kj}|i_{A}\rangle|j_{B}\rangle \nonumber \\
&=& \sum_{k=1}^{r}\lambda_{k}|\phi_{k}^{A}\rangle|\psi_{k}^{B}\rangle,
\end{eqnarray}
where, the vectors $\{|\phi_{k}^{A}\rangle\}$ constitute an orthonormal set in $\cH_{A}$ and $\{|\psi_{j}^{B}\rangle\}$ constitute an orthonormal set in $\cH_{B}$, since $U$ and $V$ are unitary matrices.

Since the state $|\psi\rangle\in \cH_{AB}$ is normalized, the corresponding coefficient matrix $A$ is such that $\sum_{i,j}|a_{ij}|^{2} = 1$. This is turn implies that $\sum_{k=1}^{r}\lambda_{k}^{2} = 1$.
\end{proof}

Given any density operator (state) $\rho_{AB} \in \cH_{AB}$, the {\it reduced state} on $\cH_{A}$, denoted as $\rho_{A}$, is obtained by tracing out over an orthonormal basis in $\cH_{B}$: $\rho_{A} = \tr_{B}[\rho_{AB}]$. Thus, $\rho_A$ is the operator on $\cH_A$ which satisfies, for all $|u\rangle, |v\rangle \in \cH_{A}$, the condition
\[ \langle v|\rho_A|u \rangle = \sum_{j}  \langle v|\langle f_{j} |\rho_{AB} |u \rangle |f_j \rangle, \] for some (and in fact any) orthonormal basis $\{|f_j \rangle\}$ of $\cH_B$. Similarly, the reduced state $\rho_{B}$ on $\cH_{B}$, is obtained by tracing out over an orthonormal basis in $\cH_{A}$: $\rho_{B} = \tr_{A}[\rho_{AB}]$. The states $\rho_{A}, \rho_{B}$ are also called {\it marginal states}, since they are analogues of the marginal distributions of a joint distribution in classical probability theory.

Now, consider the marginals of a pure state $|\psi\rangle \in \cH_{AB}$. It is a simple exercise to show that the marginals of $|\psi\rangle$ are not pure in general, they are {\it mixed} states.
\begin{exer}
Show that the reduced states of the density operator $|\psi\rangle\langle\psi|$ corresponding to the pure state $|\psi\rangle = \sum_{k=1}^{r}\lambda_{k}|\phi_{k}^{A}\rangle|\psi_{k}^{B}\rangle$, are given by
\begin{equation}
 \rho_{A} = \sum_{k=1}^{r}\lambda_{k}^{2}|\phi_{k}^{A}\rangle\langle \phi_{k}^{A}|; \;  \rho_{B} = \sum_{k=1}^{r}\lambda_{k}^{2}|\psi_{k}^{B}\rangle\langle \psi_{k}^{B}| \label{eq:marginals}
\end{equation}
\end{exer}

Thus, the marginals of a pure state $|\psi\rangle \in \cH_{AB}$ are no longer pure, they are mixed states, of rank equal to the rank of the coefficient matrix $A$. Contrast this with the classical case, where the marginals of the extreme points of the set of joint distributions are in fact extreme points of the set of distributions over the individual sample spaces. This important departure from classical probability theory, leads naturally to the notion of {\bf entanglement}.

We can associate with any pure state $|\psi\rangle \in \cH_{AB}$ a probability distribution $\{\lambda_{1}^{2}, \lambda_{2}^{2}, \ldots, \lambda_{k}^{2}\}$, via the Schmidt decomposition proved in Theorem~\ref{thm:schmidt}. Recall that the Shannon entropy of a probability distribution $\{p_{1}, p_{2}, \ldots, p_{k}\}$ is defined as $H(p) = -\sum_{i}p_{i}\log p_{i}$. Correspondingly, the {\bf von Neumann entropy}\index{von Neumann entropy} of a quantum state is defined as $S(\rho) = -\tr\rho\log\rho$. We will study the properties of this and other quantum entropies in greater detail in Section~\ref{sec:SSA}. For now, it suffices to note that for a pure state $|\psi\rangle \in \cH_{AB}$ of a joint system, $S(|\psi\rangle\langle\psi|) = 0$, whereas for the marginals $\rho_{A}, \rho_{B}$ in Eq.~\eqref{eq:marginals},
\[ S(\rho_{A}) = S(\rho_{B}) = -\sum_{k=1}^{r}(\lambda_{k}^2)\log(\lambda_{k}^{2}). \]
Again, note the quantitative departure from classical probability theory: given a joint distribution over two sample spaces $A,B$, the Shannon entropy of the joint distribution is always greater than the Shannon entropy of the marginals, that is, $H(AB) \geq H(A)$. But in quantum systems, while the pure states of a joint system will always have zero entropy, their reduced states will have non-zero entropy in general.

For a bipartite pure state $|\psi\rangle$, the Schmidt decomposition provides a way of quantifying the the deviation of $|\psi\rangle$ away from a product pure state. The number of product states in the Schmidt decomposition and the relative weights assigned to the different product states, quantify the entanglement of state $|\psi\rangle$.

\begin{defn}[Schmidt Rank]\index{Schmidt! number, pure states}
Given a bipartite pure state $|\psi\rangle \in \cH_{AB}$ with a Schmidt decomposition $|\psi\rangle = \sum_{k=1}^{r}\lambda_{k}|\phi_{k}^{A}\rangle|\psi_{k}^{B}\rangle$,
\begin{itemize}
\item[(i)] The number $r$ of non-zero coefficients $\lambda_{k}$ is defined to be the {\bf Schmidt rank}\index{Schmidt! rank} of the state $|\psi\rangle$.
\item[(ii)] A bipartite pure state $|\psi\rangle$ is said to be {\bf entangled} if it has Schmidt rank greater than one.
\item[(iii)] $S(\rho_{A}) = S(\rho_{B})$ is a measure of the entanglement of the pure state $|\psi\rangle$, where $\rho_{A} = \tr_{B}[|\psi\rangle\langle\psi|]$.
\end{itemize}
\end{defn}

For a $d$-level system, the von Neumann entropy of any state is bounded from above by $\log d$. Therefore, the maximum entanglement of a bipartite pure state $|\psi\rangle \in \cH_{AB}$ is $\log\min(d_{A}, d_{B})$.
\begin{exer}
If $\min(d_{A}, d_{B}) = m$, the {\bf maximally entangled state} in $\cH_{AB}$ is
\begin{equation}
 |\psi\rangle = \frac{1}{\sqrt{m}}\sum_{i=1}^{m}|i_{A}\rangle|i_{B}\rangle, \label{eq:maxentangled}
\end{equation}
where $\{|i_{A}\rangle\}$ and $\{|i_{B}\rangle\}$ are orthonormal bases in $\cH_{A}$ and $\cH_{B}$.
\end{exer}

We have thus far restricted our discussion to {\it pure} bipartite states. We can formally define product and entangled states in general as follows: a state $\rho \in \cS(\cH_{AB})$ is called a {\it product state} if it is of the form $\rho = \rho_{A}\otimes\rho_{B}$, where $\rho_{A}\in \cS(\cH_{A})$ and $\rho_{B}\in \cS(\cH_{B})$; if not, $\rho$ is said to be {\it entangled}. There are several interesting questions that arise in this context. For example, given a pair of states $\rho_{A} \in \cS(\cH_{A})$, $\rho_{B}\in \cS(\cH_{B})$, consider the following convex set.
\begin{equation}
 \cC(\rho_{A}, \rho_{B}) = \{\rho \in \cS(\cH_{AB})\; \vert \; \tr_{B}[\rho] = \rho_{A}, \tr_{A}[\rho] = \rho_{B}\}.
\end{equation}
Then, what is the least value $\min \{S(\rho)|\rho \in \cC(\rho_{A}, \rho_{B})\}$ of the von Neumann entropy of states $\rho$ that belong to this convex set? The maximum value is attained when the systems $A$ and $B$ are unentangled, that is, $\rho = \rho_{A}\otimes\rho_{B}$, so that $S(\rho) = S(\rho_{A}) + S(\rho_{B})$. We know this is indeed the maximum possible value because the von Neumann entropy satisfies the {\bf strong subadditivity property}: $S(\rho) \leq S(\rho_{A}) + S(\rho_{B})$ for any $\rho \in \cH_{AB}$. The strong subadditivity of the von Neumann entropy is discussed in greater detail in~\ref{sec:SSA}. Estimating the minimum value $\min \{S(\rho)|\rho \in \cC(\rho_{A}, \rho_{B})\}$  is in general a hard problem, though some estimates have been obtained for special cases~\cite{KRP05}. Interestingly, the analogous problem in classical probability theory also remains an open!

\subsection{Unitary Bases, EPR States and Dense Coding}\label{sec:EPR_dense}

For a finite-dimensional Hilbert space $\cH$ of dimension $d$, consider $\cB(\cH)$, the set of all bounded linear operators on $\cH$. $\cB(\cH)$ is also a Hilbert space with the inner product between any two operators $X, Y \in \cB(\cH)$ defined as $\langle X | Y\rangle = \tr[X^{\dagger} Y]$. $\cB(\cH)$ has a {\it unitary orthogonal basis}\index{unitary basis}, that is, it admits a family of unitary operators $\{W_{\alpha}, \alpha = 1,2,3,\ldots, d^{2}\}$, such that $\tr[W^{\dagger}_{\alpha}W_{\beta}] = d \delta_{\alpha\beta}$. In coding theory, such a basis is called a unitary error basis.

Consider the bipartite system $\cH \otimes \cH$ composed of two copies of $\cH$. Let $|\psi\rangle$ denote a maximally entangled pure state in $\cH\otimes \cH$, which is written as per Eq.~\eqref{eq:maxentangled} as,
\[ |\psi\rangle = \frac{1}{\sqrt{d}}\sum_{i}|i\rangle|i\rangle.\]
Such a maximally entangled state is often referred to as an {\bf Einstein-Podolosky-Rosen (EPR) state}\index{EPR states}, after the famous trio - Einstein, Podolsky and Rosen - who were the first to note that states like these exhibit strange properties which are uniquely quantum. We now show how a unitary orthogonal basis for $\cB(\cH)$ can in fact be used to construct a basis of EPR states, which lie at the heart of several quantum communication and cryptographic tasks.

Consider the states $|\psi_{\alpha}\rangle$ generated by applying the unitaries $\{W_{\alpha}\}$ to one half of the maximally entangled state $|\psi\rangle$, as follows:
\begin{equation}
 |\psi_{\alpha}\rangle  = \frac{1}{\sqrt{d}}\sum_{i=1}^{d}(W_{\alpha}|i\rangle)|i\rangle, \; \alpha = 1, 2, \ldots, d^{2}. \label{eq:epr}
\end{equation}
In other words, the state $|\psi\rangle$ is being acted upon by the operators $\{W_{\alpha}\otimes I\}$. It is then a simple exercise to show that the states $|\psi_{\alpha}\rangle$ are mutually orthogonal and maximally entangled.

\begin{exer}[Basis of EPR states]
The states $\{|\psi_{\alpha}\rangle, \alpha = 1,2,\ldots, d^{2}\}$ defined in Eq.~\eqref{eq:epr} constitute an orthonormal basis -- called the basis of EPR states -- for $\cH\otimes \cH$. Furthermore, each $|\psi_{\alpha}\rangle$ is maximally entangled, that is,
\[\tr_{1,2}\left[|\psi_{\alpha}\rangle\langle\psi_{\alpha}|\right] = \frac{I}{d}, \; \forall \alpha,\] whether the partial trace is taken over the first or the second Hilbert space.
\end{exer}

A remark on uniqueness and existence of unitary orthogonal bases: note that a given unitary orthogonal basis is invariant under scalar multiplication (with a scalar of unit modulus), permutation, and conjugation by a fixed operator ($W_{\alpha} \rightarrow \Gamma W_{\alpha} \Gamma^{\dagger}$). It is an interesting question to classify the distinct unitary orthogonal bases upto this equivalence. A simple construction of such a unitary error basis in any dimension $d$ is as follows. Consider the group $\mathbb{Z}_{d}$ and identify $\cH \equiv \ell_{2}(\mathbb{Z}_{d})$. Then, consider the translation and rotation operators $U_{a}, V_{b}$ ($a,b \in \mathbb{Z}_{d}$), namely,
\[ U_{a}\delta_{x} = \delta_{x +a} ; \quad V_{b}\delta_{x} = b(x) \delta_{x} ,\]
where $\{\delta_{x} \; | \; \forall x \in \mathbb{Z}_{d}\}$ is the standard orthonormal basis for $\cH$ and $b$ is a character of the group $\mathbb{Z}_{d}$. The products $\{U_{a}V_{b} | \; a,b \in \mathbb{Z}_{d} \}$, which are the Weyl operators of the abelian group $\mathbb{Z}_{d}$, form a unitary error basis.

{\bf Superdense coding}\index{superdense coding} is a simple yet important example of a quantum communication task that is made possible using EPR states. The goal of the task is for one party $A$ ({\it Alice}) to communicate $d^{2}$ classical messages to another party $B$ ({\it Bob}), by physically sending across just one quantum state. Say $\cH_{A}$ is the $d$-dimensional Hilbert space associated with Alice's system and $\cH_{B}$ the $d$-dimensional Hilbert space corresponding to Bob's system. We assume that Alice and Bob share a bipartite EPR state, namely, $|\psi_{AB}\rangle = \frac{1}{\sqrt{d}}\sum_{i=1}^{d}|i_{A}\rangle|i_{B}\rangle$. In other words, the joint system $\cH_{AB}$ is in state $|\psi_{AB}\rangle$. Now, if Alice wants to send the message $\alpha \in \{1,2,\ldots, d^{2}\}$, she simply applies the {\it unitary gate}\footnote{Since deterministic state changes are effected by unitary operators in quantum theory, unitaries are often referred to as {\it quantum gates} in the quantum computing literature.} $W_{\alpha}$ on
her half of the state. The joint state of $\cH_{AB}$  is then transformed to $|\psi_{\alpha}\rangle$.

Alice then sends across her half of the state to Bob. To {\it decode} the message, Bob has to perform a measurement and obtain a classical output. As discussed in Sec.~\ref{sec:qstates}, corresponding to the orthonormal basis $\{|\psi_{\alpha}\rangle\}$, we have a projective measurement characterized by the collection of one-dimensional projections $\{|\psi_{\alpha}\rangle\langle\psi_{\alpha}|\}$. When Bob performs this measurement on any state $\rho$, the probability of obtaining outcome $\alpha$ is $\tr[\rho |\psi_{\alpha}\rangle\langle\psi_{\alpha}|]$. Thus, given the state $|\psi_{\alpha}\rangle$, Bob will correctly decode the message $\alpha$ with probability one. The idea of dense coding was first proposed in~\cite{BW92}; for recent pedagogical discussions, see~\cite{NCbook,KRPbook2}.

As an aside, another interesting application of the unitary group is in proving {\it universality} of a set of quantum gates. For a composite system made up of $k$ finite-dimensional Hilbert spaces $\cH_{1}\otimes\cH_{2}\otimes\ldots\otimes\cH_{k}$, consider the unitary group $\cU(\cH_{1}\otimes\cH_{2}\otimes\ldots\otimes\cH_{k})$. Let $\cU_{i,j}$ be the set of unitary operators $\{U_{ij}\}$ that act only on Hilbert spaces $\cH_{i}\otimes \cH_{j}$, leaving the other Hilbert spaces unaffected. $\cU_{i,j}$ is a subgroup of $\cU$. The {\it universality theorem} states that every unitary operator $U \in \cU$ can be decomposed as a product $U = U_{1}U_{2}\ldots U_{N}$, where each $U_{m}$ is an element of $\cU_{i,j}$ for some $i,j \in 1, \ldots, k.$ That is, any unitary operator acting on the composite system of $k$ Hilbert spaces can be written as a product of unitaries that act non-trivially only on two of the $k$ Hilbert spaces.

\section{Schmidt rank of bipartite entangled states}\index{Schmidt! rank}
Using the concepts defined in Sec.~\ref{sec:schmidt}, we will now explore a few interesting problems relating to the Schmidt rank of bipartite entangled states.

\subsection{Subspaces of minimal Schmidt rank}
Given a pair of finite-dimensional Hilbert spaces $\cH_{1}\otimes \cH_{2}$, a question of interest is to construct subspaces $\cS \subset \cH_{1}\otimes\cH_{2}$ such that every pure state $|\psi\rangle \in \cS$ has Schmidt rank $\geq k$. In particular, what is the maximum dimension of a subspace $\cS$ whose pure states are all of Schmidt rank greater than or equal to $k$? Note that any state that has support on such a subspace $\cS$ will necessarily be a highly entangled state. Since entangled states are communication resources (as seen in the case of dense coding, for example), this question assumes importance in quantum information theory.

Let $\mathbb{M}(m,n)$ denote the set of $m\times n$ complex matrices. For elements $X, Y \in \mathbb{M}(m,n)$, the inner product is defined as $\langle X|Y\rangle = \tr[X^{\dagger}Y]$. Suppose the Hilbert spaces $\cH_{1}$ and $\cH_{2}$ are of dimensions $m$ and $n$ respectively, there is a natural identification between $\cH_{1}\otimes\cH_{2}$ and $\mathbb{M}(m,n)$. To see the explicit correspondence between elements of $\cH_{1}\otimes\cH_{2}$ and $\mathbb{M}(m,n)$, define a {\it conjugation} $J$ on $\cH_{2}$. For product vectors $|u\rangle|v\rangle \in \cH_{1}\otimes\cH_{2}$, the conjugation acts as follows:
\begin{equation}
 \Gamma(J): |u\rangle|v\rangle \rightarrow |u\rangle\langle J v|. \label{eq:correspondence}
\end{equation}
Since the product vectors $|u\rangle|v\rangle$ are total in $\cH_{1}\otimes\cH_{2}$, $\Gamma(J)$ defines the correspondence $\cH_{1}\otimes\cH_{2} \rightarrow \mathbb{M}(m,n)$. The corresponding operators $|u\rangle\langle Jv|$ span $\mathbb{M}(m,n)$. Furthermore, $\Gamma(J)$ is inner product preserving, and is therefore a unitary isomorphism between $\cH_{1}\otimes \cH_{2}$ and $\mathbb{M}(m,n)$. The following exercise is a simple consequence of this isomorphism.
\begin{exer}
Show that the Schmidt rank of any pure state $|\psi\rangle \in \cH_{1}\otimes\cH_{2}$ is equal to the rank of $\Gamma(J)(|\psi\rangle)$.
\end{exer}

The problem of finding subspaces for which every vector has a Schmidt rank greater than or equal to $k$, thus reduces to the following matrix-theoretic problem: Construct subspaces $\cS \subset \mathbb{M}(m,n)$, with the property that every non-zero element of these subspaces is of rank greater than or equal to $k$. In particular, the problem is to find the maximum dimension of the subspace $\cS\subset \mathbb{M}(m,n) $, whose non-zero elements are all of rank greater than or equal to $k$. While the problem remains open for a general $k$, we will present an example of such a construction for $k=2$.

Let $\{W_{\alpha}, \alpha=1,\ldots,n^{2}\}$ denote a unitary orthogonal basis for the space of square matrices $\mathbb{M}_{n} \equiv \mathbb{M}(n,n)$. Consider a positive, self-adjoint matrix $\Phi \in \mathbb{M}_{n}$ with the spectral resolution,
\[ \Phi = \sum_{\alpha=1}^{n^{2}}p_{\alpha}|W_{\alpha}\rangle\langle W_{\alpha}| ,\]
with distinct eigenvalues $p_{\alpha} > 0$. Now, we can construct a subspace $\cS \subset \mathbb{M}_{2n}$, of square matrices in $\mathbb{M}_{2n} \equiv \mathbb{M}(2n,2n)$, whose non-zero elements are of rank $k \geq 2$.
\begin{thm}\label{thm:subspace2}
Consider $\cS \subset \mathbb{M}(2n,2n)$ defined as follows:
\begin{equation}
 \cS = \left\{ X \equiv \left(\begin{array}{c|c}
A & \Phi(B)\\
\hline
B & A
\end{array}\right), \; A, B \in \mathbb{M}_{n}\right\}. \label{eq:rank2}
\end{equation}
The elements $X \in \cS$ satisfy ${\rm rank}(X) \geq 2$, for arbitrary $A,B \in \mathbb{M}_{n}$.
\end{thm}
\begin{proof}
Suppose $B=0$ and $A\neq 0$. Then, the off-diagonal blocks are zero, but the diagonal blocks are non-zero and ${\rm rank}(X) = 2$ for all $A \in \mathbb{M}_{n}$. If $A = 0$ but $B \neq 0$, then, since $\Phi$ is non-singular, $\Phi(B) \neq 0$. Therefore, ${\rm rank}(X) \geq 2$, $\forall \, B \in \mathbb{M}_{n}$.

Now, it remains to show that ${\rm rank}(X) \geq 2$, for all $A,B \neq 0$ in $ \mathbb{M}_{n}$. Suppose ${\rm rank}(X(A,B)) =1$. Such a matrix $X$ is of the form \[\left[\begin{array}{c}
                                                                                                                                                                   |u\rangle \\
|v\rangle
                                                                                                                                                                                     \end{array}\right]\left[\begin{array}{cc} |u'\rangle |v'\rangle \end{array}\right] = \left(\begin{array}{cc}
|u\rangle\langle u| & |u\rangle\langle v'| \\
|v\rangle\langle u'| & |v\rangle\langle v'| \end{array}\right),
\]
where, $|u\rangle, |v\rangle$ are column vectors of length $n$. Comparing with the definition of $X$ in Eq.~\eqref{eq:rank2}, we see that,
\begin{equation}
|u\rangle\langle u'| = |v\rangle\langle v'| = A, \qquad \Phi(|v\rangle\langle u'|) = |u\rangle\langle v'|.  \label{eq:2equality}
\end{equation}
The first equality implies that there exists scalars $c,c' \neq 0$ such that  $|v\rangle = c|u\rangle$ and $|v'\rangle = c'|u'\rangle$. The second equality in Eq.~\eqref{eq:2equality} thus becomes
\[ \Phi(|v\rangle\langle u'|) = c'c^{-1}|v\rangle\langle u'|.\]
This implies that $|v\rangle\langle u'|$ an eigenvector of $\Phi$. However, the eigenvectors of $\Phi$ belong to the unitary error basis, and thus $\Phi$ cannot have such a rank-one matrix as its eigenvector. Thus, the assumption that $X(A,B)$ as defined in Eq.~\eqref{eq:rank2} if of rank one leads to a contradiction. This proves the claim.
\end{proof}

The above construction in fact leads to a more interesting property of the set of matrices in $\mathbb{M}_{2n}$.
\begin{thm}
 The space of $2n\times2n$  matrices $\mathbb{M}_{2n}$ is a direct sum of the form $\mathbb{M}_{2n} = \cS \oplus \cS^{\perp}$, where, ${\rm rank}(X) \geq 2$, for every $ 0 \neq X \in \cS\cup\cS^{\perp}$.
\end{thm}
\begin{proof}
Consider an element of $\cS^{\perp}$, the orthogonal complement of $\cS$, written in block matrix form as:
\[ Y = \left(\begin{array}{c|c}
K & L \\
\hline
M & N  \end{array}\right).\]
Since $\tr[Y^\dagger X] = 0$ for any $X \in \cS$ and $Y \in \cS^{\perp}$, we have,
\begin{eqnarray}
 \tr\left[\left(\begin{array}{c|c}
        K^\dagger & M^\dagger \\
\hline
L^{\dagger} & N^{\dagger}
       \end{array}\right)\left(\begin{array}{c|c}
A & \Phi(B)\\
\hline
B & A
\end{array}\right)\right] &=& 0 \nonumber \\
\Rightarrow \tr [K^{\dagger}A + M^{\dagger}B + L^{\dagger}\Phi(B) + N^{\dagger}A] &=& 0, \; \forall \; A, B. \label{eq:adjoint}
\end{eqnarray}
Setting $B=0$, we have,
\[\tr[(K^{\dagger} + N^{\dagger})A] = 0, \; \forall \; A \; \Rightarrow K + N =0 .\]
Similarly, setting $A=0$, we have,
\[\tr[(M + \Phi(L))^{\dagger}B] = 0, \; \forall B \; \Rightarrow M + \Phi(L) = 0. \]
Thus, every element of the orthogonal complement of $\cS$ is of the form,
\[ Y = \left(\begin{array}{c|c}
              K & L \\
\hline
- \Phi(L) & -K
             \end{array}\right).
\]
By the same argument as in the proof of Theorem~\ref{thm:subspace2}, ${\rm rank}(Y) \geq 2$, for all $Y \in \cS^{\perp}$.
\end{proof}

The above result for $\mathbb{M}_{2n}$ naturally leads to the following problem for the more general matrix spaces $\mathbb{M}(m,n)$.
\begin{ques}
Identify the quadruples $(m,n,r,s)$ for which it is possible to find a direct sum decomposition $\mathbb{M}(m,n) = \cS \oplus \cS^{\perp}$ of $\mathbb{M}(m,n)$, into a direct sum of subspace $\cS$ and its orthogonal complement $\cS^{\perp}$, such that,
\[{\rm rank}(X) \geq r, \; \forall \; 0 \neq X \in \cS, \;\; {\rm rank}(Y) \geq s, \; \forall \; 0 \neq Y \in \cS^{\perp}.\]
\end{ques}

Finally, we return to the question of the maximum dimension of the subspace $\cS$ whose elements have Schmidt rank greater than or equal to $k$.
\begin{thm}\label{thm:minRank}
 Let $\cL \subset \mathbb{M}(m,n)$ be such that $\forall \, 0\neq X \in \cL$, ${\rm rank}(X) \geq k + 1$. Then, ${\rm dim}(\cL) \leq (m-k)(n-k)$.
\end{thm}
\begin{proof}
Let $S_{k}$ denote the {\it variety} of matrices of rank less than or equal to $k$ in $\mathbb{M}(m,n)$. Then, it is known ~\cite{Harris} that ${\rm dim}(S_{k}) = k(m+n-k)$. Let $\cL$ be a subspace ($\cL \subset \mathbb{M}(m,n)$) such that for all $0 \neq X \in \cL$, ${\rm rank}(X) \geq k+1$. Note that, $S_{k}\cap \cL = {0}$. Using a generalization of Azoff's Theorem in algebraic geometry~\cite{DMR}, we have,
\[{\rm dim}(S_{k}) + {\rm dim}(\cL) \leq mn.\]
The dimension of any such subspace $\cL$ is therefore given by
\begin{equation}
 {\rm dim}(\cL) \leq (m-k)(n-k).
\end{equation}
\end{proof}

Furthermore, we can in fact explicitly construct a subspace $\cL_{0}$ whose dimension is exactly equal to the maximum value $(m-k)(n-k)$. Consider polynomials $P$ such that ${\rm deg}(P)\leq p-k-1$. Construct diagonal matrices of order $p\times p$ with entries $D = {\rm diag}(P(z_{1}, P(z_{2}), \ldots, P(z_{p}))$, where $z_{1}, z_{2}, \ldots, z_{p}$ are all distinct. The linear space of such diagonal matrices, denoted by $\cD_{p,k}$ is of dimension $p-k$. At most $p-k-1$ entries of such a diagonal matrix $D$ can be zero. Therefore, ${\rm rank}(D) \geq k+1, \; \forall \; D \in \cD_{p,k}$. The diagonal complementary space $\cD_{p,k}^{\perp}$ is of dimension $k$.

Now the construction proceeds as follows. Any matrix $X \in \cL_{0}$ has $m+n-1$ diagonals. Fill all diagonals of length less than or equal to $k$ with zeros. Let the least length of a non-vanishing diagonal be $p$. Choose the entries of this diagonal of length $p$ from some matrix $D \in \cD_{p,k}$. Consider the $p\times p$ minor of such an $X$. This is either a lower triangular or an upper triangular matrix. The diagonal of this $p\times p$ minor has the property that at most $p-k-1$ entries are zero. Therefore, rank of such a $p\times p$ minor is greater than or equal to $k+1$. By this construction, every non-zero element $X \in \cL_{0}$ is of rank at least $k+1$. Therefore,
\[{\rm dim}(\cL_{0}) = mn - k(k+1) - (m+n-1 -2k)k = (m-k)(n-k), \]
where the second term enumerates the zero entries of the diagonals of length $k$ and the final term enumerates the zero entries of the non-vanishing diagonals.

This construction and the result of Theorem~\ref{thm:minRank} thus imply the following result for the quantum information theoretic problem of finding subspaces of high entanglement.
\begin{thm}
Given a composite system $\cH_{1}\otimes\cH_{2}$ with ${\rm dim}(\cH_{1}) = m$ and ${\rm dim}(\cH_{2}) = n$ and subspaces $\cS\subset \cH_{1}\otimes\cH_{2}$,
\begin{equation}
\max\left\{{\rm dim}(\cS) \; \vert \; {\rm Schmidt \; rank}(|\psi\rangle) \geq k, \; \forall \; \psi \in \cS \right\} = (m-k)(n-k).
\end{equation}
\end{thm}


\section{Schmidt Number of Mixed States}

We now move beyond pure states and study entanglement in mixed states. Recall that pure states are the extreme points of the set $\cS(\cH)$ of positive operators of unit-trace. A general state $\rho \in \cS(\cH)$ is thus a mixture of pure states; it is a mixed state, with $\rho \geq 0$ and $\tr[\rho] =1$. For the composite system $\cH_{1}\otimes \cH_{2}$, with ${\rm dim}(\cH_{1}) = m$ and ${\rm dim}(\cH_{2}) = n$, consider the sets
\begin{eqnarray}
 \cS_{k} &=& \left\{ |\psi\rangle\langle\psi| \in \cS(\cH_{1}\otimes \cH_{2}) \; \vert \; {\rm Schmidt \; No.}(|\psi\rangle) \leq k \right\}, \label{eq:Sk_defn} \\
\tilde{\cS}_{k} &=& \left\{\int_{{\rm Schmidt \;  No.}(\psi)\leq k} |\psi\rangle\langle\psi| \; \mu \; d(\psi)\right\},
\end{eqnarray}
where $\mu$ is a probability distribution. $\tilde{\cS}_{k}$ is thus the set of all mixed states that are convex combinations of the pure states in $\cS_{k}$. We know from the Schmidt decomposition that for $|\psi\rangle \in \cH_{1}\otimes\cH_{2}$, $1 \leq k \leq {\rm min}(m,n) $. Note that $\tilde{\cS}_{k}$ is a compact convex set in the real linear space of Hermitian operators on $\cH_{1}\otimes\cH_{2}$, of dimension $m^{2}n^{2}$. It then follows from Carath\'eodory's Theorem that every $\rho \in \tilde{\cS}_{k}$ can be expressed as
\[ \rho = \sum_{j=1}^{m^{2}n^{2}+1}p_{j}|\psi_{j}\rangle\langle\psi_{j}|, \]
with $|\psi_{j}\rangle \in \cS_{k}$, $p_{j}\geq 0$, and $\sum_{j} p_{j}=1$. It is therefore enough to consider finite convex combinations of pure states to represent the elements of $\tilde{\cS}_{k}$.

\begin{defn}[Schmidt Number]\index{Schmidt! number, mixed states}
A state $\rho$ is defined to have Schmidt number $k$ if $\rho \in \tilde{\cS}_{k+1}\setminus\tilde{\cS}_{k}$.
\end{defn}
Estimating the Schmidt number for arbitrary mixed states is in general a hard problem. A related question of interest is whether it is possible to construct a {\it test} to identify states $\rho \notin \tilde{\cS}_{k}$. Such a test would be a special case of the class of {\bf entanglement witnesses}~\cite{Hor1}\index{entanglement witness} which are well-studied in the quantum information literature (see~\cite{Hor2} for a recent review).

In the following Section, we prove the Horodecki-Terhal criterion~\cite{TH}\index{Horodecki-Terhal criterion} which shows that $k$-positive maps can act as entanglement witnesses for states with Schmidt number exceeding $k$. Such a result was first proved for $k=1$~\cite{Hor1} and later extended to general $k$.

\subsection{Test for Schmidt number $k$ using $k$-positive maps}\index{$k$-positive maps}

By the geometric Hahn-Banach Theorem, in the real linear space of Hermitian operators we can construct a linear functional $\Lambda$ such that there exists a hyperplane corresponding to $\Lambda(.) = c$ (for some constant $c$) that separates $\rho$ and the convex set $\tilde{\cS}_{k}$. In other words, there exists a linear functional $\Lambda$ such that,
\begin{equation}
\Lambda(\rho) < c \leq \Lambda(\sigma), \; \forall \; \sigma \in \tilde{\cS}_{k}. \label{eq:ranktest}
\end{equation}
Any linear functional $\Lambda$ on the Banach space of Hermitian operators can be written as $\Lambda(X) = \tr[XA]$, where $A$ is a Hermitian operator. Therefore Eq.~\eqref{eq:ranktest} implies, for any $\rho \notin \tilde{\cS}_{k}$,
\begin{equation}
 \tr[\rho A] < c \leq \tr[\sigma A] \; \Rightarrow \; \tr[\rho H] < 0 \leq \tr[\sigma H], \; \forall \; \sigma \in \tilde{\cS}_{k}, \label{eq:Htest}
\end{equation}
where, $H = A - c I$, is a Hermitian operator on $\cH_{1}\otimes\cH_{2}$. Defining the map $\Gamma_{H}: \cB(\cH_{1})\rightarrow \cB(\cH_{2})$ such that,
\[ \Gamma_{H}(X) = \tr_{\cH_{1}}[H(X\otimes I_{\cH_{2}})], \; \forall \; X \in \cB(\cH_{1}),\]
we can rewrite $H$ as follows~\cite{Hor1}.
\begin{prop}
Given an orthonormal basis of rank-$1$ operators $\{E_{\alpha}\}$ in $\cB(\cH_{1})$, $H$ can be written as,
\begin{equation}
  H = \sum_{\alpha} E_{\alpha}\otimes \Gamma_{H}(E_{\alpha}^{\dagger}). \label{eq:Hdefn}
\end{equation}
\end{prop}
\begin{proof}
Choose orthonormal bases $\{E_{\alpha}\} \in \cB(\cH_{1})$ and $\{F_{\beta}\} \in \cB(\cH_{2})$. Then, $\{E_{\alpha}\otimes F_{\beta}\}$ constitutes an orthonormal basis for $\cB(\cH_{1}\otimes\cH_{2})$. Therefore, the operator $H$ can be written as,
\begin{eqnarray}
H &=& \sum_{\alpha,\beta}\tr\left[H(E_{\alpha}^\dagger \otimes F_{\beta}^{\dagger})\right](E_{\alpha}\otimes F_{\beta}) \nonumber \\
&=& \sum_{\alpha, \beta} \tr_{\cH_{2}}\left[\tr_{\cH_{1}}[H (E_{\alpha}^{\dagger}\otimes I_{\cH_{2}})]F_{\beta}^{\dagger}\right](E_{\alpha}\otimes F_{\beta})  = \sum_{\alpha}E_{\alpha} \otimes \Gamma_{H}(E_{\alpha}^\dagger) . \nonumber
\end{eqnarray}
\end{proof}

We now show (following~\cite{TH}) that the separation in Eq.~\eqref{eq:Htest} between the set $\tilde{\cS}_{k}$ and states $\rho \notin \tilde{\cS}_{k}$ can be realized using a $k$-positive map from $\cB(\cH_{1})\rightarrow \cB(\cH_{2})$.
\begin{thm}\label{thm:ranktest_main}
There exists a $k$-positive map $\Lambda_{0}: \cB(\cH_{1})\rightarrow \cB(\cH_{2})$ such that~\footnote{A remark on notation: we use $\Id$ to denote the identity map $\Id: \cB(\cH)\rightarrow \cB(\cH)$, with $\Id(\rho) = \rho$ for any $\rho \in \cB(\cH)$.}
\begin{eqnarray}
\tr\left[(\Id\otimes \Lambda_{0})(\sum_{\alpha}E_{\alpha}\otimes E_{\alpha}^{\dagger})\rho\right] &<& 0, \; \rho \notin \tilde{\cS}_{k}, \label{eq:k-positivity1} \\
\tr\left[(\Id\otimes \Lambda_{0})(\sum_{\alpha}E_{\alpha}\otimes E_{\alpha}^{\dagger})\sigma\right] &\geq& 0, \; \forall \; \sigma \in \tilde{\cS}_{k}, \label{eq:k-positivity2}
\end{eqnarray}
where $\{E_{\alpha}\}$ is an orthonormal basis of rank-$1$ operators in $\cB(\cH_{1})$.
\end{thm}
\begin{proof}
First we show that the operator $H$ that defines the separation in Eq.~\eqref{eq:Htest} can be written in terms of an operator $\Lambda_{0}: \cB(\cH_{1})\rightarrow \cB(\cH_{2})$.
\begin{lem}\label{lem:ranktest}
If $H$ is an operator with the property
\begin{equation}
\tr[\rho H] < 0 \leq \tr[\sigma H], \; \forall \; \sigma \in \tilde{\cS}_{k}, \label{eq:ranktest2}
\end{equation}
then $H$ is of the form
\begin{equation}
 H = (\Id \otimes \Lambda_{0})(\sum_{\alpha}E_{\alpha}\otimes E_{\alpha}^{\dagger}) = (\Id\otimes\Lambda_{0})(mP), \label{eq:lambda_0}
 \end{equation}
for some operator $\Lambda_{0}: \cB(\cH_{1})\rightarrow \cB(\cH_{2})$, and a projection $P$.
\end{lem}
\begin{proof}
Suppose we consider an orthonormal basis $\{|e_{i}\rangle\}$ for $\cH_{1}$ and choose the rank-$1$ operators $E_{\alpha}$ to be $E_{\alpha} = |e_{i}\rangle\langle e_{j}|$. Then, $E_{\alpha}^{\dagger} = |e_{j}\rangle\langle e_{i}|$. Let $T$ denote the {\it transpose} operation with respect to the $\{|e_{i}\rangle\}$ basis, then, $E_{\alpha}^{\dagger} = T(E_{\alpha})$. Therefore, for this specific choice of basis, Eq.~\eqref{eq:Hdefn} implies,
\begin{equation}
H = \sum_{\alpha} E_{\alpha} \otimes (\Gamma_{H}\circ T)(E_{\alpha}) = (\Id \otimes \Lambda_{H})(\sum_{\alpha}E_{\alpha}\otimes E_{\alpha}^{\dagger}) ,
\end{equation}
where we have defined $\Lambda_{0}(.) \equiv (\Gamma_{H} \circ T)(.)$. It is now a simple exercise to check that $\sum_{\alpha}E_{\alpha}\otimes E_{\alpha}^{\dagger} = mP$, where $P$ is a projection operator.
\begin{exer}
Define the rank-$1$ operators $E_{\alpha} = |e_{i}\rangle\langle e_{j}| \in \cB(\cH_{1})$, where $\{|e_{i}\rangle\}$ is an orthonormal basis for the space $\cH_{1}$. Then, the operator
\[ P  =\frac{1}{m}\sum_{\alpha} E_{\alpha}\otimes E_{\alpha}\]
is a projection ($P^{\dagger} = P = P^{2}$).
\end{exer}
\end{proof}

\noindent The theorem is proved once we show that the operator $\Lambda_{0}$ defined above is indeed  $k$-positive.
\begin{lem}\label{lem:k-positivity}
The operator $\Lambda_{0}$ defined in Eq.~\eqref{eq:lambda_0} corresponding to a $H$ that satisfies Eq.~\eqref{eq:ranktest2}, is $k$-positive on $\cB(\cH_{1})\rightarrow \cB(\cH_{2})$.
\end{lem}
\begin{proof}
Lemma~\ref{lem:ranktest} implies that for a chosen basis $\{ |e_{i}\rangle\} \in \cH_{1}$, any state $\sigma \in \tilde{\cS}_{k}$ satisfies,
\[ \tr\left[(\Id\otimes\Lambda_{0})\left( \sum_{i,j}|e_{i}\rangle\langle e_{j}| \otimes |e_{i}\rangle\langle e_{j}| \right) \; \sigma \right] \geq 0.   \]
If we choose $\sigma$ to be a pure state $\sigma \equiv |\psi\rangle\langle\psi|$,  the corresponding vector has the Schmidt decomposition $|\psi\rangle = \sum_{r=1}^{k}|u_{r}\rangle|v_{r}\rangle$. Then, the above condition becomes,
\begin{eqnarray}
&& \sum_{i,j} \tr\left[|e_{i}\rangle\langle e_{j}| \otimes \Lambda_{0}(|e_{i}\rangle\langle e_{j}|) \; \left(\sum_{r,s=1}^{k}|u_{r}\rangle\langle u_{s}| \otimes |v_{r}\rangle\langle v_{s}| \right) \right] \nonumber \\
&=& \sum_{i,j}\sum_{r,s}\langle e_{j}|u_{r}\rangle \langle u_{s}|e_{i}\rangle \;  \langle v_{s} |\Lambda_{0}(|e_{i}\rangle\langle e_{j}|)| v_{r}\rangle \geq 0 . \label{eq:ranktest3}
\end{eqnarray}
Now,  note that $\sum_{i} \langle u_{s}|e_{i}\rangle |e_{i}\rangle = \sum_{i}|e_{i}\rangle\langle e_{i}| J \, u_{s}\rangle$, where $J$ is a conjugation. Then, since $\{|e_{i}\rangle \}$ is an orthonormal basis, $\sum_{i} \langle u_{s}|e_{i}\rangle |e_{i}\rangle = |J\, u_{s}\rangle$. Defining $|u'_{s}\rangle = |J \, u_{s}\rangle$, the inequality in Eq.~\eqref{eq:ranktest3} becomes
\begin{equation}
\sum_{r,s=1}^{k}\langle v_{s}| (\Lambda_{0}| u'_{s}\rangle\langle u_{r}'|) |v_{r}\rangle \geq 0,
\end{equation}
for all such sets of vectors $\{|u_{r}\rangle\}, \{|v_{r}\rangle\}$. Recalling the definition of $k$-positivity (see Defn.~\ref{def:k-positive}), this is true if and only if $\Lambda_{0}$ is $k$-positive.
\end{proof}
\end{proof}

Finally, Theorem~\ref{thm:ranktest_main} implies the existence of a $k$-positive map that can act as an {\it entanglement witness} for Schmidt number $k$ in the following sense~\cite{TH}.
\begin{thm}[Schmidt Number Witness]\index{Theorem! Schmidt number witness}
\begin{itemize}
\item If $\Lambda : \cB(\cH_{2})\rightarrow \cB(\cH_{1})$ is a $k$-positive operator, then for any $\sigma \in \tilde{\cS}_{k}$,
\begin{equation}
  (\Id_{\cH_{1}}\otimes \Lambda)(\sigma) \geq 0.
\end{equation}
\item If $\rho \notin \tilde{\cS}_{k}$, there exists a $k$-positive map $\Lambda : \cB(\cH_{2}) \rightarrow \cB(\cH_{1})$, such that,
\begin{equation}
  (\Id_{\cH_{1}}\otimes \Lambda)(\rho) \ngeq 0. \label{eq:non-positive}
\end{equation}
Such a map $\Lambda$ is an {\bf entanglement witness} for Schmidt number $k$.
\end{itemize}
\end{thm}
\begin{proof}
The first statement of the theorem is easily proved, and is left as an exercise.

\noindent To prove the second statement, we note that if a map $\Lambda$ is $k$-positive, then so is its adjoint.
\begin{exer}
If $\Lambda: \cB(\cH_{1})\rightarrow \cB(\cH_{2})$ is $k$-positive, then, $\Lambda^{\dagger}: \cB(\cH_{2})\rightarrow \cB(\cH_{1})$ is also $k$-positive.
\end{exer}
Since we know from Lemma~\ref{lem:k-positivity} that the map $\Lambda_{0}$ defined in Eq.~\eqref{eq:lambda_0} is $k$-positive, the above exercise implies that the map $\Lambda_{0}^\dagger: \cB(\cH_{2})\rightarrow \cB(\cH_{1})$ is also $k$-positive. Then, it follows that
\begin{equation}
 (\Id_{\cH_{1}} \otimes \Lambda_{0}^{\dagger})(\sigma) \geq 0, \;  \forall \; \sigma \in \tilde{\cS}_{k}.
\end{equation}
Going back to Eq.~\eqref{eq:k-positivity1} and taking the adjoint of the operator in the trace, we have,
\[ \tr\left[(\Id\otimes\Lambda^{\dagger}_{0})(\rho)\left(\sum_{\alpha}E_{\alpha}\otimes E_{\alpha}^{\dagger} \right)\right] < 0.\]
This implies, $(\Id \otimes \Lambda_{0}^{\dagger})(\rho) \ngeq 0.$ We have thus constructed a $k$-positive map $\Lambda_{0}^{\dagger}$ with the desired property.
\end{proof}

In order to check if a given state $\rho$ is of Schmidt number strictly greater than $k$, we need a $k$-positive map that satisfies Eq.~\eqref{eq:non-positive}. Finding such $k$-positive maps is indeed a hard problem. There is however a vast body of work on positive but not completely positive maps, which can act as entanglement witnesses in the quantum information literature~\cite{Hor2}. The classical example of such a test for entanglement is using the transpose operation which is positive, but not completely positive~\cite{Peres}.

\subsection{Schmidt Number of generalized Werner States}\label{sec:Werner}\index{Schmidt! number, Werner states}

Given a $d$-dimensional Hilbert space $\cH$, consider the bipartite system $\cH\otimes\cH$ formed using two copies of $\cH$. The {\bf Werner states}~\cite{Wer}\index{Werner states} which we define below, are an interesting single parameter family of bipartite states in $\cH\otimes \cH$.

For any unitary operator $U$ in $\cH$, let $\Pi(U) : U \rightarrow U \otimes \bar{U}$ be a representation of the unitary group $\cU(\cH)$.
\begin{prop}
Any state $\rho$ commuting with every $U\otimes\bar{U}$ has the form
\begin{equation}
 \rho = \rho_{F} \equiv F|\psi_{0}\rangle\langle\psi_{0}| + (1-F) \frac{\left(I - |\psi_{0}\rangle\langle\psi_{0}|\right)}{d^{2}-1}, \label{eq:werner}
\end{equation}
where $|\psi_{0}\rangle = \frac{1}{\sqrt{d}}\sum_{i=1}^{d}|ii\rangle$ is the maximally entangled state in $\cH\otimes \cH$.
\end{prop}
\begin{proof}
This follows from the isomorphism between $\cH\otimes \cH$ and $\cB(\cH)$. Let $\cJ$ denote conjugation with respect to a chosen basis $\{|i\rangle\} \in \cH$. Defining the map,
\[\Gamma(\cJ) : |u\rangle|v\rangle \rightarrow |u\rangle\langle \cJ v|, \]
we see that,
\[ \Gamma(\cJ)\Pi(U)\Gamma(\cJ)^{-1}(.) = U(.)U^{\dagger}. \]
Thus the action of the $U\otimes \bar{U}$ corresponds to the unitary action $X \rightarrow UXU^{\dagger}$ for all $X \in \cB(\cH)$. We denote this representation of the unitary group  as $\tilde{\Pi}(U)$.

Now note that the commutant $\{\tilde{\Pi}(U), \; U \in \cU(\cH)\}'$ is of dimension two, and is spanned by $I$ and $X \rightarrow \tr(X) I = \langle I |X\rangle |I\rangle$. In turn, the commutant $\{\Pi(U), U\in \cU(\cH)\}'$ is spanned by $I_{\cH\otimes\cH}$ and $|\psi_{o}\rangle\langle\psi_{0}|$, where $ |\psi_{0}\rangle = \frac{1}{\sqrt{d}}\sum_{i=1}^{d}|ii\rangle$, is the maximally entangled pure state in the canonical basis. This in turn implies that a state $\rho$ that commutes with all the unitaries $\Pi(U)$ must be a convex combination as given in Eq.~\eqref{eq:werner}.
\end{proof}

The state $\rho_{F}$ is the generalized {\bf Werner State}. When $F \geq \frac{1}{d^{2}}$, $\rho_{F}$ can be rewritten as
\[ \rho_{F} = p |\psi_{0}\rangle\langle\psi_{0}| + (1-p)\frac{I}{d^{2}}, \]
where $I$ is the identity operator on $\cH\otimes\cH$. Note that while the Schmidt rank of the maximally entangled state is $d$, the Schmidt rank of the {\it maximally mixed state} $I/d^{2}$ is $1$. Such a mixture of $|\psi_{0}\rangle\langle\psi_{0}|$ and $I/d^{2}$ can be realized by the averaging the action of the unitary group $\Pi(U)$ on state $\rho$, that is, by performing the following operation:
\[ \int dU (U \otimes \bar{U}) \rho (U \otimes \bar{U})^{\dagger}. \]

Our goal is to now evaluate the Schmidt number of $\rho_{F}$. We first note that when $F=1$, $\rho_{F}$ has Schmidt number $d$. Furthermore, for any state $|\psi\rangle \in \cH\otimes\cH$,
\begin{equation}
 \langle\psi|\rho_{F}|\psi\rangle = \frac{1}{d^{2}-1}\left[(d^{2}F-1)|\langle\psi_{0}|\psi\rangle|^{2} + 1-F\right].
\end{equation}
It is then an easy exercise to evaluate the maximum value of $\langle\psi|\rho_{F}|\psi\rangle$.
\begin{exer}\label{exer:rhoF_max}
 If $F \geq \frac{1}{d^{2}}$,
\[ \max_{|\psi\rangle \in \cH\otimes\cH} \langle\psi|\rho_{F}|\psi\rangle = F, \] which is attained when $|\psi\rangle = |\psi_{0}\rangle$. If $F \leq \frac{1}{d^{2}}$,  \[ \max_{|\psi\rangle \in \cH\otimes\cH}\langle\psi|\rho_{F}|\psi\rangle  = \frac{1-F}{d^{2}-1},\] and this value is attained when $|\psi\rangle$ is orthogonal to $|\psi_{0}\rangle$.
\end{exer}

The central result we prove in this section is the following result due to Terhal-Horodecki~\cite{TH}.
\begin{thm}\label{thm:FSchmidt}
Let $\frac{k-1}{d} \leq  F \leq \frac{k}{d}$, with $F \geq \frac{1}{d^{2}}$. Then, the Schmidt number of $\rho_{F}$ is equal to $k$.
\end{thm}
Thus, even though the parameter $F$ changes continuously, the Schmidt number of $\rho_{F}$ changes at discrete values of $F$, remaining constant for the intermediate values.

Before venturing to prove Theorem~\ref{thm:FSchmidt}, we first note the following general result which proves a lower bound on the Schmidt number of a general state $\rho \in \cH\otimes\cH$, given a bound on its expectation value with a maximally entangled state.
\begin{prop}\label{prop:minSchmidt}
 Suppose for a state $\rho \in \cH\otimes\cH$ and a maximally entangled state $|\psi\rangle \in \cH\otimes\cH$, $\langle\psi|\rho|\psi\rangle > \frac{k}{d}$. Then,
\begin{equation}
{\rm Schmidt \; number} (\rho) \geq k + 1.
\end{equation}
\end{prop}
The proof follows from the following Lemma and its corollary discussed below.
\begin{lem}\label{lem:A}
Let $|\psi\rangle = \sum_{j=1}^{d}\lambda_{j}|x_{j}\rangle|y_{j}\rangle$ be a state in Schmidt form. Then,
\begin{equation}
 \max\{ |\langle\psi|\psi'\rangle|, \; |\psi'\rangle : {\rm max. \; entangled \; state}\} = \frac{1}{\sqrt{d}}\sum_{j}\lambda_{j}.
\end{equation}
\end{lem}
\begin{proof}
Choose $|\psi'\rangle = \frac{1}{\sqrt{d}}\sum_{j=1}^{d}|\xi_{j}\rangle|\eta_{j}\rangle$, where, $\{|\xi_{j}\rangle\}$ and $\{|\eta_{j}\rangle\}$ are any two orthonormal bases. Then,
\[ \langle\psi|\psi'\rangle = \frac{1}{\sqrt{d}}\sum_{j,l}\lambda_{j}\langle x_{j}|\xi_{l}\rangle\langle y_{j}|\eta_{l}\rangle .\]
Choose unitaries $U, V$, such that $U|\xi_{l}\rangle = |l\rangle$, $V|\eta_{l}\rangle = |l\rangle$. Let $C$ denote conjugation with respect to the standard basis $\{|l\rangle\}$. Then,
\begin{eqnarray}
 \langle \psi|\psi'\rangle &=& \frac{1}{\sqrt{d}}\sum_{j,l}\lambda_{j}\langle U \, x_{j}|l\rangle\langle V \, y_{j}|l\rangle \nonumber \\
&=& \frac{1}{\sqrt{d}}\sum_{j,l}\lambda_{j} \langle U \, x_{j}|l\rangle\langle l |CV \, y_{j}\rangle \nonumber \\
&=& \frac{1}{\sqrt{d}}\sum_{j} \lambda_{j}\langle U \, x_{j}|CV \, y_{j}\rangle.
\end{eqnarray}
Note that, since $|U \, x_{j}\rangle, |CV \, y_{j}\rangle$ are unit vectors and $\lambda_{j}\geq 0$, $|\langle \psi|\psi'\rangle | \leq \frac{1}{\sqrt{d}}\sum_{j}\lambda_{j}$.

Now, let $|\xi_{l}\rangle = |l\rangle$, for all $l=1, \ldots, d$, so that $U \equiv I$. We choose the unitary $V$ that satisfies $VC |y_{j}\rangle = |x_{j}\rangle$, for all $j=1,\ldots, d$. Then, we have,
\begin{eqnarray}
\langle\psi|\psi'\rangle &=& \frac{1}{\sqrt{d}} \sum_{j,l}\lambda_{j}\langle x_{j}|l\rangle\langle C \, \eta_{l}|C \, y_{j}\rangle \nonumber \\
&=& \frac{1}{\sqrt{d}} \sum_{j,l}\lambda_{j}\langle x_{j}|l\rangle\langle VC \, \eta_{l}|VC \, y_{j}\rangle \nonumber \\
&=& \frac{1}{\sqrt{d}} \sum_{j,l} \lambda_{j}\langle x_{j}|l\rangle\langle l |x_{j}\rangle = \frac{1}{\sqrt{d}} \sum_{j=1}^{d} \lambda_{j},
\end{eqnarray}
where we have used, $C|\eta_{l}\rangle = V^{\dagger}|l\rangle $, so that, $|\eta_{l}\rangle = CV^{\dagger}|l\rangle$.
\end{proof}

We now have the following Corollary to Lemma~\ref{lem:A}
\begin{cor}\label{cor:A}
Let ${\rm Schmidt \; Number}(|\psi\rangle) = k$. Then,
\begin{equation}
\max\{|\langle \psi|\psi'\rangle|, \; |\psi'\rangle : {\rm max. \; entangled \; state}\} \leq \sqrt{\frac{k}{d}}.
\end{equation}
\end{cor}
\begin{proof}
$|\psi\rangle = \sum_{j=1}^{k}\lambda_{j}|u_{j}\rangle|v_{j}\rangle$, in Schmidt form, with $\lambda_{j} = 0$, if $j\geq k+1$. By Lemma~\ref{lem:A},
\begin{eqnarray}
\max\{|\langle \psi|\psi'\rangle|, \; |\psi'\rangle : {\rm max. \; entangled \; state}\} &\leq& \frac{1}{\sqrt{d}}\left(\lambda_{1} + \lambda_{2}+ \ldots + \lambda_{k}\right) \nonumber \\
&\leq& \frac{1}{\sqrt{d}} k^{1/2}\left(\lambda_{1}^{2} + \lambda_{2}^{2} + \ldots + \lambda_{k}^{2}\right)^{1/2}  = \sqrt{\frac{k}{d}}.
\end{eqnarray}
\end{proof}

It is now easy to prove Prop.~\ref{prop:minSchmidt}.
\begin{proof}
Assume, to the contrary, that the Schmidt number of $\rho$ is less than $k$. Recall that such a state $\rho$ can be written as a convex combination of pure states $|\phi\rangle\langle\phi| \in \cS_{k}$, where the convex set $\cS_{k}$ was defined in Eq.~\eqref{eq:Sk_defn}, as follows:
\[ \rho = \int_{\cS_{k}} |\phi\rangle\langle\phi| \mu d(\phi) \; .\]
If $|\psi\rangle$ is maximally entangled, then, Corollary~\ref{cor:A} implies
\begin{equation}
 \langle\psi|\rho|\psi\rangle = \int_{\cS_{k}} |\langle\phi|\psi\rangle|^{2} \mu d\phi \leq \frac{k}{d}.
\end{equation}
Therefore, if $\rho$ has the property that $\langle \psi|\rho|\psi\rangle \geq \frac{k}{d}$ for any maximally entangled state $|\psi\rangle$, \[ {\rm Schmidt \; number}(\rho) \geq k \].
\end{proof}

Now, consider the Werner state $\rho_{F}$ with $F=\frac{k}{d}$ for some $k\leq d$. It is left as an exercise for the reader to check that  such a state can be obtained by averaging the action of the unitary group $\Pi(U)$ on an  entangled state of Schmidt rank $k$.
\begin{exer}
Show that for $|\psi_{k}\rangle = \frac{1}{\sqrt{k}}\sum_{i=1}^{k}|ii\rangle$,
\begin{equation}
 \int (U\otimes\bar{U})|\psi_{k}\rangle\langle\psi_{k} |(U\otimes \bar{U})^{\dagger} = \rho_{k/d}, \label{eq:rhok}
\end{equation}
where $\rho_{k/d}$ is the Werner state with $F = \frac{k}{d}$.
\end{exer}
Clearly, ${\rm Schmidt \; Number}(\rho_{k/d}) \leq k$. We will now prove that the Schmidt number of $\rho_{k/d}$ is in fact exactly equal to $k$.

\begin{prop}
The state $\rho_{k/d}$ defined in Eq.~\eqref{eq:rhok} has Schmidt number equal to $k$.
\end{prop}
\begin{proof}
We know from Ex.~\ref{exer:rhoF_max}, that the state $\rho_{k/d}$ satisfies
\[ \max\{\; \langle \psi|\rho_{k/d}|\psi\rangle, \; |\psi\rangle : \; {\rm max. \; entangled \; state}\} = \frac{k}{d} \geq \frac{k-1}{d}. \]
Then, it is a direct consequence of Prop.~\ref{prop:minSchmidt} that,
\[ {\rm Schmidt \; number}(\rho) \geq k. \]
But we already know that ${\rm Schmidt \; number}(\rho) \leq k$, thus proving that ${\rm Schmidt \; number}(\rho) = k.$
\end{proof}

The final ingredient in proving Theorem~\ref{thm:FSchmidt} is the following property of Werner states with $\frac{1}{d^{2}} \leq F \leq \frac{k}{d}$.
\begin{prop}\label{prop:werner_convex}
For $\frac{1}{d^{2}} \leq F \leq \frac{k}{d}$, the Werner state $\rho_{F}$ is a convex combination of $\frac{I}{d^{2}}$ ($I$ being the identity operator on $\cH\otimes\cH$) and $\rho_{k/d}$. \end{prop}
\begin{proof}
Let,
\[ \rho_{F} = (1-\theta) \frac{I}{d^{2}} + \theta \rho_{k/d}. \]
Recall that $\rho_{k/d}$ is itself a convex combination of the maximally entangled state $|\psi_{0}\rangle\langle\psi_{0}|$ and the state orthogonal to it, as stated in Eq.~\eqref{eq:werner}. Therefore,
\begin{eqnarray}
\rho_{F} &=& (1-\theta)\frac{I}{d^{2}} + \theta\left[\left(1- \frac{k}{d}\right)\frac{I - |\psi_{0}\rangle\langle\psi_{0}|}{d^{2}-1} + \frac{k}{d}|\psi_{0}\rangle\langle\psi_{0}| \right] \nonumber \\
&=& \left[\frac{1-\theta}{d^{2}} + \theta\frac{1 - k/d}{d^{2}-1}\right](I - |\psi_{0}\rangle\langle\psi_{0}|)  + \left[\frac{1-\theta}{d^{2}} + \theta\frac{k}{d}\right]|\psi_{0}\rangle\langle\psi_{0}| .
\end{eqnarray}
Comparing with the canonical form of the Wener state given in Eq.~\eqref{eq:werner}, we see that the parameter $F$ corresponding to $\rho_{F}$ is,
\[ F  = \frac{1-\theta}{d^{2}} + \theta \frac{k}{d}.\]
In other words, the parameter $\theta$ satisfies
\[ \theta = \frac{F - \frac{1}{d^{2}}}{\frac{k}{d} - \frac{1}{d^{2}}}.\]
Clearly, $0 \leq \theta \leq 1$ iff
\[ \frac{1}{d^{2}} \leq F \leq \frac{k}{d}.\]
\end{proof}

Finally, we note an important corollary of Prop.~\ref{prop:werner_convex}.
\begin{cor}\label{cor:werner}
${\rm Schmidt \; number}(\rho_{F}) \leq k$ if $\frac{1}{d^{2}} \leq F \leq \frac{k}{d}$.
\end{cor}

The proof of Theorem~\ref{thm:FSchmidt} now follows easily. Recall that the Werner state $\rho_{F}$ for $\frac{k-1}{d}\leq F\leq \frac{k}{d}$ satisfies (Ex.~\ref{exer:rhoF_max})
\[ \langle\psi|\rho_{F}|\psi\rangle \geq \frac{k-1}{d}.\]
Then, Prop.~\ref{prop:minSchmidt} implies
\[ {\rm Schmidt \; number}(\rho_{F}) \geq k .\]
However, Prop.~\ref{prop:werner_convex} and Cor.~\ref{cor:werner} imply that
\[ {\rm Schmidt \; number}(\rho_{F}) \leq k , \]
thus proving Theorem~\ref{thm:FSchmidt}.

An interesting open question is to estimate the Schmidt number of the Werner state $\rho_{F}$ for the parameter range $0\leq F \leq \frac{1}{d^{2}}$.


\chapter{Operator Systems}

\bigskip \noindent
{\bf Chapter Abstract:}

{\it In the first two sections of this chapter, which is based on
  lectures by Vern Paulsen, we will take a closer look at the results
  of Choi and how they explain some results on quantum error
  correction. In particular, we will look at the Knill-Laflamme result
  and Shor's code. One novel aspect of our approach, is that we will
  introduce Douglas' Factorization Theorem, which can be used to
  replace many calculations.

A quantum channel \index{quantum channel} is always defined as a
completely positive, trace preserving mapping. The natural setting to
discuss completely positive mappings is an {\em operator system}. So
our last three sections will be an introduction to some topics in the
theory of operator systems that we believe are important for anyone
interested in quantum information theory. The chapter concludes with
some remarks on the equivalence of the Connes Embedding Problem,
Kirchberg's conjecture on $C^*(\F_\infty)$ having the Weak Expectation
Property, and other statements, including even a nuclearity assertion
about some finite dimensional operator systems.}

\section{Theorems of Choi}

\noindent {\bf Notations and Conventions}

Vectors $v \in \mbb{C}^n$ will be treated as column vectors $v =
\left[ \begin{array}{c} \alpha_1 \\ \vdots \\ \alpha_n\end{array}
    \right]$. In short, $\mbb{C}^n\equiv M_{n\times 1}$. And for such a
column vector $v$, its adjoint gives a row vector $v^*:=
[\overline{\alpha}_1, \ldots, \overline{v}_n]$.

\noindent {\em Bra and Ket Notations:} For $v, w \in \mbb{C}^n$, we shall,
following the physicists, write $|v\rangle := v$ and $\langle w | : =
w^*$. Thus, if $v =
\left[ \begin{array}{c} \alpha_1 \\ \vdots \\ \alpha_n\end{array}
    \right]$ and $w
    = \left[ \begin{array}{c} \beta_1 \\ \vdots \\ \beta_n\end{array} \right]$,
    then $v^*w = \sum_j
\overline{\alpha}_i \beta_i = \langle v | w \rangle$ and $wv^*=
         [\beta_i \overline{\alpha}_j ]=: | w \rangle \langle v|.$ (In
         particular, our inner products will be linear in the second
         variable and conjugate linear in first.)

\subsection{Douglas Factorization}

\bigskip We know by the spectral theorem that for positive semidefinite matrix $P \in M_n^+$, if $u_1,\cdots,u_r$ are non-zero eigenvectors
($r = rank (P)$) with eigenvalues $p_1,\cdots, p_r$ and if
$v_i=\sqrt{p_i}u_i$, then $P=\sum_{i=1}^r v_iv_i^*$. This is known as
the {\em spectral decomposition,} \index{spectral decomposition} but
there are many ways to decompose a positive semi definite matrix as a
sum of rank-1 matrices. Let $P=\sum_{i=1}^m w_iw_i^*$ be another such.
To find a relation between the $v_i$'s and $w_i$'s, we need the
following proposition:

\begin{prop}[Douglas' Factorization Theorem] \index{Theorem! Douglas' Factorisation}
 Let $A,B$ be two bounded operators on a Hilbert space $H$ such that
 $B^*B \le A^*A$. Then there exists a unique $C \in B(R(A) , R(B))$
 such that $||C|| \le 1$ and $CA=B$ (where $R(A)$ is range of $A$).
\end{prop}

Observe that the conclusion in the above proposition of Douglas says
loosely that $B^*B \le A^*A $ implies $ A$ ``divides'' $B$ (in the above sense).

\begin{proof}
Define $C(Ah)=Bh$ for $h \in H$. Then well-definedness of the map $C$
and the other assertions of the proposition follow from the inequality in the
hypothesis.
\end{proof}

\begin{prop}\label{changeofbasis}
Let $P \in M_n^+$ with $rank(P)=r$ and representations
    $P=\sum_{i=1}^r v_iv_i^*= \sum_{i=1}^m w_iw_i^*$. Then:
 \begin{enumerate}
  \item  there is an isometric `change of
    basis', that is, there exists a unique isometry $ U=[u_{i,j}]_{m \times
    r}$, i.e., $ U^*U=1_r$, such that $w_i=\sum_{i=1}^ru_{i,j}v_j$ for all $1 \leq i \leq
    m$.
 \item $span \{v_1,\cdots,v_r\}=span\{w_1,\cdots,w_m\}=R(P)$.
 \end{enumerate}
\end{prop}

\begin{proof} (1)
Let $V=[ v_1\  \vdots\ v_2 \ \vdots\  \cdots\   \vdots\   v_r ]$ and
$W= [ w_1\ \vdots \ w_2\ \vdots\  \cdots\  \vdots\  w_m ]$.  Note that each
vector $v_i$ (or $w_i$) is an $n \times 1$ column vector, thus
$V \in M_{n \times r}$ and $ W \in M_{n \times m}$.

The set $\{v_i: i=1,\cdots,r\}$ is linearly independent since $\mathrm{rank}(P)=r$; in
particular, $V:\mathbb{C}^r \rightarrow \mathbb{C}^n$ is injective and
$V^*:\mathbb{C}^n \rightarrow \mathbb{C}^r$ is thus surjective.

Also $P=VV^*=WW^*$; so, by Douglas' factorization, there exists a
contraction $U_1 : R(V^*)=\mathbb{C}^r \rightarrow R(W^*) \subset
\mathbb{C}^m$ such that $W^*=U_1V^*$. For the same reason, there also
exists a contraction $C:R(W^*) \rightarrow R(V^*) = \mbb{C}^r$ such that
$V^*=CW^*$. Hence $(CU_1)V^*=V^*$. As $V^*$ is surjective, this shows that $
CU_1=1_r$ everywhere.

Hence $C$ and $U_1$ must be partial isometries and $CU_1 = 1_r$
implies the injectivity of $U_1$; we, therefore, see that $U_1$ must be an
isometry. The complex conjugate matrix $U = \bar{U_1}$ is seen to
satisfy the requirements of the proposition.

(2) follows from the fact that $U$ as above is an isometry.
\end{proof}

\subsection{Choi-Kraus Representation and Choi Rank}

\begin{thm} [Choi's first theorem] \label{choi1} \index{Theorem! Choi's first }
 Let $\Phi:M_n \rightarrow M_d$ be linear. The following conditions on $\Phi$ are equivalent:
\begin{enumerate}
 \item $\Phi$ is completely positive.
 \item $\Phi$ is $n$-positive.
 \item $P_{\Phi}=(\Phi(E_{i,j})) \in M_n(M_d)^+$ where $E_{i,j}$ are the standard matrix units of $M_n$.
 \item $\Phi(X)= \sum_{i=1}^r A_iXA_i^*$ for some $r,d$ and $A_i \in M_{d \times n}$ (Choi-Kraus representation). \index{Choi-Kraus! representation}
\end{enumerate}
\end{thm}

\begin{proof}
 It is easy to see that $(1) \Rightarrow (2) \Rightarrow (3)$ and $(4)
 \Rightarrow (1)$. So we will only prove $(3) \Rightarrow (4)$.

Let $r=\mathrm{rank}(P_{\Phi})$. Since, $P_{\Phi} \in
M_n(M_d)^+=M_{nd}^+$, as above, there exist vectors $v_i \in
\mathbb{C}^{nd} , 1 \leq i \leq r $ such that $P_{\Phi} = \sum_{i=1}^r
v_iv_i^*$. Suppose $v_i=\begin{pmatrix} \alpha_1^i \\ \vdots
\\ \alpha_n^i \end{pmatrix} \in \mathbb{C}^{nd}$, where each $ \alpha_j^i
= \begin{pmatrix} \alpha_j^i(1) \\ \vdots
  \\ \alpha_j^i(d) \end{pmatrix} \in \mathbb{C}^d$.

Let $A_i=[ \alpha_1^i \ \vdots\ \cdots\
\vdots \ \alpha_n^i ]_{d \times n}$. It is now easy to check that
$\Phi(E_{i,j}) = \sum_l \alpha_i^l
(\alpha_j^l)^*= \sum_{l=1}^r A_lE_{i,j}A_l^*$ for all $1 \leq i,j \leq n$. Hence (4) holds.
\end{proof}

\begin{rem}\label{choi1conv}
If we write $P_{\Phi}$ as $\sum_{i=1}^m w_iw_i^*$, where $w_i
= \begin{pmatrix} \beta_1^i \\ \vdots \\ \beta_n^i \end{pmatrix} \in
\mbb{C}^{nd}$, with $\beta^i_j \in \mbb{C}^d$, then
$\Phi(X)=\sum_{i=1}^m B_iXB_i^*$, where $B_i:= [
  \beta_1^i\ \vdots\ \cdots\ \vdots\ \beta_n^i ]_{d \times n}$. Conversely, if $\Phi(X)=\sum_{i=1}^m B_iXB_i^*$ for some $B_i= [
  \beta_1^i \ \vdots \ \cdots \ \vdots\ \beta_n^i ]_{d \times n}$, and
if we let $w_i = \begin{pmatrix} \beta_1^i \\ \vdots
  \\ \beta_n^i \end{pmatrix}$, then $P_{\Phi}=\sum_{i=1}^m w_iw_i^*$.
\end{rem}

We now apply the Choi-Kraus result to characterize quantum channels.

\begin{prop}
Let $\mcal{E}:M_n \rightarrow M_n$ be CPTP \index{CPTP} (i.e., CP and
trace preserving). Then there exist $E_i
\in M_n$, $1 \leq i \leq r$ satisfying $\sum_{i=1}^rE_i^*E_i=1_n$ such
that $ \mcal{E}(X)=\sum_{i=1}^rE_iXE_i^*$ for all $X \in M_n$.
\end{prop}

\begin{proof}
 Since $\mcal{E}$ is CP, from Theorem \ref{choi1}, we know of the
 existence of matrices $E_i \in M_n$ such that $\mcal{E}$ is decomposed as a
 sum as above. However, $\mcal{E}$ is trace preserving $\Rightarrow
 tr(\mcal{E}(X))=tr(X)\ \forall\  X \Rightarrow tr(\sum_{i=1}^rE_iXE_i^*)=tr(X)\
 \forall\ X \Rightarrow tr(\sum_{i=1}^rE_i^*E_iX)=tr(X)\ \forall\  X
 \Rightarrow \langle \sum_{i=1}^rE_i^*E_i, X^* \rangle = \langle 1_n, X^*
 \rangle, \forall X$ and hence $\sum_{i=1}^rE_i^*E_i=1_n$.
\end{proof}

\begin{rem}
 It is easy to see that conversely, $\sum_{i=1}^rE_i^*E_i=1_n
 \Rightarrow \mcal{E}$ is trace preserving.
\end{rem}

\begin{defn}[Choi rank of a CP map] \index{Choi rank} Let $\Phi:M_n \rightarrow
 M_d$ be CP. Then, by the above theorem, $\Phi(X) = \sum_{i=1}^q
 B_iXB_i^*$ for some matrices $B_i \in M_{d,n}$, $1 \leq i \leq q$.
The Choi rank of $\Phi$ is given by $\mathrm{cr}(\Phi) := \min\{q: \Phi(X)
 = \sum_{i=1}^q B_iXB_i^*\}$.
\end{defn}

\begin{prop}[Choi]\label{choirank}
In the above set up, $\mathrm{cr}(\Phi)=rank(P_{\Phi})$.
\end{prop}

\begin{proof}
 This follows from Theorem \ref{choi1},  Remark \ref{choi1conv} and  Proposition \ref{changeofbasis}.
\end{proof}

\begin{thm} [Choi's second theorem]\label{choi2} \index{Theorem! Choi's second}
 Let $\Phi \in CP( M_n ,M_d)$ with $\mathrm{cr}(\Phi) =
 r$. Suppose $\Phi(X)= \sum_{i=1}^r A_i X A_i^*= \sum_{i=1}^m B_i X B_i^*$
 are two Choi-Kraus representations of $\Phi$.  Then there exists a
 unique matrix $U=(u_{i,j}) \in M_{m\times r}$ such that $U^*U=1_r,
 B_i=\sum_{j=1}^r u_{i,j}A_j$ and  $span\{A_1,\cdots,A_r\}=span\{B_1,\cdots,B_m\}$.
\end{thm}

\begin{proof}
 This follows from Proposition \ref{changeofbasis}, Theorem \ref{choi1} and Remark \ref{choi1conv}.
\end{proof}

\subsubsection{Example : Binary Case  Quantum Error Detection/Correction}

This is an introductory binary example to motivate the quantum error
correction/detection of the next section.  Let us take a binary string
of $0$s and $1$s of length, say 5, e.g. $s=(0,1,0,1,1)
\in \mathbb{Z}_2^5$. We want to transmit this. Some errors
may occur in transmission. We want to detect/correct that error.

One way to detect/correct error is to encode the given string into a
larger vector.  One  famous binary error correcting codes
is the \textbf{Majority Rule Code}: \index{Majority Rule Code}

We start with our original vector $s$ of length $r$ (here $r=5$) and
encode it within a vector of length $nr$ for some odd $n$, where each
digit gets repeated $n$ times consecutively.

Say $n=3$ and $s=(0,1,0,1,1)$ then this vector gets encoded into
\[s_1=(0,0,0;1,1,1;0,0,0;1,1,1;1,1,1)\] (note that the encoded vectors
form a 5-dimensional subspace of a 15-dimensional space).

After transmission, suppose the output vector turns out to be, say
\[s_2=(1,0,0;1,1,1;0,0,1;0,0,1;1,1,0).\]

We decode/recover the string by choosing the digit (0 or 1) that
appears as a majority among each block of $n(=3)$ consecutive
digits. Thus, in each block of three consecutive digits, the majority
rules.  Here the recovered string will be $s_3=(0,1,0,0,1)$.

Schematically:
$s \xrightarrow{encode} s_1 \xrightarrow{transmit} \ldots \xrightarrow{receive} s_2
\xrightarrow{decode} s_3.$

Note that there is one error and in fact to have one error in the
output there has to be at least two (i.e., more than n/2) errors within some block of $n$
digits.  Knowing probabilities of those errors will help to understand
how effective a code is.

Alternatively, instead of encoding cleverly, one could pick a clever
5-dimensional subspace of the 15-dimensional space and then any
embedding of $\mbb{Z}_2^5$ onto that subspace would be the encoding..

Note that this code requires that we ``clone'' each digit three
times. Thus, it violates the ``no cloning'' rule and could not be
implemented on a quantum machine.

In the next section, we will present Shor's code, which is a quantum
code that is robust to many types of errors.


\section{Quantum error correction}

If we assume that any errors that occur must also be the result of some quantum event, then it is natural to assume that errors are
also the result of the action of a CPTP map acting on the states. Thus,
an {\em error operator} \index{error operator} will be a CPTP map $\mcal{E}$. The usual strategy/protocol in quantum error correction
is the following:

As in the above example, we don't expect to be able to correct all
errors that occur.  If $v \in V$ is a state in a `protected' subspace
$V$, we want to correct $\mcal{E}(|v\rangle \langle v|)$, which is the
error that has happened to the state. To do this, we seek a recovery
operator $\mcal{R} :M_n
\xrightarrow{CPTP} M_n$ such that $\mcal{R}(\mcal{E}(|v \rangle \langle
v|))=|v\rangle \langle v|, \forall v \in V$.

The approach appears somewhat naive since it appears that it assumes
that we know the error explicitly.  However using Choi's theorem it
turns out, that correcting $V$ against one assumed error leads to its
correction against a whole family of errors. I believe that this is
the key element of Knill-Laflamme error correction.

\subsection{Applications of Choi's Theorems to Error Correction}

\begin{thm}[Knill-Laflamme] \label{krill1} \index{Theorem! Knill - La Flamme}
Let $V \subset \mathbb{C}^n$ be a subspace and $P:\mathbb{C}^n
\rightarrow V$ be its orthogonal projection.  Let $\mcal{E}:M_n
\rightarrow M_n$ be CPTP (an error map) given by
$\mcal{E}(X)=\sum_{i=1}^m E_iXE_i^*$ for all $X \in M_n$.  Then there
exists a (recovery) \index{recovery operator} CPTP map $\mcal{R}:M_n \rightarrow M_n$ such that
$\mcal{R}(\mcal{E}(PXP))=PXP,\ \forall\ X \in M_n$ if, and only if,  $
PE_i^*E_jP=\alpha_{i,j}P$ for some $\alpha_{i,j} \in \mathbb{C}$, $ 1
\leq i, j \leq m$.
\end{thm}

\begin{proof}
 $(\Rightarrow)$ Suppose a recovery map $\mcal{R}$ has a
  representation $\mcal{R}(W)=\sum_{l=1}^q A_lWA_l^*$ for all $ W \in
  M_n$ with $ \sum A_l^*A_l =1$.  Since $\mcal{R}$ and $\mcal{E}$ are
  both CPTP, so is $\mcal{R} \circ\mcal{ E}$. Further, the equality
  $(\mcal{R}\circ \mcal{E})(PXP) = \mcal{R}(\mcal{E}(PXP)) =
  \sum_{i=1}^m \sum_{l=1}^q (A_lE_iP)X(PE_i^*A_l^*) = PXP$, gives two Choi-Kraus
  representations for $\mcal{R}\circ \mcal{E}$, and $PXP$ is of
  minimal length. Hence, by Theorem \ref{choi2}, there
  exists a unique isometry $U =[\beta_{ij}]_{mq \times 1}$ such that
  $A_lE_iP=\beta_{l i}P$. Since $U=\begin{pmatrix} \beta_{11} \\ \vdots
  \\ \beta_{1m} \\ \beta_{21} \\ \vdots \\ \beta_{2m} \\ \vdots
  \\ \beta_{q1} \\ \vdots \\ \beta_{qm} \end{pmatrix}$ is an isometry,
  we have $\sum_{l,i} |\beta_{li}|^2= U^*U= 1$. In particular, $
  (PE_i^*A_l^*)(A_lE_jP)=(\bar{\beta}_{li}P)(\beta_{lj}P)=\bar{\beta}_{li}\beta_{lj}P$
  and, using $\sum_l A_l^* A_l = 1_n$, we have
\[
\sum_{l=1}^q \bar{\beta_{li}}\beta_{lj}P=\sum_{l=1}^q
PE_i^*(A_l^*A_l)E_jP=PE_i^*E_jP.
\]
Simply take $\alpha_{i,j} = \sum_l \bar{\beta}_{li}\beta_{lj}$ for the
conclusion.
\vspace*{2mm}

\noindent $(\Leftarrow)$ We may clearly assume $P \neq 0$ so $tr(P) > 0$, as the
Theorem is vacuously true in the contrary case!  Suppose there exist
$\alpha_{i,j} \in
\mathbb{C}$, $ 1 \leq i, j \leq m$ such that
$PE_i^*E_jP=\alpha_{i,j}P$ for all $ 1 \leq i, j \leq m$. Then
$(\alpha_{i,j}P)=(PE_i^*E_jP)=\begin{pmatrix} PE_1^* \\ \vdots
\\ PE_m^* \end{pmatrix} \begin{pmatrix} E_1P, & \ldots, &
  E_mP \end{pmatrix} \ge 0$. This implies $(\alpha_{i,j}) \in M_m^+$
(since $[\alpha_{i,j}] = \frac{1}{tr_n(P)} (id \otimes tr_n)
([\alpha_{i,j}]\otimes P)$); also, $\sum_i \alpha_{i,i}P = \sum_i
PE_i^*E_iP= P^2 = P \Rightarrow tr_m((\alpha_{i,j})) =  \sum_i \alpha_{i,i}
=1$, i.e., $[\alpha_{ij}]$ is a density matrix.

Thus, $(\alpha_{i,j})$ is unitarily diagonalizable, i.e., there exists
a unitary $U=(u_{i,j}) \in M_m$ such that $U(\alpha_{i,j})U^* = D
=\ \mathrm{diag}(d_{11},\cdots, d_{mm})$, with $d_{ii} \ge 0 ~\forall i$ and $\sum_{i=1}^m d_{ii} = 1$. Set $F_i=\sum_{j=1}^m
\bar{u}_{i,j}E_j$, $ 1\leq i \leq m$. Hence, by Choi (or by direct
calculation), $\sum_{i=1}^m F_iXF_i^*= \mcal{E}(X)$ for all $X \in
M_n$.  Also, note that
\begin{eqnarray*}
PF_i^*F_jP &=& P (\sum_{k,l} u_{i,k}E_k^* \bar{u}_{j,l}E_l )P
\\
&=& \sum_{k,l}
u_{i,k} \bar{u}_{j,l}PE_k^*E_lP\\ & = & \sum_{k,l}
u_{i,k}\alpha_{k,l}\bar{u}_{j,l}P\\
&=& d_{ij}P~.
\end{eqnarray*}
So we see that if we define
\[V_i = \left\{ \begin{array}{ll} 0 & \mbox{if } d_{ii} = 0\\
\frac{1}{\sqrt{d_{ii}}} F_iP & \mbox{otherwise}
\end{array} \right. ~,\]
we find that the $V_i$'s are partial isometries with $V_i^*V_j = \delta_{ij} P ~\forall i,j$ (i.e., all non-zero ones among them having initial space equal to $V$)  and with $V_iV_i^* = F_i(V)$ (i.e., with pairwise orthogonal final spaces). Hence we see that
$R = \sum_i V_iV_i^* = \sum \frac{1}{d_{ii}}F_iPF_i^*$ is an
orthogonal projection. Let $Q=1-R$.

Define $\mcal{R}:M_n \rightarrow M_n$ by $\mcal{R}(X)=\sum V_i^*XV_i +
QXQ$. Clearly,  $\mcal{R}$ is a unital CPTP map.

We want to show that $\mcal{R}(\mcal{E}(PXP))=PXP ~ \forall\ X \in
M_n$. To see this, it is enough to show that for $v \in V,
\mcal{R}(\mcal{E}(vv^*))=vv^*$ (since $\{PXP : X \in M_n\}=$ all matrices living
on $V$ while $\{vv^*:v \in V\}=$ all rank-$1$ projections living on
$V$). For this, compute thus:
\begin{eqnarray*}
\mcal{R}(\mcal{E}(vv^*)) &=& \mcal{R}(\mcal{E}(Pvv^*P))\\
&=& \sum_i V_i^* \left( \sum_j F_jPvv^*PF_j^* \right) V_i + Q \left( \sum_j F_jPvv^*PF_j^* \right)  Q\\
&=& \sum_{i,j} V_i^*F_jP vv^* PF_j^*V_i^* + \sum_j QF_jPvv^*PF_jQ\\
&=& \sum_{i,j} d_{jj} V_i^*V_j vv^*V_i^* V_j + \sum_j d_{jj} QV_jPvv^*PV_jQ\\
&=& \sum_j d_{jj} Pvv^*P + 0\\
&=& vv^*~,
\end{eqnarray*}
as  desired.
\end{proof}

\begin{thm}\label{krill2}
 Let $\mcal{R}$ be the recovery operator constructed as in Theorem
 \ref{krill1} and $\tilde{\mcal{E}}:M_n \rightarrow M_n$ be another
 error operator admitting a representation
 $\tilde{\mcal{E}}(X)=\sum_{i=1}^p\tilde{E}_iX\tilde{E}_i^*$ with
 $\tilde{E}_i \in \mathrm{span}\{E_1,\ldots, E_m\}$, $1 \leq i \leq
 p$. Then, also  $\mcal{R} \tilde{\mcal{E}}(PXP)=PXP,\ \forall X \in M_n$.
\end{thm}

\begin{proof}
We have $\tilde{E_i} \in \mathrm{span} \{E_1,\cdots, E_m\}$, $1 \leq i
\leq p$ and since $\tilde{\mcal{E}}$ is CPTP, they must satisfy $\sum
\tilde{E_i}^* \tilde{E_i}=1$. Recall that - with the $F_i$ as in the proof of Theorem \ref{krill1} -
we may deduce from Theorem \ref{choi2} that $\mathrm{span}\{E_1,\cdots,E_m\} = \mathrm{span}\{F_1,\cdots, F_m\}$
and $ PF_i^*F_jP= \delta_{i,j} d_{ii}P$, $ 1 \leq i, j \leq m$; so,
\begin{equation}\label{*}
V_i^* F_j P= \delta_{i,j} \sqrt{d_{ii}} P.
\end{equation}

Write $\tilde{E}_i=\sum_{l=1}^m \beta_{i,l} F_l$.  Then, $1=\sum_k
\tilde{E}_k^* \tilde{E}_k = \sum_k (\sum_{l=1}^m \bar{\beta}_{k,l}
F_l^*)(\sum_{j=1}^m \beta_{k,j} F_j)$. So,
\begin{eqnarray*}
P &=& P1P   =   P(\sum_k (\sum_{l=1}^m \beta_{k,l} F_l^*)(\sum_{j=1}^m \beta_{k,j} F_j))P\\
& = & \sum_k \sum_{l,j} \bar{\beta_{k,l}} \beta_{k,j} PF_l^* F_j P\\
& = &  \sum_k \sum_{j=1}^m |\beta_{k,j}|^2d_{jj}P.
\end{eqnarray*}
Hence,
\begin{align}\label{**}
\sum_k \sum_{j=1}^m |\beta_{k,j}|^2 d_{jj} = 1.
\end{align} Hence,
\begin{eqnarray*}
\mcal{R}(\tilde{\mcal{E}}(PXP))&=& \sum_{i,j} V_i^* \tilde{E}_jP X P \tilde{E}_j^*V_i\\
&=& \sum_{i,j,k,l} \overline{\beta_{lj}}\beta_{kj}V_i^*F_l P X P F_k^* V_i\\
&=& \sum_{i,j} |\beta_{ij}|^2 d_{ii} PXP \ \ \ (\text{by Eqn.}\ (\ref{*})\\
& =& PXP ~\forall X \in M_n\ \ \ \ \ (\text{by Eqn.}\ (\ref{**}).
 \end{eqnarray*}
\end{proof}

\subsection{Shor's Code : An Example}

We consider the `protected subspace' $V \subset \mathbb{C}^2 \otimes
\cdots \otimes \mathbb{C}^2 \text{ (9 copies)}$ with
$\mathrm{dim}(V)=2$ given by \index{Shor's Code}
\[
V = span \{|0_L \rangle, |1_L \rangle\},
\]
where $0_L= \frac{1}{2 \sqrt{2}}((|000 \rangle + |111 \rangle) \otimes
(|000\rangle + |111\rangle ) \otimes (|000 \rangle + |111 \rangle))$
and $1_L = \frac{1}{2 \sqrt{2}}((|000 \rangle - |111 \rangle) \otimes
((|000 \rangle - |111 \rangle) \otimes (|000 \rangle - |111
\rangle))$. (Notation: Fixing an orthonormal basis $\{|0\rangle, |1\rangle\}$ for
$\mbb{C}^2$, we write $|000 \rangle$ for $|0 \rangle \otimes |0
\rangle \otimes |0 \rangle$ and so on.)

We consider the Pauli basis of $\mathbb{C}^2 \otimes \cdots \otimes
\mathbb{C}^2 \text{ (9 copies)}$ constructed as follows:

Take the basis of $M_2$ (which is orthonormal with respect to the
  normalised trace-inner-product) consisting of $X = \begin{pmatrix} 0
  & 1 \\ 1 & 0 \end{pmatrix}, Y= \begin{pmatrix} 0 & i \\ -i &
  0 \end{pmatrix}, Z= \begin{pmatrix} 1 & 0 \\ 0 & -1 \end{pmatrix}$
  and $1_2$. (These are regarded as maps on $\mbb{C}^2$ with respect
  to the orthonormal basis $\{|0\rangle,|1\rangle\}$.

1-Pauli elements: For $i=1, \cdots,9$, these are the the $2^9 \times
2^9$ unitary self-adjoint matrices defined by
\begin{align*}
 & 1 = 1_2 \otimes 1_2 \otimes \cdots \otimes 1_2 \\ &X_i = 1_2 \otimes
  1_2 \otimes \cdots \otimes X \otimes \cdots \otimes 1_2\  (X \text{ at }
  i^{th} \mathrm{ position})\\ &Y_i = 1_2 \otimes 1_2 \otimes \cdots \otimes
  Y \otimes \cdots \otimes 1_2\  (Y \text{ at } i^{th} \mathrm{
    position})\\ &Z_i = 1_2 \otimes 1_2 \otimes \cdots \otimes Z \otimes
  \cdots \otimes 1_2\  (Z \text{ at } i^{th} \mathrm{ position}).
\end{align*}
Let us list the above $1$-Paulis \index{$1$-Paulis} as $U_1,\ldots, U_{28}$. Define
$\mcal{E}: M_{2^9} \rightarrow M_{2^9} $ by
$\mcal{E}(X)=\frac{1}{28}\sum_{i=1}^{28}U_iXU_i^*$. Then, it is easily seen
that $\mcal{E}$ is a CPTP map (being an average of $\ast$-automorphisms).

\begin{prop}
With this notation, we have that
\[V, X_1V,\ldots,X_9V,Y_1V,\ldots,Y_9V,Z_1V\] are
all mutually orthogonal and $Z_{3k+i\big|_V} = Z_{3k +j \big|_V}$ for $k = 0,1,2,$ and $0 \le i,j < 3$.
\end{prop}
\begin{proof}  Exercise.
\end{proof}

$\bullet$ Notice that $PU_i^*U_jP \in \{0,P\} ~\forall 1 \leq i,j \leq
28$ and hence Theorem \ref{krill1} ensures the existence of a recovery
operator $\mcal{R}$ satisfying $\mcal{R} \circ \mcal{E} (PXP) = PXP
~\forall X \in M_n$.

By Theorem \ref{krill2}
above, for this protected space $V$ and error map $\mcal{E}$, we have
$\mcal{R}(\tilde{\mcal{E}}(PXP))=PXP$ for any error map
$\tilde{\mcal{E}} : M_n \to M_n$ given by $ \tilde{\mcal{E}}(X)=\sum
\tilde{E}_iX\tilde{E}_i^*$, where $\tilde{E}_i \in
\mathrm{span}\{1$-Pauli basis$\}$.

$\bullet$ The 1-Paulis contain in their span any operator of the form
$1 \otimes \cdots \otimes 1 \otimes A \otimes 1 \cdots \otimes 1$. So
while the Shor code may not fix all errors, it does fix all errors in
${span}\{1$-Paulis$\}$. Thus, if each term in the error operator acts
on only one of the qubits, then this subspace will be protected from
this error and the decoding map will recover the original encoded
qubit.

\begin{rem}
See the work of David Kribs, et al for various generalizations of the
Knill-Laflamme theory, including infinite dimensional versions of this
theory.

A more refined version of the Knill-Laflamme theory than we have
stated is to consider protected {\it operator subsystems} of the
matrices and encodings of states into such subsystems.

In the next lecture we will introduce this concept.
\end{rem}


\section{Matrix ordered systems and Operator systems}
Recall that a typical quantum channel \index{quantum channel} on $M_n$
looks like $\mcal{E}(X) = \sum_{i = 1 }^r E_i X E_i^*$ for some
matrices $E_i \in M_n$, which are not unique. However, given any such
representation of $\mcal{E}$, the space $\mcal{S}:= \text{span }\{E_j
: 1\leq j \leq r \}$ remains the same. Moreover, the space $\mcal{S}$
contains $1$ and is closed under taking adjoints.

Duan, Severini and Winter have argued that various concepts of quantum
capacities of the channel $\mcal{E}$ really only depend on this
subspace $\mcal{S},$ i.e., if two channels generate the same subspace
then their capacities should be the same. Thus, capacities are
naturally functions of such subspaces.

Moreover, in the extensions of the Knill-Laflamme theory, it is
exactly such subspace that are the protected subspaces, i.e., the
subspaces that one wants to encode states into so that they can be
recovered after the actions of some error operators.

A subspace of $M_n$(or more generally, $B(\mcal{H})$) that contains
$1$ and is closed under the taking of adjoints is called an {\it
  operator system.} \index{operator system} These are also the natural
domains and ranges of completely positive maps.

Thus, the concept of an operator system plays an important role in the
study of completely positive maps and, in particular, in QIT. For this
reason we want to introduce their general theory and axiomatic
definitions.

Finally, when we study $M_n= B(\mbb{C}^n)$ we know that the positive
operators of rank one, represent the states of the underlying space
and that positive operators of trace one represent the mixed
states. But when we focus on a more general operator system, what
exactly is it the states of?  One viewpoint is to just regard it as a
restricted family of states of the underlying space.  But this is very
problematical since many operator subsystems of $M_n$ have no rank one
positives and others that have plenty of rank one positives, still
have trace one positives that cannot be written as sums of rank ones!
The correct answer to what is an operator system the states of
involves introducing the concept of {\it (ordered) duals.}

Motivated by these issues and some structural properties of
$M_n$, we introduce a bunch of abstract definitions.

 \begin{defn}
A $\ast$-vector space \index{$\ast$-vector space} is a complex vector space $\mcal{V}$ with a map $\ast :
\mcal{V} \to \mcal{V}$ satisfying
\begin{enumerate}
\item $(v + w)^* = v^* + w^*$;
\item $(\lambda v)^* = \bar{\lambda} v^*$; and
\item $(v^*)^* = v$
\end{enumerate}
for all $v, w \in \mcal{V}$, $\lambda \in \mbb{C}$. The self adjoint
elements of such a space is denoted by $$\mcal{V}_h = \{ v \in
\mcal{V}: v = v^*\}.$$
\end{defn}
$\bullet$ As usual, we have a {\em Cartesian decomposition}
 \index{Cartesian decomposition} in a $\ast$-vector space $\mcal{V}$
 given by $v = \frac{v + v^*}{2} + i \frac{v-v^*}{2i}$ for all $v \in
 \mcal{V}$ and we call these the {\it real} and {\it imaginary} parts
 of $v.$

$\bullet$ Given a $\ast$-vector space $\mcal{V}$ and $n \geq 1$, its
 $n$th-amplification $M_n(\mcal{V}),$ which is just the set of $n
 \times n$ matrices with entries from $\mcal{V}$ inherits a canonical
 $\ast$-vector space structure and is naturally identified with $M_n
 \otimes \mcal{V}$.

\begin{defn}
A matrix ordering \index{matrix ordering} on a $\ast$-vector space
$\mcal{V}$ is a collection $\mcal{C}_n \subset M_n(\mcal{V})_h$, $ n
\geq 1$ satisfying:
\begin{enumerate}
\item $\mcal{C}_n$ is a cone for all $n \geq 1$;
\item $\mcal{C}_n \cap (-\mcal{C}_n) = (0)$ for all $ n \geq 1$; and
\item $A P A^* \in \mcal{C}_{k}$ for all $A \in M_{k, n}, P \in
  \mcal{C}_n$, $k,n \geq 1$.
\end{enumerate}
A $\ast$-vector space $\mcal{V}$ with a matrix ordering as above is
called a matrix ordered space.
\end{defn}

\begin{rem}
Some authors also add the following axiom in the definition of a
matrix ordering:

$ (\tilde{4})\quad \mcal{C}_n - \mcal{C}_n = M_n(V)_h. $

\noindent However, we will abstain from its use because of reasons
that will get clear while discussing `dual of an operator system'.
\end{rem}
\begin{exer}
Let $P \in \mcal{C}_n$ and $Q \in \mcal{C}_k$. Then
$\left[ \begin{array}{cc} P & 0 \\ 0 & Q \end{array}\right] \in \mcal{C}_{n+k}$.
\end{exer}

\begin{exam}
\begin{enumerate}\label{example}
\item $\mcal{V} = B(H)$ with usual adjoint structure and $\mcal{C}_n :=
  M_n(B(H))^+= B(H \otimes \mbb{C}^n)^+$, $n \geq 1$ is matrix ordered.
\item Any subspace $\mcal{V} \subset B(H)$ such that $\mcal{V}$ is
  closed under taking adjoints provides $\mcal{V}$ with the natural induced
  matrix ordering $\mcal{C}_n := M_n(B(H))^+ \cap M_n(\mcal{V})$, $n
  \geq 1$.
 \end{enumerate}
\end{exam}

\begin{defn}\label{def:OS}
An operator system \index{operator system} is a subspace $\mcal{S}
\subset B(H)$ that is closed under taking adjoints and contains the
unit $e_{\mcal{S}}:=id_H$, for some Hilbert space $H$, together with
the matrix ordering given in Example \ref{example}(2).
\end{defn}

$\bullet$ Usually, whenever the matrix ordering is clear from the
context, we simply write $M_n(\mcal{V})^+$ for $\mcal{C}_n$, $n \geq 1$.

\begin{defn}
Given matrix ordered spaces $\mcal{V}$ and $\mcal{W}$, a linear map
$\varphi : \mcal{V} \to \mcal{W}$ is said to be $n$-positive if its
$n$th-amplication $\varphi^{(n)}: M_n(\mcal{V}) \to M_n(\mcal{W})$,
$[v_{ij}] \mapsto [\varphi(v_{ij})]$ is positive, i.e., $\varphi^{(n)}
(M_n(\mcal{V})^+) \subset M_n(\mcal{W})^+$. $\varphi$ is said to be
completely positive (in short, CP) if $\varphi$ is $n$-positive for all $n \geq 1$.
\end{defn}

$\bullet$ Clearly, if we include Axiom $(\tilde{4})$ in the axioms of
matrix ordering, then every CP map is $\ast$-preserving.

\begin{defn}
Two matrix ordered spaces $\mcal{V}$ and $\mcal{W}$ are said to be
completely order isomorphic \index{completely! order isomorphic} if
there is a completely positive linear isomorphism $\varphi : \mcal{V}
\to \mcal{W}$ such that $\varphi^{-1}$ is also completely positive.
\end{defn}

Given two operator systems $\mcal{S}_1, \mcal{S}_2$ on possibly
different Hilbert spaces, we identify them as the ``same'' operator
system when they are completely order isomorphic via a unital complete
order isomorphism.

\subsection{Duals of Matrix ordered spaces}
Let $\mcal{V}$ be a matrix ordered space. Let $\mcal{V}^d$ be the
space of linear functionals on the vector space $\mcal{V}$.

{\em $\ast$-structure:} For each $f \in \mcal{V}^d$ define $f^* \in
\mcal{V}^d$ by $f^*(v) = \overline{f(v^*)},\ v \in \mcal{V}$. This
makes $\mcal{V}^d$ into a $\ast$-vector space.

{\em Matrix ordering:} Given a matrix of linear functionals $[f_{ij}]
\in M_n(\mcal{V}^d)$, identify it with the map $\Phi: \mcal{V} \to
M_n$ given by $\Phi (v) = [f_{ij}(v)]$, $v \in \mcal{V}$. Let
\[
\mcal{C}_n = \{[f_{ij}] \in M_n(\mcal{V}^d) : \Phi: \mcal{V} \to M_n\   is\
  CP\}.
\]
Then, $\mcal{V}^d$ together with above cones forms a matrix ordered
space.

This matrix ordered space is what is meant by the {\it matrix-ordered
  dual} \index{matrix ordered dual} of $\mcal{V}.$

$\bullet$ There are many other ways of making $\mcal{V}^d$ into a
matrix ordered space. However, the above structure has better
compatibility with respect to some important operations on matrix ordered
spaces.

\begin{rem}
In general, if we require $\mcal{V}$ to also satisfy axiom $(\tilde{4})$ in the definition of a
matrix ordering, then it can still be the case that $\mcal{V}^d$ does not satisfy $(\tilde{4}).$
\end{rem}

An immediate compatibility of the above dual structure is seen in the
following:
\begin{prop}
Let $\varphi : \mcal{V} \to \mcal{W}$ be a CP map between two matrix
ordered spaces. Then the usual dual map \index{adjoint! CP map} $\varphi^d :
\mcal{W}^d \to \mcal{V}^d$ is also CP.
\end{prop}

\begin{proof}
Let $n \geq 1$ and $[f_{ij}] \in M_n(\mcal{W}^d)^+$. Then
$(\varphi^d)^{(n)} ([f_{ij}]) = [f_{ij}\circ \varphi]$. Now,
$[f_{ij}\circ \varphi] :\mcal{V} \to M_n$ is given by $v \mapsto
[f_{ij}\circ \varphi(v)]$, which being a composite of the CP maps
$\mcal{V}\stackrel{\varphi}{\to} \mcal{W} \stackrel{[f_{ij}]}{\to}
M_n$ is again CP. Thus, $(\varphi^d)^{(n)} ([f_{ij}]) \in
M_n(\mcal{V}^d)^+$. In particular, $\varphi^d$ is CP.
\end{proof}

Consider the matrix algebra operator system $\mcal{V} =
M_p=\mcal{L}(\mbb{C}^p)$. Let $\{ E_{ij} : 1 \leq i, j \leq p\}$ be
the system of matrix units for $\mcal{V}$. Via this basis, we can
infact identify $\mcal{V}^d$ with $\mcal{V}$ itself. Formally, let
$\{\delta_{ij}\} \subset \mcal{V}^d$ be the dual basis for
$\{E_{ij}\}$.  Given $A = [a_{ij}] \in \mcal{V}$, define $f_A \in
\mcal{V}^d$ by $f_A = \sum_{ij}a_{ij} \delta_{ij}$. Thus, $a_{ij} =
f_A(E_{ij})$ and we see that $A$ is the usual ``density'' matrix of
the linear functional $f_A.$ Note that, for $B = [b_{ij}] \in
\mcal{V}$, we have $f_A(B) = \sum_{ij} a_{ij}b_{ij} =
tr_p(A^tB)$. Clearly, $f_A (B) \geq 0$ for all $B \geq 0$ if and only
if $A \geq 0$. Define $\Gamma : \mcal{V} \to \mcal{V}^d$ by $\Gamma(A)
= f_A$. Then, we have the following:

\begin{thm}\cite[Theorem~6.2]{PTT09}
The map $\Gamma : M_p(\mbb{C}) \to M_p(\mbb{C})^d$ as constructed
above is a complete order isomorphism.
\end{thm}
A natual question to ask at this stage would be whether any other
basis for $M_p$ works equally well or not? And, quite surprisingly,
the answer is not very clear!

Let $\mcal{B} = \{ B_{rs} : 1 \leq r,s \leq p\}$ be any other basis
for $M_p$ and $\{ \eta_{rs} : 1\leq r, s \leq p\} $ its dual
basis. Define $\Gamma_{\mcal{B}} : M_p \to (M_p)^d$ by
$\Gamma_{\mcal{B}}(A) = \sum_{rs} a_{rs}\eta_{rs}$, where $A =
\sum_{rs}a_{rs}B_{rs}$. Then, it is not difficult to find a basis $\mcal{B}$
such that $\Gamma_{\mcal{B}}$ is not a complete order
isormorphism. However, it will be interesting to answer the following:

\begin{ques}
\begin{enumerate}
\item What are the necessary and sufficient conditions on the basis
  $\mcal{B}$ such that the map $\Gamma_{\mcal{B}}$ as above is a complete
  order isomorphism?
\item If the map $\Gamma_{\mcal{B}}$ is positive, then is it automatically
  completely positive?
\end{enumerate}
\end{ques}

See \cite{PS2} for some work on this problem.

 Note that, under the above complete order isomorphism $\Gamma$, we
 have $\Gamma (1) = tr_p$. We will soon see that $1$ and $tr_p$ have
 some interesting significane on the structures of $M_p$ and
 $(M_p)^d$, respectively.

\subsection{Choi-Effros Theorem}

Arveson introduced the concept of an operator system in
1969. Choi-Effros were the first to formally axiomatize the
theory. Their axiomatic characterization follows.

\begin{thm}\cite[Theorem $4.4$]{CE77} \index{Theorem! Choi-Effros}
Let $\mcal{V}$ be a matrix ordered space with an element $e\in
\mcal{V}_h$ satisfying:
\begin{enumerate}
\item For each $n \geq 1$ and $H \in M_n(\mcal{V})_h$, there exists an
  $ r > 0$ such that $r\, \mathrm{diag}(e, e, \ldots, e)_{n\times n} + H \in
  M_n(\mcal{V})^+$. (Such an $e$ is called a matrix order unit.)
\item If $H \in M_n(\mcal{V})_h$ satisfies $r\, \mathrm{diag}(e, e,
  \ldots, e)_{n\times n} + H \in M_n(\mcal{V})^+$ for all $r > 0$,
  then $H \in M_n(\mcal{V})^+$. (Such a matrix order unit is called an
  Archimedean matrix order unit.) \index{Archimedean matrix order unit}
\end{enumerate}
Then there exists a Hilbert space $H$ and a CP map $\varphi : \mcal{V}
\to B(H)$ such that $\varphi (e)= id_H$ and $\varphi$ is a complete
order isomorphism onto its range.
\end{thm}

$\bullet$ The converse of this theorem is clear: Every operator system
clearly is matrix ordered with $e_{\mcal{S}}$ as an Archimedean matrix
order unit. And, the above theorem of Choi and Effros allows us to
realize every Archimedean matrix ordered space with an operator
system. We will thus use the terminology operator system for an
Archimedean matrix ordered space and vice versa.

\begin{thm}\cite[$\S 4$]{CE77}\label{unit}
Let $\mcal{S} \subset B(H)$ be a finite dimensional operator
system. Then,
\begin{enumerate}
\item there exists an $f \in (\mcal{S}^d)^+$ such that $f(p) > 0$ for
  all $ p \in \mcal{S}^+ \setminus \{ 0\}$; and
\item any such $f$ is an Archimedean matrix order unit for the matrix
  ordered space $\mcal{S}^d$.
\end{enumerate}
\end{thm}

\begin{rem}
\begin{enumerate}
\item When $\mcal{S} \subset M_p$ is an operator system for some $p
  \geq 1$, then $tr_p$ also works as an Archimedean matrix order unit
  for the matrix ordering on $\mcal{S}^d$; thus, $(\mcal{S}^d,
  tr_p)$ is an operator system. However, the above theorem becomes
  important as, quite surprisingly, there exist finite dimensional
  operator systems (\cite[$\S 7$]{CE77}) that cannot be embedded in
  matrix algebras as operator sub-systems.
\item In general, for an infinite dimensional operator system
  $\mcal{S}$, it is not clear whether $\mcal{S}^d$ admits an
  Archimedean matrix order unit or not.
\end{enumerate}
\end{rem}

For a finite dimensional operator system $\mcal{S}$, it is easily
checked that $(\mcal{S}^d)^d = \mcal{S}.$ Thus, returning to our
motivating question: if we start with an operator system, then the
object that it is naturally the set of states on is $\mcal{S}^d.$


\section{Tensor products of operator systems}

In the usual axioms for quantum mechanics, if Alice has a quantum
system represented as the states on a (finite dimensional) Hilbert
space $H_A$ and Bob has a quantum system represented as the states on
a Hilbert space $H_B$ then when we wish to consider the combined
system it has states represented by the Hilbert space $H_A \otimes_2
H_B$ where we've introduced the subscript 2 to indicate that this is
the unique Hilbert space with inner product satisfying,
\[ \langle h_A \otimes h_B | k_A \otimes k_B \rangle = \langle h_A | k_A \rangle \cdot \langle h_B \otimes k_B \rangle.\]
As vector spaces we have that
\[ \mcal{B}( H_A \otimes_2 H_B) = \mcal{B}(H_A) \otimes \mcal{B}(H_B),\]
and since the left hand side is an operator system, this tells us
exactly how to make an operator system out of the two operator systems
appearing on the right hand side.

If $P \in \mcal{B}(H_A)^+$ and $Q \in \mcal{B}(H_B)^+$ then $P \otimes
Q \in \mcal{B}(H_A \otimes_2 H_B)^+.$ But there are many positive
operators in $\mcal{B}(H_A \otimes_2 H_B)^+,$ even of rank one, i.e.,
vector states, that can not be expressed in such a simple fashion and
this is what leads to the important phenomenon known as {\it
  entanglement} which you've undoubtedly heard about in other
lectures.

Now suppose that we are in one of the scenarios, such as in coding
theory or capacity theory, where Alice and Bob do not both have all
the operators on their respective Hilbert spaces but instead are
constrained to certain operator subsystems, $\mcal{S}_A \subseteq
\mcal{B}(H_A)$ and $\mcal{S}_B \subseteq \mcal{B}(H_B).$ When we wish
to consider the bivariate system that includes them both then as a
vector space it should be $\mcal{S}_A \otimes \mcal{S}_B,$ but which
elements should be the states? More importantly, since we want to
study quantum channels on this bivariate system, we need to ask: What
should be the operator system structure on this bivariate system?

There is an easy answer, one could identify
\[
\mcal{S}_A \otimes \mcal{S}_B \subseteq \mcal{B}(H_A \otimes_2 H_B),
\]
and when one does this we see that it is an operator subsystem, i.e.,
it contains the identity operator and is closed under the taking of
adjoint. This operator system is denoted by $\mcal{S}_A \otimes_{sp}
\mcal{S}_B$ and is called their {\it spatial tensor product.} \index{tensor product! spatial}

Unfortunately, in general,
\[
\big( \mcal{S}_A \otimes_{sp} \mcal{S}_B \big)^d \ne \mcal{S}_A^d
\otimes_{sp} \mcal{S}_B^d,
\]
so we need at least one other tensor product to explain what are the
states on a tensor product.

Attempts by researchers such as Tsirelson \cite{tsirelson1980,
  tsirelson1993} to determine the sets of density matrices that are
the outcomes of various multipartite quantum settings, and various
works in operator algebras, argue for several other ways to form the
tensor product of operator systems.

Thus, we are lead to consider more general ways that we can form an
operator system out of a bivariate system.

Given operator systems $(\mcal{S}, e_{\mcal{S}})$ and $(\mcal{T},
e_{\mcal{T}})$, we wish to take the vector space tensor product
$\mcal{S} \otimes \mcal{T}$, endow it with a matrix ordering
$\{\mcal{C}_n \subset M_n(\mcal{S} \otimes \mcal{T}) : n \geq 1\}$
such that $\mcal{S} \otimes \mcal{T}$ together with these cones and
$e_{\mcal{S}}\otimes e_{\mcal{T}}$ forms an operator system.
\begin{defn}

Given operator systems $(\mcal{S}, e_{\mcal{S}})$ and $(\mcal{T},
e_{\mcal{T}})$, by an \index{tensor product! operator systems} operator
system tensor product $\tau$ we mean a family of cones
$\mcal{C}_n^{\tau} \subset M_n(\mcal{S} \otimes \mcal{T})$, $ n \geq
1$ such that $(\mcal{S} \otimes \mcal{T}, \{\mcal{C}_n^\tau\},
e_{\mcal{S}}\otimes e_{\mcal{T}})$ is an operator system satisfying:
\begin{enumerate}
\item $P \otimes Q= [p_{ij}\otimes q_{kl}] \in \mcal{C}_{nm}^{\tau}$ for all
  $P =[p_{ij}] \in M_n(\mcal{T})^+$, $Q =[q_{kl}]\in M_m(\mcal{T})^+$,
  $n, m \geq 1$.
\item $\varphi \otimes \psi \in CP(\mcal{S}\otimes_\tau \mcal{T},\,
  M_{n} \otimes M_{m} = M_{nm})$ for all $\varphi \in CP(\mcal{S},
  M_{n})$, $\psi \in CP(\mcal{T}, M_{m}) $, $n, m \geq 1$.
\end{enumerate}
 \end{defn}

\begin{rem}
Condition $(2)$ in the above definition is analogous to the
reasonableness axiom of Grothendieck for Banach space tensor products.
\end{rem}

\begin{defn} A tensor product $\tau$ of operator systems is said to be
\begin{enumerate}
\item \underline{functorial} if \begin{enumerate}
\item it is defined for any two operator systems; and
\item $\varphi \ot \psi \in CP(\mcal{S}_1 \otimes_{\tau} \mcal{S}_2,
  \mcal{T}_1 \otimes_{\tau} \mcal{T}_2)$ for all $\varphi \in
  CP(\mcal{S}_1 , \mcal{T}_1 )$, $ \psi \in CP( \mcal{S}_2,
  \mcal{T}_2)$.
\end{enumerate}
\item \underline{associative} if $(\mcal{S}_1 \otimes_{\tau}
  \mcal{S}_2)\otimes_{\tau}\mcal{S}_3$ is canonically completely order
  isomorphic to $\mcal{S}_1 \otimes_{\tau}
  (\mcal{S}_2\otimes_{\tau}\mcal{S}_3)$ for any three operator systems
  $\mcal{S}_i$, $i = 1, 2, 3$.
\item \underline{symmetric} if the flip map gives a complete order
  isomorphism $\mcal{S}\otimes_\tau \mcal{T} \simeq
  \mcal{T}\otimes_\tau \mcal{S}$ for any two operator systems
  $\mcal{S}$ and $ \mcal{T}$.
\end{enumerate}
\end{defn}

$\bullet$ Suppose a $\ast$-vector space $\mcal{W}$ has two matrix
orderedings $\{\mcal{C}_n\}$ and $\{\mcal{C}'_n\}$; then $(\mcal{W},
\{\mcal{C}_n\})$ is thought of as to be ``bigger'' than $(\mcal{W},
\{\mcal{C}_n'\})$ if the identity map $id_{\mcal{W}} : (\mcal{W},
\{\mcal{C}_n\}) \to (\mcal{W}, \{\mcal{C}_n'\})$ is CP; or,
equivalently, if $\mcal{C}_n \subset \mcal{C}_n'$ for all $n \geq 1$.

Note that this notion is parallel to the fact that if $||\cdot||_1$ and
$||\cdot||_2$ are two norms on a complex vector space $X$, then
$||\cdot||_1 \leq ||\cdot||_2$ if and only if the (closed) unit balls with
respect to these norms satisfy $B_1(X, ||\cdot||_2) \subset B_1(X,
||\cdot||_1) $.

\subsection{Minimal tensor product of operator systems}

Let $(\mcal{S}, e_{\mcal{S}})$ and $(\mcal{T}, e_{\mcal{T}})$ be two
operator systems. For each $p \geq 1$, set
\[
\mcal{C}^{\min}_p = \Big\{ [u_{ij}] \in M_p(\mcal{S} \otimes \mcal{T}) :
     [(\varphi\otimes\psi)(u_{ij})] \in M_{nmp}^+,\ \forall \varphi \in
     CP(\mcal{S}, M_n),\, \psi\in CP(\mcal{T}, M_m), n, m \geq 1 \Big\}.
\]
\begin{thm}\cite{KPTT11}\index{tensor product! minimal (operator systems)}
With above setup,
\begin{enumerate}
\item $\{ \mcal{C}^{\min}_p\}$ is an operator system tensor product on
  $\mcal{S}\ot\mcal{T}$ and we denote the consequent operator system
  by $\mcal{S}\ot_{\min}\mcal{T}$.
\item $\ot_{\min}$ is the smallest operator system tensor product in
  the sense that, if $\{ \mcal{C}^{\tau}_p\}$ is any other operator
  system tensor product on $\mcal{S}\ot\mcal{T}$, then $ \mcal{C}^{\tau}_p
  \subset \mcal{C}^{\min}_p$ for all $p \geq 1$.
\item if $\mcal{S} \subset B(H)$ and $\mcal{T} \subset B(K)$ for some
  Hilbert spaces $H$ and $K$, then the spatial tensor product
  $\mcal{S}\ot_{sp}\mcal{T} \subset B(H \ot K)$ is completely order
  isomorphic to the minimal tensor product
  $\mcal{S}\ot_{\min}\mcal{T}$.
\item $\ot_{\min}$ is functorial, associative and symmetric.
\item if $A$ and $B$ are unital $C^*$-algebras, then their minimal
  tensor product as operator systems is completely order isomorphic to
  the image of $A \ot B$ in $A \ot_{C^*\text{-}\min}B$.
\end{enumerate}
\end{thm}

\subsection{Maximal tensor product of operator systems}

Let $(\mcal{S}, e_{\mcal{S}})$ and $(\mcal{T}, e_{\mcal{T}})$ be two
operator systems. For each $n \geq 1$, consider
\[
\mcal{D}^{\max}_n = \Big\{ X^*([p_{ij}]\otimes [q_{kl}])X : [p_{ij}] \in
M_r(\mcal{S})^+ , [q_{kl}]\in M_s(\mcal{T})^+, X\in M_{rs,
  n}(\mbb{C}), r, s \geq 1 \Big\}.
\]
It can be seen that $\{ \mcal{D}^{\max}_n\}$ gives a matrix ordering
on $\mcal{S}\ot\mcal{T}$ and that $e_{\mcal{S}}\ot e_{\mcal{T}}$ is a
matrix order unit for this ordering. However, there exist examples where
$e_{\mcal{S}}\ot e_{\mcal{T}}$ fails to be an Archimedean matrix order
unit. We, therefore, Archimedeanize the above ordering, by
considering:
\[
\mcal{C}^{\max}_n =\Big\{ [u_{ij}] \in M_n(\mcal{S}\ot \mcal{T}) :
\delta\, \mathrm{diag}(e_{\mcal{S}}\ot e_{\mcal{T}}, e_{\mcal{S}}\ot
e_{\mcal{T}}, \ldots, e_{\mcal{S}}\ot e_{\mcal{T}}) + [u_{ij}] \in
\mcal{D}^{\max}_n, \ \forall \delta > 0 \Big\}.
\]
\begin{thm}\cite{KPTT11} With above set up, \index{tensor product! maximal}
\begin{enumerate}
\item $\{ \mcal{C}^{\max}_n\}$ is an operator system tensor product on
  $\mcal{S}\ot\mcal{T}$ and we denote the consequent operator system
  by $\mcal{S}\ot_{\max}\mcal{T}$.
\item $\ot_{\max}$ is the largest operator system tensor product in the sense
  that if $\{ \mcal{C}^{\tau}_n\}$ is any other operator system tensor product
  on $\mcal{S}\ot\mcal{T}$, then $ \mcal{C}^{\max}_n \subset
  \mcal{C}^{\tau}_n$ for all $n \geq 1$, i.e. $\max$ tensor product is
  the largest operator system tensor product.
\item $\ot_{\max}$ is functorial, associative and symmetric.
\item if $A$ and $B$ are unital $C^*$-algebras, then their maximal
  tensor product as operator systems is completely order isomorphic to
  the image of $A \ot B$ in $A \ot_{C^*\text{-}\max}B$.
\end{enumerate}
\end{thm}

\begin{rem}
For any $C^*$-algebra $A \subset B(H)$, the matrix ordering that it
inherits does not depend (upto complete order isomorphism) upon the
embedding or the Hilbert space $H$.
\end{rem}

$\bullet$ At this point, we must mention that there is a big
difference between the operator space maximal tensor product and the
operator system maximal tensor product. This can be illustrated by an
example:

For $n, m \geq 1$, the operator system maximal tensor product of $M_n$
and $M_m$ equals their $C^*$-maximal tensor product, whereas their
operator space maximal tensor product does not.

\begin{thm}[CP Factorisation Property]\cite{HP11} \label{cpap} \index{CP Factorisation Property}
Let $\mcal{S} \subset B(H)$ be an operator system. Then
$\mcal{S}\ot_{\min} \mcal{T} = \mcal{S}\ot_{\max} \mcal{T}$ for all
operator systems $\mcal{T}$ if and only if there exist nets of UCP
maps $\varphi_{\lambda} : \mcal{S} \to M_{n_\lambda}$ and
$\psi_{\lambda} : M_{n_\lambda}\to \mcal{S}$, $\lambda \in \Lambda$
such that
\[
||\psi_{\lambda} \circ \varphi_{\lambda} (s) -s || \to 0,\ \forall \,
s \in \mcal{S}.
\]
\end{thm}

\begin{rem}\label{norm}
 One can even avoid the above embedding $\mcal{S} \subset B(H)$, and
 give a characterization for $\mathrm{CPFP}$, alternately, by
 considering the norm
\[
||s|| : =\inf \Big\{ r > 0: \begin{pmatrix} re & s \\ s^* &
  re \end{pmatrix} \in M_2(\mcal{S})^+ \Big\},\ s \in \mcal{S}.
\]
\end{rem}

$\bullet$ We had remarked earlier that there exist finite dimensional
operator systems which can not be embedded in matrix algebras. The
surprise continues as, unlike $C^*$-algebras, not all finite
dimensional operator systems are $(\min, \max)$-nuclear in the sense
of $\S$\ref{lattice} - \cite[Theorem $5.18$]{KPTT11}. However, we have
the following useful fact.
\begin{lem}\cite{HP11, KPTT11}
Matrix algebras are $(\min, \max)$-nuclear as operator systems, i.e.,
$M_n \ot_{\min} \mcal{S} = M_n\ot_{\max} \mcal{T}$ for any operator
system $\mcal{T}$.
\end{lem}
\begin{proof} (Sketch!)
One basically identifies $M_k(M_n \ot \mcal{S})$ naturally with
$M_{kn} \ot \mcal{S}$ and then with some serious calculation shows
that $\mcal{D}_k^{\max}(M_n\ot\mcal{S}) = M_{kn}(\mcal{S})^+ =
\mcal{C}_k^{\min}(M_n\ot\mcal{S})$.
\end{proof}

\begin{rem}\cite{HP11}
In fact, the only finite dimensional $(\min, \max)$-nuclear operator
systems are the matrix algebras and their direct sums.
\end{rem}

\begin{rem}\label{bidual}
For a finite dimensional opertor system $\mcal{S}$, the canonical
isomorphism $\mcal{S} \ni x \mapsto \hat{x} \in (\mcal{S}^d)^d$, where
$\hat{x}(f) :=f(x)$ for all $f \in \mcal{S}^d$, is a complete order
isomorphism and $\widehat{e}_{\mcal{S}}$ is an Archimedean matrix
order unit for $(\mcal{S}^d)^d$.
\end{rem}
The requirement for the complete order isomorphism of the above map is
that $[x_{ij}] \in M_n(\mcal{S})^+$ if and only if $[\widehat{x}_{ij}]
\in M_n((\mcal{S}^d)^d)^+$; and this can be deduced readily from the
following fact:

\begin{lem}\cite[Lemma $4.1$]{KPTT11}\label{pos-lemma} For any
operator system $\mcal{S}$ and $P \in M_n(\mcal{S})$, $P \in
M_n(\mcal{S})^+$ if and only if $\varphi^{(n)}(P) \in M_{nm}^+$ for
all $\varphi \in \mathrm{UCP}(\mcal{S}, M_m)$ and $m \geq 1$.
\end{lem}

For vector spaces $\mcal{S}$ and $\mcal{T}$ with $\mcal{S}$ finite
dimensional, we have an identification between $\mcal{S}\ot \mcal{T}$
and the space of all linear maps from $\mcal{S}^d$ into $\mcal{T}$ by
identifying the element $u = \sum_i s_i \ot t_i \in \mcal{S}\ot
\mcal{T}$ with the map $ \mcal{S}^d \ni f \stackrel{L_u}{\mapsto}
\sum_i f(s_i)t_i \in \mcal{T}$.

\begin{lem}\label{C-min}\cite[Lemma $8.4$]{KPTT11}
Let $\mcal{S}$ and $\mcal{T}$ be operator systems with $\mcal{S}$
finite dimensional and let $[u_{ij}] \in M_n( \mcal{S}\ot
\mcal{T})$. Then $[u_{ij}] \in M_n( \mcal{S}\ot_{\min} \mcal{T})^+$ if
and only if the map $\mcal{S}^d \ni f \mapsto [L_{{u}_{ij}}(f)] \in
M_n(\mcal{T})$ is CP.
\end{lem}
\begin{proof}
($\Rightarrow$) Let $[u_{ij}] \in M_n( \mcal{S}\ot_{\min}
  \mcal{T})^+$, $k \geq 1$ and $[f_{rs}] \in
  M_k(\mcal{S}^d)^+$. Suppose $u_{ij} = \sum_p s^{ij}_p \ot
  t^{ij}_p$. We need to show that $X:=[L_{u_{ij}}]^{(k)}([f_{rs}]) \in
  M_k(M_n(T))^+$. We will again appeal to Lemma \ref{pos-lemma}.  Let
  $m \geq 1$ and $\varphi \in \mathrm{UCP}(\mcal{T}, M_m)$.  Then, for
  each $ 1 \leq k, l \leq m$, there exists a unique $\varphi_{kl} \in
  \mcal{T}^d$ such that $\varphi(t) = [\varphi_{kl}(t)]$ for all $ t
  \in \mcal{T}$; and, thus
\begin{eqnarray*}
\varphi^{(kn)} (X) & = & \left[ [\varphi\circ L_{{u}_{ij}}(f_{rs})
  ]_{ij} \right]_{rs}\\ & = & \left[ \left[ [\varphi_{kl}\circ
      L_{{u}_{ij}}(f_{rs})]_{kl} \right]_{ij} \right]_{rs}\\ & = &
\left[ \left[ \left[\varphi_{kl}(\sum_p f_{rs}(s^{ij}_p) t^{ij}_p
      )\right]_{kl} \right]_{ij} \right]_{rs}\\ & = & \left[ \left[
    [f_{rs}(\sum_p s^{ij}_p \varphi_{kl}(t^{ij}_p) )]_{kl}
    \right]_{ij} \right]_{rs}\\ & = & \tau ([f_{rs}]^{(nm)}
\left(\left[ \left[\sum_p s^{ij}_p \varphi_{kl}(t^{ij}_p) \right]_{kl}
  \right]_{ij} \right)),
\end{eqnarray*}
where $ [f_{rs}]^{(nm)} $ denotes the $nm$-amplification of the CP map
$\mcal{S}\ni s \mapsto [f_{rs}(s)] \in M_k$ and $\tau$ is the
canonical flip $\ast$-isomorphism $M_{nmk}\simeq M_n\ot M_m \ot M_k
\simeq M_k \ot M_n\ot M_m \simeq M_{knm}$. Next, since $\ot_{\min}$ is
functorial, $id \ot \varphi : \mcal{S}\ot_{min} \mcal{T} \to
M_m(\mcal{S})$ is UCP, and, note that under the complete order
isomorphism $\theta : M_n \ot M_m \ot_{min} \mcal{S} \simeq M_n
\ot_{min} \mcal{S} \ot M_m$, $\theta ( \left[ \left[\sum_p s^{ij}_p
    \varphi_{kl}(t^{ij}_p) \right]_{kl} \right]_{ij}) = (id \ot
\varphi)^{(n)} ([u_{ij}])$. Thus, $\left[ \left[\sum_p s^{ij}_p
    \varphi_{kl}(t^{ij}_p) \right]_{kl} \right]_{ij} \in
M_{nm}(\mcal{S})^+$ and we conclude that $\varphi^{(kn)}(X) \geq
0$. In particular, $\mcal{S}^d \ni f \mapsto [L_{u_{ij}}(f)] \in
M_n(\mcal{T})$ is CP.

($\Leftarrow$) Conversely, suppose the map $\mcal{S}^d \ni f \mapsto
[L_{u_{ij}}(f)] \in M_n(\mcal{T})$ is CP. Let $k, m \geq 1$, $\varphi
\in \mathrm{UCP}(\mcal{S}, M_k)$ and $\psi \in \mathrm{UCP}(\mcal{T},
M_m)$. We need to show that $(\varphi \ot \psi)^{(n)} ([u_{ij}]) \in
M_{nkm}^+$. As above, there exist $\varphi_{rs} \in \mcal{S}^d$,
$\psi_{uv} \in \mcal{T}^d$ for $ 1 \leq r, s \leq k$ and $ 1 \leq u, v
\leq m$ such that $\varphi(s) = [\varphi_{rs}(s)]$ and
$\psi(t)=[\psi_{uv}(t)]$ for all $s \in \mcal{S}, t \in
\mcal{T}$. Also, since $\varphi$ and $\psi$ are UCP, we have
$[\varphi_{rs}] \in M_k(\mcal{S}^d)^+$ and $[\psi_{uv}] \in
M_m(\mcal{T}^d)^+$. Suppose $u_{ij} = \sum_p s^{ij}_p \ot
t^{ij}_p$. Then

\begin{eqnarray*}
(\varphi \ot \psi)^{(n)} ([u_{ij}]) & = & \left[ \sum_p
    \varphi(s^{ij}_p) \ot \psi(t^{ij}_p) \right]_{ij}\\ & = & \left[
    \left[ \left[ \sum_p \varphi_{rs}(s^{ij}_p) \psi_{uv}(t^{ij}_p)
        \right]_{uv} \right]_{rs} \right]_{ij}\\ & = & \left[ \left[
      \left[ \psi_{uv}(\sum_p \varphi_{rs}(s^{ij}_p)t^{ij}_p)
        \right]_{uv} \right]_{rs} \right]_{ij}\\ & = & \left[ \left[
      \left[ \psi_{uv}( L_{u_{ij}}(\varphi_{rs})) \right]_{uv}
      \right]_{rs} \right]_{ij}\\ & = & \left[ \left[ [
        \psi_{uv}](L_{u_{ij}}(\varphi_{rs})) \right]_{rs}
    \right]_{ij}\\ & = & [ \psi_{uv}]^{(kn)} \left[ \left[
      (L_{u_{ij}}(\varphi_{rs})) \right]_{rs} \right]_{ij}\\ & = & [
    \psi_{uv}]^{(kn)} \circ \tau \left(
  [L_{u_{ij}}]^{(k)}([\varphi_{rs}])\right),
\end{eqnarray*}
where $\tau$ is the canonical complete order isomorphism $ \tau : M_k
\ot M_n \ot_{min} \mcal{T} \simeq M_n \ot M_k \ot_{min}
\mcal{T}$.

Thus, with all the data that we have at our disposal, we
immediately conclude that $(\varphi \ot \psi)^{(n)} ([u_{ij}]) \geq 0$
and hence $[u_{ij}] \in M_n(\mcal{S}\ot_{min}\mcal{T})^+$.
\end{proof}

\begin{lem}\label{lemmas}
\begin{enumerate}
\item Let $\mcal{S}$ and $\mcal{T}$ be operator systems and $u \in
  \mcal{D}^{\max}_1$. Then, there exist $n \geq 1$, $[p_{ij}] \in
  M_n(\mcal{S})^+$ and $[q_{ij}] \in M_n(\mcal{T})^+$ such that $u =
  \sum_{i,j=1}^n p_{ij}\otimes q_{ij}$.
\item Let $\mcal{F}$ be a finite dimensional operator system and $\{
  v_1, \ldots, v_n\}$ be a basis of $\mcal{F}$ with a dual basis $\{
  \delta_1, \ldots, \delta_n\}$. Then $\sum_i \delta_i \otimes v_i \in
  (\mcal{F}^d \otimes_{\min} \mcal{F})^+$.
\end{enumerate}
\end{lem}
\begin{proof} (1)
By defintion, there exist $n, m \geq 1$, $X \in M_{1, nm}$,
$P=[p_{ij}] \in M_n(\mcal{S})^+$ and $Q = [q_{rs}] \in
M_m(\mcal{T})^+$ such that $u = X(P \ot Q)X^*$. We first note that,
adding suitable zeros to $X$ and $P$ or $Q$ according as $n$ or $m$ is
smaller among them, we can assume that $m = n$. Thus, we have $X =
[x_{11}, x_{12}, \ldots, x_{1n},x_{21}, \ldots, x_{nn}] \in M_{1,
  n^2}$, $P=[p_{ij}] \in M_n(\mcal{S})^+$ and $Q = [q_{rs}] \in
M_n(\mcal{T})^+$ such that $u = X(P \ot Q)X^*$. Consider $\tilde{P}
\in M_{n^2}(\mcal{S})$ and $\tilde{Q}\in M_{n^2}(\mcal{T})$ given by
$\tilde{P}_{(i,r),(j,s)} = x_{ir}p_{ij}\bar{x}_{js}$ and
$\tilde{Q}_{(i,r),(j,s)} = q_{rs}$, $1 \leq i, j , r , s \leq
n$. Clearly $\tilde{P} \in M_{n^2}(\mcal{S})^+$ and $\tilde{Q} \in
M_{n^2}(\mcal{T})^+$ as under the identification $M_{n^2}(\mcal{S})
\simeq M_n(\mcal{S}) \ot M_n $, we have $\tilde{P} = \tilde{X}(P \ot
J_n )\tilde{X}^*$, where $\tilde{X} = \text{diag}( x_{11}, \ldots,
x_{1n}, x_{21}, \ldots, x_{nn}) \in M_{n^2}$, $J_n \in M_n$ is the
positive semi-definite matrix with entries $(J_n)_{ij}= 1 $ for all $
1 \leq i, j \leq n$; and, under the identification $M_{n^2}(\mcal{T})
\simeq M_n \ot M_n(\mcal{T})$, we have $\tilde{Q} = J_n \ot
Q$. Finally
\[
\sum_{(i,r),(j,s)}\tilde{P}_{(i,r),(j,s)} \ot \tilde{Q}_{(i,r),(j,s)}
= \sum_{i,j,r,s,} x_{ir}p_{ij}\bar{x}_{js} \ot q_{rs}= \sum_{i,j,r,s,}
x_{ir}(p_{ij} \ot q_{rs}) \bar{x}_{js}= X(P \ot Q)X^* = u.
 \]
(2) Let $ u = \sum_i \delta_i \ot v_i$. By Lemma \ref{C-min}, we just
 need to show that the map $$(\mcal{F}^d)^d \ni \hat{x} \mapsto
 L_u(\hat{x}) = \sum_i \hat{x}(\delta_i)  v_i = \sum_i
 \delta_i(x)v_i = x \in \mcal{F}$$ is CP, which is precisely the
 complete order isomorphism in Remark \ref{bidual}.
\end{proof}
\vspace*{2mm}

\noindent{\em Proof of Theorem \ref{cpap}:} ($\Rightarrow$) By
functoriality of the tensor products $\ot_{\min}$ and $\ot_{\max}$,
and by nuclearity of matrix algebras, for any operator system
$\mcal{T} \subset B(K)$, we have CP maps
\[
\mcal{S}\ot_{\min}\mcal{T} \stackrel{\varphi_{\lambda} \ot \,
  id_{\mcal{T}}}{\longrightarrow} M_{n_{\lambda}} \ot_{\min} \mcal{T}
= M_{n_{\lambda}} \ot_{\max} \mcal{T} \stackrel{\psi_{\lambda} \ot \,
  id_{\mcal{T}}}{\longrightarrow} \mcal{S} \ot_{\max} \mcal{T}.
\]
In particular, their composition $(\psi_{\lambda} \circ
\varphi_{\lambda}) \ot id_{\mcal{T}}: \mcal{S}\ot_{\min}\mcal{T} \to
\mcal{S}\ot_{\max}\mcal{T} $ is CP. Then, by the characterization of
the norm on an operator system as given in Remark \ref{norm}, one sees
that the norm $||\cdot||_{\mcal{S}\ot_{\max}\mcal{T}}$ induced on
$\mcal{S}\ot\mcal{T}$ by the operator system
$\mcal{S}\ot_{\max}\mcal{T}$ is a sub-cross norm and thus
$(\psi_{\lambda} \circ \varphi_{\lambda}) \ot id_{\mcal{T}} (z)$
converges to $z$ for all $z \in \mcal{S}\ot_{\min}\mcal{T}$. In
particular, $id: \mcal{S}\ot_{\min}\mcal{T} \to
\mcal{S}\ot_{\max}\mcal{T}$ is CP and we obtain
$\mcal{S}\ot_{\min}\mcal{T} = \mcal{S}\ot_{\max}\mcal{T}$.
\vspace*{1mm}

Conversely, suppose $\mcal{S}\ot_{\min} \mcal{T} = \mcal{S}\ot_{\max}
\mcal{T}$ for all operator systems $\mcal{T}$. Let $\mcal{F} \subset
\mcal{S}$ be a finite dimensional operator sub-system ($1_{\mcal{F}} =
1_{\mcal{S}}$). Using the fact that the tensor products $\ot_{\min}$
and $\ot_{\mathrm{sp}}$ coincide, we have
\[
\mcal{F}^d\ot_{\min} \mcal{F} \subset\mcal{F}^d \ot_{\min}\mcal{S} =
\mcal{F}^d\ot_{\max}\mcal{S}.\ \
\]
In particular, for a basis $\{ v_i\}$ of $\mcal{F}$ with dual basis
$\{\delta_i\}$, we see,  by Lemma \ref{lemmas}(2),  that $\sum_i
\delta_i \otimes v_i \in (\mcal{F}^d \ot_{\max} \mcal{S})^+ =
\mcal{C}_1^{\max} (\mcal{F}^d\ot \mcal{S})$.

By Theorem \ref{unit}, fix an $f \in \mcal{F}^d$ that plays the role
of an Archimedean matrix order unit for $\mcal{F}^d$. Since $r f$ is
also an Archimedean matrix order unit for any $r > 0$, we can assume
that $f(e_{\mcal{S}} ) = 1 $. So, for all $\delta > 0$, $$\delta (f
\otimes e_{\mcal{S}}) + \sum_i \delta_i \otimes v_i \in
\mcal{D}_1^{\max},$$ which implies, by Lemma \ref{lemmas}(1), that for
each $\delta > 0$ there exist $n \geq 1$, $[f_{ij}] \in M_n(\mcal{F}^d)^+$ and
$[p_{ij}] \in M_n(\mcal{S})^+$ such that $$ \delta(f\ot e_{\mcal{S}})
+ \sum_i \delta_i \otimes v_i = \sum_{ij}f_{ij} \otimes p_{ij}.$$

 Recall that $[f_{ij}] \in M_n(\mcal{F}^d)^+$ if and only if the map $
 \mcal{F} \ni v \stackrel{\Phi}{\mapsto} [f_{ij}(v)] \in M_n$ is CP.

\noindent \underline{\em Claim:} We can choose $[f_{ij}]$ in such way that the
corresponding map $\Phi$ is UCP.

\noindent \underline{\em Proof of claim:} Let $Q = [q_{ij}] \in M_n$
be the support projection of the positive semi-definite matrix
$[f_{ij}(e_{\mcal{S}})] \in M_n$, i.e., $Q: \mbb{C}^n \to \mbb{C}^n$
is the orthogonal projection onto the subspace
$[f_{ij}(e_{\mcal{S}})]\mbb{C}^n \subset \mbb{C}^n$. Note that
$Y:=[f_{ij}(e_{\mcal{S}})]$ is invertible in $QM_nQ$. Let $Y^{-1}$
denote the inverse of $[f_{ij}(e_{\mcal{S}})]$ in $QM_nQ$. Also, if $p
= \mathrm{rank}\, Q$, let $U^*QU = \mathrm{diag}(I_p, O)$ be the
diagonalization of $Q$, where $U$ is a unitary matrix. Let $X =
[x_{11}, x_{12}, \ldots, x_{1n},x_{21}, \ldots, x_{nn}] \in M_{1,
  n^2}$, be given by $X_{ij} = \delta_{i,j}$. Then,
\begin{eqnarray*}
\sum_{ij}f_{ij} \otimes p_{ij} & = & X([f_{ij}] \ot [p_{ij}]) X^* \\ &
= & X \left(Y^{1/2}U\left(\begin{array}{c} I_p \\ 0 \end{array}
\right) \ot I_n\right) \cdot \left[(I_p, 0) U^* Y^{-1/2} [f_{ij}]
  Y^{-1/2}U \left(\begin{array}{c} I_p \\ 0 \end{array} \right) \ot
  [p_{ij}]\right]\cdot \\ & & \left( Y^{1/2}U\left(\begin{array}{c}
  I_p \\ 0 \end{array} \right)\ot I_n \right)^*X^*
\end{eqnarray*}
and \[
(I_p, 0) U^* Y^{-1/2} [f_{ij}(e_{\mcal{S}})] Y^{-1/2}U
    \left(\begin{array}{c} I_p \\ 0 \end{array} \right) = I_p.
\]
Hence the claim.

So, by Arveson's extension Theorem, there exists a UCP map
$\tilde{\Phi} : \mcal{S} \to M_n$ such that
$\tilde{\Phi}_{|_{\mcal{S}}} = \Phi$.

Now, consider the linear map $\Psi : M_n \to \mcal{S}$ sending $E_{ij}
\mapsto p_{ij}$ for all $1 \leq i, j \leq n$. By (Choi's) Theorem
\ref{choi}, this map is CP. Then, for $v \in \mcal{F}$, we have
\[
\Psi \circ \Phi (v) =
\Psi ([f_{ij}(v)]) = \sum_{ij}f_{ij}(v)p_{ij}.
\]
On the other hand, via the canonical identification between
$\mcal{F}^d \ot \mcal{S} $ and the space of linear transformations
from $\mcal{F}^d$ into $\mcal{S}$, we have $\sum_{ij}f_{ij}(v)p_{ij} =
(\sum_{ij} f_{ij} \ot p_{ij} )(v) = \delta (f \ot e_{\mcal{S}}) (v)+
(\sum_i \delta_i \ot v_i)(v) = \delta f(v) e_{\mcal{S}} + v$ for all $
v \in \mcal{F}$.  Also, since $f \in (\mcal{F}^d)^+$, we have $f :
\mcal{F} \to \mbb{C}$ is UCP (by our choice); so, $||f||_{cb} = 1$
and $ |f(v)| \leq ||v||$ for all $v \in \mcal{F}$.  Consider the
directed set $$\Lambda = \{(\mcal{F}, \delta): \mcal{F} \subset
\mcal{S}, \mathrm{operator
  \ sub\text{-}system\ with}\ \mathrm{dim}\,\mcal{F} < \infty, \delta
>0 \}$$ with respect to the partial order $\leq$ given by
\[
(\mcal{F}_1, \delta_1) \leq (\mcal{F}_2, \delta_2) \Leftrightarrow
 \mcal{F}_1 \subseteq \mcal{F}_2\  \mathrm{and}\ \delta_1 > \delta_2.
\]

Thus, for each $\lambda= (\mcal{F}, \delta) \in \Lambda$, there exist
$n_{\lambda}\geq 1$, $\varphi_{\lambda} \in UCP(\mcal{S},
M_{n_{\lambda}})$ and $\psi_{\lambda}' \in CP( M_{n_{\lambda}},
\mcal{S})$ satisfying $||\psi_{\lambda}' \circ \varphi_{\lambda} (v) -
v|| \leq \delta ||v||$ for all $ v \in\mcal{F}$. In particular,
$||\psi_{\lambda}' \circ \varphi_{\lambda} (v) - v|| \to 0$ for every
$v \in \mcal{S}$. Since each $\varphi_{\lambda}$ is unital,
$\psi_{\lambda}'(I_{n_{\lambda}})$ converges to $e_{\mcal{S}}$. Fix a
state $\omega_{\lambda}$ on $M_{n_{\lambda}}$ and set
\[
\psi_{\lambda} (A) = \frac{1}{||\psi_{\lambda}'||} \psi_{\lambda}'(A)
+ \omega_{\lambda} (A) (e_{\mcal{S}} - \frac{1}{||\psi_{\lambda}'||}
\psi_{\lambda}'(I_{n_{\lambda}})). \] Now, $\psi_{\lambda}\in UCP
(M_{n_{\lambda}}, \mcal{S})$ for all $\lambda \in \Lambda$ and we
still have $||\psi_{\lambda} \circ \varphi_{\lambda} (v) - v|| \to 0$
for every $v \in \mcal{S}$.  \hfill $\Box$
\vspace*{3mm}

\noindent {\bf An example of a nuclear operator system that is not a
  $C^*$-algebra\cite{HP11}}

Let $\mcal{K}_0 = \overline{\mathrm{span}}\{ E_{i,j} : (i, j ) \neq
(1, 1)\} \subset B(\ell^2)$, where $E_{i,j}$ are the standard matrix
units. Consider
\[
\mcal{S}_0 := \{ \lambda I + T : \lambda \in \mbb{C}, T \in \mcal{K}_0\}
\subset B(\ell^2),
\]
the operator system spanned by the identity operator and $\mcal{K}_0$.
In \cite{HP11}, it has been proved that $\mcal{S}_0$ is a $(\min,
\max)$-nuclear operator system and is not unitally completely order
isomorphic to any unital $C^*$-algebra.


\section{Graph operator systems}

Graphs, especially the confusability graph, play an important role in
Shannon's information theory.  In the work of \cite{DSW} on quantum
capacity, they associate an operator system with a graph and show that
many of Shannon's concepts have quantum interpretations in terms of
these {\it graph operator systems.} \index{graph operator system} Many
concepts that we wish to deal with become much more transparent in
this setting.

Given a (finite) graph $G$ with vertices $V=\{1, 2, \ldots, n\}$ and edge set
$\mcal{E} \subset V \times V$ (edges are not ordered; thus, $(i, j)
\in \mcal{E} \Rightarrow (j, i)\in  \mcal{E}$), we have an operator system
$\mcal{S}_G$ given by
\[
\mcal{S}_G:= \text{span}\left\{ \{E_{ij} : (i,j)\in \mcal{E}\} \cup \{
E_{ii} : 1 \leq i \leq n\}\right\}
\subset M_n.
\]

Note that $\mcal{S}_G$ does not depend upon the number of bonds
between two edges and the loops, if any. Thus, we assume that the
graphs that we consider are simple, i.e., they have no loops or
parallel edges.
\vspace*{2mm}

\noindent {\em Dual of a graph operator system.} \index{graph operator
  system! dual} As a vector space, the operator system dual of a graph
operator system $\mcal{S}_G$ can be identified with a vector subspace
of $(M_n)^d$ as
\[
\mcal{S}_G^d = \text{span} \left\{ \{\delta_{ij} : (i,j) \in
\mcal{E}\} \cup \{ \delta_{ii} : 1 \leq i \leq n\} \right\} \subset
(M_n)^d.
\]
But it is not clear what the matrix order should be on this subspace.
We will show that it is {\bf not} the induced order.  That is, while
this is a natural vector space inclusion as operator systems,
\[ \mcal{S}_G^d \not\subset M_n^d  \, \, !\]

\begin{exam}
If we consider the graph $G$ with vertex set $V=\{ 1, 2, \ldots,
n\}$ and edge set $\mcal{E} = \{ (1,2), (2,3), (3,4), \ldots, (n-1,
n), (n, n-1), \ldots, (3,2), (2,1)\}$, then
\[
\mcal{S}_G = \{ \mathrm{tridiagonal\ matrices}\} = \{ [a_{ij}]\in M_n :
  a_{ij} = 0 \ \mathrm{for}\  |i-j| > 1\}.
\]
\end{exam}

\noindent {\bf Observation:} For the above graph $G$, given $f =
\sum_{1 \leq i, j \leq n, |i-j|\leq 1} b_{ij} \delta_{ij} \in
\mcal{S}_G^d$, setting $b_{ij} = 0$ for $|i-j|>1$, one would like to
know the conditions on the tridiagonal matrix $B=[b_{ij}]$ such that
$f \in (\mcal{S}_G^d)^+$.  It turns out that $f \in (\mcal{S}_G^d)^+$
if and only if the matrix $B$ is {\it partially positive,} i.e., we can
choose the off-tridiagonal entries of $B$ to get a positive
semi-definite matrix:

  Indeed, if $f \in (\mcal{S}_G^d)^+$, then $f :\mcal{S}_G \to
  \mbb{C}$ is CP if and only if (by Arveson's extension theorem) it
  extends to a CP map $\tilde{f} : M_n \to \mbb{C}$. Suppose
  $\tilde{f} = \sum_{i,j} \tilde{b}_{ij} \delta_{ij}$ for some
  $\tilde{b}_{ij} \in \mbb{C}$. Then, we have $[\tilde{b}_{ij}] =
  \tilde{f}^{(n)}([E_{ij}]) \in (M_n)^+$ and, since
  $\tilde{f}_{|_{\mcal{S}_G}} = f$, we also see that $\tilde{b}_{ij} =
  b_{ij}$ for all $|i-j| \leq 1$.

Conversely, if $\tilde{B}=[\tilde{b}_{ij}]$ is a positive
semi-definite matrix with $\tilde{b}_{ij} = b_{ij}$ for all $|i-j|
\leq 1$, then, by Theorem \ref{choi}, the Schur multiplication map
\index{Schur multipliers} $S_{\tilde{B}} : M_n \ni [x_{ij}] \mapsto
      [\tilde{b}_{ij}x_{ij}] \in M_n $ is CP. By Theorem \ref{choi}
      again, the map $M_n \ni [x_{ij}] \mapsto \sum_{i,j}x_{ij} \in
      \mbb{C}$ is also CP; hence, the map $M_n \ni [x_{ij}]\mapsto
      \sum_{i,j}\tilde{b}_{ij}x_{ij} \in \mbb{C}$ is CP, i.e.,
      $\tilde{f}:=\sum_{i,j}\tilde{b}_{ij} \delta_{ij} \in (M_n^d)^+$,
      and, since $\tilde{f}_{|_{\mcal{S}_G}} = f$, we have $f \in
      (\mcal{S}_G^d)^+$.

 In general, a tridiagonal matrix $B$ can be extended to a positive
 semi-definite matrix if and only if each $2\times2$ submatrix of $B$
 is positive semi-definite \cite{PPS}. Thus, we have a very clear
 picture in this case of which linear functionals are positive.

\hfill $\Box$

The situation for a general graph is as follows:

Given a graph $G$ then a functional $f: \mcal{S}_G \to \mbb{C}$ has
the form $f = \sum_{(i,j) \in \mcal{E} \text{ or } i=j} b_{i,j}
\delta_{i,j}.$ Set $B= \sum_{(i,j) \in \mcal{E} \text{ or } i=j}
b_{i,j} E_{i,j}.$ Then $f \in (S_G^d)^+$ if and only if there exists an $ M
\in \mcal{S}_G^{\perp}$ such that $B + M \in M_n^+$.

This result can be put into the language of {\it partially defined
  matrices.} \index{partially defined matrices} Notice that when we
have $f: \mcal{S}_G \to \mbb{C}$ and we try to form the density matrix
of this functional, $\big( f(E_{i,j}) \big),$ then because only some
of the matrix units $E_{i,j}$ belong to the space $\mcal{S}_G$ we only
have some of the entries of our matrix specified.  This is what is
meant by a partially defined matrix, i.e., a matrix where only some
entries are given and the rest are viewed as free variables.  Choosing
the matrix $M$ above is tantamount to choosing values for the
unspecified entries.  In the language of partially defined matrices,
this is called {\it completing the matrix.}

Thus, what we have shown is that each functional gives rise to a
partially defined density matrix and that the positive functionals on
$\mcal{S}_G$ are precisely those whose density matrices can be
completed to positive semidefinite matrices.

\section{Three more operator system tensor products}

There are at least three more operator system tensor products that are
important in the operator algebras community and are likely to have
some importance for quantum considerations as well.  In particular,
the one that we call the commuting tensor product is important for the
study of the Tsirelson conjectures \cite{tsirelson1993, jnppsw}.

Moreover, there are many important properties of operator systems that
are equivalnet to the behaviour of the operator system with respect to
these tensor products. We give a very condensed summary of this theory
below.

\subsection{The commuting tensor product $\ot_{\text{c}}$.} \index{tensor product! commuting}
 Let $\mcal{S}$ and
$\mcal{T}$ be operator systems and $\varphi :\mcal{S} \to B(H)$, $\psi
: \mcal{T} \to B(H)$ be UCP maps with commuting ranges, i.e.,
$\varphi(s) \psi (t) = \psi(t) \varphi(s)$ for all $s \in \mcal{S}$,
$t \in \mcal{T}$. Define $\varphi \odot \psi : \mcal{S} \ot \mcal{T}
\to B(H)$ by $(\varphi \odot \psi )(s \ot t) = \varphi(s) \psi (t)$, $
s \in \mcal{S}$, $t \in \mcal{T}$. Consider

$ \mcal{C}^c_n = \{ [u_{ij}] \in M_n(\mcal{S}\ot \mcal{T}) : (\varphi
\odot \psi)^{(n)} ([u_{ij}]) \in B(H^{(n)})^+ , \text{ Hilbert spaces
} H,$ \\ \hspace*{2mm} \hfill $ \varphi \in \mathrm{UCP}(\mcal{S}, B(H)), \psi
\in \mathrm{ UCP}(\mcal{T}, B(H)) \text{ with commuting ranges} \} .  $

\begin{thm} With above set up,
\begin{enumerate}
\item $\{\mcal{C}_n^{\mathrm{c}}\}$ is an operator tensor on $\mcal{S} \ot
  \mcal{T}$ and the consequent operator sytem is denoted by
  $\mcal{S}\ot_{\mathrm{c}}\mcal{T}$.
\item $\ot_{\mathrm{c}}$ is functorial and symmetric.
\end{enumerate}
\end{thm}

\begin{rem}
Explicit examples showing non-associativity of the tensor product
$\ot_{\text{c}}$ are not known yet.
\end{rem}

\subsection{The tensor products $\ot_{\text{el}}$ and $\ot_{\text{er}}$.} \index{tensor product! left enveloping}

Let $\mcal{S} \subset B(H)$ and $\mcal{T} \subset B(K)$ be operator
systems. Then the inclusion $\mcal{S} \ot \mcal{T} \subset
B(H)\ot_{\max} \mcal{T}$ induces an operator system on $\mcal{S} \ot
\mcal{T}$, which is refered as {\em enveloping on the left} and is
denoted by $\mcal{S} \ot_{\text{el}} \mcal{T}$. Likewise, the inclusion
$\mcal{S} \ot \mcal{T} \subset \mcal{S}\ot_{\max}B(K) $ induces an
operator system on $\mcal{S} \ot \mcal{T}$ called {\em enveloping on
  the right} and is denoted by $\mcal{S} \ot_{\text{er}} \mcal{T}$. \index{tensor product! right enveloping}

It requires some work to establish that the operator system tensor
products $\ot_{\text{el}}$ and $\ot_{\text{er}}$ of $\mcal{S}$ and $
\mcal{T}$ do not depend (upto complete order isomorphisms) on the
embeddings $\mcal{S} \subset B(H)$ and $\mcal{T} \subset B(K)$.

$\bullet$ $\ot_{\text{el}}$ and $\ot_{\text{er}}$ are both
functorial. However, it is not clear whether they are associative or
not?

$\bullet$ $\mcal{S} \ot_{\text{el}} \mcal{T} \simeq
\mcal{T} \ot_{\text{er}} \mcal{S} $ as operator systems.

\subsection{Lattice of operator system tensor products}\label{lattice}
We have the following lattice structure among the above five operator
sytem tensor products:
\[
\min \leq \text{el, er} \leq \text{c} \leq \max.
\]

Given functorial operator system tensor products $\alpha$ and $\beta$,
an operator system $\mcal{S}$ is said to be $(\alpha, \beta)$-nuclear
provided $\mcal{S} \ot_{\alpha} \mcal{T} \simeq \mcal{S} \ot_{\beta}
\mcal{T}$ as operator spaces for all operator systems $\mcal{T}$.

Recall that a $C^*$-algebra is nuclear if and only if it satisfies
completely positive approximation property (CPAP). We saw in Theorem
\ref{cpap}, that this generalizes to the context of operator system as
well, i.e., an operator system $\mcal{S}$ is $(\min, \max)$-nuclear if
and only if it satisfies completely positive factorization property
(CPFP).  In general, $(\alpha, \beta)$-nuclearity of operator systems
have some futher analogous structural characterizations.

\section{Some characterizations of operator system tensor products}

\subsection{Exact operator systems.} Analogous to the notion of \index{operator system! exact}
exactness for $C^*$-algebras and operator systems, there is a notion
of exactness for operator systems as well - see \cite{KPTT10}.

\begin{thm}\cite[Theorem $5.7$]{KPTT10}
An operator system $\mcal{S}$ is $1$-exact if and only if it is $(\min,
\mathrm{el})$-nuclear.
\end{thm}

For an operator system $\mcal{S}$, we consider its Banach space dual
$\mcal{S}^*$ and endow it with a matrix ordering as we did for
$\mcal{S}^d$ above. We repeat the process to endow its double dual
$\mcal{S}^{**}$ also with a matrix ordering. It is not very difficult
to see that the canonical embedding $\mcal{S} \subset \mcal{S}^{**}$
is a complete order isomorphism onto its image. Also, it is a fact -
\cite[Proposition $6.2$]{KPTT10}- that $\widehat{e}_{\mcal{S}}$ is an
Archimedean matrix order unit for $\mcal{S}^{**}$.

\subsection{Weak Expectation Property (WEP).}
\begin{lem}\cite[Lemma $6.3$]{KPTT10}
Let $\mcal{S}$ be an operator system. Then the following are
equivalent:
\begin{enumerate}
\item There exists an inclusion $\mcal{S} \subset B(H)$ such that the
  canonical embedding $\iota : \mcal{S} \to \mcal{S}^{**}$ extends to
  a CP map $\tilde{\iota} : B(H) \to \mcal{S}^{**}$.
\item For every operator system inclusion $\mcal{S} \subset \mcal{T}$,
  the map $\iota : \mcal{S} \rightarrow \mcal{S}^{**}$ extends to a CP map
  $\widetilde{\iota} : \mcal{T} \rightarrow \mcal{S}^{**}$
\item The
canonical embedding $\iota : \mcal{S} \subset \mcal{S}^{**}$ factors
through an injective operator system by UCP maps, i.e., there
is an injective operator system $\mcal{T}$ and UCP maps $\varphi_1 :
\mcal{S} \to \mcal{T}$ and $\varphi_2 : \mcal{T} \to \mcal{S}^{**}$
such that $\iota = \varphi_2 \circ \varphi_1$.
\end{enumerate}
\end{lem}

\begin{defn} \index{operator system! Weak Expectation Property (WEP)}
An operator system $\mcal{S}$ is said to have weak expectation
property (WEP) if it satisfies any of the equivalent conditions above.
\end{defn}

\begin{thm}\cite{KPTT10, Han}
Let $\mcal{S}$ be an operator system. Then $\mcal{S}$ possesses
$\mathrm{WEP}$ if and only if it is $(\mathrm{el}, \max)$-nuclear.
\end{thm}

\begin{thm}\cite[Theorem $6.11$]{KPTT10}
Let $\mcal{S}$ be a finite dimensional operator system. Then the
following are equivalent:
\begin{enumerate}
\item  $\mcal{S}$ possesses $\mathrm{WEP}$;
\item  $\mcal{S}$ is $(\mathrm{el}, \max)$-nuclear;
\item  $\mcal{S}$ is $(\min, \max)$-nuclear;
\item  $\mcal{S}$ is completely order isomorphic to a $C^*$-algebra; and
\item $\mcal{S} \ot_{\mathrm{el}} \mcal{S}^* = \mcal{S}
  \ot_{\max} \mcal{S}^*$.
\end{enumerate}
\end{thm}

\subsection{Operator system local lifting property (OSLLP).} \index{operator system! OSLLP} Let $H$
be an infinite dimensional Hilbert space and $\mcal{S}$ an operator
system. Let $u \in \mathrm{UCP}(\mcal{S}, Q(H))$, where $Q(H)$ is the
Calkin algebra $Q(H) = B(H)/K(H)$. The operator system $\mcal{S}$ is
said to have OSLLP if, for each such $u$, every finite dimensional
operator sub-system $\mcal{F} \subset \mcal{S}$ admits a lifting
$\tilde{u} \in \mathrm{UCP}(\mcal{F}, B(H))$ such that $\pi \circ
\tilde{u} = u_{|_{\mcal{F}}}$, where $\pi : B(H) \to Q(H)$ is the
canonical quotient map.

\begin{thm}\cite[Theorems $8.1$, $8.5$]{KPTT10}
Let  $\mcal{S}$ be an operator system. Then the following are equivalent:
\begin{enumerate}
\item $\mcal{S}$ possesses $\mathrm{OSSLP}$;
\item $\mcal{S}\ot_{\min}B(H)
= \mcal{S}\ot_{\max}B(H)$ for every Hilbert space $H$; and
\item $\mcal{S}$ is $(\min, \mathrm{er})$-nuclear.
\end{enumerate}
\end{thm}

\subsection{Double commutant expectation property (DCEP).} \index{operator system! DCEP} An
operator system $\mcal{S}$ is said to have DCEP if for every completely
order embedding $\varphi : \mcal{S}\ra B(H)$ there exists a completely
positive mapping $E : B(H) \ra \varphi(\mcal{S})''$ fixing $\mcal{S}$,
i.e., satisfying $ E\circ \varphi = \varphi$.

\begin{thm}\cite{KPTT10}
An operator system $\mcal{S}$ possesses $\mathrm{DCEP}$ if and only if it is
$(\mathrm{el}, \mathrm{c})$-nuclear.
\end{thm}
\begin{rem}
In particular, since ${\mathrm{el}} \leq {\mathrm{c}}$, an operator system $\mcal{S}$ is $1$-exact and
possesses $\mathrm{DCEP}$ if and only if it is $(\min, \mathrm{c})$-nuclear. Of
course, it will be desirable to have a better characterization for
$(\min, \mathrm{c})$-nuclearity.
\end{rem}
\begin{defn}
An unordered graph $G = (V, \mcal{E})$ is said to be a chordal graph
if every cycle in $G$ of length greater than $3$ has a chord, or,
equivalently, if $G$ has no minimal cycle of length $\geq 4$.
\end{defn}

\begin{thm}\cite{KPTT11}
If $G$ is a chordal graph, then $\mcal{S}_G$ is $(\min,
\mathrm{c})$-nuclear.
\end{thm}

\begin{ques}
\begin{enumerate}
\item For which graphs $G$ are the operator systems $\mcal{S}_G$ $(\min,
\mathrm{c})$-nuclear?
\item Consider the graph $G$ consisting of a quadrilateral. Clearly
  $G$ is not chordal. Is the graph operator system $\mcal{S}_G$
  $(\min, \mathrm{c})$-nuclear?
\item Is every graph operator system $(\min, \mathrm{c})$-nuclear?
\item For arbitrary operator system tensor products $\alpha$ and
  $\beta$, which graphs give $(\alpha, \beta)$-nuclear graph operator
  systems?
\item For a graph $G$, study the above problems for the dual operator
  system $\mcal{S}_G^d$ as well.
\item Obtain characterizations for $(\alpha, \beta)$-nuclearity of
  operator systems for the remaining cases.
\item If a graph operator system $\mcal{S}_G$ is $(\alpha,
  \beta)$-nuclear, identify the tensor products $\eta$ and $\xi$ (if
  any) such that $\mcal{S}_G^d$ is $(\eta, \xi)$-nuclear.
\end{enumerate}
\end{ques}

\section{Operator system tensor products and the conjectures of Kirchberg and Tsirelson}
\subsection{Special operator sub-systems of the free group  $C^*$-algebras}

Let $\mbb{F}_n$ be the free group on $n$ generators, say, $\{ g_1,
g_2, \ldots, g_n\}$. For any Hilbert space $H$, any choice of $n$
unitaries $\{ U_1, U_2, \ldots, U_n\}$ in $ B(H)$ gives a (unitary)
representation $\pi : \mbb{F}_n \to B(H)$ of $\mbb{F}_n$ sending $g_i$
to $U_i$ for all $1 \leq i \leq n$. Recall, the full group
$C^*$-algebra \index{group $C^*$-algebra! full} $C^*(\mbb{F}_n)$ is
the closure of the group algebra $\mbb{C}[\mbb{F}_n]$ in the norm
obtained by taking supremum over all (unitary) representations of the
group $\mbb{F}_n$. Let
\[
\mcal{S}_n = \text{span}\{{1}, g_1, \ldots, g_n, g_1^*, \ldots, g_n^*
\} \subset C^*(\mbb{F}_n).
\]
Clearly, $\mcal{S}_n$ is a $(2n+1)$-dimensional operator system.
\subsection{Kirchberg's Conjecture}

A famous conjecture of Kirchberg states that the full group $C^*$-algebra \index{Kirchberg's conjecture}
$C^*(\mbb{F}_\infty)$ has WEP - \cite{Kir94}. It attracts immense
importance from its equivalence with some other important conjectures
in the world of Operator Algebras and now thanks to the work of \cite{jnppsw} we now know that Tsirelson's attempts at determining the possible sets of density matrices for quantum outcomes is also related.

In fact, it is now known that if Tsirelson's conjectures are true then necessarily Kirchberg's conjecture is true. For a physicists perspective on these issues see, \cite{fritz}.

\begin{thm}\label{kirch}\cite{Kir94} The following statements are equivalent:
\begin{enumerate}
\item Connes' Embedding Theorem: Every $II_1$-factor with separable \index{Connes' Embedding Problem}
  predual can be embedded as a subfactor in to the free ultraproduct
  of the hyperfinite $II_1$-factor.
\item $C^*(\mbb{F}_n) \ot_{\mathrm{sp}}C^*(\mbb{F}_n) = C^*(\mbb{F}_n)
  \ot_{C^*\text{-}\max}C^*(\mbb{F}_n)$ for all $ n \geq 1$.
\item $C^*(\mbb{F}_n)$ has $\mathrm{WEP}$ for all $n \geq 1$.
\item $C^*(\mbb{F}_\infty)$ has $\mathrm{WEP}$. (Kirchberg's conjecture)
\item Every $C^*$-algebra is a quotient of a $C^*$-algebra with $\mathrm{WEP}$.
\end{enumerate}
\end{thm}

To the above list, the techniques of operator system tensor products,
has contributed the following (``seemingly simpler'') equivalent
statements:

\begin{thm} \cite{KPTT10}
The following statements are equivalent:
\begin{enumerate}
\item $C^*(\mbb{F}_\infty)$ has $\mathrm{WEP}$.
\item $\mcal{S}_n$ is $(\mathrm{el}, \mathrm{c})$-nuclear for all $n \geq 1$.
\item $\mcal{S}_n \ot_{\min} \mcal{S}_n = \mcal{S}_n \ot_{\mathrm{c}}
  \mcal{S}_n$ for all $ n \geq 1$.
\item Every $(\min, \mathrm{er})$-nuclear operator system is $(\mathrm{el}, \mathrm{c})$-nuclear.
\item Every operator system possessing $\mathrm{OSLLP}$ possesses $\mathrm{DCEP}$.
\end{enumerate}
\end{thm}

In the above list of equivalences, the equivalence $(1) \leftrightarrow
(2)$ in Theorem \ref{kirch} is the deepest link and was proved first
by Kirchberg in \cite{Kir93}. An essential part of the proof of this
equivalence involved the followng deep theorem due to Kirchberg:

\begin{thm}\cite{Kir93}
$C^*(\mbb{F}_n) \ot_{\mathrm{sp}} B(H) = C^*(\mbb{F}_n)
  \ot_{C^*\text{-}\max}B(H)$ for all $n \geq 1$ and for all Hilbert spaces $H$.
\end{thm}

$\bullet$ Quite surprisingly, making use of the notion of {\em
  quotient of an operator system}, a relatively much easier proof of
Kirchberg's Theorem has been obtained in \cite{fp}.

\subsection{Quotient of an operator system.}

The idea of quotient of an operator system comes from the requirement
that given operator systems $\mcal{S}$ and $ \mcal{T}$, and a UCP
$\varphi : \mcal{S} \to \mcal{T}$, we would like to have a quotient
operator system $\mcal{S} / \mathrm{ker}\varphi$ such that the
canonical quotient map $q : \mcal{S}\to \mcal{S} /
\mathrm{ker}\varphi$ is UCP and so is the factor map $\tilde{\varphi}
: \mcal{S} / \mathrm{ker}\varphi \to \mcal{T}$. \index{operator system! quotient}

It also gives a way to explain the duals of graph operator systems, they are actually quotients of the matrix algebra.

\begin{defn}[Quotient map]
Let $\mcal{S}$ and $\mcal{T}$ be operator systems. Then a UCP map
$\varphi : \mcal{S}\ra \mcal{T}$ is said to be a quotient map if
$\varphi$ is surjective and the canonical factor map $\tilde{\varphi}
: \mcal{S}/\mathrm{ker}\varphi \to \mcal{T}$ is a complete order
isomorphism. In other words, $\mcal{T}$ is a quotient of $\mcal{S}$.
\end{defn}

\begin{exam}
Let $\mcal{T}_{n+1} = \{\mathrm{tridiagonal }\, (n+1) \times (n+1)\,
\mathrm{ matrices}\} \subset M_{n+1}$ and $$\mcal{K}_{n+1} =\{\mathrm{trace}
\  0\ \mathrm{ diagonal}\  (n+1) \times (n+1)\ \mathrm{ matrices }\} \subset
M_{n+1}.$$ Consider $\varphi : \mcal{T}_{n+1} \to \mcal{S}_n$ given by
$\varphi (E_{ii}) = \frac{1}{n} 1$, $\varphi (E_{i, i+1}) = \frac{1}{n}
g_i$, $\varphi (E_{i+1, i}) = \frac{1}{n} g_i^*$. Clearly $\varphi$ is
onto and UCP. Also, $\mathrm{ker}\,\varphi = \mcal{K}_{n+1}$. It is a fact
that $\varphi$ is a quotient map in above sense. In particular,
$\mcal{S}_n$ is completely order isomorphic to $\mcal{T}_{n+1} /
\mcal{K}_{n+1}$.

\end{exam}


\chapter{Quantum Information Theory}

\bigskip \noindent
{\bf Chapter Abstract:}

\bigskip \noindent
{\it In this chapter, based on the lectures by Andreas Winter, we survey four different areas in quantum information theory in which ideas from the theory of operator systems and operator algebras play a natural role. The first is the problem of  zero-error communication over quantum channels, which makes use of concepts from operator systems and quantum error correction. Second, we study the strong subadditivity property of quantum entropy and the central role played by the structure theorem of matrix algebras in understanding its equality case. Next, we describe different norms on quantum states and the corresponding induced norms on quantum channels. Finally, we look at matrix-valued random variables, prove Hoeffding-type tail bounds  and describe the applications of such matrix tail bounds in the quantum information setting.}

\section{Zero-error Communication via Quantum Channels}\index{zero-error! communication}

A quantum channel\index{quantum channel}\index{quantum channel} $\cT$ is a completely positive, trace-preserving (CPTP)\index{CPTP} map from the states of one system ($A$) to another ($B$). Specifically, $\cT : \cB(\cH_{A})\rightarrow \cB(\cH_{B})$ is a CPTP map from the set of bounded linear operators on Hilbert space $\cH_{A}$ to operators in $\cH_{B}$. In this section we focus on the problem of zero-error communication using quantum channels. We begin with a brief review of preliminaries including the idea of purification and the Choi-Jamiolkowski isomorphism\index{Choi-Jamiolkowski! isomorphism}.

\begin{defn}[Purification]\label{def:purification}\index{purification}
Given any positive semi-definite operator $\rho \geq 0$ in $\cB(\cH_{A})$, suppose there exists a vector $|v\rangle \in \cH_{A} \otimes \cH_{A'}$, where $\cH_{A'}$ is simply an auxiliary Hilbert space, such that $\tr_{A'}[|v\rangle\langle v|] = \rho$. The vector $| v\rangle$ in the extended Hilbert space is said to be a {\bf purification} of the operator $\rho$.
\end{defn}
When $\tr[\rho] = 1$, that is, when $\rho$ is a valid quantum state, then the corresponding vector $|v\rangle$ is a pure state of the extended Hilbert space, satisfying $\langle v |v\rangle =1$.

In order to obtain a purification of a state $\rho$, it suffices to have the dimensions of the auxiliary space $\cH'$ to be equal to the rank of $\rho$. To see this, suppose $\rho$ has a spectral decomposition $\rho = \sum_{i}r_{i} |e_{i}\rangle\langle e_{i}|$, then a purification of $\rho$ is simply given by
\[ |v\rangle = \sum_{i}\sqrt{r_{i}}|e_{i}\rangle\otimes|e_{i}\rangle.\]
The purification of a given state $\rho$ is not unique. Suppose there exists another purification $|w\rangle \in \cH_{A}\otimes\cH_{B}$, via a different extension of the Hilbert space $\cH_{A}$, such that $\tr_{B}[|w\rangle\langle w|] = \rho$. Then, there exists a unique isometry $U: \cH_{A'} \rightarrow \cH_{B}$ such that~\footnote{A remark on notation: throughout this chapter we use $I$ to denote the identity operator and $\Id$ to denote the identity map, for example, $\Id_{A}:\cB(\cH_{A})\rightarrow \cB(\cH_{A})$.}
\[ |w\rangle = (\Id \otimes U)|v\rangle.\]
This isometry can in fact be obtained from the Choi uniqueness theorem (Theorem~\ref{choi2}).

We next define the Choi-Jamiolkowski matrix corresponding to a quantum channel $\mathcal{T}: \cB(\cH_{A}) \rightarrow \cB(\cH_{B})$. This provides an alternate way to obtain the Stinespring dilation\index{Stinespring dilation} of the CPTP map $\cT$, discussed earlier in Theorem 1.1.8.
\begin{defn}[Choi-Jamiolkowski Matrix]\index{Choi-Jamiolkowski! matrix}
Let $\{|i\rangle\}$ denote an orthonormal basis for $\cH_{A}$. Consider the (non-normalized) maximally entangled state: $|\Phi_{AA'}\rangle = \sum_{i}|ii\rangle \in \cH_{A} \otimes \cH_{A'}$, with $\cH_{A'}$ chosen to be isomorphic to $\cH$.
The Choi-Jamiolkowski matrix corresponding to a CPTP map $\cT: \cB(\cH_{A})\rightarrow \cB(\cH_{B})$ is then defined as
\begin{equation}
 J_{AB} := (\Id \otimes \cT)\Phi,
 \end{equation}
where, the operator $\Phi = |\Phi\rangle \langle \Phi| \in \cB(\cH_{A} \otimes \cH_{A'})$ is simply
\[ \Phi := \sum_{i,j}|i\rangle\langle j | \otimes |i\rangle\langle j| = \sum_{i,j}|ii\rangle\langle jj|. \]
\end{defn}
Complete positivity of $\cT$ implies that $J \geq 0$, and the trace-preserving condition on $\cT$ implies,
\begin{equation}
 \tr_{B}[J] = I_{A} = \tr_{A'}(\Phi). \label{eq:CJ_pure}
\end{equation}
Pick a purification $|G_{ABC}\rangle \in \cH_{A}\otimes\cH_{B}\otimes\cH_{C}$ of the matrix $J$. Thus, $|G_{ABC}\rangle\langle G_{ABC}|$ is a rank-one, positive operator satisfying $\tr_{C}[|G_{ABC}\rangle\langle G_{ABC}|] = J_{AB}$. $\cH_{C}$ is any auxiliary Hilbert space whose dimension is ${\rm dim}(\cH_{C}) \geq {\rm rank}(J)$. Then, Eq.~\eqref{eq:CJ_pure} implies that
\[ \tr_{BC}[|G_{ABC}\rangle\langle G_{ABC}|] = I_{A} = \tr[\Phi_{AA'}].\]
Therefore, there exists an isometry $U : \cH_{A'} \rightarrow \cH_{B}\otimes \cH_{C}$ such that
\[ |G_{ABC}\rangle = (\Id \otimes U)|\Phi_{AA'}\rangle. \]
Since the Choi matrix $J_{AB}$ corresponding to the map $\cT$ is unique, we have for any $X \in \cB(\cH_{A})$,
\[\cT(X) = \tr_{C}[UXU^{\dagger}].\]
The isometry thus corresponds to the Stinespring dilation of the map $\cT$. Furthermore, we also obtain the Choi-Kraus decomposition\index{Choi-Kraus! decomposition} (Theorem 1.2.1) of the map $\cT$ by noting that the isometry can be rewritten as  $U = \sum_{i}E_{i}\otimes |v_{i}\rangle$, where $\{|v_{i}\rangle\}$ is an orthonormal basis for $\cH_{C}$. Thus,
\[\cT(X) = \tr_{C}[UXU^{\dagger}] = \sum_{i}E_{i}XE_{i}^{\dagger}, \; \forall \; X \in \cB(\cH_{A}).\]
The non-uniqueness of the Kraus representation is captured by the non-uniqueness of the choice of basis $\{|v_{i}\rangle\}$.

In physical terms, this approach to the Choi-Kraus decomposition offers an important insight that CP maps can in fact be used to represent noisy interactions in physical systems. Examples of such noise processes include sending a photon through a lossy optic fibre or a spin in a random magnetic field. Any physical noise affecting a system $A$ is typically thought of as resulting from unwanted interaction with an {\it environment} which is represented by the system $C$ here. The total evolution of the system $+$ environment is always {\it unitary} (a restriction of the isometry $U$) and the noise results from the act of performing the partial trace which physically corresponds to the fact that we do not have access to complete information about the environment.

\subsection{Conditions for Zero-error Quantum Communication}

In the context of quantum communication, a quantum channel described by the CPTP map $\cT: \cB(\cH_{A})\rightarrow \cB(\cH_{B})$ represents a {\it process}: it takes as input, states $\rho \in \cB(\cH_{A})$ and produces corresponding states $\cT(\rho) \in \cB(\cH_{B})$ as output. It could model an information transmission process that transmits some set of input signals from one location to another, or, it could model a data storage scenario in which some information is input into a noisy memory at one time to be retrieved later.

A classical channel\index{classical channel} $\cN$, in the Shannon formulation, is simply characterized by a kernel or a probability transition function $N(Y|X)$. $\{N(y|x) \geq 0\}$ are the conditional probabilities of obtaining output $y\in Y$ given input $x \in X$, so that, $\sum_{y}N(y|x) = 1$. $X$ and $Y$ are often called the input and output {\it alphabets} respectively. The probabilities $\{N(y|x)\}$ thus completely describe the classical channel $\cN$.

The quantum channel formalism includes a description of such classical channels as well. For example, let $\{|x\rangle\}$ denote an orthonormal basis for $\cH_{A}$ and $\{|y\rangle\}$ denote an orthonormal basis for $\cH_{B}$, where the labels $x$ and $y$ are drawn form the alphabets $X$ and $Y$ respectively. Then, corresponding to the channel $\cN \equiv \{N(y|x)\}$, we can construct the following map on states $\rho \in \cB(\cH_{A})$:
\[ \cT(\rho) = \sum_{x,y}N(y|x) \; |y\rangle\langle x| (\rho) |x\rangle\langle y|. \]
It is easy to see that $\cT: \cB(\cH_{A})\rightarrow \cB(\cH_{B})$ is a CPTP map with Kraus operators $E_{xy} = \sqrt{N(y|x)}\; |y\rangle\langle x|$, which maps diagonal matrices to diagonal matrices. Any non-diagonal matrix in $\cB(\cH_{A})$ is also mapped on to a matrix that is diagonal in the $\{|y\rangle\}$-basis. Classical channels are thus a special case of quantum channels.

Apart from the action of the CPTP map, a quantum communication protocol also includes an {\it encoding} map at the input side and a {\it decoding} map at the output. Given a set of messages $\{m=1,2,.\ldots, q\}$, the encoding map assigns a quantum state $\rho_{m} \in \cB(\cH_{A})$ to each message $m$. The decoding map has to identify the message $m$ corresponding to the output $\cT(\rho_{m})$ of the channel $\cT$. In other words, the decoding process has to extract classical information from the output quantum state; this is done via a {\it quantum measurement}. Recall from the discussion in Sec.~\ref{sec:qstates}, that the outcome $M$ of a measurement of state $\cT(\rho_{m})$ is a random variable distributed according to some classical probability distribution. Here, we are interested in {\bf zero-error communication}, where the outcome $M$ is equal to the original message $m$ with probability $1$.

Zero-error transmission via general quantum channels was originally studied in~\cite{Med, BS} and more recently in~\cite{CCA, CLMW, Duan}. In this section we first review some of this earlier work, highlighting the role of operator systems in the study of zero-error communication. In the next section, we focus on the recent work of Duan {\it et al}~\cite{DSW} where a quantum version of the Lov\'asz $\vartheta$-function is introduced in the context of studying the zero-error capacity of quantum channels.

Firstly, note that the requirement of zero-error communication imposes the following constraint on the output states $\{\cT(\rho_{m})\}$.
\begin{exer}
There exists a quantum measurement $\cM$ in $\cH_{B}$ such that the outcome $M$ corresponding to a measurement of state $\cT(\rho_{m})$ is equal to $m$ with probability $1$, if and only if the ranges of the states $\{\cT(\rho_{m})\}$ are mutually orthogonal.
\end{exer}
Since the states $\{\cT(\rho_{m})\}$ are positive semi-definite operators, the fact their ranges are mutually orthogonal implies the following condition:
\begin{equation}
 \tr[\cT(\rho_{m})\cT(\rho_{m'})] = 0, \; \forall \; m\neq m' \, . \label{eq:orth_cond}
\end{equation}
Suppose we choose a particular Kraus representation for $\cT$, so that $\cT(\rho) = \sum_{i}E_{i}\rho E_{i}^{\dagger}$. It then follows from Eq.~\eqref{eq:orth_cond} that,
\[ \sum_{i,j} \tr[E_{i}\rho_{m}E_{i}^{\dagger}E_{j}\rho_{m}E_{j}^{\dagger}] = 0, \; \forall \; m\neq m' \, . \]
Note that the orthogonality condition on the ranges of the output states further implies that the input states $\{\rho_{m}\}$ can be chosen to be rank-one operators  ($\{\rho_{m} = |\psi_{m}\rangle\langle \psi_{m}|\}$) without loss of generality. Therefore, the condition for zero-error communication becomes
\begin{eqnarray}
 \vert \langle \psi_{m}|E_{i}^{\dagger}E_{j}|\psi_{m'}\rangle\vert^{2} &=& 0, \; \forall m \neq m' \, . \nonumber \\
\Rightarrow \langle \psi_{m}|E_{i}^{\dagger}E_{j}|\psi_{m'}\rangle &=& 0, \; \forall m \neq m', \; \forall \; i,j \, . \nonumber \\
\Rightarrow \tr[|\psi_{m'}\rangle\langle \psi_{m}| E_{i}^{\dagger} E_{j}] &=& 0, \; \forall m \neq m', \; \forall \; i,j \, .
\end{eqnarray}
We have thus obtained the following condition for zero-error communication using the quantum channel $\cT$.
\begin{lem}[Condition for Zero-error Communication]\index{zero-error communication}
Given a channel $\cT$ with a choice of Kraus operators $\{E_{i}\}$, zero-error communication via $\cT$ is possible if and only if the input states $\{|\psi_{m}\rangle \in \cH_{A}\}$ to the channel satisfy the following: $\forall \; m\neq m'$, the operators $|\psi_{m'}\rangle\langle\psi_{m}| \in \cB(\cH_{A})$ must be orthogonal to the span
\begin{equation}
 S := {\rm span}\{E_{i}^{\dagger}E_{j}, \; i,j\}, \label{eq:span}
\end{equation}
with orthogonality defined in terms of the Hilbert-Schmidt inner product~\footnote{The Hilbert-Schmidt inner product\index{Hilbert-Schmidt inner product} between two operators $A,B$ is simply the inner product defined by the trace, namely, $\langle A,B\rangle = \tr[A^{\dagger}B]$.}.
\end{lem}

Note that $S \subset \cB(\cH_{A})$, and, $S = S^{\dagger}$. Further, since the channel $\cT$ is trace-preserving, $\sum_{i}E_{i}^{\dagger}E_{i} = I$, so that $S \ni I$. This implies that $S$ is an {\bf Operator System}\index{operator system}, as defined in Def.~\ref{def:OS}. Since all Kraus representations for $\cT$ give rise to the same subset $S$, the above condition is unaffected by the non-uniqueness of the Kraus representation.

The {\bf Complementary Channel} $\hat{\cT}$ and its dual $\hat{\cT}^{*}$ corresponding to a channel $\cT$, are defined as follows.
\begin{defn}[Complementary Channel]\index{quantum channel! complementary}
Suppose the channel $\cT$ is given by $\cT(\rho) = \tr_{C}[V\rho V^{\dagger}]$, where $V : \cH_{A} \otimes \cH_{B}\otimes\cH_{C}$ is the Stinespring isometry. Then, the Complementary Channel $\hat{\cT}: \cB(\cH_{A})\rightarrow \cB(\cH_{C})$ is defined as:
\[\hat{\cT}(\rho) = \tr_{B}[V \rho V^{\dagger}]. \]
\end{defn}
The {\bf dual}\index{quantum channel! dual} map $\hat{\cT}^{*}$ is defined via $\tr[\rho\hat{\cT}^{*}(X)] = \tr[\hat{\cT}(\rho)X^{\dagger}]$. Then, the following result was shown in~\cite{DSW}.
\begin{obs}
Given a channel $\cT$ with a complementary channel $\hat{\cT}$,
\[ S = \hat{\cT}^{*}(\cB(\cH_{C})), \; \hat{\cT}^{*}: \cB(\cH_{C}) \rightarrow \cB(\cH_{A}),\]
where $S$ is the operator system defined in Eq.~\eqref{eq:span}, and, $\hat{\cT}^{*}$ is the dual to the map $\hat{\cT}$.
\end{obs}
Furthermore, it turns out that every operator system can be realized in this manner~\cite{Duan, CCA}: 
\begin{obs}
Given an operator system $S\subset \cB(\cH_{A})$, there exists a CPTP map $\cT$ with a choice of Kraus operators $\{E_{i}\}$ such that $S = {\rm span}\{E_{i}^{\dagger}E_{j}, \; i,j\}$.
\end{obs}

\subsection{Zero-error Capacity and Lovasz $\vartheta$ Function}

Using the condition for zero-error communication, we next quantify the maximum number of messages $m$ that can be transmitted reliably through the channel $\cT$. We begin with the notion of a {\it quantum independence number} originally defined in~\cite{DSW}.
\begin{defn}[Independence Number of S]\label{def:ind_no}\index{operator system! independence number}\index{independence number! quantum}
Given an operator system $S = {\rm span}\{E_{i}^{\dagger}E_{j}\}$, the independence number $\alpha (S)$ is defined as the maximum value of $q$, such that
 \begin{equation}
  \exists \quad \{|\psi_{1}\rangle, |\psi_{2}\rangle, \ldots, |\psi_{q}\rangle\} : \forall \; m \neq m', \; |\psi_{m}\rangle\langle\psi_{m'}| \perp S. \label{eq:ind_no}
 \end{equation}
\end{defn}
To understand better the motivation for this definition, consider the example of the classical channel again. If $\cT$ is classical, the Kraus operators of the channel are $\{E_{xy} = \sqrt{N(y|x)} \; |y\rangle\langle x|\}$. The operator system $S$ in this case is given by
\[ S  = {\rm span}\{E^{\dagger}_{x'y'}E_{xy}\} = {\rm span}\{\sqrt{N(y|x)}\sqrt{N(y'|x')} |x\rangle\langle y|y' \rangle\langle x' |\}. \]
Note that, $E^{\dagger}_{x'y'}E_{xy} \neq 0$ iff $x=x'$ (so that $y=y'$) or $N(y|x)N(y|x') \neq 0$, whenever $x \neq x'$. The latter condition naturally leads to the notion of the {\it Confusability Graph} associated with a classical channel, and its {\it Independence Number}.

\begin{defn}[Confusability Graph, Independence Number]\index{Confusability graph}\index{independence number! classical}
\begin{itemize}
\item[(i)] The confusability graph of a classical channel $N$ is the graph $G$ with vertices $x \in X$ and edges $x \sim x'$ iff $\exists \; y \in Y$ such that $N(y|x)N(y|x') \neq 0$.
\item[(ii)] A set of vertices $X_{0} \subset X$ such that no pair of vertices in the set is has an edge between them is said to be an independent set. The maximum size of an independent set of vertices $X_{0}$ in $G$ is called the Independence Number $\alpha(G)$ of the graph $G$.
\end{itemize}
\end{defn}
The name comes from the fact that the edges $x \sim x'$ of the graph correspond to {\it confusable} inputs, that is, inputs $x,x'$ that map to the same output $y$. Thus, the operator system $S$ corresponding to this classical quantum channel $\cT$ carries information about the structure of the underlying graph $G$:
\begin{equation}
{\rm span}\{|x\rangle\langle x'|: x = x' \; {\rm or} \; x \sim x' \} \label{eq:class_span}
\end{equation}
More generally, using this definition, every graph $G$ gives rise to an operator system $S$. Furthermore, if two such operator systems are isomorphic, then the underlying graphs are also isomorphic in a combinatorial sense. Thus, for the special case of an operator system $S$ coming from a classical channel as in Eq.~\eqref{eq:class_span},
\[ \alpha (S) = \alpha (G).\]
Def.~\ref{def:ind_no} of the independence number is therefore a generalization of the notion of the independence number of a graph.

Given a graph $G$, estimating its independence number $\alpha(G)$ is known to be an NP-complete problem. Similarly, it was shown that~\cite{BS} estimating $\alpha(S)$ for an operator system $S$ arising from a quantum channel $\cT$ is a QMA-complete problem. QMA (Quantum Merlin-Arthur) is the class of problems that can be solved by a quantum polynomial time algorithm given a quantum witness; it is the quantum generalization of probabilistic version of NP. Rather than estimating $\alpha(S)$, in what follows we seek to find upper bounds on $\alpha(S)$.

We first rewrite the condition in Eq.~\eqref{eq:ind_no} in terms of positive semi-definite operators. Note that $|\psi_{m}\rangle\langle\psi_{m'}| \perp S$ implies that
\[ M = \sum_{m \neq m'}|\psi_{m}\rangle\langle\psi_{m'}| \perp S.\]
Since $S \ni I$, the states $\{|\psi_{m}\rangle, m=1,2,\ldots,q\}$ satisfying Eq.~\eqref{eq:ind_no} are mutually orthogonal. Also, the operator $M + \sum_{m}|\psi_{m}\rangle\langle\psi_{m}|$ is positive semi-definite. Therefore,
\[ 0 \leq M + \sum_{m=1}^{q}|\psi_{m}\rangle\langle\psi_{m}| = M + I , \]
where the number $q$ is simply $q = \parallel M + \sum_{m}|\psi_{m}\rangle\langle\psi_{m}| \parallel$. Recalling that $\alpha(S)$ is simply the maximum value of $q$, we have,
\begin{equation}
 \alpha(S) \leq \max_{\substack{ \left \{M \in \cB(\cH_{A}) : M\perp S, \right. \\ \left. I + M \geq 0 \right\}}} \parallel M + I \parallel = \vartheta(S). \label{eq:lovasz_no}
\end{equation}
The quantity on the RHS, defines a {\bf quantum $\vartheta$-function}\index{Lov\'asz $\vartheta$ function}, as a straightforward generalization of a well known classical quantity. If the operator system $S$ arises from a graph $G$ as in Eq.~\eqref{eq:class_span}, then, $\vartheta(S) = \vartheta(G)$, where, $\vartheta(G)$ is the Lov\'asz number of the graph $G$~\cite{Lov79}. It was shown by Lov\'asz that $\vartheta(G)$ is an upper bound for the independence number $\alpha(G)$ and is in fact a semidefinite program~\cite{SDP}\index{semidefinite program}. In other words, for the case that $S$ arises from a graph $G$, the optimization problem that evaluates $\vartheta(S)$ has a well behaved objective function and the optimization constraints are either linear ($M \perp S$) or semi-definite ($I + M \geq 0$). However for a general operator system $S$, the optimization problem that evaluates $\vartheta(S)$ is not a semidefinite program (SDP) in general.

In order to define the {\bf zero-error capacity} of a channel, we move to the asymptotic setting, where we consider several copies ($n$) of the channel in the limit of $n\rightarrow \infty$. The operator system corresponding to the $n$-fold tensor product of a channel is simply the $n$-fold tensor product of the operator systems associated with the original channel. Then, the capacity is formally defined as follows.
\begin{defn}[Zero-error Capacity]\index{zero-error! capacity}
 The zero-error capacity ($C_{0}(S)$) of a channel with associated operator system $S$ is the maximum number of bits that can be transmitted reliably per channel use, in the asymptotic limit.
\[C_{0}(S) := \lim_{n\rightarrow \infty} \frac{1}{n}\alpha(S^{\otimes n}). \]
\end{defn}
The capacity is even harder to compute than the independence number. In fact, it is not even known if $C_{0}(S)$ is in general a computable quantity in the sense of Turing! For classical channels, where the operator system $S$ arises from the confusability graph $G$, Lov\'asz obtained an upper bound on $C_{0}(S)$. In this case, the Lov\'asz number satisfies,
\[ \vartheta(G_{1} \times G_{2}) = \vartheta(G_{1})\vartheta(G_{2}),\]
which immediately implies that $C_{0}(S) \leq \vartheta (G)$. Till date, this remains the best upper bound on the zero-error capacity of a classical channel.

To gain familiarity with the quantum $\vartheta$-function, we evaluate it for two simple examples of quantum channels.
\begin{example}\label{ex:entire}
Consider the channel whose Kraus operators span the entire space $\cB(\cH_{A})$. The corresponding operator system is given by $S = \cB(\cH_{A})$. Then, the only $M\geq 0$ satisfying $M \perp S$ is in fact $M = 0$. Therefore, by the definition in Eq.~\eqref{eq:lovasz_no}, $\vartheta(S) \equiv \vartheta(\cB(\cH_{A})) = 1$.
\end{example}
Next we consider the case where $S$ is a multiple of the identity.
\begin{example}\label{ex:identity}
Suppose $S = \mathbb{C}I$. Then, the operators $M \perp S$ have to satisfy $\tr[M] = 0$. Without loss of generality, we may assume $M = {\rm diag}(\lambda_{1}, \ldots, \lambda_{n})$ with $\sum_{i}\lambda_{i} = 1$. Suppose the eigenvalues are ordered as follows: $\lambda_{1}\geq \lambda_{2}\geq \ldots \lambda_{n}$. Furthermore, the positive semi-definiteness constraint on $(I + M)$ implies $\lambda_{i} \geq -1$. Therefore,
\begin{eqnarray}
 \parallel I + M \parallel &=& 1 + \lambda_{1} \leq n. \nonumber \\
\Rightarrow \vartheta(S) = \vartheta(\mathbb{C}I) &=& n = {\rm dim}(\cH_{A}).
\end{eqnarray}
\end{example}

We note the following interesting property of the $\vartheta$-function.
\begin{lem}\label{lem:theta_tensor}
Given two operator systems $S_{1}$ and $S_{2}$,
\begin{equation}
 \vartheta(S_{1}\otimes S_{2}) \geq \vartheta(S_{1})\vartheta(S_{2}).
\end{equation}
\end{lem}
\begin{proof}
 Suppose the operator $M_{1}\perp S_{1}$ achieves $\vartheta(S_{1})$ and $M_{2}\perp S_{2}$ achieves $\vartheta(S_{2})$. Then, define,
\[ I + M := (I + M_{1})\otimes (I + M_{2}). \]
By definition, $M \perp S_{1}\otimes S_{2}$. Since the norm is multiplicative under tensor product,
\[ \vartheta (S_{1}\otimes S_{2}) \geq \parallel I + M \parallel = (\parallel I + M_{1}\parallel)(\parallel I + M_{2} \parallel) = \vartheta(S_{1})\vartheta(S_{2}). \]
\end{proof}

To see that the inequality in Lemma~\ref{lem:theta_tensor} can be strict, consider the case where $S = I_{n} \otimes \cB(\mathbb{C}^{n})$. Then, it is a simple exercise to evaluate $\vartheta(S)$.
\begin{exer}
Show that \[ \vartheta (I_{n}\otimes \cB(\mathbb{C}^{n})) = n^{2}.\]
Hint: Recall dense coding (Sec.~\ref{sec:EPR_dense})!
\end{exer}
The product of the individual $\vartheta$-functions evaluated in examples~\ref{ex:entire} and~\ref{ex:identity} is much smaller:
\[ \vartheta(I_{n}) \vartheta(\cB(\mathbb{C}^{n})) = n < \vartheta (I_{n}\otimes \cB(\mathbb{C}^{n})).\]
Thus we have a simple example of the {\it non-multiplicativity} of the quantum $\vartheta$-function.

This non-multiplicativity motivates the definition of a modified $\vartheta$-function which can be thought of as a completion of $\vartheta(S)$.
\begin{defn}[Quantum Lov\'asz $\vartheta$-function]
For any operator system $S$, define the quantum Lov\'asz function as follows.
\begin{equation}
\tilde{\vartheta}(S) = \sup_{\cH_{C}}\vartheta(S \otimes \cB(\cH_{C})). \label{eq:tilde_theta}
\end{equation}
\end{defn}
Note how the definition is reminiscent of norm-completion. To clarify the the operational significance of this modified $\vartheta$-function, we define a related independence number in a modified communication scenario.
\begin{defn}[Entanglement-Assisted Independence Number]\label{def:tilde_alpha}\index{independence number! entanglement-assisted}
 For any operator system $S \subset \cB(\cH_{A})$, the quantity $\tilde{\alpha}(S)$ is the maximum value of $q$ such that $\exists$ Hilbert spaces $\cH_{A_{0}}, \cH_{C}$ and isometries $\{V_{1}, V_{2}, \ldots V_{q}\}: \cH_{A_{0}}\rightarrow \cH_{A}\otimes \cH_{C}$, such that,
\[ \forall \; m\neq m', V_{m}\rho V{m}^{\dagger} \, \perp \,  S\otimes \cB(\cH_{C}), \; \forall \; \rho \in \cS(\cH_{A_{0}}).  \]
\end{defn}
Physically, $\tilde{\alpha}(S)$ corresponds to the maximum number of messages that can be transmitted reliably in an {\it entanglement-assisted} communication problem, which allows for some entanglement to be shared beforehand between the sender and the receiver.

It is easy to see that $\alpha(S) \leq \tilde{\alpha}(S)$, since the non-multiplicativity of the $\vartheta$-function implies $\vartheta(S)\leq \tilde{\vartheta}(S)$. Furthermore, $\tilde{\vartheta}(S)$ is an upper bound for $\tilde{\alpha}(S)$, just as $\vartheta(S)$ is to $\alpha(S)$.

\begin{exer}
Given an operator system $S$, the $\tilde{\vartheta}$-function is an upper bound for the entanglement-assisted independence number:
\[ \tilde{\alpha}(S) \leq \tilde{\vartheta}(S).\]
\end{exer}
The following simple observations are left as exercises.
\begin{exer}[Larger operator systems have a smaller independence number:]
Given two operator systems $S_{1}, S_{2}$ such that $S_{1}\subset S_{2}$, $\alpha(S_{1}) \geq \alpha(S_{2})$. This also holds for the $\tilde{\alpha}$, $\vartheta$ and $\tilde{\vartheta}$ functions.
\end{exer}
\begin{exer}[Upper bounds:]
Show that for an operator system $S \subset \cB(\cH_{A})$, (a) $\vartheta(S) \leq {\rm dim}(\cH_{A})$, and, (b) $\tilde{\vartheta}(S) \leq ({\rm dim}(\cH_{A}))^{2}$. Equality holds in both cases when $S = \mathbb{C}I_{{\rm dim}(\cH_{A})}$.
\end{exer}
In fact, when $S = \mathbb{C}I_{2}$, $\tilde{\alpha}(S) = ({\rm dim}(\cH_{A}))^{2}$, which is easily proved using the idea of superdense coding (Sec.~\ref{sec:EPR_dense}).

\begin{example}[$\tilde{\alpha}(S)$ for a Qubit]
For example, consider the simplest case of a qubit, for which ${\rm dim}(\cH_{A}) = 2$ and $S = \mathbb{C}I_{2}$. Since this is the entanglement-assisted communication scenario, suppose the sender and receiver share the maximally entangled state
\[ |\psi\rangle_{AC} = \frac{1}{\sqrt{2}}(|00\rangle + | 11\rangle).\]
The sender modifies the part of the state that belongs to her subsytem via conjugation by a unitary $V_{m} \in \{\Id_{2}, X, Y , Z\}$. These four operators constitute a basis for the space of $2\times 2$ matrices, and were discussed earlier in the construction of Shor's code. It is easy to check that the states $\{(V_{m}\otimes \Id)|\psi\rangle,\; m=1, \ldots, 4\}$ are mutually orthogonal. Thus, once the sender sends across these modified states, the receiver can perfectly distinguish them, implying that  $\tilde{\alpha}(\mathbb{C}I_{2}) = 4$.

Comparing with Def.~\ref{def:tilde_alpha}, we see that $\cH_{A_{0}} = \cH_{A}$ in this case, so that the isometries $V_{m}$ become unitaries and the input state is $\rho = \tr_{C}[|\psi\rangle\langle\psi|]$.
\end{example}

Following~\cite{DSW}, we now study the $\tilde{\vartheta}$-function in greater detail and show that it is indeed a true generalization of the classical Lov\'asz number. Note that the non-multiplicativity of the $\vartheta$-function carries over to the $\tilde{\vartheta}$-function. Given two operator systems $S_{1}$ and $S_{2}$,
\[ \tilde{\vartheta}(S_{1}\otimes S_{2}) \geq \tilde{\vartheta}(S_{1})\tilde{\vartheta}(S_{2}). \]
However, unlike $\vartheta(S)$, $\tilde{\vartheta}(S)$ can be computed efficiently via a SDP, analogous to the classical Lov\'asz number $\vartheta(G)$.
\begin{thm}
For any operator system $S$, $\tilde{\vartheta}(S)$ is a semidefinite program\index{semidefinite program}.
\end{thm}
\begin{proof}
We first show that $\vartheta(S \otimes \cB(\cH_{C}))$ is in fact a semidefinite program. For $S\subset \cB(\cH_{A})$, recall that,
\begin{equation}
 \vartheta(S\otimes \cB(\cH_{C})) = \max_{|\phi\rangle \in \cH_{A}\otimes \cH_{C}} \langle \psi | I \otimes M |\psi\rangle,  \label{eq:theta_tensor}
 \end{equation}
subject to $\parallel |\psi\rangle \parallel = 1$ , $I + M \geq 0$ and $M \perp S \otimes \cB(\cH_{C})$. The last constraint is equivalent to the constraint that $M \in S^{\perp}\otimes \cB(\cH_{C})$. The objective function in Eq.~\eqref{eq:theta_tensor} is a multi-linear function and the constraint on the norm of $|\phi\rangle$ is a non-linear one. To recast this as a SDP, we use the following trick.

Assume ${\rm dim}(\cH_{C}) = {\rm dim}(\cH_{A}) $. Then, $\exists \rho \in \cS(\cH_{C})$, such that for $|\Phi\rangle = \sum_{i}|i\rangle|i\rangle$,
\[ |\psi\rangle = (I \otimes \sqrt{\rho})|\Phi\rangle .\]
Inserting this in Eq.~\eqref{eq:theta_tensor}, the objective function becomes,
\begin{equation}
 \langle \Phi| I \otimes \rho + M' |\Phi\rangle, \; M' = (I \otimes \sqrt{\rho})M (\Id \otimes \sqrt{\rho}) \in S \otimes \cB(\cH_{C}), \label{eq:sdp1}
\end{equation}
with the constraints $I \otimes \rho + M' \geq 0$, $\rho \geq 0$ and $\tr[\rho] =1$. This defines a semidefinite program~{SDP}.

Since $\tilde{\vartheta}(S)$ involves a further maximization over $\cH_{C}$ (Eq.~\eqref{eq:tilde_theta}), we can assume without loss of generality that ${\rm dim}(\cH_{C}) \geq {\rm dim}(\cH_{A})$. If ${\rm dim}(\cH_{C}) > {\rm dim}(\cH_{A})$, we can still identify a subspace of $\cH_{C}$ where we can construct the state $|\Phi\rangle$ and set up the same optimization problem as in Eq.~\eqref{eq:sdp1}. Therefore,
\begin{eqnarray}
\tilde{\vartheta}(S) &=& \max_{\rho, M'}\langle \Phi | I \otimes \rho + M' |\Phi\rangle,  \nonumber \\
{\rm such \;  that} &:& I \otimes \rho + M' \geq 0,  \qquad M' \perp S \otimes \cB(\cH_{C}), \nonumber \\
&& \rho \geq 0, \qquad \tr[\rho] =1. \label{eq:sdp2}
\end{eqnarray}
This is indeed an SDP since the objective function is linear and the constraints are either semi-definite or linear.
\end{proof}

It was further shown by Duan {\it et al}~\cite{DSW} that the {\bf dual}\footnote{The dual of a convex program~\cite{BV1, BV2} is obtained by introducing a variable for each constraint, and a constraint for every co-efficient in the objective function. Thus, the objective function and the constraints get interchanged in the dual problem.} to the optimization problem in Eq.~\eqref{eq:sdp2} is also an SDP of the following form:
\begin{eqnarray}
\min_{Y \in S \otimes \cB(\cH_{C})} && \lambda, \nonumber \\
 {\rm such \; that} &:& \lambda I\geq \tr_{A}[Y], \; Y \geq |\Phi\rangle\langle\Phi|. \label{eq:dual}
\end{eqnarray}
Furthermore, strong duality holds here, implying that the two optimization problems in Eqs.~\eqref{eq:sdp2} and~\eqref{eq:dual} are equivalent:  $\max\langle\Phi|I\otimes \rho + M'|\Phi\rangle \equiv \min \lambda$.

The form of the dual in Eq.~\eqref{eq:dual} also implies that
\[ \tilde{\vartheta}(S_{1}\otimes S_{2}) \leq \tilde{\vartheta}(S_{1})\tilde{\vartheta}(S_{2}).\]
Therefore we see from the SDP formulation that the $\tilde{\vartheta}$-function is in fact {\it multiplicative}:
\[ \tilde{\vartheta}(S_{1}\otimes S_{2}) = \tilde{\vartheta}(S_{1})\tilde{\vartheta}(S_{2}). \]

We conclude this section with a few open questions. One important question is regarding the largest dimension  of a non-trivial operator system $S$. Consider for example $S \subset M_{n}$ such that $S^{\perp} = {\rm diag}(n-1, -1, \ldots, -1)$. It trivially follows that $\tilde{\vartheta}(S) = n = \vartheta(S)$. However, $\alpha(S) = 1$ since there does not exist a rank-$1$ operator in $S^{\perp}$. It can be shown that $\tilde{\alpha}(S) = 2$. This example already shows that there could exist a large gap between the independence numbers $\alpha(S), \tilde{\alpha}(S)$ and their upper bounds $\vartheta(S), \tilde{\vartheta}(S)$. An important open question is therefore that of finding the {\bf entanglement-assisted zero-error capacity} of the operator system $S$, which is defined as follows:
\[ \cC_{0E}(S) :=  \lim_{k \rightarrow \infty} \frac{1}{k}\log\tilde{\alpha}(S^{\otimes k}). \]
 From the values of $\tilde{\alpha}(S)$ and $\tilde{\vartheta}(S)$, it is clear that $1 \leq \cC_{0E} \leq \log n$. However, even the simple question of whether $\tilde{\alpha}(S\otimes S) = 4$ or $\tilde{\alpha} > 4$ remains unanswered for any value of $n$. Another line of investigation is to explore further the connection between graphs and operator systems that emerges naturally in this study of zero-error quantum communication.

\section{Strong Subadditivity and its Equality Case}\label{sec:SSA}


The von Neumann entropy\index{von Neumann entropy}  of a state $\rho \in \cS(\cH_{A})$  of a finite-dimensional Hilbert space $\cH_{A}$ is defined as
\begin{equation}
 S(\rho) := - \tr[\rho\log \rho] = -\sum_{i}\lambda_{i}\log \lambda_{i},
\end{equation}
where $\textrm{\bf spec}(\rho) = \{\lambda_{1},\ldots,\lambda_{|A|}\}$. It easy to see that $S(\rho) \geq 0$, with equality iff $\rho = |\psi\rangle\langle\psi|$, a pure state. By analogy with classical Shannon entropy, it also possible to define joint and conditional entropies for composite quantum systems. The von Neumann entropy of a joint state $\rho_{AB} \in \cS(\cH_{A}\otimes\cH_{B})$ (where $\cH_A$ and $\cH_B$ are finite-dimensional Hilbert spaces) is thus defined as $S(\rho_{AB}) = -\tr(\rho_{AB}\log\rho_{AB})$.

In this section we will focus on some important inequalities regarding the entropies of states of composite system, and understand the structure of the states that saturate them.

\begin{defn}[\bf Subadditivity]\index{subadditivity}
For any joint state of a bipartite system $\rho_{AB} \in \cS(\cH_A \otimes \cH_B)$  with the {\it reduced states} given by the partial traces $\rho_{A} = \tr_{B}[\rho_{AB}]$ and $\rho_{B} = \tr_{A}[\rho_{AB}]$,
\begin{equation}\label{eq:subadd}
 S(\rho_{A}) + S(\rho_{B}) - S(\rho_{AB}) \geq 0.
\end{equation}
\end{defn}
The quantity on the LHS is in fact the {\bf Quantum Mutual Information}\index{quantum! mutual information} $I(A:B)$ between systems $A$ and $B$, so that the inequality can also be restated as the positivity of the mutual information.
\begin{equation}\label{eq:mutualinfo}
 I(A:B) := S(\rho_{A}) + S(\rho_{B}) - S(\rho_{AB}) \geq 0
\end{equation}
Alternately, the subadditivity inequality can also be expressed in terms of the {\bf Relative Entropy} between the joint state $\rho_{AB}$ and the product state $\rho_{A}\otimes\rho_{B}$.
\begin{defn}[{\bf Relative Entropy}]\index{relative entropy}
The quantum relative entropy between any two states $\rho$ and $\sigma$ ($\rho, \sigma \in \cS(\cH_A)$) is defined as
\begin{equation}\label{eq:relentropy}
 D(\rho\parallel\sigma) := \tr[\rho(\log\rho -\log\sigma)],
\end{equation}
when $\textrm{\bf Range}(\rho) \subset \textrm{\bf Range}(\sigma)$, and $\infty$ otherwise~\cite{Ume}.
\end{defn}
Note that this definition generalizes the Kullback-Liebler divergence~\cite{KLdistance} of two probability distributions, just as the von Neumann entropy generalizes the Shannon entropy.

\begin{exer}
Show that $I(A:B) = D(\rho_{AB}\parallel\rho_A\otimes\rho_B)$.
\end{exer}
{\bf Solution:}
\begin{eqnarray}
I(A:B) &=& S(\rho_A) + S(\rho_{B}) - S(\rho_{AB}) \nonumber \\
&=& -\tr_{A}[\rho_A\log\rho_A] - \tr_{B}[\rho_{B}\log\rho_{B}] + \tr[\rho_{AB}\log\rho_{AB}] \nonumber \\
&=& \tr[\rho_{AB}\log\rho_{AB}] - \tr[\rho_{AB}(\log\rho_{A}+\log\rho_{B})] \nonumber \\
&=& \tr[\rho_{AB}\log\rho_{AB}] - \tr[\rho_{AB}\log(\rho_{A}\otimes\rho_{B})] \nonumber \\
&=& D(\rho_{AB}\parallel\rho_A\otimes\rho_B),
\end{eqnarray}
where the last step follows from the observation that $\log(\rho_{A}\otimes\rho_{B}) = \log\rho_{A} + \log\rho_{B} $. $\blacksquare$

Thus, subadditivity is simply the statement that the relative entropy $D(\rho_{AB}\parallel\rho_A\otimes\rho_B) \geq 0$ is positive. Positivity of relative entropy is easy to prove, using the observation that $D(\rho\parallel\sigma) \geq \tr[(\rho - \sigma)^{2}]$.

\begin{exer}
 Show that $D(\rho\parallel\sigma) \geq \tr[(\rho - \sigma)^{2}] \geq 0$, with equality iff $\rho = \sigma$.
\end{exer}

Thus the subadditivity inequality becomes an inequality iff $\rho_{AB} = \rho_A \otimes \rho_{B}$, in other words, if and only if the joint state $\rho_{AB}$ of the system is in fact a product state.

The subadditivity inequality for two quantum systems can be extended to three systems. This result is often known as {\bf Strong Subadditivity}, and is one of the most important and useful results in quantum information theory.
\begin{thm}[\bf Strong subadditivity~\cite{liebruskai}]\index{Theorem! Strong Subadditivity}
 Any tripartite state $\rho_{ABC} \in \cS(\cH_A\otimes \cH_B\otimes \cH_C)$, with reduced density operators $\rho_{AB} = \tr_{C}[\rho_{ABC}]$, $\rho_{BC} = \tr_{A}[\rho_{ABC}]$,  and $\rho_{B} = \tr_{C}[\rho_{ABC}]$, satisfies,
\begin{equation}\label{eq:SSA}
 S(\rho_{AB}) + S(\rho_{BC}) - S(\rho_{B}) - S(\rho_{ABC}) \geq 0
\end{equation}
\end{thm}
\begin{proof}
To prove the strong subadditivity property, it will again prove useful to interpret the LHS as a mutual information, in particular, the {\bf Conditional Mutual Information}\index{conditional mutual information}  $I(A:C|B)$ --- the mutual information between systems $A$ and $C$ given $B$ -- which is defined as
\begin{eqnarray}\label{eq:condMI}
I(A:C|B) &=& I(A:BC) - I(A:B) \nonumber \\
&=& S(\rho_A) + S(\rho_{BC}) - S(\rho_{ABC}) - S(\rho_{A}) - S(\rho_{B}) + S(\rho_{AB}) \nonumber \\
&=& S(\rho_{AB}) + S(\rho_{BC}) - S(\rho_{B}) - S(\rho_{ABC})
\end{eqnarray}
Thus strong subadditivity is proved once we establish the positivity of the conditional mutual information. This in turn is proved by first expressing $I(A:C|B)$ in terms of relative entropies.
\begin{eqnarray}
I(A:C|B) &=& D(\rho_{ABC}\parallel \rho_A\otimes\rho_{BC}) - D(\rho_{AB}\parallel\rho_A\otimes\rho_B) \nonumber \\
&=& D(\rho_{ABC}\parallel \rho_A\otimes\rho_{BC}) \nonumber \\
&& \qquad \qquad - \; D(\tr_{C}[\rho_{ABC}]\parallel\tr_{C}[ \rho_A\otimes\rho_{BC}]) \label{eq:REtoSSA}
\end{eqnarray}
The final step is to then use Uhlmann's theorem~\cite{uhlmann} (proved earlier by Lindblad~\cite{lindblad1975} for the finite-dimensional case) on the monotonicity of the relative entropy.
\begin{thm}[{\bf Monotonicity of Relative Entropy under CPTP maps}]\index{Theorem! Uhlmann's}
 For all states $\rho$ and $\sigma$ on a space $\cH$, and quantum operations $\cT:\cL(\cH)\rightarrow\cL(\cK)$,
\begin{equation}\label{eq:uhlmann}
D(\phi\parallel\sigma) \geq D(\cT(\phi)\parallel \cT(\sigma)).
\end{equation}
\end{thm}
To obtain strong subadditivity, one simply has to make the correspondences $\phi \equiv \rho_{ABC}$, $\sigma \equiv \rho_A\otimes\rho_{BC}$, and $\cT \equiv \tr_{C}$ (recall that the partial trace is indeed a CPTP map). Then, Uhlmann's theorem implies
\begin{equation}\label{eq:monotonicity}
D(\rho_{ABC}\parallel \rho_A\otimes\rho_{BC}) - D(\tr_{C}[\rho_{ABC}]\parallel\tr_{C}[ \rho_A\otimes\rho_{BC}]) \geq 0,
\end{equation}
which in turn implies strong subadditivity through Eq.s~\eqref{eq:REtoSSA} and~\eqref{eq:condMI}. In fact, it was shown by Ruskai~\cite{ruskai2002} that the contractive property of the quantum relative entropy under CPTP maps and the positivity of conditional mutual information are equivalent statements.
\end{proof}

Finally, we note that the corresponding inequalities in the classical setting, namely, the subadditivity and the strong subadditivity of the Shannon entropy are easier to prove, since they follow almost directly from the concavity of the $\log$ function.

\subsection{Monotonicity of Relative Entropy : Petz Theorem}\index{relative entropy! monotonicity}

The above discussion on strong subadditivity implies that the question of finding the conditions under which equality holds for Eq.~\eqref{eq:SSA}, translates to that of finding the conditions under which Eq.~\eqref{eq:uhlmann} describing the monotonicity of the relative entropy (under the partial trace operation) is saturated. Note that there is a trivial case of such an equality, namely, if there exists a quantum operation $\hat{\cT}$  which maps $\cT(\phi)$ to $\phi$ and $\cT(\sigma)$ to $\sigma$. It was shown by Petz~\cite{petz1986,petz1988} that this is in fact the only case of such an equality.
\begin{thm}[{\bf Saturating Uhlmann's theorem}]\index{Theorem! Petz}
For states $\phi,\sigma\in\cH$, and the CPTP map $\cT: \cL(\cH)\rightarrow\cL(\cK)$,
$D(\phi\parallel\sigma) = D(\cT(\phi)\parallel\cT(\sigma))$ if and only if 
\begin{equation}
(\hat{\cT}\circ\cT)(\phi) = \phi, 
\end{equation}
for the map $\hat{\cT}: \cL(\cK)\rightarrow \cL(\cH)$  given by
\begin{equation}\label{eq:petzmap}
X \rightarrow \sqrt{\sigma}\cT^{\dagger}\left([\cT(\sigma)]^{-1/2}X[\cT(\sigma)]^{-1/2}\right)\sqrt{\sigma}, \quad \forall  \; X \in \cL(\cH)
\end{equation}
where $\cT^\dagger$ is the adjoint map\footnote{For any CPTP map $\cT$ in a  Kraus representation $\{T_{i}\}$, the adjoint map is another CPTP map with Kraus operators $\{T^{\dagger}_{i}\}$.} of $\cT$, and the map $\cT$ is assumed to be such that $\cT(\sigma) > 0$.
\end{thm}
Note that $(\hat{\cT}\circ\cT)(\sigma) = \sigma$. In other words, the CPTP map $\hat{\cT}$ is defined to be the partial inverse of $\cT$ on the state $\sigma$. Petz's theorem states that such a map saturates monotonicity if and only if  $\hat{\cT}$ is also a partial inverse for $\cT$ on $\phi$.

To apply the Petz theorem for strong subadditivity, the relevant map $\cT$ is the partial trace over system $C$. Thus, $\cT \equiv \tr_{C} : \cL(ABC) \rightarrow \cL(AB)$. The adjoint map is therefore $\cT^\dagger : \cL(AB)\rightarrow \cL(ABC)$, with $\cT^{\dagger}(X) = X \otimes I_{C}$.  The states $\phi$ and $\sigma$ are respectively, $\phi = \rho_{ABC}$ and $\sigma = \rho_A\otimes\rho_{BC}$. Thus the Petz map defined in Eq.~\eqref{eq:petzmap} is given by
\begin{equation}\label{eq:That}
\hat{\cT} \equiv \Id_{A}\otimes\cT''_{B\rightarrow BC}
\end{equation}
with the map $\cT'' : \cL(B)\rightarrow\cL(BC)$ given by,
\begin{equation}\label{eq:T''}
\cT'' (X) = \sqrt{\rho_{BC}}(\rho_{B}^{-1/2}X\rho_{B}^{-1/2}\otimes I_{C})\sqrt{\rho_{BC}}, \quad \forall \; X \in \cL(\cH_B).
\end{equation}
It is easy to verify that $(\hat{\cT}\circ\cT)(\sigma) = \hat{\cT}(\rho_A \otimes \rho_{B}) = \rho_{A}\otimes\rho_{BC}$.  Thus, Petz's theorem states that the inequality in Eq.~\eqref{eq:monotonicity} is saturated if and only if the joint state $\rho_{ABC}$ is such that
\begin{equation}\label{eq:statecondition}
(\hat{\cT}\circ\cT)(\rho_{ABC}) = \rho_{ABC} \quad \Rightarrow \hat{\cT}(\rho_{AB}) \equiv (\Id_{A}\otimes\cT'')(\rho_{AB}) = \rho_{ABC},
\end{equation}
for the map $\cT''$ defined in Eq.~\eqref{eq:T''}.

To summarize, Petz's theorem explicitly gives a condition on the states that saturate Uhlmann's theorem on the monotonicity of the quantum relative entropy under the action of a CPTP map $\cT$. This is done by constructing a map $\hat{\cT}$ corresponding to the map $\cT$ such that $(\hat{\cT}\circ\cT)$ leaves both the states in the argument of relative entropy function invariant. This in turn gives us a handle on states that saturate strong subadditivity. Specifically, a tripartite state $\rho_{ABC}$  saturates the strong subadditivity inequality in Eq.~\eqref{eq:SSA} if and only if it satisfies Eq.~\eqref{eq:statecondition}.

{\bf Remark:} Note that the condition $(\Id_{A}\otimes\cT'')(\rho_{AB}) = \rho_{ABC}$ can in fact be thought to characterize a {\it short Quantum Markov Chain}\index{quantum! Markov chain} (see~\cite{haydenpetz} for a more rigorous definition). Given three classical random variables $A,B,C$, they form a short Markov chain $A\rightarrow B\rightarrow C$ iff $P_{AC|B}(a,c|b) = P_{A|B}(a|b)P_{C|B}(c|b)$, that is, the random variable $C$ is independent of $A$.  Analogously, Eq.~\eqref{eq:statecondition} implies that the tripartite state $\rho_{ABC}$ is such that the state of system $C$ depends only on the state of $B$ (via the map $\cT''$) and is independent of system $A$. In fact, the classical conditional mutual information $I(A:C|B)$ vanishes iff $A\rightarrow B\rightarrow C$ form a Markov chain in the order. Petz's Theorem is simply a quantum analogue of this statement!

\subsection{Structure of States that Saturate Strong Subadditivity}

The structure of the tripartite states that satisfy Eq.~\eqref{eq:statecondition}, can in fact be characterized more explicitly as follows.

\begin{thm}\label{thm:structure}
A state $\rho_{ABC} \in \cL(\cH_A\otimes \cH_B\otimes \cH_C)$ satisfies strong subadditivity (Eq.~\eqref{eq:SSA}) with equality if and only if there exists a decomposition of subsystem $B$ as
\begin{equation}
\cH_B = \bigoplus_{j} \cH_{b_{j}^{L}}\otimes \cH_{b_{j}^{R}},
\end{equation}
into a direct sum of tensor products, such that,
\begin{equation}\label{eq:structure}
\rho_{ABC} = \bigoplus_{j} q_{j} \rho^{(j)}_{Ab^{L}_{j}} \otimes \rho^{(j)}_{b_{j}^{R}C},
\end{equation}
with states $\rho^{(j)}_{Ab^{L}_{j}} \in \cH_A \otimes \cH_{b_{j}^{L}}$ and $\rho^{(j)}_{b^{R}_{j}C} \in \cH_{b_{j}^{R}}\otimes \cH_C$, and a probability distribution $\{q_{j}\}$.
\end{thm}

This result due to Hayden et al.~\cite{haydenpetz} provides an explicit characterization of tripartite states that form a short quantum Markov chain. Eq.~\eqref{eq:structure} states that for a given $j$ in the orthogonal sum decomposition of the Hilbert space $\cH_{B}$, the state on $\cH_{A}$ is independent of the state in $\cH_{C}$.

\begin{proof}
The sufficiency of the condition is immediately obvious --- it is easy to check that states $\rho_{ABC}$ of the form given in Eq.~\eqref{eq:structure} indeed saturate the inequality in Eq.~\eqref{eq:SSA}.

To prove that such a structure is necessary, first note that Eq.~\eqref{eq:statecondition}, can be simplified to a condition on $\rho_{AB}$ alone:
\begin{eqnarray}
(\Id_{A}\otimes\cT'')(\rho_{AB}) &=& \rho_{ABC} \nonumber \\
\Rightarrow (\Id_{A}\otimes\tr_{C}\circ\cT'')(\rho_{AB}) &=& \rho_{AB}.\label{eq:fixedpt}
\end{eqnarray}
Define $\Phi \equiv \tr_{C}\circ\cT''$, a quantum operation on $\cS(\cH_{B})$. Then, states $\rho_{AB}$ satisfy a fixed point like equation: $(\Id_{A}\otimes\Phi)\rho_{AB} = \rho_{AB}$. For any operator $M$ on $\cH_{A}$ with $0\leq M \leq I$, define,
\begin{equation}
\sigma = \frac{1}{p}\tr_{A}\left(\rho_{AB}(M\otimes I)\right), \quad p = \tr\left(\rho_{AB}(M\otimes I)\right). \label{eq:sigma}
\end{equation}
Then,
\begin{equation}\label{eq:equivalence}
(\Id_{A}\otimes\Phi)\rho_{AB} = \rho_{AB}   \Longleftrightarrow  \Phi(\sigma) = \sigma.
\end{equation}
Varying the operator $M$ thus gives a family of fixed points ${\bf M} = \{\sigma\}$ for $\Phi$.


The next step is to make use of the structure theorem for the (finite-dimensional) matrix algebra\index{structure theorem} of the fixed points. More specialized to the situation on hand, is a result due to Koashi and Imoto~\ref{thm:KI} that makes use of the structure theorem to characterize the structure of a family of invariant states. See Subsection~\ref{subsec:KIthm} for a formal statement and proof of this result.

As discussed in Theorem~\ref{thm:KI} below, the family of invariant states ${\bf M}\equiv\{\sigma\}$ induces a decomposition on the Hilbert space $\cH_{B}$:
\begin{equation}\label{eq:decomposition}
 \cH_{B} = \bigoplus_{j}\cH_{b_{j}^{L}}\otimes\cH_{b_{j}^{R}},
\end{equation}
such that every state $\sigma \in \mathbf{M}$ can be written as
\begin{equation}
 \sigma = \bigoplus_{j}q_{j}(\sigma)\rho_{j}(\sigma)\otimes\omega_{j},
\end{equation}
with states $\rho_{j}(\sigma) \in \cH_{b_{j}^{L}}$ that depend on $\sigma$ and $\omega_{j}\in\cH_{b_{j}^{R}}$ that are constant for all $\sigma$. The equivalence in Eq.~\eqref{eq:equivalence} implies that $\rho_{AB} \in {\rm Span}\{\xi \otimes \sigma, \; \xi \in \cS(\cH_{A})\}$. This in turn implies the following structure for $\rho_{AB}$:
\begin{equation}\label{eq:state_structure}
 \rho_{AB} = \bigoplus_{j}q_{j}\rho_{Ab_{j}^{L}}\otimes\omega_{b_{j}^{R}}.
\end{equation}
Finally, Theorem~\ref{thm:KI} also gives the following decomposition for the map $\Phi$ (see Eq.~\eqref{eq:invariant_map2}):
\begin{equation}\label{eq:invariant_map}
 \Phi = \bigoplus_{j}\Id_{\cH_{b_{j}^{L}}}\otimes \Phi_{j}.
\end{equation}
Since $\Phi = \tr_{C}\circ\cT''$, this implies
\begin{eqnarray}
 \rho_{ABC} &=& (\Id_{A}\otimes\cT'')(\rho_{AB}) \nonumber \\
 &=& \bigoplus_{j}q_{j}\rho_{Ab_{j}^{L}}\otimes\cT''(\omega_{j}) \nonumber \\
&=& \bigoplus_{j}q_{j}\rho_{Ab_{j}^{L}}\otimes\rho_{b_{j}^{R}C} \quad (\rho_{b^{R}_{j}C} = \cT''(\omega_{j})),
\end{eqnarray}
as desired.
\end{proof}

This explicit characterization of the states saturating strong subadditivity offers an interesting insight into the physical properties of these states. In particular, it gives a neat condition for the tripartite state $\rho_{ABC}$ to be separable along the $A-C$ cut.
\begin{cor}
 For a state $\rho_{ABC}$ satisfying strong subadditivity with equality, that is,
\begin{equation}
 I(A:C|B) = S(\rho_{AB}) + S(\rho_{BC}) - S(\rho_{ABC}) - S(\rho_{B}) = 0,
\end{equation}
the marginal state $\rho_{AC}$ is separable. Conversely, for every separable state $\rho_{AC}$ there exists an extension $\rho_{ABC}$ such that $I(A:C|B)=0$.
\end{cor}

\subsection{An Operator Algebraic Proof of the Koashi-Imoto Theorem}\label{subsec:KIthm}

Recall first the structure theorem for matrix algebras which lies at the heart of the Koashi-Imoto characterization. For notational consistency, we state the theorem for the finite-dimensional Hilbert space $\cH_{B}$.

\begin{lem}\label{lem:matrix}
Let $\cA$ be a $*$-subalgebra of $\cB(\cH_{B})$, with a finite dimensional $\cH_{B}$. Then, there exists a direct sum decomposition
\begin{equation}
\cH = \bigoplus_{j}\cH_{B_{j}} = \bigoplus_{j}\cH_{b_{j}^{L}}\otimes\cH_{b_{j}^{R}},
\end{equation}
such that
\begin{equation}
\cA = \bigoplus_{j}\cB(\cH_{b_{j}^{L}})\otimes I_{b_{j}^{R}}.
\end{equation}
For $X\in\cB(\cH)$, any completely positive and unital projection $P^{*}$ of $\cB(\cH)$ into $\cA$ is of the form
\begin{equation}\label{eq:proj_decomp}
P^{*}(X) = \bigoplus_{j}\tr_{b_{j}^{R}}(\Pi_{j}X\Pi_{j}(I_{b_{j}^{L}}\otimes \omega_{j}))\otimes I_{b_{j}^{R}},
\end{equation}
with $\Pi_{j}$ being projections onto the subspaces $\cH_{b_{j}^{L}}\otimes\cH_{b_{j}^{R}}$, and $\omega_{j}$ being states on $\cH_{b_{j}^{R}}$.
\end{lem}
\begin{proof}
See~\cite{takesaki}.
\end{proof}


\begin{thm}[{\bf Koashi-Imoto~\cite{koashi_imoto}}]\label{thm:KI}\index{Theorem! Koashi-Imoto}
Given a family of states ${\bf M} = \{\sigma\}$ on a finite-dimensional Hilbert space $\cH_{B}$, there exists a decomposition of $\cH_{B}$ as
\begin{equation}
 \cH_{B} = \bigoplus_{j}\cH_{B_{j}} = \bigoplus_{j}\cH_{b^{L}_{j}}\otimes \cH_{b^{R}_{j}},
\end{equation}
into direct sum of tensor products, such that:\\
(a) The states $\{\sigma\}$ decompose as
\begin{equation}
 \sigma = \bigoplus_{j}q_{j}\rho_{j}(\sigma)\otimes \omega_{j},
\end{equation}
where $\rho_{j}$ are states on $\cH_{b_{j}^{L}}$, $\omega_{j}$ are states on $\cH_{b_{j}^{R}}$ which are constant for all $\sigma$ and $q_{j}$ is a probability distribution over $j$.

(b) Every $\Phi$ which leaves the set $\{\sigma\}$ invariant, is of the form
\begin{equation}\label{eq:invariant_map2}
\Phi|_{\cH_{B_{j}}} = \Id_{b_{j}^{L}}\otimes\Phi_{j},
\end{equation}
where the $\Phi_{j}$ are CPTP maps on $\cH_{b_{j}^{R}}$ such that $\Phi_{j}(\omega_{j}) = \omega_{j}$.
\end{thm}

Note that Statement (a) of the theorem was first proved by Lindblad~\cite{lindblad}, and the proof here closely follows his approach.

\begin{proof}
First, consider the fixed points\index{fixed points} of the dual (or adjoint) map of $\Phi$. Since $\Phi^{*}$ is a CP unital map, the fixed points form an algebra, unlike the fixed points of the CPTP map $\Phi$ which are just a set of states. Furthermore, the following Lemma holds.
\begin{lem}
If $\Phi^{*}(X) = \sum_{i}B_{i}^{*}X B_{i}$, then, $\cI_{\Phi} = \{X:\Phi^{*}(X) = X\}$ is a $*$-subalgebra and is in fact equal to the commutant\index{commutant} of the Kraus operators of $\Phi^{*}$:
\begin{equation}
 \cI_{\Phi} = \{B_{i}, B_{i}^{*}\}'  \label{eq:commutant}
\end{equation}
\end{lem}
\begin{proof}
Using the Kraus representation, it is easy to see that $\Phi^{*}(X^{*}) = (\Phi^{*}(X))^{*} = X^{*}$. For $X\in\cI_{\Phi}$, direct computation gives:
 \begin{equation}
\sum_{i}\left[X,B_{i}\right]^{*}\left[X,B_{i}\right] = \cF^{*}(X^{*}X) - X^{*}X = 0.
 \end{equation}
 The last equality follows from the fact that the family of invariant states ${\bf M} = \{\sigma\}$ contains a faithful state. Since the LHS is a sum of positive terms, all of them must be $0$, so that $\left[X,B_{i}\right] = 0$ for all $i$. Similarly, $\left[X, B_{i}^{*}\right] = 0$ for all $i$, which immediately implies $\cI_{\cF} \subset \{B_{i}, B_{i}^{*}\}'$. The converse relation that $\{B_{i},B_{i}^{*}\}' \subset \cI_{\cF}$ is trivial.
\end{proof}

Iterating the map $\Phi^{*}$ asymptotically many times gives a projection of the full algebra onto the subalgebra of the fixed points~\footnote{This is essentially a {\it mean ergodic theorem} for the dual of a quantum operation. See~\cite{haydenpetz} for a proof.}. That is,
\begin{equation}\label{eq:projection}
\cP^{*} = \lim_{N\rightarrow\infty}\frac{1}{N}\sum_{n=1}^{N}(\Phi^{*})^{n},
\end{equation}
is such that $\cP^{*}(\sigma) = \sigma$, for $\sigma$ defined in Eq.~\eqref{eq:sigma}. And the adjoint of $\cP^{*}$,
\begin{equation}
 \cP = \lim_{N\rightarrow\infty}\frac{1}{N}\sum_{n=1}^{N}\Phi^{n},
\end{equation}
is also a map that leaves the family of states ${\bf M} \equiv \{\sigma\}$ invariant.
Using the decomposition for $\cP^{*}$ from the structure theorem above (Eq.~\eqref{eq:proj_decomp}), the map $\cP$ has a similar decomposition, as follows. For all $\xi \in \cB(\cH_{B})$,
\begin{equation}
 \cP(\xi) = \bigoplus_{j}\tr_{b_{j}^{R}}\left(\Pi_{j}\xi\Pi_{j}\right)\otimes \omega_{j},
\end{equation}
where $\Pi_{j}$ are projectors onto the subspaces $\cH_{B_{j}}$ and $\omega_{j} \in \cS(\cH_{b_{j}^{R}})$. Since $\cP_{0}(\sigma) = \sigma$ for all $\sigma \in {\bf M}$, we obtain the desired form of the states:
\begin{equation}
 \sigma = \cP_{0}^{*}(\sigma) = \bigoplus_{j}q_{j}\rho_{j}(\sigma)\otimes\omega_{j}.
\end{equation}

To characterize the properties of a general $\Phi$ that leaves a set of states ${\bf M} \equiv \{\sigma\}$ invariant, we look at the set of such maps ${\bf \Phi} \equiv \{\Phi\}$ that leave ${\bf M}$ invariant. Taking a suitable convex combination of the maps $\{\Phi\}$ such that the corresponding invariant algebra $\cI$ is the smallest gives the most stringent decomposition of $\cH_{B}$. Define $\cI_{0} = \bigcap_{\cF\in\mathbf{F}}\cI_{\cF}$ which is a $*$-subalgebra itself. Because all dimensions are finite, this is actually a finite intersection: $\cI_{0} = \cI_{\cF_{1}}\cap\ldots\cap\cI_{\cF_{M}}$. Consider for example,
\begin{equation}
 \cF_{0} = \frac{1}{M}\sum_{\mu=1}^{M}\cF_{\mu}.
\end{equation}
>From Lemma~\ref{lem:matrix}, there exists the following decomposition of $\cI_{0}$:
\begin{equation}
 \cI_{0} = \cB(\cH_{b_{j}^{L}})\otimes I_{b_{j}^{R}},
\end{equation}
Since $\cI_{0}\subset\cI_{\Phi}$, $\Phi^{*}|_{\cI_{0}} = \Id_{\cI_{0}}$. Explicitly, for $\rho \in \cS(\cH_{b_j^{L}})$,
\begin{equation}
 \Phi^{*}(\rho \otimes I_{\cH_{b_{j}^{R}}}) = \rho\otimes I_{\cH_{b_{j}^{R}}}.
\end{equation}
Now, consider $\mu \in\cB(\cH_{b_j^{R}})$ such that $0\leq \mu \leq I$. Then,
\begin{equation}\label{eq:Tstructure}
 0\leq \Phi^{*}(\rho\otimes \mu) \leq \Phi^{*}(\rho \otimes I_{\cH_{b_{j}^{L}}} ) = \rho \otimes I_{\cH_{b_{j}^{R}}} = I_{\cH_{b_{j}^{L}}}\otimes I_{\cH_{b_{j}^{R}}}
\end{equation}
This implies that $\Phi^{*}$ maps $\cB(\cH_{b_j^{L}}\otimes\cH_{b_j^{R}})$ into itself for all $j$, and hence the same applies to $\Phi$. Applying Eq.~\eqref{eq:Tstructure} to a rank-one projector $|\psi\rangle\langle\psi| \in \cS(\cH_{b_{j}^{L}})$, we get,
\begin{equation}
 \Phi^{*}(|\psi\rangle\langle\psi| \otimes \mu) = |\psi\rangle\langle\psi|\otimes \mu',
\end{equation}
with $\mu'$ depending linearly on $\mu$, and independent of $|\psi\rangle\langle\psi|$. Thus, for states $\rho \in \cS(\cH_{b_{j}^{L}})$ and $\sigma \in \cS(\cH_{b_{j}^{R}})$,
\begin{equation}
 \Phi(\rho\otimes\sigma) = \rho\otimes\Phi_{j}(\sigma).
\end{equation}
Applying this to the invariant states $\sigma$ gives the invariance of $\omega_{j}$ under $\Phi_{j}$.
\end{proof}

\section{Norms on Quantum States and Channels}

We motivate this section by recalling the mathematical formalism of quantum error correction discussed in detail in Sec.3.2. Given a {\it noise channel} (a CPTP map) $\cT : \cB(\cH_{A}) \rightarrow \cB(\cH_{B})$, the objective of error correction is to (a) identify a subspace $\cH_{A_{0}}\subset \cH_{A}$, which translates into an embedding $\cI_{0}$ of matrices on $\cH_{A_{0}}$ completing them as matrices on $\cH_{A}$, and, (b) find a CPTP decoding map $\cD : \cB(\cH_{A}) \rightarrow \cB(\cH_{A_{0}})$. Perfect error correction demands that any state $\rho \in \cB(\cH_{A_{0}})$ should therefore transform as
\[ \cD(\cT(\cI_{0}(\rho)))  = \rho, \; \forall \; \rho \in \cB(\cH_{A_{0}}). \]

Mathematically, this gives rise to a closed, algebraic theory. But in practice, the condition of perfect error correction is rather hard to achieve. The best we can hope for is that the overall map on $\rho$ (including noise and decoding) is something close to an identity map. In order to describe such an almost perfect error correction, we need an appropriate choice of norms and metrics on states and channels. In particular, for some choice of embedding $\cI_{0}$ and decoding map $\cD$, we should be able to quantify how close $\cD(\cT(\cI_{0}(\rho)))$ is to the state $\rho$, as well as, the closeness of the map $\cD\circ\cT\circ\cI_{0}$ to the identity map $\Id_{A_{0}}$ on $\cH_{A_{0}}$.

We begin with a few standard definitions.
\begin{defn}[Operator Norm]\index{operator norm}
For any operator $X$ on a finite dimensional Hilbert space, the operator norm is defined as
\begin{equation}
\parallel X \parallel = {\rm largest \; singular \; value}(X) = \parallel X \parallel_{\infty}
\end{equation}
This is a limiting case ($p\rightarrow \infty$) of the non-commutative {\bf $L^{p}$ norms}:
\begin{equation}
\parallel X \parallel_{p} = (\tr[|X|^{p}])^{1/p}, \quad |X| = \sqrt{X^{\dagger}X}.
\end{equation}
\end{defn}
We also recall the following well known inequalities. For any $X \in \cB(\cH_{A})$, and $p \leq q$,
\[ \parallel X\parallel_{\infty} \leq \parallel X\parallel_{p} \leq \parallel X \parallel_{q} \leq \parallel X\parallel_{1} \leq \sqrt{{\rm dim}(\cH_{A})}\parallel X \parallel_{2} .\]

The norm of choice to quantify distance between quantum states is the $L^{1}$-norm, which is often referred to as the trace-norm.
\begin{defn}[Distance between states]\index{trace-norm}
The distance between any two states $\rho, \sigma \in \cS(\cH_{A})$ is quantified via the $L^{1}$-norm (trace-norm) as $\parallel \rho -\sigma\parallel_{1}$.
\end{defn}
The most compelling reason for this choice of definition is that it captures exactly how well the states $\rho, \sigma$ can be distinguished in an experimental setting. To see this clearly, we briefly discuss the problem of {\bf quantum state discrimination}\index{quantum! state discrimination}.

State discrimination is an information-theoretic game which can be described as follows. Assume that some experimental procedure prepares, with equal probability, one of two states $\rho$ or $\sigma$. The goal of the game is to minimize the error in guessing which of these two states was prepared, based on a suitable measurement. Since the goal is to distinguish between two states, it suffices for the measurement process to have just two outcomes. Without loss of generality, we can characterize the measurement apparatus by a pair of positive operators $\{M, I - M\}$, with $0 \leq M \leq I$. Even if a more complex measuring apparatus were to be used, the outcomes will have to grouped together so that one set corresponds to $\rho$ and the other to $\sigma$. Say, $M$ corresponds to guessing $\rho$. Then, the error probability is given by,
\begin{equation}
p_{\rm err} = \frac{1}{2}\tr[\rho (I -M)] + \frac{1}{2}\tr[\sigma M].
\end{equation}
Clearly, $p_{\rm err} \leq \frac{1}{2}$, where the upper bound corresponds to making a random guess without really performing a measurement (that is, $M = I$). Defining the bias $\beta$ to be $\beta = 1 - 2p_{\rm err}$, we have,
\begin{equation}
\beta = \tr[(\rho - \sigma)M].
\end{equation}
The optimal measurement is the one that maximizes $\eta$. It is easy to see that the optimal $M$ is simply the projection onto the subspace where $(\rho - \sigma) \geq 0$. Therefore,
\begin{equation}
\beta_{\rm max} = \frac{1}{2}\parallel \rho - \sigma \parallel_{1}.
\end{equation}
This proves that the minimum error probability is given by
\begin{equation}
\min [p_{\rm err}] = \frac{1}{2}\left( 1 - 2\parallel \rho - \sigma \parallel_{1}\right), \label{eq:tracenorm_p}
\end{equation}
thus providing an operational meaning to the trace-norm~\cite{HH}. In mathematical terms, we have simply realized the fact that the trace-norm is the dual of the operator norm:
\begin{equation}
\parallel \xi\parallel_{1} = \max_{\parallel X \parallel \leq 1} \vert \tr[\xi X] \vert.
\end{equation}
This is easily seen once we can rewrite the bias as follows:
\begin{equation}
\beta = \frac{1}{2}\tr[(\rho - \sigma)(2M - I)].
\end{equation}
Since $(\rho - \sigma)$ is Hermitian, to maximize $\beta$, it suffices to maximize over all Hermitian operators of norm less than unity.

Note that $\parallel \rho -\sigma\parallel_{1}$ is an upper bound on the Kolmogrov distance (the $L^{1}$ distance) between the measurement statistics corresponding to a given measurement on $\rho$ and $\sigma$. The optimizing measurement that achieves equality is in fact the one that minimizes $p_{\rm err}$ in the state discrimination problem.

We move on to distances between channels, or more generally, linear maps $\cL : \cB(\cH_{A})\rightarrow \cB(\cH_{B})$. The trace-norm on states naturally induces the following norm on maps.
\begin{equation}
\parallel \cL \parallel_{1\rightarrow 1} \; = \; \max_{\parallel \xi \parallel_{1}} \parallel \cL(\xi)\parallel_{1}.
\end{equation}
The so-called {\bf diamond norm}\index{diamond norm} for quantum channels which is used often in the quantum information literature is defined as the completion of this induced norm.
 \begin{equation}
\sup_{C} \parallel \cL \otimes \cI_{C}\parallel_{1\rightarrow 1} \; := \;  \parallel\cL\parallel_{\diamond}.
\end{equation}
It follows from the duality between the trace-norm and the operator norm, the $\diamond$-norm is the dual of the CB-norm\index{completely! bounded (CB) norm} (Defn.1.1.2). 

The $\diamond$-norm also has an operational significance similar to that of the trace-norm. Consider $\cL = \cT_{1} -\cT_{2}$, where $\cT_{1}: \cB(\cH_{A}) \rightarrow \cB(\cH_{B})$ and $\cT_{2}: \cB(\cH_{A}) \rightarrow \cB(\cH_{B})$ are both CPTP maps. We seek to evaluate $\parallel \cL\parallel_{1\rightarrow 1} \equiv \parallel \cT_{1} - \cT_{2}\parallel_{1\rightarrow 1}$. Noting that it suffices to maximize over operators of unit-trace rather than operators of unit-norm, we have,
\begin{equation}
\parallel \cT_{1} - \cT_{2}\parallel_{1\rightarrow 1} \; =  \max_{\rho \in \cS(\cH_{A})} \parallel \cT_{1}(\rho) - \cT_{2}(\rho)\parallel_{1}.
\end{equation}
This quantity is simply twice the maximum bias of distinguishing between the channels\index{quantum channel! distinguishing} $\cT_{1}$ and $\cT_{2}$, with respect to a restricted set of strategies, namely, the ``prepare and measure the output'' strategies.

Similarly, when we consider the $\diamond$-norm $\parallel\cT_{1} - \cT_{2}\parallel_{\diamond}$, it again suffices to minimize over all states in the extended Hilbert space $\cH_{A} \otimes \cH_{C}$, including of course the entangled states. Therefore,
\begin{equation}
\parallel \cT_{1} - \cT_{2}\parallel_{\diamond} \; = \max_{\rho \in \cS(\cH_{A} \otimes \cH_{A'})}\parallel (\cT_{1}\otimes \cI_{A'})(\rho) - (\cT_{2}\otimes \cI_{A'})(\rho) \parallel_{1}.
\end{equation}
This is also twice the maximum bias of distinguishing $\cT_{1}$ and $\cT_{2}$, but allowing for all possible strategies.

Thus, both the simple induced norm\index{induced $L_{1}$ norm} ($\parallel.\parallel_{1}$) and the $\diamond$-norm have a nice operational interpretation in the quantum information theoretic setting. The mathematical properties of these norms have interesting physical consequences in this operational setting, and vice-versa. For example, if we take Eq.~\eqref{eq:tracenorm_p} as the definition of the trace-norm, it immediately follows that the norm must be contractive under CPTP maps. If not, we could always perform the measurement after the action of a CPTP map and distinguish the states with lesser probability of error. Similarly, it follows from the operational interpretation that the $1\rightarrow 1$-norm or the $\diamond$-norm cannot increase when the channel is preceded or succeeded by the action of another CPTP map.

We now present an example where the simple induced norm ($\parallel.\parallel_{1\rightarrow 1}$) and its norm completion (the $\diamond$-norm) behave rather differently. Let $\cT_{1} : M_{d}\rightarrow M_{d}$ and $\cT_{2} : M_{d}\rightarrow M_{d}$ be two CPTP maps acting on the space of $d\times d$ matrices. We define these maps by their action on the following state $\rho \in M_{d}\otimes M_{d}$:
\[ \rho = |\phi\rangle\langle \phi| = \frac{1}{d}\sum_{i,j=1}^{d}|ii\rangle\langle jj|.\]
Note that $|\phi\rangle = \frac{1}{\sqrt{d}}\sum_{i}|ii\rangle$ is the maximally entangled state in $M_{d}\otimes M_{d}$. Under the action of $\cT_{1}$ and $\cT_{2}$, $\rho$ transforms as:
\begin{eqnarray}
 (\cT_{1}\otimes \Id)\rho &:=& \alpha = \frac{1 - F}{d(d-1)}, \nonumber \\
(\cT_{2}\otimes \Id)\rho &:=& \sigma = \frac{1 + F}{d(d+1)}, \label{eq:example}
\end{eqnarray}
where, $F = \sum_{i,j}|ij\rangle\langle ji|$ is the so-called SWAP operator -- the unitary operator that interchanges the two Hilbert spaces. $F$ has eigenvalues $\pm 1$, and so, $1 + F$ and $1-F$ are projections onto two mutually orthogonal subspaces. Thus, upto suitable normalization, $\alpha$ is the projector onto the antisymmetric subspace and $\sigma$ is the projector onto the symmetric subspace of $M_{d}\otimes M_{d}$.

It suffices to define the maps $\cT_{1}, \cT_{2}$ via Eq.~\eqref{eq:example}, since $\alpha$ and $\sigma$ are simply the Choi-Jamiolkowski matrices corresponding to $\cT_{1}$ and $\cT_{2}$ respectively. Since we have already identified a state $\rho$ which is mapped onto orthogonal subspaces by $\cT_{1}\otimes \cI$ and $\cT_{2}\otimes \cI$, the $\diamond$-norm is easily computed.
\begin{equation}
 \parallel \cT_{2} - \cT_{1}\parallel_{\diamond}  = \parallel \sigma - \alpha\parallel_{1} = 2. \label{eq:diamond_max}
\end{equation}
This follows from the fact that trace-norm between any two density matrices cannot be larger than $2$, and the maximum value of $2$ is attained iff the states are orthogonal.

On the other hand, to compute $\parallel\cT_{1} - \cT_{2}\parallel_{1\rightarrow 1}$ we still need to perform a maximization over states:
\[ \parallel\cT_{1} - \cT_{2}\parallel_{1\rightarrow 1} = \max_{\rho}\parallel (\cT_{1} - \cT_{2})\rho \parallel_{1}.\]
Since the trace-norm is convex, the maximum will be attained for an extreme point on the set of states, namely a pure state. In fact, it suffices to evaluate $\parallel (\cT_{1} - \cT_{2})|0\rangle\langle 0|\parallel_{1}$ for some arbitrary fixed state $|0\rangle\langle 0|$, because of the following observation. The states $\alpha, \sigma$ have the property that they are left invariant under conjugation by any unitary $U\otimes U \in \cU(M_{d}\otimes M_{d})$:
\[(U\otimes U)\alpha(U\otimes U)^{\dagger} = \alpha,\] and similarly for $\sigma$. The action of the map $\cT_{1}$ on any state $X \in M_{d}$ is obtained from the corresponding Choi-Jamiolkowski matrix $\alpha$ as follows:
\[ \cT_{1}(X) = d\tr_{2}[(I\otimes X^{T})\alpha],\]
where the partial trace is over the second system. Thus, the invariance of $\alpha$ under conjugation by unitaries translates into a covaraince property for the channels, and so the norm $\parallel (\cT_{1} - \cT_{2})\rho\parallel_{1}$ is the same for $\rho$ and $U\rho U^{\dagger}$. Therefore, we have,
\begin{equation}
\parallel \cT_{1} - \cT_{2}\parallel_{1\rightarrow 1} = \parallel (\cT_{1} - \cT_{2}) |0\rangle\langle 0|\parallel_{1} = \frac{4}{d+1}.
\end{equation}
The value of the RHS follows from the fact that the states $\cT_{i}(|0\rangle\langle 0|), i=1,2$ for some arbitrary rank-$1$ projector $|0\rangle\langle 0|$, are close to the maximally mixed state $\frac{\Id}{d}$.

This example clearly highlights that the naive $\parallel .\parallel_{1\rightarrow 1}$-norm can in fact be smaller than the $\diamond$-norm by a factor that scales as the dimension of the system. Operationally, this difference between the two norms has an important consequence. Eq.~\eqref{eq:diamond_max} implies that the two channels in this example can in fact be distinguished perfectly, provided the experimenter has access to the right equipment. In particular, the experimenter needs to be able to create the maximally entangled state $|\phi\rangle = \frac{1}{\sqrt{d}}\sum_{i}|ii\rangle$, store it and in the end be able to perform measurements on the composite system. In short, the experimenter needs access to a quantum computer! On the other hand, if the experimenter is restricted to performing local measurements, the distinguishability, which is now characterized by $\parallel \cT_{1} - \cT_{2} \parallel_{1\rightarrow 1}$, is rather small.

This observation also has significance in a quantum cryptographic setting. Suppose we consider the problem of distinguishing the states $\alpha, \sigma \in M_{d}\otimes M_{d}$, when the experimenter is restricted to performing only local operations and using classical communication. This restriction defines a set of operations called {\bf LOCC}, which is in fact a convex subset of $\cB(\cH)$ whose elements lie between $[-I, I]$. Just as the trace-norm quantifies the distinguishability allowing for arbitrary possible measurements, the distinguishability of the states subject to this restricted set\footnote{Any convex, centrally symmetric subset $\mathbb{M} \subset [-\cI, \cI] \subset \cB(\cH)$ with the property $-\mathbb{M} = \mathbb{M}$, gives rise to a norm, the dual of which is a measure of the distinguishability.} can be characterized by the so-called LOCC-norm. Denoting this norm as $\parallel . \parallel_{{\rm LOCC}}$, it was shown that~\cite{DLT} for the states $\alpha, \sigma$ defined in the above
example,
\[ \parallel \alpha - \sigma\parallel_{{\rm LOCC}} \; = \; \frac{4}{d+1}.\]
It is thus possible to encode information in two perfectly distinguishable states $\alpha$ and $\sigma$, which are however close to indistinguishable under a restricted set of operations (like LOCC). This observation leads to a quantum cryptographic scheme called {\bf quantum data-hiding}~\cite{DLT}.

\section{Matrix-valued Random Variables}\index{matrix-valued! random variables}

We begin with a brief introduction to the standard theory of large deviations in classical probability theory. Recall that for some set $\{X_{1}, X_{2}, \ldots, X_{n}\}$ of independent, identically distributed (i.i.d.) random variables taking values in $\mathbb{R}$, the {\bf law of large numbers} states that
\[ \lim_{n\rightarrow \infty}\frac{1}{n}\sum_{i=1}^{n}X_{i} = \mathbb{E}[X] \equiv \mu, \]
with probability $1$. In statistics and other applications, we are often interested in how significantly the empirical mean on the LHS deviates from the expectation value $\mu$. Here, we focus on large deviations. In particular, for finite $n$, we are interested in bounding the probability that the empirical mean is larger than some $\alpha > \mu$. All known bounds of this form require some additional knowledge about the distribution of the $X$.

The simplest and most commonly encountered large deviation setting is one where the $X_{i}$'s are bounded, so that all higher moments exist. Specifically, assuming that $X_{i} \in [0,1]$, the following bound\index{large deviation bound} is known to be asymptotically optimal.
\begin{equation}
 {\rm Pr}\left\{\frac{1}{n}\sum_{i=1}^{n}X_{i} > \alpha\right\} \leq e^{-nD(\alpha\parallel \mu)} , \; (\alpha > \mu) \label{eq:large_dev}
 \end{equation}
where,
\[ D(\alpha\parallel\mu) = \alpha \ln\frac{\alpha}{\mu}+ (1-\alpha)\ln \frac{1-\alpha}{1-\mu} \]
is the relative entropy between two binary variables. Note that $D(\alpha\parallel\mu) \geq 0$ for $\alpha > \mu$ with equality holding iff $\alpha = \mu$. The upper bound in Eq.~\eqref{eq:large_dev} can further be simplified to
\[ {\rm Pr}\left\{\frac{1}{n}\sum_{i=1}^{n}X_{i} > \alpha\right\} \leq e^{-2n(\alpha - \mu)^{2}} .\]
We briefly sketch the proof of the large deviation bound in Eq.~\eqref{eq:large_dev}.

\begin{proof}
Introducing a real parameter $t>0$, it is easy to see that,
\[ {\rm Pr}\left\{\frac{1}{n}\sum_{i=1}^{n} X_{i} > \alpha \right\} = {\rm Pr}\left\{ e^{t\sum_{i}X_{i}} > e^{nt\alpha}\right\}\]
Note that we could have chosen any monotonic function of the two quantities $\frac{1}{n}\sum_{i=1}^{n}X_{i}$ and $\alpha$ on the RHS. It turns out that the choice of the exponential function is indeed asymptotically optimal.

The next step is to use Markov's inequality which states that the probability that a positive-valued random variable $Y$ is greater than some positive constant $a$ is upper bound by $\mathbb{E}[Y]/a$. This implies,
\[
{\rm Pr}\left\{\frac{1}{n}\sum_{i=1}^{n}X_{i} > \alpha\right\} \leq   \frac{\mathbb{E}[e^{t\sum_{i}X_{i}}]}{e^{nt\alpha}} = \left(\mathbb{E}[e^{tX - t\alpha}]\right)^{n},
\] where the second equality follows from the fact that the $X_{i}$ are independent and identically distributed.

>From the convexity of the exponential function, it follows that for fixed $\mu$ and $X \in [0,1]$, the expectation on the RHS is maximized when $X$ is a Bernoulli variable with ${\rm Pr}(X=0) = 1- \mu$, ${\rm Pr}(X=1) = \mu$. Therefore,
\[ {\rm Pr}\left\{\frac{1}{n}\sum_{i=1}^{n}X_{i} > \alpha\right\} \leq  \left( (1-\mu)e^{-t\alpha} + \mu e^{t(1-\alpha)}\right)^{n}.\]
Optimizing the RHS over the parameter $t$ yields the desired bound in Eq.~\eqref{eq:large_dev}.
\end{proof}

The same large deviation bound holds for real-valued vector variables, upto a dimensional constant that comes from the union bound, when we seek to quantify the deviation of each co-ordinate from its mean value. We would like to obtain similar tail bounds for matrices.

\subsection{Matrix Tail Bounds}\index{matrix-valued! tail bounds}

Consider a set of matrices $X_{1}, X_{2}, \ldots, X_{n}, \ldots \in \cB(\cH)$ that are i.i.d., where each $X_{i} \in [0, I]$ and $\cH$ is a finite dimensional Hilbert space. Note that this setting is rather different from that of random matrix theory, where the entries of the matrix are chosen in an i.i.d fashion. Here, we do not care about the distribution over the entries of the matrices, rather each $X_{i}$ is chosen independently from the same distribution. Here again, in the large-$n$ limit, the empirical mean converges to the expectation value $\mathbb{E}[X] = M \in [0, I]$ with probability $1$.
\[ \frac{1}{n}\sum_{i=1}^{n}X_{i} \rightarrow \mathbb{E}[X] \equiv M .\]

The corresponding large deviation problem seeks to bound the following probability:
\[ {\rm Pr}\left\{ \frac{1}{n}\sum_{i=1}^{n}X_{i} \nleq A \right\},\]
given a positive matrix $A > M$. In order to obtain such a bound, we first need a matrix version of Markov's inequality. The following lemma is left as an exercise; it simply follows from the proof of the standard (classical) Markov's inequality.
\begin{lem}[Matrix Markov Inequality]\index{Markov inequality}
Let $X\geq 0$ be a positive semi-definite matrix random variable with $\mathbb{E}(X) = M$. Then, for some positive definite matrix $A > 0$,
\begin{equation}
{\rm Pr}\left\{ X \nleq A \right\} \leq \tr[MA^{-1}] . \label{eq:matrix_markov}
 \end{equation}
\end{lem}

We will henceforth assume that $\mathbb{E}[X] = M \geq \mu I$, for $\mu > 0$. This is a natural assumption to make in sampling problems. If the probability of a certain event is too small, then we will have to sample a very large number of times to get an estimate of this probability. Assuming that the expectation value is larger than a certain minimum value excludes such rare events. In our setting, this assumption implies that the mean does not have very small eigenvalues. Then, defining the operator $Y = \mu M^{-1/2}XM^{-1/2}$, we see that $Y \in [0, I]$, and, $\mathbb{E}[Y] = \mu I$. Further,
\[ {\rm Pr}\left\{ X \nleq A \right\} = {\rm Pr}\left\{ Y \nleq A' \right\}, \quad A' = \mu M^{-1/2}AM^{-1/2}.\]
Thus, without loss of generality, we have an operator $A'$ with the property that it is strictly larger than the expectation value $\mathbb{E}[X]$. All of the eignevalues of $A'$ are strictly larger than $\mu$, so that, $A' \geq \alpha I > \mu I$.

In this setting, assuming $\mathbb{E}[X] = M = \mu I$, the following large deviation bound was proved in~\cite{AW}.
\begin{eqnarray}
{\rm Pr}\left\{ \frac{1}{n}\sum_{i=1}^{n}X_{i} \nleq \alpha I \right\} &\leq& de^{-nD(\alpha\parallel\mu)}, \quad {\rm for} \; \alpha > \mu, \label{eq:matrix_ldev}\\
{\rm Pr}\left\{ \frac{1}{n}\sum_{i=1}^{n}X_{i} \ngeq \alpha I \right\} &\leq& de^{-nD(\alpha\parallel\mu)}, \quad {\rm for} \; \alpha < \mu . \nonumber
\end{eqnarray}
Before proceeding to the proof, we note some useful facts about the relative entropy function.
\begin{rem}
The relative entropy $D(\alpha\parallel\mu)$ satisfies:
\begin{itemize}
\item When $(\alpha - \mu)$ is fixed,
\[ D(\alpha\parallel\mu) \geq 2(\alpha-\mu)^{2} .\]
\item When $\mu$ is small, for $\alpha = (1+ \epsilon)\mu$, a stronger bound holds:
\[ D(\alpha\parallel\mu) \geq c\mu\epsilon^{2},\]
for some constant $c$.
\end{itemize}
\end{rem}
\begin{proof}
We closely follow the arguments in~\cite{AW} where the matrix tail bounds were originally proved. First, note that,
\begin{equation}
 {\rm Pr}\left\{ \frac{1}{n}\sum_{i=1}^{n}X_{i} \nleq \alpha I \right\} \leq {\rm Pr}\left\{ e^{t\sum_{i=1}^{n}X_{i}} \nleq e^{n t \alpha} I \right\}. \label{eq:ldev1}
 \end{equation}
In the simplified setting we consider, these two probabilities are in fact equal. However, equality does not hold in general, since the function $x \rightarrow e^{x}$ is not a matrix monotone. The inverse function $x \rightarrow \ln x$ is indeed a matrix monotone. This follows from two simple facts: first, for matrices $X, Y$, $0\leq X \leq Y$ implies $Y^{-1} \leq X^{-1}$. Using this, along with the well known integral representation of the log function:
\[ \ln x  = \int_{t=1}^{x}\frac{dt}{t} = \int_{t=0}^{\infty} dt \left(\frac{1}{t+1} - \frac{1}{x+t}\right),\]
it is easy to prove that $\ln x$ is a matrix monotone.

Invoking the Markov bound (Eq.~\eqref{eq:matrix_markov}) in Eq.~\eqref{eq:ldev1}, we have,
\[ {\rm Pr}\left\{ \frac{1}{n}\sum_{i=1}^{n}X_{i} \nleq \alpha \Id \right\} \leq \tr[\mathbb{E}[e^{t\sum_{i=1}^{n}X_{i}}]]e^{-nt\alpha}.\]
Note that we cannot make use of the independence argument at this stage because the sum on the exponent does not factor into a product of exponents in this case. The operators $X_{i}$ could be non-commuting in general. Instead, we use the Golden-Thompson inequality.
\begin{lem}[Golden-Thompson Inequality]\index{Golden-Thompson inequality}
\begin{equation}
\tr[e^{A + B}] \leq \tr[e^{A}e^{B}] .
\end{equation}
\end{lem}
Using this, we have,
\begin{eqnarray}
{\rm Pr}\left\{ \frac{1}{n}\sum_{i=1}^{n}X_{i} \nleq \alpha I \right\} &\leq& \tr[\mathbb{E}[e^{tX_{n}}]\mathbb{E}[e^{t\sum_{i=1}^{n-1}X_{i}}]] e^{-nt\alpha} \nonumber \\
&\leq & \parallel \mathbb{E}[e^{tX_{n}}] \parallel\tr[\mathbb{E}[e^{t\sum_{i=1}^{n-1}X_{i}}]] e^{-nt\alpha}, \label{eq:ldev2}
\end{eqnarray}
where we have upper bounded the operator $e^{tX_{n}}$ by its norm in the final step. Repeating these steps iteratively, we have,
\begin{equation}
{\rm Pr}\left\{ \frac{1}{n}\sum_{i=1}^{n}X_{i} \nleq \alpha I \right\} \leq d \left(\parallel \mathbb{E}[e^{tX}]\parallel e^{-t\alpha}\right)^{n}. \nonumber
\end{equation}
Finally, noting that $\mathbb{E}[e^{tX}]$ is maximized when $X$ is Bernoulli distributed over $\{0,\Id\}$, and then optimizing over $t$, we get,
\begin{equation}
{\rm Pr}\left\{ \frac{1}{n}\sum_{i=1}^{n}X_{i} \nleq \alpha I \right\} \leq d \left( [1 - \mu + \mu e^{t}]e^{-t\alpha}\right)^{n} \leq de^{-nD(\alpha\parallel \mu)}. \label{eq:ldev3}
\end{equation}
\end{proof}

A stronger tail bound is obtained by avoiding the Golden-Thompson inequality and using Lieb concavity instead~\cite{tropp}. In particular, in place of Eq.~\eqref{eq:ldev2}, we have,
\[ {\rm Pr}\left\{ \frac{1}{n}\sum_{i=1}^{n}X_{i} \nleq \alpha I \right\} \leq  e^{-nt\alpha} \tr[(\mathbb{E}[e^{tX}])^{n}] .\]
This bound is clearly better than the bound in Eq.~\eqref{eq:ldev3} when the operator $e^{tX}$ has one dominating eigenvalue and the other eigenvalues are quite small. We refer to~\cite{tropp} for further details of this approach.

Such matrix tail bounds were originally obtained by Ahlswede and Winter in the context of a specific quantum information theoretic problem~\cite{AW}. Here we will focus on two other applications, namely in destroying correlations~\cite{GPW} and in quantum state merging.

\subsection{Destroying Correlations}

An important question in quantum information theory as well as statistical physics is to quantify the amount of correlations in a bipartite quantum system. One approach to quantifying the correlations in a bipartite state $\rho_{AB} \in \cS(\cH_{A}\otimes \cH_{B})$, proposed in~\cite{GPW}, is to understand the process by which $\rho_{AB}$ can be transformed to an uncorrelated (product) state of the form $\rho_{A}\otimes \rho_{B}$. The fundamental idea is to shift from characterizing the states to characterizing processes $ \cT : \rho_{AB} \rightarrow \rho_{A}\otimes\rho_{B}$, in particular, to quantify the amount of randomness that has to be introduced to effect such a transformation.

For sufficiently large systems $\cH_{A}, \cH_{B}$, there always exists a {\it global} unitary operation, that is, a unitary operator on $\cH_{A}\otimes \cH_{B}$, that maps $\rho_{AB} \rightarrow \rho_{A}\otimes \rho_{B}$ deterministically. On the other hand, if our physical model allows only for {\it local} unitary conjugations of the type $U_{A} \otimes I_{B}$ and mixtures thereof, a natural question to ask is, how much randomness is required for such a transformation\footnote{We could equally well have considered unitaries of the type $I_{A} \otimes U_{B}$, and the same results would hold.}. Note that a single local unitary conjugation cannot change the amount of correlations in a bipartite state. However, if we insert a finite amount of randomness by constructing maps involving probabilistic mixtures of local unitaries, a bipartite correlated state can indeed get mapped to a product state.

We thus consider CP maps of the following type:
\[\cT: X \rightarrow \sum_{i}p_{i}(U_{i}\otimes I)X(U_{i}\otimes I)^{\dagger}, \]
where $\{U_{i}\}$ are drawn from the set of unitary matrices on $\cH_{A}$. To see a concrete example of such a map, consider the maximally two-qubit state:
\[ \rho_{\rm ent} = \frac{1}{2}(|00\rangle + |11\rangle)(\langle 00 | +\langle 11|).\]
Then, choosing the unitaries $U_{0} = I$, $U_{1} = X$, $U_{2} = Y$ and $U_{3} = Z$, where $X, Y, Z$  are the Pauli matrices defined in Sec 3.2.2., and choosing the probabilities $p_{i} = \frac{1}{4}$, we get,
\[ \cT(\rho_{\rm ent}) = \sum_{i}p_{i}U_{i}\rho U_{i}^{\dagger} = \frac{I}{2}\otimes\frac{I}{2}. \]
We thus have a simple example of transforming a maximally entangled state into a product state via a probabilistic mixture of unitaries.

In order to quantify the amount of noise or randomness involved in the process, one approach is to simply use the difference in the (quantum) entropies between the final and initial states. However, it might be more meaningful to consider the classical entropy of the probability distribution $\{p_{i}\}$, which really captures the thermodynamic cost associated with the process. Physically, the erasing of correlations is in fact a consequence of erasing the knowledge of the probability distribution $\{p_{i}\}$ associated with the unitaries $\{U_{i}\}$. {\it Landauer's principle} states that there is an energy cost associated with erasing information, and in the asymptotic limit, this cost is proportional to the entropy of the distribution $\{p_{i}\}$.

\begin{defn}\index{randomizing map}
Given $n$ copies of a bipartite state $\rho_{AB}$, the family $(p_{i}, U_{i})_{i=1}^{N}$ of probabilities $p_{i}$ and unitaries $U_{i}\in \cU(\cH_{A})$ is said to be $\epsilon$-randomizing for $\rho_{AB}^{\otimes n}$ if the associated map $\cT$ is such that
\begin{equation}
\parallel \sum_{i=1}^{N}p_{i}(U_{i}\otimes I)\rho_{AB}^{\otimes n}(U_{i}^{\dagger}\otimes I) - \tilde{\rho}_{A} \otimes \rho_{B}^{\otimes n} \parallel \leq \epsilon, \label{eq:ep_random}
\end{equation}
where $\tilde{\rho}_{A} = \sum_{i}p_{i}U_{i}\rho_{A}^{\otimes n}U_{i}^{\dagger} \in \cH_{A}^{\otimes n}$.
\end{defn}
We are interested in quantifying the size of the smallest such $\epsilon$-randomizing family $(p_{i}, U_{i})$ for the state $\rho_{AB}^{\otimes n}$, which we denote as $N(n, \epsilon)$.
\begin{defn}
$N(n, \epsilon)$ is defined to be the smallest $N$ such that $\exists$ an $\epsilon$-randomizing family $(p_{i}, U_{i})_{i=1}^{N}$ for $\rho_{AB}^{\otimes n}$.
\end{defn}

Specifically, we are interested in the asymptotic behavior of $\frac{\log N(n,\epsilon)}{n}$ in the limit $\epsilon\rightarrow 0$ and $n\rightarrow \infty$, as way to quantify the correlation of the bipartite state $\rho_{AB}$. Following the work of Groisman {\it et al}~\cite{GPW}, we first prove the following lower bound.
\begin{thm}
Given $n$ copies of a bipartite state $\rho_{AB} \in \cS(\cH_{A} \otimes \cH_{B})$, any $\epsilon$-randomizing set $\{U_{i}\}_{i=1}^{N(n,\epsilon)}$ must have at least $N(n,\epsilon)$ unitaries, where $N(n,\epsilon)$ satisfies,
\begin{equation}
\liminf_{\substack{n\rightarrow\infty \\ \epsilon \rightarrow 0}} \; \frac{1}{n}\log N(n,\epsilon) \geq I(A:B)_{\rho},
\end{equation}
where, $I(A:B)_{\rho} = S(\rho_{A}) + S(\rho_{B}) - S(\rho_{AB})$ is the mutual information of the state $\rho_{AB}$ (see Eq.~\eqref{eq:mutualinfo}).
\end{thm}
\begin{proof}
Let $(p_{i}, U_{i})_{i=1}^{N(n,\epsilon)}$ be an $\epsilon$-randomizing family for $\rho_{AB}^{\otimes n}$. The corresponding CP map $\cT$ is given by
\[ \cT(\rho_{AB}^{\otimes n}) = \sum_{i=1}^{N}p_{i}(U_{i}\otimes I)\rho_{AB}^{\otimes n}(U_{i}^{\dagger}\otimes I).\]
Then, it can be shown that the von Neumann entropy of the transformed state satisfies,
\begin{equation}
S(\cT(\rho_{AB}^{\otimes n})) \leq S(\rho_{AB}^{\otimes n}) + H(p_{1}, \ldots, p_{N}),  \label{eq:upbound}
\end{equation}
where $H(p_{1},\ldots,p_{N})$ is the Shannon entropy of the distribution $\{p_{i}\}$. The inequality above follows from the concavity of the von Neumann entropy, in particular,
\[ S(\sum_{i}p_{i}\rho_{i}) \leq \sum_{i}p_{i}S(\rho_{i}) + H(p_{1}, p_{2}, \ldots, p_{N}) .\]

>From the definition in Eq.~\eqref{eq:ep_random}, the final state $\cT(\rho_{AB}^{\otimes n})$ is close to the product state $\tilde{\rho}_{A}\otimes \rho_{B}^{\otimes n}$. This implies, via the Fannes inequality,
\[ S(\tilde{\rho}_{A}\otimes\rho_{B}^{\otimes n}) \leq  S(\cT(\rho_{AB}^{\otimes n})) + \epsilon\log[(d_{A})^{n}(d_{B})^{n}], \]
where $d_{A} = {\rm dim}(\cH_{A})$ and $d_{B} = {\rm dim}(\cH_{B})$. Further, from the definition of $\tilde{\rho}_{A}$, we have,
\[ S(\tilde{\rho}_{A}\otimes\rho_{B}^{\otimes n})  = nS(\rho_{A}) + n S(\rho_{B}).\]

These observations along with the inequality in Eq.~\eqref{eq:upbound} imply,
\[ nS(\rho_{A}) + n S(\rho_{B}) - n\epsilon\log(d_{a}d_{B}) \leq nS(\rho_{AB}) + \log N(n,\epsilon) .\]
Thus we have the following lower bound:
\begin{equation}
\frac{\log N(n,\epsilon)}{n} \geq  I(A:B)_{\rho} - O(\epsilon), \end{equation}
where $I(A:B)_{\rho} = S(\rho_{A}) + S(\rho_{B}) - S(\rho_{AB})$ is the mutual information of the state $\rho_{AB}$. This in turn implies the asymptotic lower bound,
\begin{equation}
\liminf_{\substack{n\rightarrow\infty \\ \epsilon \rightarrow 0}} \frac{1}{n}\log N(n,\epsilon) \geq I(A:B) .
\end{equation}
\end{proof}

To show that the mutual information is also an upper bound, we make use of the matrix sampling bound proved in the previous section. We will also need to invoke the typicality principle, which we recall here.
\begin{defn}[Typicality Principle]\index{typicality principle}
Given $n$ copies of a state $\rho_{AB}$, for all $\epsilon$ and large enough $n$, there exists a truncated state $\rho_{\hat{A}\hat{B}}^{(n)} \in \cS(\hat{\cH}_{A}^{\otimes n}\otimes\hat{\cH}_{B}^{\otimes n})$ with the following properties:
\begin{itemize}
 \item[(i)] $\parallel \rho^{(n)}_{\hat{A}\hat{B}} - \rho_{AB}^{\otimes n}\parallel_{1} \leq \epsilon$,
\item[(ii)] ${\rm Range} (\rho_{\hat{A}\hat{B}}^{(n)}) \subset \hat{\cH}_{A}\otimes\hat{\cH}_{B}$, where, $\hat{\cH}_{A} \subset \cH_{A}^{\otimes n}$ and $\hat{\cH}_{B} \subset \cH_{B}^{\otimes n}$ with the dimensions of the truncated spaces satisfying,
\[ d_{\hat{A}} \equiv {\rm dim}(\hat{\cH_{A}}) \leq 2^{nS(\rho_{A}) \pm \epsilon}, \qquad d_{\hat{B}} \equiv {\rm dim}(\hat{\cH_{B}}) \leq 2^{nS(\rho_{B}) \pm \epsilon} .\]
 \item[(iii)] Asymptotic Equipartition Property: $\rho_{AB}^{(n)} \approx 2^{-nS(\rho_{AB}) \pm \epsilon}$ on its range. Also,
     \[ \rho_{A}^{(n)} \approx 2^{-nS(\rho_{A})\pm \epsilon}\; I_{\hat{A}},  \quad \rho_{B}^{(n)} \approx 2^{-nS(\rho_{B})\pm \epsilon}\; I_{\hat{B}} .\]
In other words, in the truncated spaces, the reduced states have an almost flat spectrum. The global state also has a near flat spectrum on its range space.
\end{itemize}
\end{defn}

It now suffices to find a map that can destroy the correlations in the truncated state $\rho_{\hat{A}\hat{B}}^{(n)}$. We now prove that there always exists such a map which is a probabilistic mixture of $N(n,\epsilon)$ unitaries, where $N(n,\epsilon)$ satisfies the following upper bound:
\begin{equation}
 \limsup_{\substack{n\rightarrow\infty \\ \epsilon \rightarrow 0}} \frac{1}{n}\log N(n,\epsilon) \leq I(A:B).
\end{equation}
\begin{proof}
The proof is to essentially construct an ensemble from which the random unitaries that make up the map can be picked. First, choose a probability distribution $\mu$ over the set of all unitaries $\cU(\cH_{\hat{A}})$ acting on $\cH_{\hat{A}}$ with the property that
\begin{equation}
 \sigma \rightarrow \int d\mu(U) U\sigma U^{\dagger} = \frac{I_{\hat{A}}}{d_{\hat{A}}}, \forall \sigma \in \cS(\hat{\cH}_{A}).
\end{equation}
There are many measures $\mu$ that satisfy this property. Mathematically, the most useful measure is the {\it Haar measure}\index{Haar measure}, which is the unique unitary invariant probability measure. This makes the above integral manifestly unitary invariant. The smallest support of such a measure is $(d_{\hat{A}})^{2}$. Another choice of measure can be realized by considering the discrete Weyl-Hiesenberg group. This is the group of unitaries generated by the cyclic shift operator and the diagonal matrix with the $n^{\rm th}$ roots of unity along the diagonal. Either choice of measure would work for our purpose.

Next, we draw unitaries $U_{1}, U_{2}, \ldots, U_{N}$, independently at random from $\mu$. Define the operators
\[ Y_{i} = (U_{i}\otimes I)\rho_{AB}^{(n)}(U_{i}^{\dagger}\otimes I).\]
Note that the expectation values of these operators satisfy
\[ \mathbb{E}[Y_{i}] = \frac{1}{d_{\hat{A}}} I_{\hat{A}}\otimes\rho_{\hat{B}}^{(n)}, \; \forall\; i.\]
The $Y_{i}$s are not upper bounded by $I$ but by some exponential factor as stated in the typicality principle. Therefore, we rescale the $Y_{i}$s as follows.
\[ X_{i} = 2^{n(S(\rho_{AB}) -\epsilon)}Y_{i} \in [0, I].\]
Then, the expectation values of the $X_{i}$s satisfy,
\begin{eqnarray}
\mathbb{E}[X_{i}] &=& 2^{n(S(\rho_{AB} - \epsilon))}\frac{1}{d_{\hat{A}}} I_{\hat{A}}\otimes\rho_{\hat{B}}^{(n)} \nonumber \\
 &\geq& 2^{-n(I(A:B) - 3\epsilon)} I_{\hat{A}\hat{B}} \equiv \nu I_{\hat{A}\hat{B}} ,
\end{eqnarray}
where the second inequality comes from the typicality principle.

Once we have this lower bound on the expectation values, we can use the matrix tail bounds in Eq.~\eqref{eq:matrix_ldev} as follows. Let $M$ denote the expectation value $\mathbb{E}[Y_{i}]$. Then,
\begin{eqnarray}
{\rm Pr}\left\{ \frac{1}{N}\sum_{i=1}^{N}Y_{i} \nleq (1+\epsilon)M \right\} &\leq&  2d_{\hat{A}}d_{\hat{B}} e^{-ND((1+ \epsilon)\nu\parallel\nu)} \nonumber \\
 &\leq& 2(d_{A})^{n}(d_{B})^{n}e^{-c\epsilon^{2}\nu N} . \label{eq:tail_bound}
\end{eqnarray}
We can similarly get a lower bound by using the other tail bound in Eq.~\eqref{eq:matrix_ldev} for $(1-\epsilon)M$. If the bound on the RHS is less than one, then, there exist unitaries $U_{1}, U_{2}, \ldots, U_{N}$ with the property
\begin{equation}
 (1 - \epsilon)M \leq \frac{1}{N}\sum_{i=1}^{N}(U_{i}\otimes I)\rho_{AB}^{(n)}(U_{i}^{\dagger}\otimes I) \leq (1 + \epsilon)M . \label{eq:ep_random2}
\end{equation}
Taking the trace-norm on both sides, we see that the set $\{U_{1}, U_{2}\ldots, U_{N}\}$ is indeed $\epsilon$-randomizing as defined in Eq.~\eqref{eq:ep_random}. Note that the statement we have proved in Eq.~\eqref{eq:ep_random2} is infact a stronger one!

Thus, in order to achieve the $\epsilon$-randomizing property, the bound in Eq.~\eqref{eq:tail_bound} shows that it suffices to have
\[ N \geq \frac{1}{c\epsilon} 2^{-nI(A:B) + 3\epsilon}\log[(d_{A})^{n}(d_{B})^{n}] \approx 2^{nI(A:B) + \delta},\]
in the limit of large $n$.
\end{proof}

Finally, we revisit the example of the maximally entangled two-qubit state $\rho_{\rm ent}$ discussed earlier. The mutual information of this state is $I(A:B)_{\rho} = 2$. Correspondingly, we provided a set of four unitaries, which when applied with equal probabilities could transform $\rho_{\rm ent}$ into a tensor product of maximally mixed states on each qubit. The upper and lower bounds proved here show that this set is in fact optimal for destroying correlations in the state $\rho_{\rm ent}$.

In the following section, we show that the results discussed here provide an interesting approach to an important quantum information processing task, namely, state merging.

\subsection{State Merging}

We first introduce the notion of {\it fidelity} between quantum states, which is based on the overlap between their purifications (see Defn.~\ref{def:purification} above). Recall that for any $\rho, \sigma \in \cS(\cH_{A})$, there exist purifications  $|\psi\rangle, |\phi\rangle \in \cH_{A} \otimes \cH_{A'}$ such that $\tr_{A'}[|\psi\rangle\langle\psi|] = \rho$ and $\tr_{A'}[|\phi\rangle\langle\phi|] = \sigma$. Furthermore, two different purifications of the same state $\rho$ are related by an isometry. That is, if there exists $|\tilde{\psi}\rangle \in \cH_{A}\otimes \cH_{\tilde{A}}$ satisfying $\tr_{\tilde{A}}[|\tilde{\psi}\rangle\langle\tilde{\psi}|] = \rho$, then there exists a partial isometry $U: A' \rightarrow \tilde{A}$ such that $|\tilde{\psi}\rangle = (I \otimes U) |\psi\rangle$. This observation gives rise to the following definition of fidelity between states.
\begin{defn}[Fidelity]\index{quantum! fidelity}
Given states $\rho, \sigma \in \cS(\cH_{A})$, the fidelity $F(\rho, \sigma)$ is defined as
\begin{equation}
 F (\rho, \sigma) := \max_{|\psi\rangle,|\phi\rangle} \tr[(|\psi\rangle\langle\psi|)(|\phi\rangle\langle\phi|)] = \max_{|\psi\rangle,|\phi\rangle} |\langle\psi|\phi\rangle|^{2},
\end{equation}
where the maximization is over all purifications $|\psi\rangle, |\phi\rangle$ of $\rho$  and $\sigma$ respectively.
\end{defn}
Note that the optimization over purifications is in fact an optimization over unitaries acting on the auxiliary space. A particular choice of purification of the state $|\psi_{0}\rangle$ defined as
\[ |\psi_{0}\rangle = (\sqrt{\rho}\otimes I)|\Phi\rangle, \]
where $|\Phi\rangle = \sum_{i}|ii\rangle$ is a purification of the identity. Any other purification $|\psi\rangle$ is then of the form
\[ |\psi\rangle = (\sqrt{\rho}\otimes U)|\Phi\rangle,\]
where $U$ is a partial isometry. Therefore the expression for fidelity simplifies as stated below .

\begin{exer}[Properties of Fidelity]
\begin{itemize}
\item[(i)]The fidelity between states $\rho, \sigma \in \cS(\cH_{A})$ is given by,
 \begin{equation}
  F(\rho,\sigma) = \parallel \sqrt{\rho}\sqrt{\sigma}\parallel_{1}^{2}.
 \end{equation}
\item[(ii)] $P(\rho, \sigma) := \sqrt{1 - F(\rho,\sigma)}$ is a metric on $\cS(\cH_{A})$.
\item[(iii)] $P(\rho, \sigma)$ is contractive under CPTP maps.
\item[(iv)] Relation between fidelity and trace-distance:
\[ \frac{1}{2}\parallel \rho -\sigma\parallel_{1} \leq P(\rho, \sigma). \]
Thus, convergence in the trace-norm metric is equivalent to convergence in the $P$-metric.
\end{itemize}
\end{exer}

The above discussion on fidelity thus implies the following. Given a state $|\psi\rangle \in \cH_{A}\otimes \cH_{A'}$ with $\rho = \tr_{A'}[|\psi\rangle\langle\psi|]$, and, $|\phi\rangle \in \cH_{A}\otimes \cH_{A'}$ with $\sigma = \tr_{A'}[|\phi\rangle\langle\phi|]$, there exists a partial isometry $U$ such that the state $|\psi'\rangle = (I\otimes U)|\psi\rangle$ is as close to $|\phi\rangle$ as the fidelity between $\rho,\sigma$. In other words, $|\langle\phi|\psi\rangle|^{2} = F(\rho, \sigma)$.

We now define an important information theoretic primitive in the quantum setting.
\begin{defn}[State Merging]\index{quantum! state merging}
Consider $\rho_{AB} = \tr_{R}[|\psi\rangle_{RAB}\langle\psi|]$, where $|\psi\rangle_{RAB}$ is a joint state of $\cH_{R}\otimes\cH_{A}\otimes\cH_{B}$. System $R$ is simply a {\it reference} system which plays no active role in the protocol; all operations are performed by the two parties $A$ and $B$ alone. The goal of state merging is to transform the state $|\psi\rangle_{RAB}$ into $|\tilde{\psi}_{R\hat{A}\hat{B}}\rangle$ which is close in fidelity to the original state, where the systems $\hat{A}$ and $\hat{B}$ now correspond to party $B$. The protocol could use entanglement between the parties $A$ and $B$ (say, a state of Schmidt rank $2^{K}$), local operations and classical communication which are assumed to be a {\it free} resources.
\end{defn}

The goal of state merging is to transform a joint state of $A$ and $B$ (along with the reference) to a state of $B$ alone (and the reference), where system $B$ is now composed of systems $\hat{A}$ and $\hat{B}$. Note that {\bf quantum teleportation} is in fact a special case of state merging, where party $A$ conveys an unknown quantum state to $B$ using one maximally entangled state ($K$ = 1) and $2$ bits of classical communication.

Assuming that $A$ and $B$ share some entanglement initially, the actual initial state for the protocol is of the form $|\psi\rangle_{RAB} \otimes |\phi_{K}\rangle_{A_{0}B_{0}}$. To incorporate the local operations and classical communication (which we assume to be one-way, from $A$ to $B$), we can break down the protocol into three phases. First, $A$ performs some local operation which is modeled as a family of CP maps $\{\cT_{\alpha}\}$ on system $A$ alone:
\[\{\cT_{\alpha}\}_{\alpha}, \quad \cT_{\alpha} : \cB(\cH_{A}\otimes \cH_{A_{0}})\rightarrow \cB(\cH_{A}\otimes\cH_{A_{0}}), \; \sum_{\alpha}\cT_{\alpha} = I . \]
This is followed by classical communication from $A$ to $B$ and finally local operations on system $B$. The operations on $B$'s system are CPTP maps $\{\cD_{\alpha}\}$ that {\it decode} the state based on the classical information from $A$. Therefore,
\[ \cD_{\alpha}: \cB(\cH_{B}\otimes \cH_{B_{0}}) \rightarrow \cB(\cH_{\hat{A}}\otimes\cH_{\hat{B}}\otimes\cH_{B_{0}}) .\]
Thus the final state at the end of this protocol is given by,
\[ \sum_{\alpha}(\Id_{R}\otimes\cT_{\alpha}\otimes\cD_{\alpha})|\psi\rangle_{RAB}\otimes|\phi_{K}\rangle_{A_{0}B_{0}} \approx |\psi\rangle_{R\hat{A}\hat{B}}\otimes|\phi_{0}\rangle_{B_{0}}, \]
which should be close to the original state in fidelity.

Having formalized the model for state merging, we are interested in the following question: what is the smallest $K (\rho_{AB}, \epsilon)$ such that there exists a state merging protocol that achieves a fidelity greater than $1-\epsilon$ for the state $\rho_{AB}$? As before, we focus to the asymptotic, multiple copy setting. Taking $n$ copies of $\rho_{AB}$, we would like to study the quantity $\frac{K(\rho_{AB}^{\otimes n}, \epsilon)}{n}$ in the asymptotic limit of $n \rightarrow \infty$ and $\epsilon \rightarrow 0$. This is simply another way of asking, how much quantum communication is required to effect this transformation. A special case of this problem is the case where $A$ and $B$ share no entanglement; only $A$ and $R$ are entangled. This is the same as the question of quantum data compression originally studied by Schumacher~\cite{Sch}. The asymptotic bound in the special case turns out to be the entropy of system $A$.

Here, we prove the following asymptotic bound for state merging, originally shown in~\cite{HOW1, HOW2}.
\begin{thm}\index{Theorem! Quantum State Merging}
Given $n$ copies of a state $\rho_{AB}$, any state merging protocol that achieves a fidelity larger than $(1-\epsilon)$ must use an entangled state of Schmit rank at least $K(n,\epsilon)$, with $K(n,\epsilon)$ bounded by,
\begin{equation}
 \liminf_{\substack{n\rightarrow \infty \\ \epsilon\rightarrow 0}}\; \frac{1}{n}K(\rho_{AB}^{\otimes n}, \epsilon) \geq S(A|B)_{\rho}, \label{eq:s_merge1}
\end{equation}
where $S(A|B)_{\rho} = S(\rho_{AB}) - S(\rho_{B})$ is the conditional entropy of system $A$ given $B$. Conversely, there always exists a state merging protocol that that achieves fidelity greater $1 - \epsilon$ for $n$ copies of $\rho_{AB}$, with the shared entanglement between $A$ and $B$ bounded by
\begin{equation}
 \limsup_{\substack{n\rightarrow \infty \\ \epsilon\rightarrow 0}}\; \frac{1}{n}K(\rho_{AB}^{\otimes n}, \epsilon) \leq S(A|B)_{\rho}, \label{eq:s_merge2}
\end{equation}
\end{thm}
Note that the conditional entropy satisfies $S(A|B) \leq S(\rho_{A})$ with equality holding iff $\rho_{AB} = \rho_{A}\otimes \rho_{B}$ is a product state. Thus, the state merging protocol in general achieves something beyond teleportation or data compression. It is also worth noting here that state merging is in some sense a quantum version of the classical {\it Slepian-Wolf protocol}. Loosely speaking, both protocols consider a setting where the second party $B$ has some partial information about $A$, and then ask what is the minimum amount of information that has to be transmitted by $A$ so that $B$ has complete knowledge of $A$. In the classical Slepian-Wolf also, the amount of information that has to be transmitted from $A$ to $B$ to achieve complete transfer of information is indeed the conditional entropy $H(A|B)$.

Here we only prove the converse statement by explicitly constructing the state merging protocol using the $\epsilon$-randomizing maps defined in the last section.

\begin{proof}
Assume there exists an $\epsilon$-randomizing family of unitaries $\{U_{i}\}_{i=1}^{N}$ on $(\cH_{A})^{\otimes n}$, for the state $\rho_{AR} = \tr_{B}[|\psi\rangle_{RAB}\langle\psi|]$. Then,
\[ \frac{1}{N} \sum_{i=1}^{N}(I \otimes U_{i})\rho_{RA}^{\otimes n}(I\otimes U_{i}^{\dagger}) \approx \rho_{R}^{\otimes n}\otimes \rho_{A}^{\otimes n}, \]
where, $N \sim 2^{nI(A:R)}$, and, the closeness to the product state is assumed to be in the fidelity sense. This holds because of the relation between the trace-distance and the fidelity stated earlier.

The state merging protocol can now be constructed as follows. Party $A$ first prepares a uniform superposition of basis states on an auxiliary space $\cH_{C}$:
\[ |\Xi\rangle = \frac{1}{\sqrt{N}}\sum_{i=1}^{N}|i\rangle \; \in \cH_{C} .\]
Party $A$ uses this state as the source of randomness and picks the unitaries from the $\epsilon$-randomizing set as per this superposition. That is, $A$ applies the global unitary,
\[ U = \sum_{i=1}^{N} (U_{i})_{A^{n}}\otimes |i\rangle\langle i|_{C} \; \in \cU(\cH_{A}^{\otimes n}\otimes \cH_{C}) .\]
 Assuming the systems start with $n$ copies of the state $|\psi\rangle_{RAB}$, the global state after the action of this unitary is given by
\[ |\phi^{(n)}\rangle_{R^{n}A^{n}CB^{n}} = (I_{R^{n}B^{n}} \otimes U)|\psi\rangle_{RAB}^{\otimes n} . \]
Tracing out over system $C$ amounts to an averaging over the unitaries $\{U_{i}\}$. Therefore, the reduced state $\phi_{R^{n}A^{n}}$ is $\epsilon$-close to a product state: $\phi_{R^{n}A^{n}} \approx \phi_{R^{n}}\otimes \phi_{A^{n}}$.

The next step is for $A$ to teleport system $C$ to party $B$. This requires $A$ and $B$ to share an entangled state of Schmidt rank $K$, given by,
\[ K = |C| = 2^{nI(A:R)}.\]
Thus, the global state is simply a purification of $\phi_{R^{n}A^{n}}$, and party $B$ now holds the remainder of the state, namely the reduced state on systems $C$ and $B^{n}$.

Consider the product state $\phi_{R^{n}}\otimes \phi_{A^{n}}$, which we know is $\epsilon$-close to $\rho_{R^{n}A^{n}}$ in fidelity. This product state can be purified by taking a tensor product of purifications of $\phi_{R^{n}}$ and $\phi_{A^{n}}$. Since the reference system is unaffected by the protocol, the reduced state $\phi_{R^{n}}$ is still the same as the initial reduced state on the reference system, and is therefore purified by the original state $|\psi\rangle_{R\hat{A}\hat{B}}^{\otimes n}$. Further, from the typicality principle, we can assume without loss of generality that  $\phi_{A^{n}} = \frac{I_{A_{0}^{n}}}{d_{A_{0}^{n}}}$, which is purified by a maximally entangled state of rank $L = d_{A_{0}^{n}} = 2^{nS(\rho_{A})}$. Thus, we have,
\[ \phi_{R^{n}}\otimes\phi_{A^{n}} = \tr_{\hat{A}^{n}\hat{B}^{n}}[(|\psi\rangle_{R\hat{A}\hat{B}}\langle\psi|)^{\otimes n}] \otimes \tr_{B_{0}}[|\Phi_{L}\rangle_{A_{0}B_{0}}\langle\Phi_{L}|] . \]

Since the reduced states $\phi_{R^{n}A^{n}}$ and $\phi_{R^{n}}\otimes\phi_{A^{n}}$ are $\epsilon$-close in fidelity, there exists an isometry $V: B^{n}C \rightarrow \hat{A}^{n}\hat{B}^{n}B_{0}$ such that the corresponding purifications are $\epsilon$-close in fidelity. Therefore, the final step of the protocol is that $B$ applies the isometry $V$, to obtain the final state
\[ (I\otimes V) |\phi\rangle^{\otimes n}_{R^{n}A^{n}CB^{n}} \approx |\psi\rangle^{\otimes n}_{R\hat{A}\hat{B}}\otimes |\Phi_{L}\rangle_{A_{0}B_{0}}, \]
whose fidelity with the desired final state is greater than $1-\epsilon$.
\end{proof}

Thus, the protocol achieves state-merging using entangled states of Schmidt rank $I(A:R)$, and leaves behind an entangled state between $A_{0}$ and $B_{0}$ of rank $S(\rho_{A})$. The net entanglement used up in the protocol is thus
\begin{eqnarray}
 I(A:R)_{\rho} - S(\rho_{A}) &=& S(\rho_{R}) - S(\rho_{AR}) \nonumber \\
 &=& S(\rho_{AB}) - S(\rho_{B}) \equiv S(A|B)_{\rho},
 \end{eqnarray}
where the second equality follows from the total state $|\psi\rangle_{RAB}$ being a pure state. In situations where $I(A:R)_{\rho} > S(\rho_{A})$, the conditional entropy $S(A|B)_{\rho} >0$ and there is a net loss of entanglement. However, when $I(A:R)_{\rho} < S(\rho_{A})$, the conditional entropy\index{quantum! conditional entropy} $S(A|B)$ is negative and there is a net gain in entanglement! Thus, if we start with a certain amount of entanglement initially, the protocol eventually creates some entanglement via local operations and classical communication only~\cite{HOW2}.

\cleardoublepage
\phantomsection
\addcontentsline{toc}{chapter}{Index}
\printindex

\bibliographystyle{alpha}

\end{document}